\documentclass[11pt]{article}

\usepackage[utf8]{inputenc}
\usepackage[english]{babel}
\usepackage{amsmath}
\usepackage{amsfonts}
\usepackage{graphicx}
\usepackage{subfig}
\usepackage[utf8]{inputenc}
\usepackage[english]{babel}
\usepackage{amsmath}
\usepackage{amsfonts}
\usepackage{amssymb}
\usepackage{amsthm}
\usepackage{mathtools}
\usepackage{nicefrac}
\usepackage[colorlinks = true,
linkcolor = blue,
urlcolor  = blue,
citecolor = blue,
anchorcolor = blue]{hyperref}
\usepackage{cleveref}
\usepackage{bm}
\usepackage{float}
\usepackage{bbm}
\usepackage{subfig}
\usepackage{natbib}
\usepackage{wrapfig}
\usepackage[]{appendix}
\usepackage{caption}
\usepackage{setspace} 
\usepackage{animate}
\usepackage{algorithm}
\usepackage{algpseudocode}
\usepackage{pdfpages}
\usepackage{macros}
\usepackage{comment}
\usepackage{dsfont}
\usepackage[paper=a4paper,left=30mm,right=30mm,top=25.4mm,bottom=25.4mm]{geometry}
\usepackage{cleveref}
\newcommand{\easysum}[2]{\ensuremath{\underset{#1}{\overset{#2}{\sum}}}}
\newcommand{\es}[2]{\ensuremath{\underset{#1}{\overset{#2}{\sum}}}}
\newcount\colveccount
\newcommand*\colvec[1]{
        \global\colveccount#1
        \begin{pmatrix}
        \colvecnext
}
\def\colvecnext#1{
        #1
        \global\advance\colveccount-1
        \ifnum\colveccount>0
                \\
                \expandafter\colvecnext
        \else
                \end{pmatrix}
        \fi
}

\newtheorem{theorem}{Theorem}[section]
\newtheorem{lemma}[theorem]{Lemma}

\newtheorem{corollary}[theorem]{Corollary}

\newtheorem{example}[theorem]{Example}
\newtheorem{remark}[theorem]{Remark}
\newcommand{\R}{\mathbb R}
\newcommand{\msrX}{\mathcal{M}_+(\mathcal{X})}
\newcommand{\msrY}{\mathcal{M}_+(\mathcal{Y})}
\newcommand{\X}{\mathcal{X}}
\newcommand{\Y}{\mathcal{Y}}
\newcommand{\XC}{\mathcal{X}}

\newcommand{\dum}{\mathfrak{d}}

\newcommand{\KR}{\mathrm{KR}}
\newcommand{\dist}{\mathrm{dist}}
\renewcommand{\dist}{\Delta}

\DeclareMathOperator*{\ty}{\tilde{\mathcal{Y}}}

\crefformat{footnote}{#2\footnotemark[#1]#3}
 
\makeatletter
  \newcommand\lalign{\@fleqntrue}
\makeatother

\graphicspath{{./graphics/}}
\newcommand{\footremember}[2]{%
	\footnote{#2}
	\newcounter{#1}
	\setcounter{#1}{\value{footnote}}%
}
\newcommand{\footrecall}[1]{%
	\footnotemark[\value{#1}]%
}

\begin{document}
\title{Unbalanced Kantorovich-Rubinstein distance, plan, and barycenter on finite spaces: A statistical~perspective}

\author{Shayan Hundrieser \footremember{equ}{These authors contributed equally}\footremember{ims}{\scriptsize Institute for Mathematical
Stochastics, University of G\"ottingen,
Goldschmidtstra{\ss}e 7, 37077 G\"ottingen}
    \and Florian Heinemann\footrecall{equ} \footrecall{ims}
 		\and Marcel Klatt\footrecall{ims}
    \and Marina Struleva\footrecall{ims}
 		\and Axel Munk\footrecall{ims} \footnote{\scriptsize Max Planck Institute for Multidisciplinary Science, Am Fa{\ss}berg 11, 37077 G\"ottingen} \footnote{ \scriptsize University Medical Center Göttingen, Cluster of Excellence 2067 Multiscale Bioimaging - From molecular machines to networks of excitable cells}}
\maketitle

\begin{abstract}%
  We analyze statistical properties of plug-in estimators for unbalanced optimal transport quantities between finitely supported measures in different prototypical sampling models. Specifically, our main results provide non-asymptotic bounds on the expected error of empirical Kantorovich-Rubinstein (KR) distance, plans, and barycenters for mass penalty parameter $C>0$. The impact of the mass penalty parameter $C$ is studied in detail. Based on this analysis, we mathematically justify randomized computational schemes for KR quantities which can be used for fast approximate computations in combination with any exact solver. Using synthetic and real datasets, we empirically analyze the behavior of the expected errors in simulation studies and illustrate the validity of our theoretical bounds. 
\end{abstract}

\section{Introduction}
Optimal transport (OT) \citep[for a detailed mathematical discussion see e.g.][]{villani2008optimal,santambrogio2015optimal} has been a focus of attention in various research fields for a long time. More recently, its powerful geometric features promoted by improved computational tools \citep[see e.g.][]{chizat2018scaling,peyre2019computational,guo2020fast,lin2020fixed} have turned OT into a promising new tool for modern data analysis with applications in machine learning \citep{frogner2015learning,arjovsky2017wasserstein,schmitz2018wasserstein,yang2018low,vacher2021dimension}, computer vision \citep{baumgartner2018visual,kolkin2019style}, computational biology \citep{evans2012phylogenetic,gellert2019substrate,klatt2020empirical,tameling2021colocalization,bunne2023learning}, image processing \citep{pitie2007automated,bonneel2016wasserstein,tartavel2016wasserstein} and statistical inference \citep{sommerfeld2018inference,lee2018minimax,mena2019statistical,panaretos2019statistical,hallin2021distribution, hallin2022measure}, among others.

\begin{figure}[btp]
  \centering
  \subfloat[][]{\includegraphics[width=0.4\linewidth]{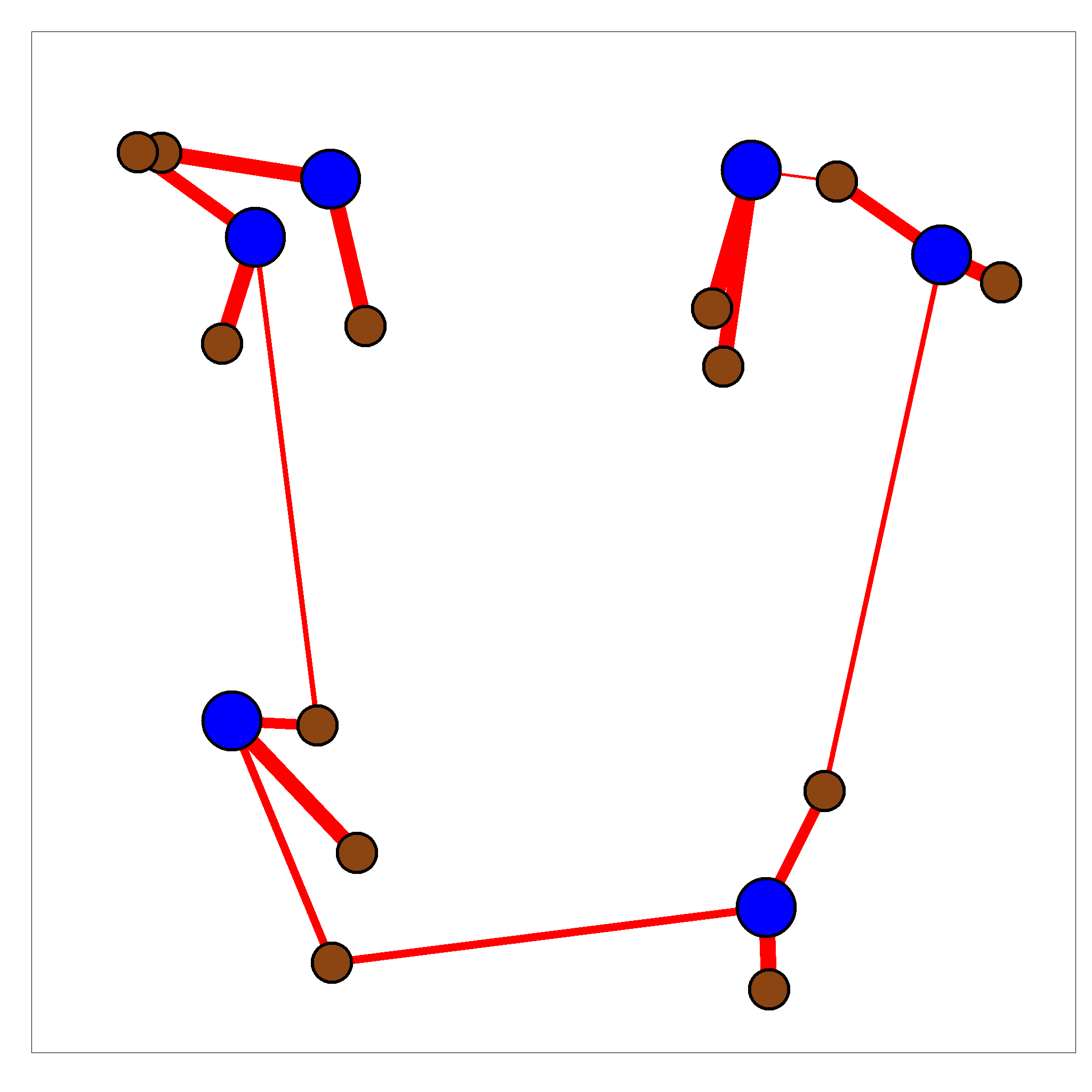}}%
  \qquad
  \subfloat[][]{\includegraphics[width=0.4\linewidth]{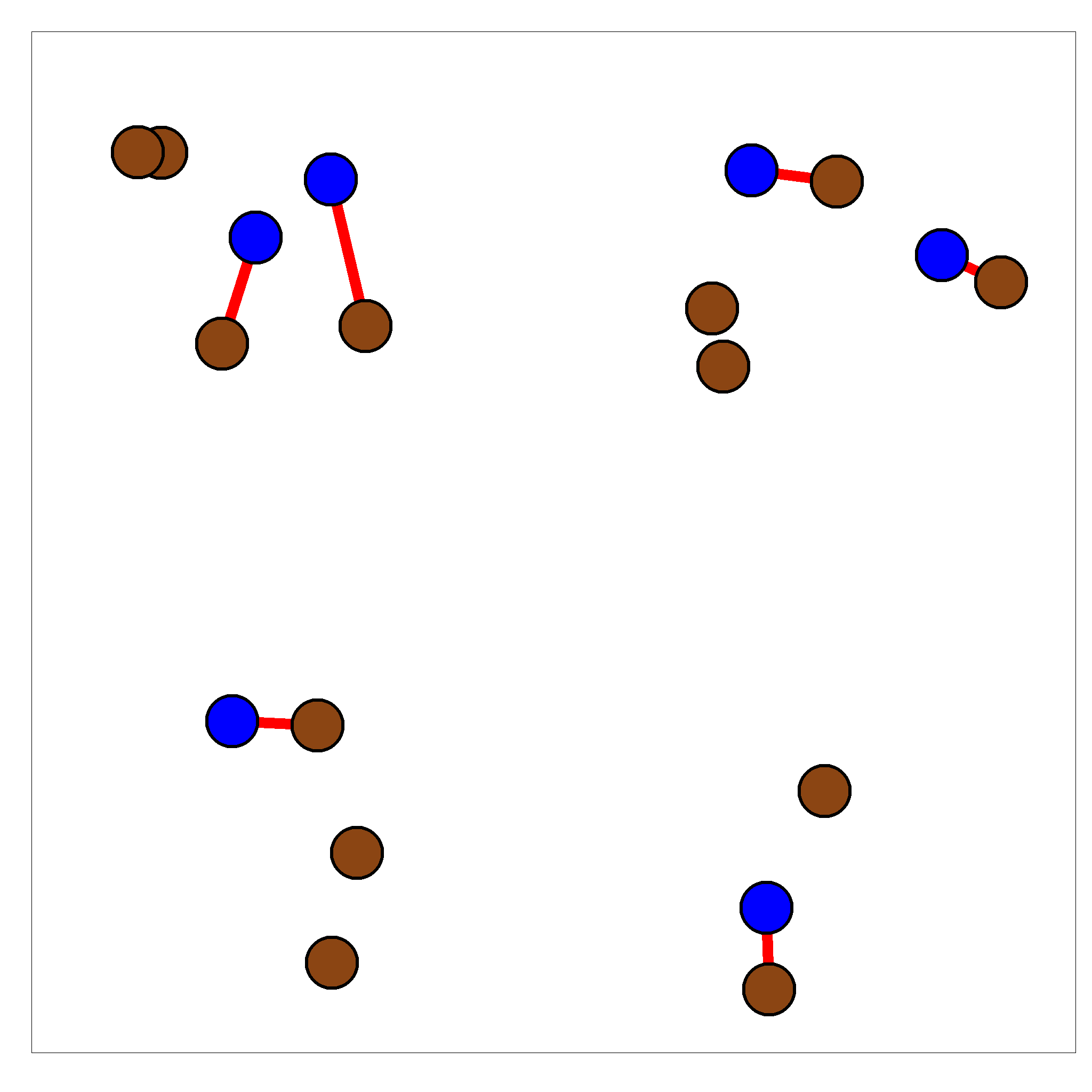}}%
  \caption{Transport between two measures (blue and brown) with their support points located in $[0,1]^2.$ The respective transport plans between them are displayed by red lines where the thickness of a line is proportional to the transported mass. \textbf{(a)} The measures have been normalized to probability measures (the blue points have mass $1/6$, the brown ones have mass $1/13$). \textbf{(b)} All points in the unnormalized measures have mass~1, therefore the UOT plan for the $(2,2)$-KRD yields a one-to-one matching between a sub-collection of points.}
  \label{fig:uotvsot}
\end{figure}

The wide range of applications also surfaced limitations of classical OT. In particular, the assumption of equal total mass intensity of the measures is often inappropriate, e.g., in the context of image processing and retrieval \citep{rubner1998metric, pele2008linear, pele2009fast, rabin2015convex}, for takings means of neuroimaging data \citep{gramfort2015fast}, when comparing radiation patterns of collider events \citep{komiske2019metric,manole2022background}, in recovery of developmental trajectories of embryonic stem cells \citep{schiebinger2019optimal,ventre2023trajectory}, and in multi-color super-resolution colocalization analysis \citep{Naas2024Multimatch}.
If standard OT methodology had been used in these contexts, all measures would have needed to be normalized in order to overcome the issue of different total mass. However, this preprocessing step would structurally change the problem and directly impact the data analysis, in particular the underlying transport plan. 
For example, when matching point clouds of different sizes the resulting plan distributes mass among several points, whereas often it is desired to match points one-to-one which is favorable in many applications (see \Cref{fig:uotvsot}). Attempts to circumvent this issue have led to a range of \emph{unbalanced optimal transport} (UOT) proposals \citep[see above applications and also][]{figalli2010optimal,liero2018optimal,chizat2018unbalanced,balaji2020robust, chapel2020, le2021,le2022,mukherjee2021outlier,heinemann2022kantorovich}. These formulations extend OT concepts to general positive measures by either fixing the total amount of mass to be transported in advance or by penalizing the hard marginal constraints inherent in OT. These approaches also give rise to  barycenters, generalizing the popular notion of \emph{OT barycenters} \citep{agueh2011barycenters} to measures of unequal mass \citep{gramfort2015fast,chizat2018scaling,friesecke2021barycenters,heinemann2022kantorovich}. 

The UOT formulation considered here is the $(p,C)$-Kantorovich-Rubinstein distance (KRD) (see \Cref{sec:krd}) whose structural properties and related $(p,C)$-barycenter have recently been studied in detail \citep{heinemann2022kantorovich}. The $(1,1)$-KRD essentially corresponds to the notion of extended Kantorovich norms \citep{kantorovich1958space,hanin1992kantorovich} considered in the context of Lipschitz spaces and signed measures. In this context, \cite{guittet2002extended} first introduced a discrete formulation of the problem, where he established a Linear Program (LP) formulation which carries over to the general $(p,C)$-KRD. 
For illustration, a comparison between the $(p,C)$-barycenter and the $p$-Wasserstein barycenter in a simple example is displayed in \Cref{fig:krvswsbary}. From a data analysis point of view we find it particularly appealing that for the $(p,C)$-KRD there is a clear geometrical connection between its penalty $C$ and the structural properties of the corresponding UOT plans and $(p,C)$-barycenter. More precisely, it is shown in \cite[Lemma 2.1]{heinemann2022kantorovich} that $C$ controls the largest scale at which mass transport is possible in an optimal plan. This interpretation of the $(p,C)$-KRD allows designing it to respect different structural properties of the data and thus makes it a prime candidate for statistical tasks in data analysis. In particular, this emphasizes the robustness of the $(p,C)$-KRD to spatial outliers, i.e., data points which are located far in the tails.  Furthermore, each support point of any $(p,C)$-barycenter is contained in a finite set characterized by the value of $C$. This allows to adapt OT solvers for the unbalanced problems.

Due to the unbalanced nature of the problem, however, the task of sampling from the underlying measures requires alternative sampling schemes and different statistical modelling. While for OT between probability measures there is a canonical sampling model by i.i.d. replications from the measures, this fails for UOT, since the considered measures are not necessarily probability measures. In this work, we address this issue and suggest a framework underpinning UOT based statistical data analysis. To this end, we analyze the $(p,C)$-KRD, the corresponding UOT plan and its barycenter in three specific statistical models motivated by applications in randomized algorithms and microscopy tasks. Notably, these models also provide a framework which potentially allows treating the alternative UOT models mentioned above. %
 Throughout, we focus on finite discrete domains, as it allows a complete analysis while imposing minimal assumptions on the ground space.

\subsection{Kantorovich-Rubinstein Quantities: Distance, Plan, and Barycenter}\label{sec:krd}

Let $(\X,d)$ be a finite metric space with finite cardinality $M\coloneqq |\XC|$ and denote by
$\msrX \coloneqq \left\lbrace \mu \in \mathbb{R}^{M} \,\mid\, \mu(x)\geq 0 \ \forall x \in \X \right\rbrace$
the set of non-negative measures\footnote{A non-negative measure on a finite space $\X$ is uniquely characterized by the values it assigns to each singleton $\{ x\}$. To ease notation we write $\mu(x)$ instead of $\mu(\{x\})$. The corresponding $\sigma$-field is always to be understood as the power set of $\mathcal{X}$.} on $\X$. The total mass of a measure $\mu\in\msrX$ is defined as $\mathbb{M}(\mu)\coloneqq \sum_{x\in\X}\mu(x)$ and the subset $\mathcal{P}(\X)\subset \msrX$ of measures with total mass one is the set of probability measures. If $\pi\in \mathcal{M}_+(\X\times\X)$ is a measure on the product space $\X \times \X$ its marginals are defined as $\pi(x,\X)\coloneqq \sum_{x^\prime} \pi(x,x^\prime)$ and $\pi(\X,x^\prime)\coloneqq \sum_{x^\prime \in \X}\pi(x,x^\prime)$, respectively. For two measures $\mu,\nu\in\msrX$ define the set of \emph{non-negative sub-couplings}~as
\begin{align}\label{eq:subcouplings}
\begin{split}
    \Pi_{\leq}(\mu,\nu)\coloneqq \lbrace \pi\in \mathcal{M}_+(\X\times\X) \mid \,&\pi(x,\X) \leq  \mu(x),\, \pi(\X,x^\prime)  \leq  \nu(x^\prime)\, \forall \, x,x^\prime\in\mathcal{X}  \rbrace.
\end{split}
\end{align}
Following \cite{heinemann2022kantorovich}, for $p\geq 1$ and a parameter $C>0$, the $(p,C)-$\emph{Kanto\-rovich-Rubinstein distance} (KRD) between two measures $\mu,\nu\in\msrX$ is defined as
\begin{align}\label{eq:krdistance}
\begin{split}
\KR_{p,C}(\mu,\nu)\coloneqq \Bigg( \min_{\pi\in \Pi_{\leq}(\mu,\nu)} &\sum\limits_{x,x^\prime \in \X}d^p(x,x^\prime)\pi(x,x^\prime)+C^p\left(\frac{\mathbb{M}(\mu)+\mathbb{M}(\nu)}{2}-\mathbb{M}(\pi)\right)\Bigg)^{\frac{1}{p}}.
\end{split}
\end{align}
For $p\geq 1$, it defines a distance on $\msrX$ robust to spatial outliers\footnote{As a toy example, let $\mu_{x, \alpha} \coloneqq \alpha\delta_{0} + (1-\alpha)\delta_{x}$ and $\nu_\alpha \coloneqq \alpha\delta_{0} + (1-\alpha)\delta_{1}$ for $x\geq 0$ and $\alpha \in [0,1]$. Then, $\textup{KRD}_{1,C}(\mu_{x,\alpha}, \nu_\alpha) = (1-\alpha)\min(|x-1|,C)$ and $W_1(\mu_{x,\alpha}, \nu_\alpha)= (1-\alpha)|x-1|,$
asserting for $x \to \infty$ and $\alpha \nearrow 1$ that 	$\textup{KRD}_{1,C}(\mu_{x,\alpha}, \nu_\alpha) \to 0$; 
but if $\alpha = 1- \min(1,x^{-1+\varepsilon})$ for $\varepsilon>0$, then $W_1(\mu_{x,\alpha}, \nu_\alpha) \to \infty$.} and naturally extends the well-known \emph{$p$-th order OT distance} defined for measures of equal total mass \citep{heinemann2022kantorovich}. Notably, by compactness of $\Pi_{\leq}(\mu, \nu)$ and continuity of the objective, an optimizer $\pi\in \Pi_{\leq}(\mu,\nu)$ of~\eqref{eq:krdistance} always exists and is called \emph{unbalanced optimal transport  plan} (UOT plan).

\label{sec:pcbary}
\begin{figure}
  \centering
  \subfloat[][]{\includegraphics[width=0.32\linewidth]{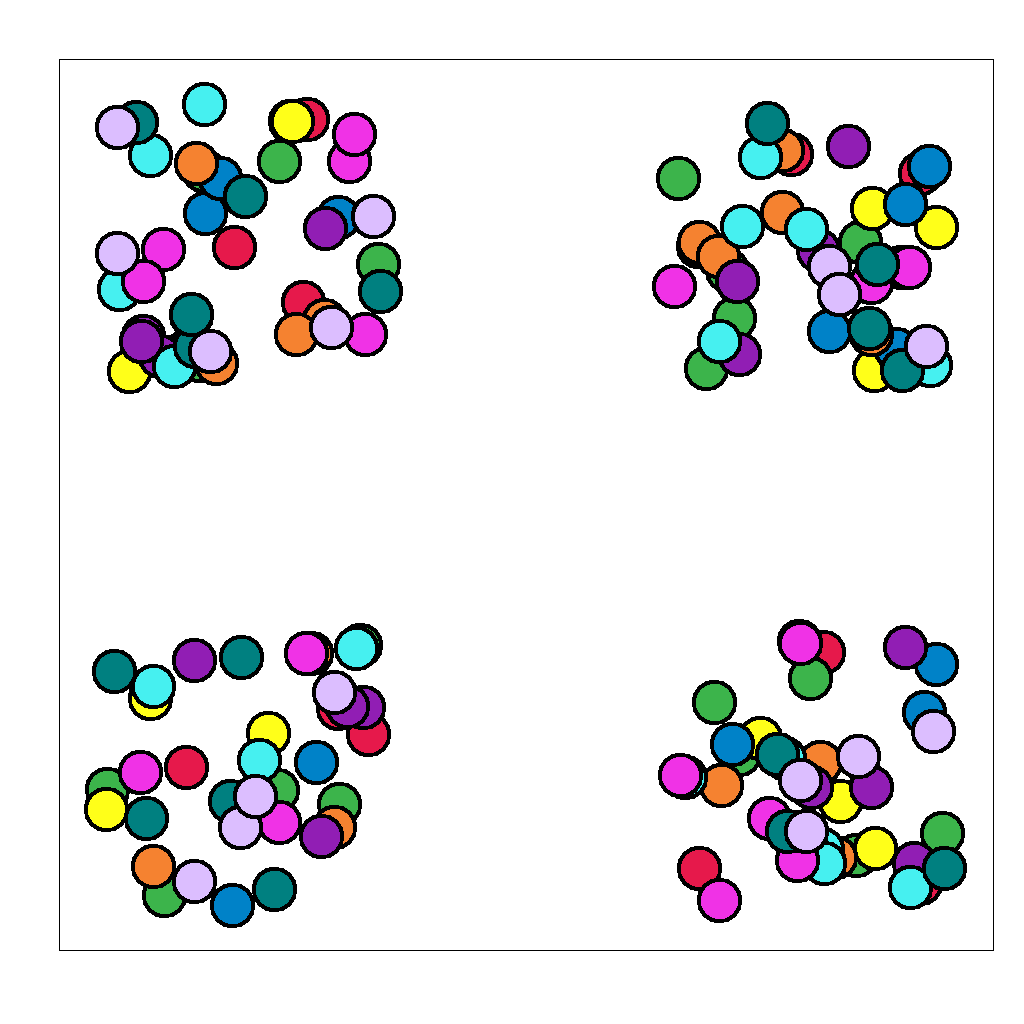}}%
  \
  \subfloat[][]{\includegraphics[width=0.32\linewidth]{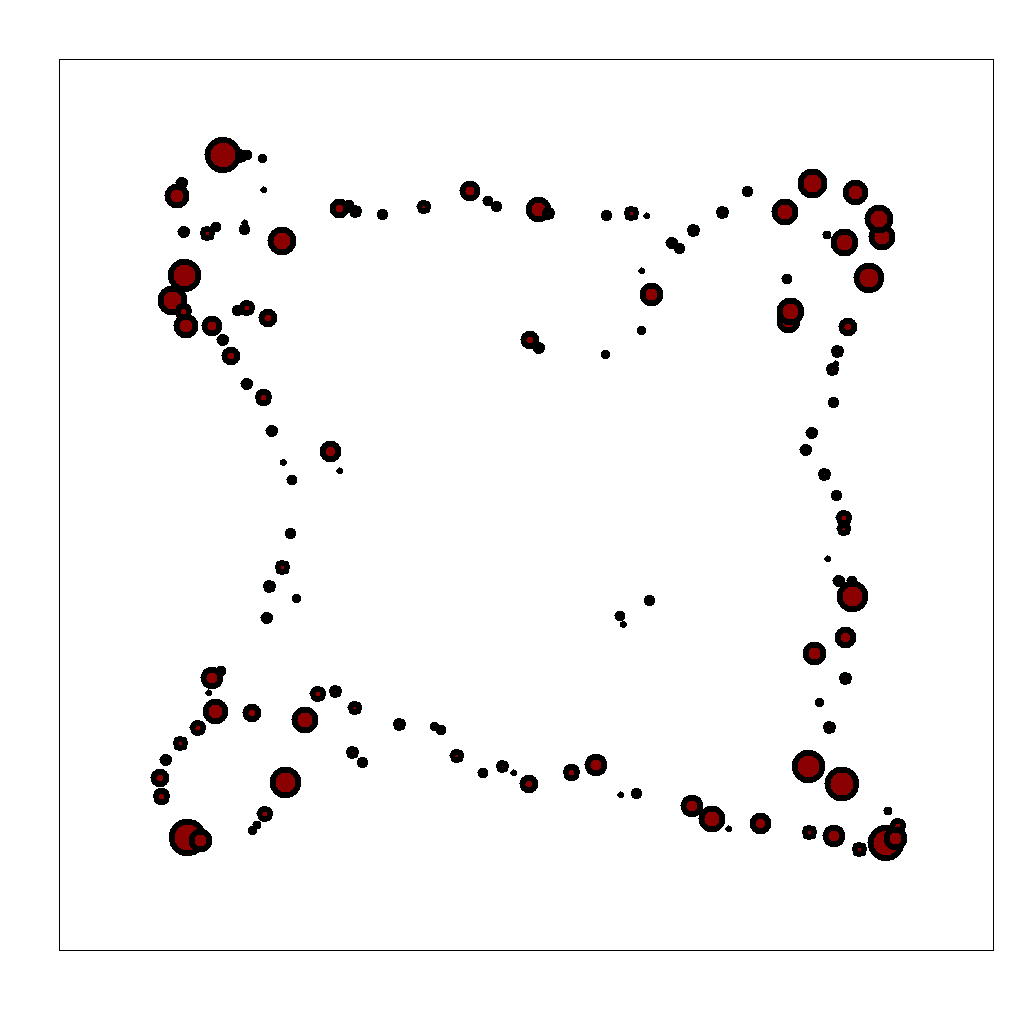}}%
    \
  \subfloat[][]{\includegraphics[width=0.32\linewidth]{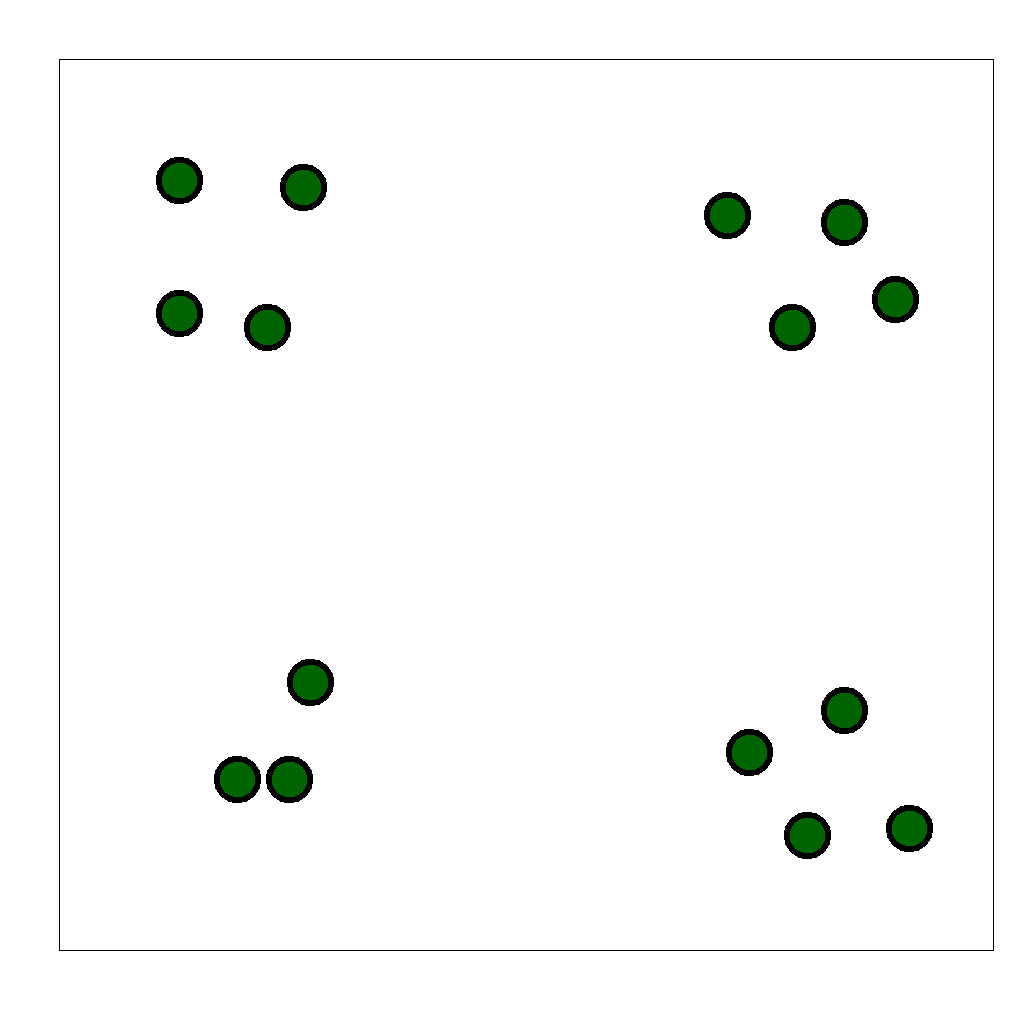}}%
  \caption{\textbf{(a)} $J=10$ measures with mass $1$ at each support point and different total mass intensities (each color corresponds to a different measure) superimposed on top of each other. \textbf{(b)} The OT barycenter (for squared Euclidean cost) of the normalized measures. \textbf{(c)} The $(2,0.3)$-barycenter of the unnormalized measures (see \eqref{eq:KRbarycenter} for a rigorous definition).}%
  \label{fig:krvswsbary}
\end{figure}
The $(p,C)$-KRD also allows defining a notion of a barycenter for a collection of measures with (potentially) different total masses. Assume $(\X,d)$ to be embedded in some connected ambient space\footnote{We assume the metric on $\X\subset \Y$ to be the metric of $\Y$ restricted to $\X$.} $(\Y,d)$, e.g., a Euclidean space, and define the \emph{$(p,C)$-Fr\'echet functional}
\begin{equation}\label{eq:KRbarycenter}
        F_{p,C}(\mu)=\frac{1}{J}\easysum{i=1}{J} \KR_{p,C}^p(\mu^i,\mu).
\end{equation}
Any minimizer of this functional in $\msrY$ is said to be a \emph{$(p,C)$-Kantorovich-Rubinstein barycenter} of $\mu_1,\dots,\mu_J$ or \emph{$(p,C)$-barycenter} for short\footnote{For the sake of readability, the weights in this definition are fixed to $1/J$. Adaptation of all results to arbitrary positive weights $\lambda_1,\dots ,\lambda_J$ summing to one is straightforward.}. The objective functional $F_{p,C}$ is referred to as (unbalanced) \emph{$(p,C)$-Fr\'echet functional} and the so-called \emph{Borel barycenter application} is defined as $T^{L,p}(x_1,\dots ,x_L) \in  \argmin_{y \in \Y} \sum_{i=1}^L d^p(x_i,y).$
Define the \emph{full centroid set}\footnote{There are scenarios where multiple sets fulfil the definition of the centroid set, since there might be multiple points that minimize the barycentric application. In this case, a fixed representative is chosen and there still exists a choice of centroid set which contains the support of the $(p,C)$-barycenter.} of the measures
\begin{align}\label{eq:KRfullcentroid}
\begin{split}
\mathcal{C}_{\KR}(J,p)=\Big\{ y \in \mathcal{Y} \ \vert \ \exists L\geq \lceil J/2 \rceil, \ \exists (i_1,\dots,i_L)\subset \{1,\dots,J\},\\ x_1,\dots,x_L: \ x_l\in \mathrm{supp}(\mu_{i_l}) \\ \forall l=1,\dots,L: \ 
y=T^{L,p}(x_1,\dots,x_L)\Big\},
 \end{split}
\end{align}
and based on it the \emph{restricted centroid set} 
\begin{align}
\begin{split}
\label{eq:KRCentroid}
\mathcal{C}_{\KR}(J,p,C)=\Big\{ &y=T^{L,p}(x_1,\dots,x_L)\in\mathcal{C}_{\KR}(J,p)\, \mid \forall 1\leq l \leq L: \\ &d^p(x_l,y)\leq C^p; \ \es{i=1}{L}d^p(x_l,y) \leq \frac{C^p (2L-J)}{2} \Big\}. 
\end{split}
\end{align}
According to \cite[Theorem 2.5]{heinemann2022kantorovich}, any $(p,C)$-barycenter is finitely supported, and its support is included in the restricted centroid set $\mathcal{C}_{KR}(J,p,C)$. This is critical, as it allows us to restrict theoretical analysis as well as computational methods to the scenario all measures involved have finite support. Moreover, it permits an additional degree of interpretability of the mass penalization parameter $C$ which might be of interest for specific applications.

\subsection{On Plug-in Estimators for Optimal Transport}\label{sec:sampling}
In practice, one often does not have access to the population measures $\mu,\nu,\mu^1,\dots ,\mu^J$, respectively, and they need to be estimated from data. For probability measures the most common statistical model assumes access to i.i.d.\ data $X_1,\dots,X_N\sim \mu$ (and similar for $\nu,\mu^1,\dots,\mu^J$). A commonly used estimator is the empirical measure $\hat{\mu}_N=(1/N)\sum_{k=1}^{N}\delta_{X_k}$, where $\delta_X$ denotes the (random) point measure at location $X$. In light of the ongoing field of statistical OT, there are various topics of interest which have been extensively analyzed. 

First, the metric property of the $p$-Wasserstein distance $W_p$ (see, e.g., \cite{villani2008optimal}) allows using it to evaluate the statistical accuracy of the estimator $\hat \mu_N$ in estimating $\mu$ \citep{dudley1969speed,fournier2015rate,weed2019sharp}. On the other hand, it is of similar interest to investigate the behavior of $W_p(\hat{\mu}_N, \hat{\nu}_N)$  as an estimator for the true functional $W_p(\mu, \nu)$, which in general yields statistically efficient estimators \citep{sommerfeld2019optimal,hundrieser2022empirical, manole2024sharp}, although for continuous settings under appropriate smoothness assumptions improvements are possible \citep{niles2022minimax}. In addition to that, considerable interest has been put in deriving distributional limits for the OT plan \citep{klatt2020limit,liu2023asymptotic} for the discrete setting as well as estimating the OT map in the continuous setting \citep{hutter2021minimax, deb2021rates, manole2021plugin} under smoothness assumptions and in the semi-discrete setting \citep{sadhu2023limit, hundrieser2022unifying,del2024central}.
For OT barycenters based on empirical measures significantly less is known, though recently some progress has been made in the context of finitely supported measures \citep{heinemann2022randomized}. Extending such results to general measures is not obvious and requires the need for alternative statistical modelling. %

\subsection{Contributions: Statistical Models and Deviation Bounds}

The key contributions are summarized as follows. First, we propose three prototypical statistical models for general measures on finite spaces: the \emph{multinomial model}, the \emph{Bernoulli model}, and the \emph{Poisson intensity model}. %
 Each model gives rise to a specific measure estimator. For each model, we then analyze the statistical performance of plug-in estimators in approximating one of the following four quantities:
\begin{enumerate}
  \item[(\textsf{A})] The population measure with respect to the $(p,C)$-Kantorovich-Rubinstein distance and in total variation norm (Section \ref{chp:sampbounds}).
  \item[(\textsf{B})] The $(p,C)$-Kantorovich-Rubinstein distance between two population measures in mean absolute deviation (Section \ref{chp:sampbounds}).
  \item[(\textsf{C})] The unbalanced optimal transport plan between  two population measures with respect to the Hausdorff distance induced by the total variation norm (Section \ref{sec:UOT_plans}).
  \item[(\textsf{D})] The $(p,C)$-barycenter of a finite collection of population measures in terms of the excess Fr\'echet function value and $(p,C)$-Kantorovich-Rubinstein distance (Section~\ref{sec:KR_barycenter}).
\end{enumerate}
In addition to our theoretical analysis we show how our results can be employed for randomized computation with statistical guarantees (\Cref{sec:applications}). %
Finally, we perform various simulations which showcase the performance of the introduced empirical estimators for different quantities above (Section \ref{sec:sims}). 

\subsubsection{Statistical models}
In the following, we formalize the three statistical models and provide specific motivations for each. For an illustration of different realizations of these estimators see  \Cref{fig:realization_estimators}.

\subsubsection*{Multinomial Model}%
For the \emph{multinomial model}, we %
consider i.i.d.\ random variables $X_1,\ldots,X_N {\sim} \frac{\mu}{\mathbb{M}(\mu)}$, where the total intensity $\mathbb{M}(\mu)$ is assumed to be known and strictly positive. The corresponding unbiased empirical estimator is then defined as 
\begin{align}\label{eq:multinomialmeasure}
\hat{\mu}_N\coloneqq \frac{\mathbb{M}(\mu)}{N} \sum_{x \in \X} \lvert\{ k\in \{1,\dots,N\} \ \vert \ X_k=x \}\rvert \ \delta_{x}.
\end{align}
This estimator is exactly the standard empirical measure associated to $\frac{\mu}{\mathbb{M}(\mu)}$ but rescaled with $\mathbb{M}(\mu)$. 
Insofar, this model extends the classical sampling approach for probability measures to measures of positive mass. A key motivation for introducing this model is resampling for randomized computation of UOT quantities. In real world data analysis it is common to encounter data (e.g.,  high-resolution images) which are out of reach for current state-of-the-art OT solvers. One idea in this scenario is to replace each measure by its empirical version and then compute the respective UOT quantity between these surrogates. 
For probability measures, this was introduced for the $p$-Wasserstein distance \citep{sommerfeld2019optimal} and the $p$-Wasserstein barycenter \citep{heinemann2022randomized}. Statistical deviation bounds allow balancing computational complexity and approximation accuracy in terms of the sample size $N$. For more details on randomized computation in the present context we refer to \Cref{sec:rand_comp}. %

\begin{figure}[h]
  \centering
  \includegraphics[width=0.95\textwidth]{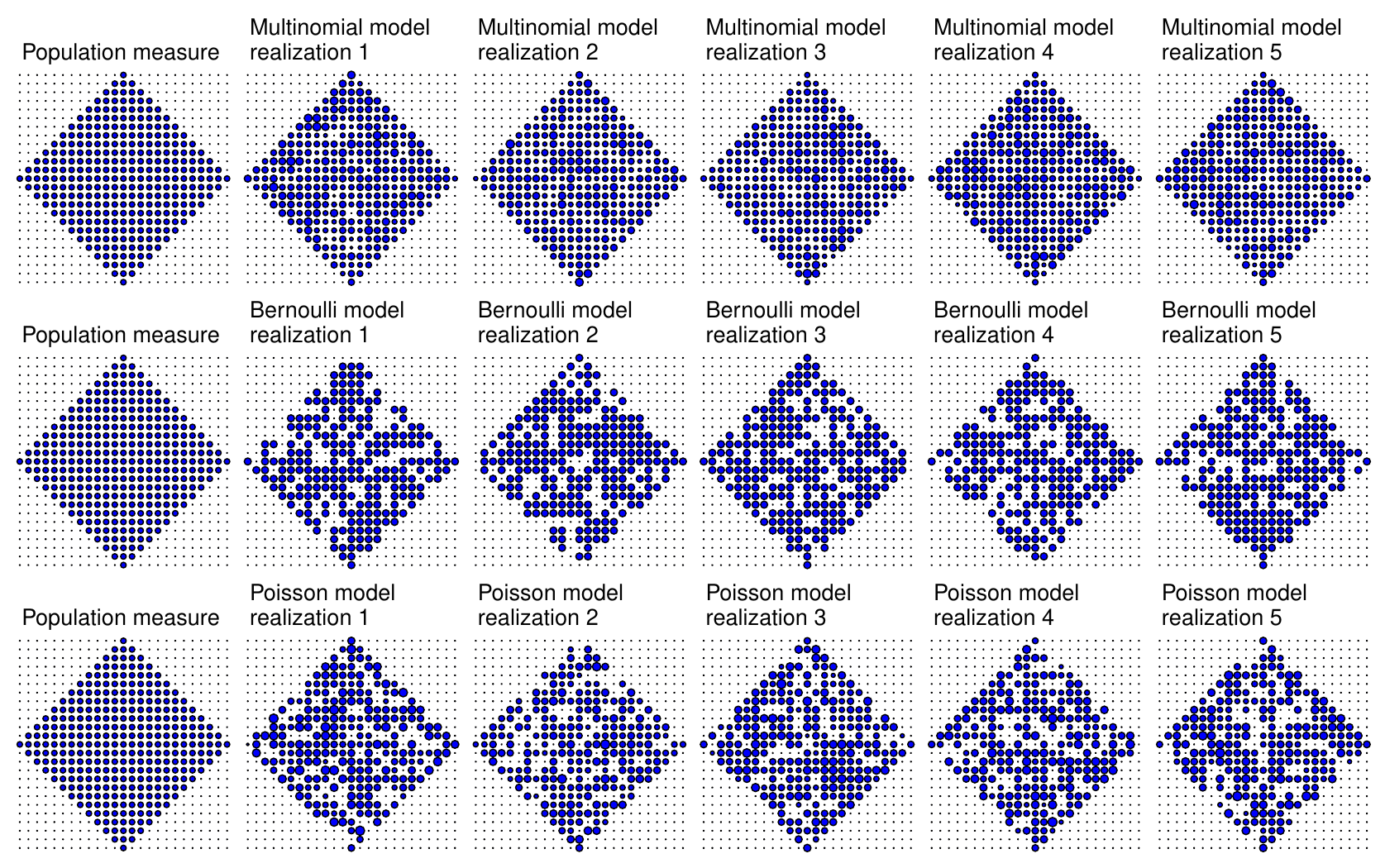}
  \caption{Realizations of measure estimators for the population measure (left column) based on the multinomial model with $N = 2000$ (top), Bernoulli model with homogeneous thinning $s = 0.75$ (middle) and Poisson intensity model with $s = 0.75$ and $t = 5$ (bottom). The thickness of each circle represents the mass assigned to the respective point.}
  \label{fig:realization_estimators}
\end{figure}

\subsubsection*{Bernoulli Model}
For the \emph{Bernoulli model} we consider measures $\mu$ with $\mu(x)=1$ for all $x \in \mathrm{supp}(\mu)$. Thus, the measure $\mu$ represents a point cloud in the ambient space $\Y$ with the total mass being the cardinality of the point cloud. We restrict the model to this setting, but we stress that generalizations to arbitrary masses at the individual locations are straightforward. For each location $x\in\mathcal{X}$ we assume to observe independent Bernoulli random variables $B_x{\sim} \mathrm{Ber}(s_x \mu(x) )$ with a fixed \emph{success probability} $s_x\in (0,1]$. We denote $s_\X:=(s_{x_1},\dots,s_{x_{M}})$ where $x_1, \dots, x_{M}$ is an enumeration of all elements of $\X$, $|\XC| = M$ and refer to $s_\X$ as \emph{success vector}. A suitable unbiased estimator for $\mu$ is defined by
\begin{align}\label{eq:bermeasure}
\hat{\mu}_{s_\X}:= \sum_{x\in\mathcal{X}} \frac{B_x}{s_x}\delta_x.
\end{align}
This estimator models a potentially \emph{spatio-heterogenous} thinning of the measure $\mu$ in terms of Bernoulli random variables such that the total mass of $\hat{\mu}_{s_\X}$ is close (but not always equal) to the total mass of $\mu$. Notably, the Bernoulli field $(B_x)_{x\in \X}$ is a prototypical model for incomplete data, where data is missing at random and the Bernoulli variables serve as labels for this. It further arises, e.g., in the context of generalized linear models where the regressor $X$ is linked to $B_x$ by a link function in a non-parametric fashion. The Bernoulli model also underlies the sampling scheme for the subsequent Poisson intensity model, which occurs, e.g.,  in various imaging devices, such as fluorescence cell microscopy. There, fluorescent markers are, e.g., chemically attached to each protein within a complex protein ensemble and then are excited with a laser beam. The resulting, emitted photons indicate the spatial position of the objects of interest in the proper experimental setup \citep{kulaitis2021resolution}. However, the marker has a limited \emph{labelling efficiency} $s_x \in (0,1]$ at each location $x\in \X$, and we only observe a location which has been labelled by the marker and finally emits photons. We refer to \cite{aspelmeier2015modern} for the further discussion on statistical aspects of high-resolution fluorescence microscopy.

\subsubsection*{Poisson Intensity Model}

For the \emph{Poisson intensity model} we fix a parameter $t>0$ and a success probability $s\in (0,1]$. Consider a collection of $\lvert \X \rvert$ independent Poisson random variables $P_x{\sim} \mathrm{Poi}(t\mu(x))$ with intensity $t\mu(x)$ at each location $x\in\mathcal{X}$ and independently from that Bernoulli random variables $B_x\sim \mathrm{Ber}(s)$ for each $x \in X$. A suitable unbiased estimator for $\mu$ is defined by
\begin{align}\label{eq:poissonmeasure}
\hat{\mu}_{t,s}:= \frac{1}{st}\sum_{x\in\mathcal{X}} B_x P_x\delta_x.
\end{align}
In contrast to the Bernoulli model, in this model the success probability is assumed to be homogeneous (as opposed to the inhomogeneous probabilities in the Bernoulli model, though such generalizations are straightforward) and the values of the population measures at each support point are not necessarily equal to one. Hence, we have two independent layers of randomness in the construction of this empirical measure. First, we draw a location $x$ with a certain probability $s$, then we observe random photon counts driven by a Poisson distribution based on the mass of $\mu$ and the value of $t$.\\
This model is motivated by various tasks in photonic imaging, for example, fluorescence microscopy, X-ray imaging and positron emission tomography (PET), see \cite{munk2020statistical} for a survey. The finite space $\X$ represents the center of bins of a detection interface used to measure the emitted photons. The value $\mu(x)$ corresponds to the integrated underlying photon intensity over its respective bin. This intensity is proportional to an external source, such as a laser duration in fluorescence microscopy and modelled by the parameter $t>0$. The Bernoulli random variable $B_x$ models the possibility that in the bin of $x$ a photon can not be recorded. This might be due to various effects that cause thinning, such as limited labelling efficiency, dead time of cameras or a loss of photons due to sparse detector tubes. The value of $P_x$ corresponds to the number of photons which have been measured at the bin of $x$. Note that besides $B_x$ there might be also additional effects present which do not disable the whole bin, but just prevent a single photon from being measured. All this causes a thinning of the process and is incorporated in the probability $s^\prime\in (0,1]$  that a single photon at any bin and any point in time can not be measured. In this case, the model can be shown to be equivalent to a Poisson model with parameter $ts^\prime >0$ instead of~$t$ \citep{aspelmeier2015modern}, %
 and is thus a special case of the general Poisson intensity model.

\subsubsection{Summary of Statistical Deviation Bounds}

Concerning approximation of population measure by its empirical counterpart, we show in \Cref{chp:sampbounds} that there exist constants $\mathcal{E}_{p,\mathcal{X},\mu}^{\mathrm{Mult}}(C), \mathcal{E}_{p,\mathcal{X},\mu}^{\mathrm{Ber}}(C),\mathcal{E}_{p,\mathcal{X},\mu}^{\mathrm{Poi}}(C)$ such that for any measure $\mu$ and its estimator $\hat{\mu}$, with $p\geq 1$,  in each of the three statistical models it holds, %
\begin{align}\label{eq:contr1}
    \mathbb{E}\left[ \KR_{p,C}\left( \hat{\mu},\mu\right)\right]\leq 
    \begin{cases}
      \begin{array}{lll}
        \mathcal{E}_{p,\mathcal{X},\mu}^{\mathrm{Mult}}(C)^{\frac{1}{p}} 
        \, N^{-\frac{1}{2p}}, &\text{ if }\hat{\mu}=\hat{\mu}_N, &\text{ (Multinomial)}\\[0ex]
        \mathcal{E}_{p,\mathcal{X},\mu}^{\mathrm{Ber}}(C)^{\frac{1}{p}} 
        \, \psi(s_\X)^{\frac{1}{p}}, &\text{ if }\hat{\mu}=\hat{\mu}_{s_\X}, &\text{ (Bernoulli)}\\[0ex]
        \mathcal{E}_{p,\mathcal{X},\mu}^{\mathrm{Poi}}(C)^{\frac{1}{p}} 
        \, \phi(t,s)^{\frac{1}{p}},  &\text{ if }\hat{\mu}=\hat{\mu}_{t,s}, &\text{ (Poisson)}    
      \end{array}
    \end{cases}
\end{align}
where for the \emph{Multinomial model} we obtain a scaling rate of $N^{-\frac{1}{2p}}$, for the \emph{Bernoulli model}
\begin{align*}
  \psi(s_\X)= \begin{cases}
  \left(2\sum_{x\in \X} (1-s_x) \right), &\text{ if } C\leq d_{\min }\coloneqq \min_{x\neq x^\prime}d(x,x^\prime), \\
 \left( \sum_{x\in \X}\frac{1-s_x}{s_x} \right)^{\frac{1}{2}}, &\text{ else}. \\
   \end{cases}
\end{align*}
and for \emph{Poisson intensity model}
\begin{align*}
  \phi(t,s)=\begin{cases}
 \left(2(1-s)\mathbb{M}(\mu)+\frac{s}{\sqrt{t}}\sum_{x\in \X}\sqrt{\mu(x)} \right), &\text{ if } C\leq d_{\min },\\%\min_{x\neq x^\prime}d(x,x^\prime), \\
\left( \frac{1}{st}\mathbb{M}(\mu)+\frac{1-s}{s}\sum_{x\in \X}\mu(x)^2 \right)^{\frac{1}{2}}, &\text{ else}.
  \end{cases}
\end{align*}
Notably, in the multinomial model for $N\rightarrow\infty$, in the Poisson model for $t\rightarrow\infty$, $s\rightarrow 1$ and in the Bernoulli model for $s_{\X}\rightarrow \mathbf{1}_{\X}$, these upper bounds tend to zero. Our approach enables an explicit characterization of the constants $\mathcal{E}_{p,\mathcal{X},\mu}^{\mathrm{Mult}}(C), \mathcal{E}_{p,\mathcal{X},\mu}^{\mathrm{Ber}}(C),\mathcal{E}_{p,\mathcal{X},\mu}^{\mathrm{Poi}}(C)$ in terms of  structural properties of the measures and space, such as  total mass intensity and covering numbers(for details see \Cref{chp:sampbounds} and Appendices \ref{sec:multi} and \ref{sec:ber}). We believe this to be particularly relevant for statistical tasks surrounding the KRD. As an example, we comment on the behavior of the constants if $\XC$ is a subset of $\RR^d$ equipped with Euclidean metric.

\begin{example}[Compact ground space in $\RR^D$] 
For finitely supported measures on the unit ball in $\RR^D$ with $|\XC|$ support points and the Euclidean distance as the metric, it holds for $p = 2$ and $D\geq 1$  up a universal constant\footnote{We write $A \lesssim B$ if there exists a universal constant $\kappa >0$ such that $A \leq \kappa B$.} (see \Cref{sec:explicitB} and Appendices \ref{sec:multi} and \ref{sec:ber}) that %
\begin{align*}
  \frac{{\mathcal{E}}_{2,\XC, \mu}^{\mathrm{Mult}}(C)}{\mathbb{M}(\mu) + \mathbb{M}(\mu)^{1/2}}, \,{\mathcal{E}}_{2,\XC, \mu}^{\mathrm{Ber}}(C),\, {\mathcal{E}}_{2,\XC, \mu}^{\mathrm{Poi}}(C) \lesssim  \begin{cases}
 	C^2, & \text{ if } C \leq d_{\min} \text{ or } D < 4,\\
 	C^2 + \log_2(|\X|), &\text{ if } C >d_{\min} \text{ and } D = 4,\\
 	C^2 + |\X|^{\frac{1}{2} - \frac{2}{D}},  &\text{ if } C >d_{\min} \text{ and } D > 4.
 \end{cases}
\end{align*}
In particular, for $D<4$ the constants are independent of the cardinality of the space $\lvert \X \rvert$, for $D=4$ the dependence is logarithmic and for $D>4$ there is a polynomial dependency in $|\XC|$. %
We note however that the upper bound adapt to the intrinsic dimension of the domain on which the finite measure is supported, e.g.,  if a measure is supported on a $D^\prime<D$ dimensional subspace of $[0,1]^D$ (e.g.,  a submanifold), then in the upper bound the dependency reduces to $D^\prime$ and where the suppressed constant depends on the domain.
\end{example}

In addition, we also establish related convergence statements with respect to the total variation norm (\Cref{thm:TVbound}).
Controlling the empirical estimator with respect to the $(p,C)$-KRD and total variation norm enables us to draw conclusions on the performance of plug-in estimators for the $(p,C)$-KRD. 
Indeed, by  reverse triangle inequality and our  stability bound in \Cref{chp:sampbounds}, Lemma \ref{lem:stability_dual} it holds
\begin{equation}\label{eq:stability bound_intro}
\begin{aligned}
   \mathbb{E}\left[ \lvert \KR_{p,C}(\hat{\mu},\hat{\nu})-\KR_{p,C}({\mu},{\nu})\rvert \right] 
  \leq & \min\Big(\mathbb{E}\left[ \KR_{p,C}(\hat{\mu},{\mu})\right] + \mathbb{E}\left[ \KR_{p,C}(\hat{\nu},{\nu})\right], \\ 
  & \quad 2C^p \KR_{p,C}^{1-p}({\mu},{\nu}) \left(\mathbb{E}\left[ \TV(\hat{\mu},{\mu})\right] + \mathbb{E}\left[\TV(\hat{\nu},{\nu})\right]\right)\Big),
\end{aligned}
\end{equation}
where $\TV(\mu,\nu)=\sum_{x\in \X} \lvert \mu(x)-\nu(x) \rvert$ is the \emph{total variation distance}. This asserts that the $\KR_{p,C}({\mu},{\nu})$ is well-approximated by  $\KR_{p,C}(\hat{\mu},\hat{\nu})$ as soon as the empirical estimators $\hat \mu$ and $\hat \nu$ approximate the population measures well (Corollary \ref{cor:estimation_KRD} and Theorem \ref{thm:estimation_KRD_separate}). %

In addition to the UOT \emph{values} as in \eqref{eq:stability bound_intro}, we also establish a novel quantitative stability bound for UOT \emph{plans} with parameters $p\geq 1$ and $C>0$. Since UOT plans are not necessarily unique, our stability result is stated for the collection of UOT plans ${\mathbf{P}}^*_{p,C}(\mu,\nu)$ and is quantified in terms of the Hausdorff distance  $\mathcal{H}_\TV$ induced by the total variation norm (\Cref{thm:UOTplan_Stability}),
\begin{align*}
  & \mathcal{H}_\TV\left({\mathbf{P}}^*_{p,C}(\hat\mu,\hat\nu),{\mathbf{P}}^*_{p,C}(\mu,\nu)\right)\\
  \leq \;& 4\,(|\X| +1)\left(\TV(\hat\mu,\mu) + |\mathbb{M}(\hat\mu) - \mathbb{M}(\mu)| +\TV(\hat\nu,\nu) + |\mathbb{M}(\hat\nu) - \mathbb{M}(\nu)|\right).
  \end{align*}
Based on this bound, we can make quantitative statements about the performance of plug-in estimators for UOT plans (\Cref{thm:UOTplan_convergence}). Notably, we also establish a quantitative stability bound for balanced OT. 
Finally, the convergence analysis of empirical estimators with respect to the $(p,C)$-KRD also enables us to draw conclusions for $(p,C)$-barycenters (defined in \eqref{eq:KRbarycenter}), see Theorems \ref{thm:frechetboundPoi} and \ref{thm:kr_boundPoi}. %
To elaborate, let $\mu^*$ be any $(p,C)$-barycenter of the population measures $\mu^1,\dots,\mu^J$ and let $\hat{\mu}^\star$ be any $(p, C)$-barycenter of the empirical measures $\hat{\mu}^1,\dots,\hat{\mu}^J$. Since neither $\mu^\star$ nor $\hat{\mu}^\star$ is necessarily unique, we quantify the error of the empirical counterpart $\hat{\mathbf{B}}^\star$ in approximating the population set $\mathbf{B}^\star$ via (recall \eqref{eq:KRbarycenter}) %
\begin{align}\label{eq:cont4}
   F_{p,C}(\hat{\mu}^\star)-F_{p,C}({\mu}^\star) \quad \text{ and }\quad  \   \sup_{\hat{\mu}^\star \in \hat{\mathbf{B}}^\star} \inf_{\mu^\star \in \mathbf{B}^\star} \KR_{p,C}(\mu^\star,\hat{\mu}^\star).
\end{align}
Notably, the term on the right-hand side quantifies the maximal difference between an empirical barycenter to the closest population barycenter. %
Although this is a slightly weaker result than the analysis in terms of the Hausdorff distance, it is sufficient for practical considerations, as it details how well the ``worst choice'' for the empirical barycenters approximates its population counterpart. 
Both expressions in \eqref{eq:cont4} can be related to the $(p,C)$-KRD error between empirical and population measure in \eqref{eq:contr1}, which enables us to quantify the performance of empirical barycenters. 

\section{Empirical Kantorovich-Rubinstein Distances}\label{chp:sampbounds}
In this section we investigate the Poisson model and analyze how fast the empirical measure estimator tends to its population counterpart in terms of the $(p,C)$-Kantorovich-Rubinstein distance and the total variation norm. These bounds allow us to quantify how fast the plug-in estimator tends to the population $(p,C)$-Kantorovich-Rubinstein distance. 
Results for the multinomial and Bernoulli model follow along the same reasoning. Corresponding deviation bounds and proofs are provided in Appendices \ref{sec:multi} and \ref{sec:ber}.

\subsection{Stability Bounds for Kantorovich-Rubinstein Distance} 
The convergence results of this section are based on novel stability bounds for the UOT distance, which we detail in the following. The first is based on a tree approximation of the space $\X$ (Lemma \ref{lem:treeapproximation}), whereas the second relies on the dual formulation of the UOT cost (Lemma \ref{lem:stability_dual}). The proofs for this section are detailed in  \Cref{app:stability_proofs}.  

\subsubsection{Stability Bound via Tree Approximation of Domain}\label{app:treeapproximation}
Let $\mathcal{T}=(V,E)$ be a rooted, ultrametric tree with height function $h:V\longrightarrow \mathbb{R}_+$ and root~$\mathsf{r}$. For two nodes $\textsf{u},\textsf{v}\in V$, denote the unique path between $\mathsf{u}$ and $\mathsf{v}$ in $\mathcal{T}$ by $\mathcal{P}(\mathsf{u},\mathsf{v})$. For a node $\textsf{v}\in V$ its \emph{children} are the elements of the set $\mathcal{C}(\textsf{v})=\left\lbrace \textsf{w}\in V \,\mid\, \textsf{v}\in \mathcal{P}(\textsf{w},\textsf{r})\right\rbrace$. The \emph{parent} $par(\mathsf{v})$ of a node $\textsf{v}$ is the unique node with $(par(\mathsf{v}),\mathsf{v})\in E$ and $h(\mathsf{v})< h(par(\mathsf{v}))$.
For any $C>0$, define the set
\begin{align}\label{eq:subtreeroot}
\mathcal{R}(C)\coloneqq \left\lbrace \textsf{v}\in V\, \mid \, h(\textsf{v})\leq C/2< h(\text{par}(\textsf{v}))\right\rbrace
\end{align} 
with the convention that $\mathcal{R}(C)=\lbrace \textsf{r} \rbrace$ if $\frac{C}{2}\geq h(\textsf{r})$. The goal is to control the $(p,C)$-KRD on the finite metric space $(\X,d)$ by bounding it from above by a dominating distance $d_\mathcal{T}$ induced\footnote{For two vertices of the tree $\mathcal{T}$, we define their distance $d_{\mathcal{T}}$ as the sum of the weights of the edges included in the unique path between the two vertices. Here, the weight of an edge joining two vertices $v$ and $par(v)$ is given by $h(par(v))-h(v)$.\label{fn:treedist}} from a tree $\mathcal{T}$ with the elements of $\X$ as vertices and a height function $h$ such that $d(x,x^\prime)\leq d_\mathcal{T}(x,x^\prime)$. In this case and by the definition of the Kantorovich-Rubinstein distance it holds for all measures $\mu,\nu\in\msrX$ that
\begin{align}\label{eq:generalKRbound}
\text{KR}_{p,C}(\mu,\nu)\leq \text{KR}_{d_\mathcal{T}^p,C}(\mu,\nu),
\end{align}
where $\text{KR}_{d_\mathcal{T}^p,C}(\mu,\nu)$ denotes the $(p,C)$-KRD w.r.t. the ground space $(\X,d_\mathcal{T})$. Moreover, if $\mathcal{T}$ is an ultrametric tree with leaf nodes $L$ and height function $h\colon V\to \mathbb{R}_+$ inducing\cref{fn:treedist} the tree metric $d_\mathcal{T}$ and the two measures $\mu^L,\nu^L\in \mathcal{M}_+(L)$ supported on the leaf nodes of $\mathcal{T}$, then it holds \citep[Theorem 2.3]{heinemann2022kantorovich} that
{\begin{align}\label{eq:KRultrametric}
\begin{split}
    &\text{KR}^p_{d_\mathcal{T}^p,C}\left(\mu^L,\nu^L\right)=\\
 &\sum_{\textsf{v}\in \mathcal{R}(C)} \Bigg(2^{p-1}\sum_{\textsf{w}\in \mathcal{C}(\textsf{v})\setminus \lbrace \textsf{v} \rbrace}  \Big(\left( h(par(\textsf{w}))^p-h(\textsf{w})^p \right) \left\vert \mu^L(\mathcal{C}(\textsf{w}))-\nu^L(\mathcal{C}(\textsf{w}))\right\vert \Big) \\[1ex]
&+\left(\frac{C^p}{2}-2^{p-1}h(\textsf{v})^p\right) \left\vert \mu^L(\mathcal{C}(\textsf{v}))-\nu^L(\mathcal{C}(\textsf{v}))\right\vert \Bigg).
\end{split}
\end{align}}
\begin{figure}
  \centering
  \subfloat[][]{\includegraphics[width=0.78\linewidth]{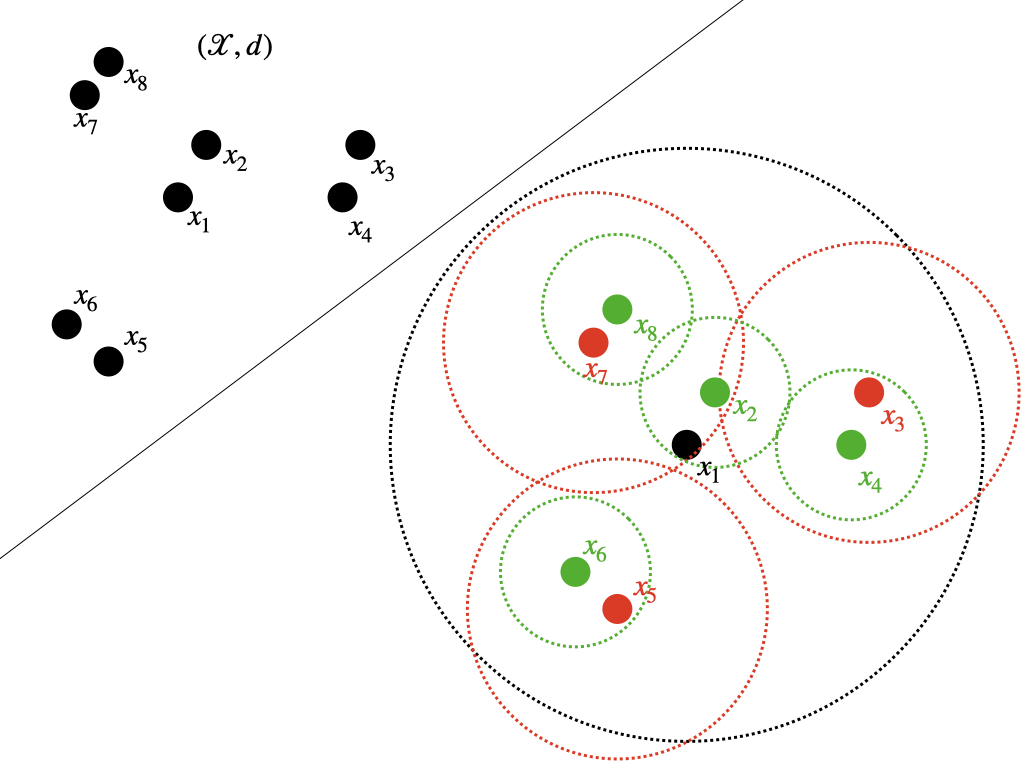}} \hspace{0.2\linewidth}
  \subfloat[][]{\includegraphics[width=0.78\linewidth]{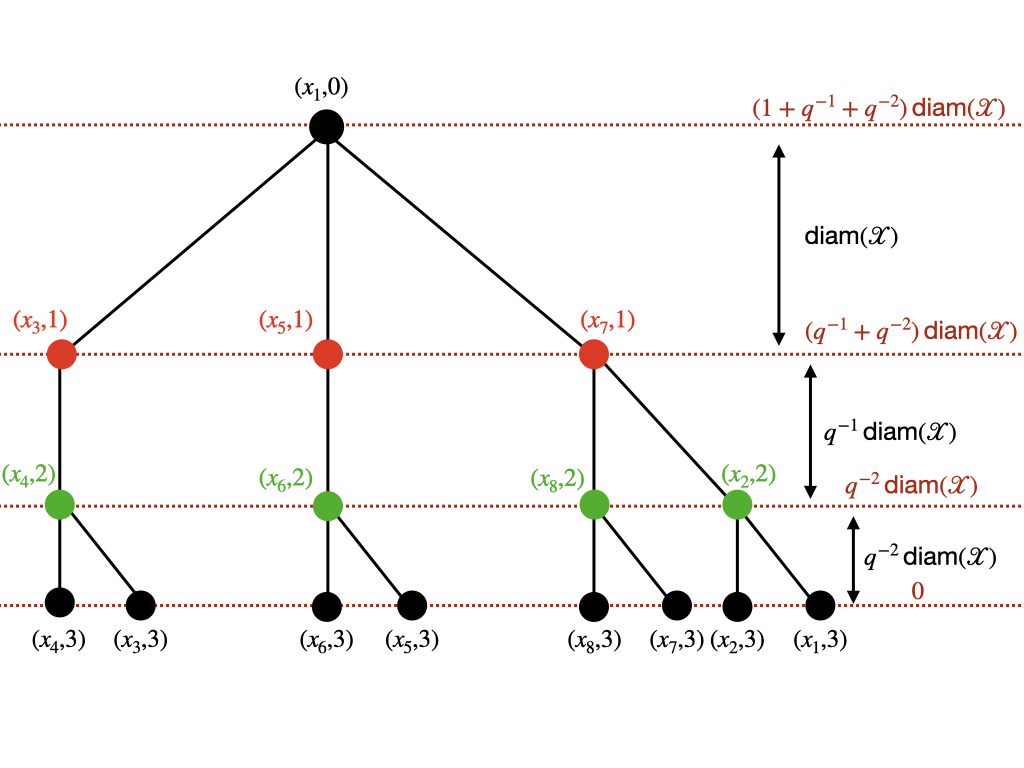}}%
  \caption{\textbf{Ground metric approximation by an ultrametric tree distance:} \textbf{(a)} A finite metric space $(\X,d)$ and its covering sets $Q_0$ (black), $Q_1$ (red) and $Q_2$ (green) for $L=2$. \textbf{(b)} Based on the covering sets from (a) an ultrametric tree is constructed. The metric space $\X$ is embedded in level $L+1=3$ and equal to all leaf nodes of that tree.}%
  \label{fig:treeapproximation}
\end{figure}

The construction of $\mathcal{T}$, such that \eqref{eq:generalKRbound} holds, is as follows. 
Fix some \emph{depth level} $L\in \mathbb{N}$. For some \emph{resolution} $q>1$ and level $j=0,\ldots,L$ define the covering set\footnote{For a metric space $(\mathcal{X},d)$ an $\varepsilon$-cover is a set of points $\left\{ x_1,\ldots,x_m\right\}\subset \mathcal{X}$ such that for each $x\in\mathcal{X}$, there exists some $1\leq i \leq m$ such that $d(x,x_i)\leq \varepsilon$. The smallest such set is denoted as $\mathcal{N}(\X, \varepsilon)$.} $Q_j\coloneqq \mathcal{N}(\X,q^{-j}\mathrm{diam}(\X))\subset \X$ and let $Q_{L+1}\coloneqq\mathcal{X}$.  Any point $x\in Q_j$ is considered as a node at level $j$ of a tree $\mathcal{T}$ and denoted as $(x,j)$ to emphasize its level position.  An illustration of this approximation is given in \Cref{fig:treeapproximation}. For level $j=0$ this yields a single element in $Q_0$ which serves as the root of the tree. For $j=0,\ldots,L$ a node $(x,j)$ at level $j$ is connected to one node $(x^\prime,j+1)$ at level $j+1$ if their distance satisfies $d(x,x^\prime)\leq q^{-j}\mathrm{diam}(\X)$ (ties are broken arbitrarily). The edge weight of the corresponding edge is set equal to $q^{-j}\mathrm{diam}(\mathcal{X})$. Consequently, the height of each node only depends on its assigned level $0\leq l\leq L+1$ and is defined as $h_{q,L}\colon \{0,\ldots,L+1\}\to\mathbb{R}$ by \footnote{The construction of the ultra-metric tree is based on approximating the underlying domain at varying precisions $(\delta_0, \dots, \delta_L)=(q^{-l}\diam(\XC))_{l = 0, \dots, L}$ with as few points as possible. If one were to select the precisions differently, the height function would change to $h(l) \coloneqq \sum_{j=l}^{L} \delta_{j}$, possibly leading to different convergence statements. Our choice is inspired by that of \cite{dereich2013constructive} and \cite{fournier2015rate}, who derived sharp rates of convergence for the empirical Wasserstein distance in Euclidean settings. Based on this choice, we establish in Section \ref{sec:explicitB} convergence results for different regimes which align with those of previous articles.}
\begin{align}\label{eq:height}
    h_{q,L}(l)=\sum_{j=l}^{L} q^{-j}\mathrm{diam}(\X)=\frac{q^{1-l}-q^{-L}}{q-1}\mathrm{diam}(\X).
\end{align}
By definition the space $\X$ is embedded in level $L+1$ as the leaf nodes of $\mathcal{T}$ with height $h_{q,L}(L+1)=0$. By a straightforward computation it holds for two points $x,x^\prime \in \X$ considered as embedded in $\mathcal{T}$ as $(x,L+1)$ and $(x^\prime,L+1)$ that
\begin{align}\label{eq:metricupper}
  d^p(x,x^\prime)\leq d_\mathcal{T}^p((x,L+1),(x^\prime,L+1)).  
\end{align}
The measures $\mu,\nu$ are embedded into $\mathcal{T}$ as measures $\mu^L,\nu^L$ supported only on leaf nodes of $\mathcal{T}$, and thus it follows from \eqref{eq:metricupper} that
\begin{align*}
    \KR_{p,C}(\mu,\nu)\leq \KR_{d_\mathcal{T}^p,C}\left(\mu^L,\nu^L\right).
\end{align*}
In combination with the closed formula from \eqref{eq:KRultrametric} this yields an upper bound on the $(p,C)$-KRD. Whenever clear from the context the notation is alleviated by writing $\textsf{v}\in Q_l$ instead of $(\textsf{v},l) \in Q_l$. 

\begin{lemma}\label{lem:treeapproximation}
Let $(\X,d)$ be a finite metric space and let $\mu,\nu\in\msrX$ with total mass $\mathbb{M}(\mu)$ and $\mathbb{M}(\nu)$, respectively. Let $p\geq 1$ and $C>0$. Then for any resolution $q>1$, depth $L\in\mathbb{N}$ and height function \eqref{eq:height} with $$h_{q,L}(k)=\frac{q^{1-k}-q^{-L}}{q-1}\mathrm{diam}(\X)$$ it holds that
\begin{align*}
    \KR^p_{p,C}(\mu,\nu)\leq \;
    &\begin{cases}
        \left(\frac{C^p}{2}-2^{p-1}h_{q,L}(0)^p\right)\left\vert \mathbb{M}(\mu)-\mathbb{M}(\nu)\right\vert\\[0.1cm] \ +B_{q,p,L,\X}(1), \\[0.1cm] \hspace{2cm}    \text{if } C\geq 2h_{q,L}(0), \\[0.4cm]
        B_{q,p,L,\X}(l), \\[0.1cm] \hspace{2cm} \text{if } 2h_{q,L}(l)\leq  C< 2h_{q,L}(l-1),\\[0.4cm]
        \frac{C^p}{2} \TV(\mu,\nu), \\[0.1cm] \hspace{2cm} \text{if } C\leq \left( 2h_{q,L}(L) \vee {\min_{x\neq x^\prime}d(x,x^\prime)}\right), \\
    \end{cases}
\end{align*}
where we define 
\begin{align*}
    B_{q,p,L,\X}(l)\coloneqq 2^{p-1}\sum\limits_{j=l}^{L+1}\sum\limits_{x\in Q_j} \left( h_{q,L}(j-1)^p-h_{q,L}(j)^p\right)\left\vert \mu^L(\mathcal{C}(x))-\nu^L(\mathcal{C}(x))\right\vert.
\end{align*}
\end{lemma}

\begin{remark}\label{rmk:boundTV_KRD}
According to \Cref{lem:treeapproximation}, for any two measures $\mu,\nu\in\msrX$ and $C\leq \left( 2h_{q,L}(L) \vee {\min_{x\neq x^\prime}d(x,x^\prime)}\right)$ it always holds that $\KR^p_{p,C}(\mu,\nu)\leq \frac{C^p}{2} \TV(\mu,\nu) \leq \frac{C^p}{2}$. In~particular, for $C\searrow 0$ this implies that $\sup_{\mu, \nu \in \msrX} \KR^p_{p,C}(\mu,\nu) \leq \frac{C}{2}\searrow 0$.
\end{remark}

\subsubsection{Stability Bound via Dual Formulation}

The stability bound based on the tree construction from previous subsection yields sharp  statements about the convergence of the measure estimator to its population counterpart. By triangle inequality, this yields sharp convergence rates for plug-in estimators of the KRD when population measures are close, but they are suboptimal when the population measures are strictly separated. In this subsection we follow a different approach to derive a stability bound for the KRD which is inspired by \cite{chizat2020} and \cite{manole2024sharp} and based on the dual representation of the KRD.

\begin{lemma} \label{lem:stability_dual}
Let $(\XC,d)$ be a finite metric space and consider two pairs of measures $\mu^1, \nu^1, \mu^2, \nu^2\in \msrX$. Let $p \geq 1$ and $C>0$. Then, it follows that
\begin{align}\label{eq:KRD_duality_stabilityBound}
  |\KR_{p,C}^p(\mu^1, \nu^1) - \KR_{p,C}^p(\mu^2, \nu^2)| 
  \leq \frac{C^p}{2}\begin{cases}4 \left(\TV(\mu^1, \mu^2)+ \TV(\nu^1, \nu^2)\right),  \\[0.1cm] \hspace{2cm}    \text{if }  C> \min\{\dist(\mu^1, \nu^1), \dist(\mu^2, \nu^2)\},\\[0.4cm]
 \left(\left|\mathbb{M}(\mu^1) - \mathbb{M}(\mu^2)\right| +  \left|\mathbb{M}(\nu^1) - \mathbb{M}(\nu^2)\right|\right),  \\[0.1cm] \hspace{2cm}  \text{ else,}
 \end{cases}
\end{align}
where we define 
$$
\dist(\mu, \nu) \coloneqq \begin{cases}
  \min_{	x \in \mathrm{supp}(\mu), x'\in \mathrm{supp}(\nu)} d(x, x'), & \text{ if } \mu \neq  0 \text{ and } \nu \neq  0, \\
  + \infty, & \text{ if } \mu = 0 \text{ or } \nu = 0. 
 \end{cases}
 $$

Moreover, under $\KR_{p,C}(\mu^1, \nu^1)\geq \delta>0$ it holds that \begin{align*}
  |\KR_{p,C}(\mu^1, \nu^1) - \KR_{p,C}(\mu^2, \nu^2)| \leq \delta^{1-p}|\KR_{p,C}^p(\mu^1, \nu^1) - \KR_{p,C}^p(\mu^2, \nu^2)|.
\end{align*}
\end{lemma}

 \subsection{$\mathbf{(p,C)}$-Kantorovich-Rubinstein Deviation Bound}\label{subsec:KR_bounds}

 In this section we first derive statistical deviation bounds for approximation error empirical estimators for their population counterpart with respect to the TV norm and $(p,C)$-KRD. We then continue by analyzing the performance of plug-in estimators. The omitted proofs of the following Theorems \ref{thm:TVbound} and \ref{thm:samplingboundKR} can be found in \Cref{app:TV_bound_proofs}.

\begin{theorem}[Expected deviation of estimator in TV and KRD]\label{thm:TVbound}
  Let $(\mathcal{X},d)$ be a finite metric space and $\mu\in\mathcal{M}_+(\mathcal{X})$ with total mass $\mathbb{M}(\mu)$. Let $\hat{\mu}_{t,s}$ be the estimator from \eqref{eq:poissonmeasure}. Then, for any $p\geq 1$ and $C>0$ it holds that
  \begin{align*}
      \mathbb{E}\left[ \KR^p_{p,C}\left(\hat{\mu}_{t,s},\mu\right)\right]\leq C^p\mathbb{E}\left[ \TV\left(\hat{\mu}_{t,s},\mu\right)\right] \leq C^p\left(2(1-s)\mathbb{M}(\mu)+\frac{s}{\sqrt{t}}\sum_{x\in \X}\sqrt{\mu(x)} \right).
  \end{align*}
  \end{theorem}

\begin{theorem}[Expected deviation of estimator in KRD]\label{thm:samplingboundKR}
Let $(\mathcal{X},d)$ be a finite metric space and $\mu\in\mathcal{M}_+(\mathcal{X})$. Let $\hat{\mu}_{t,s}$ be the estimator from \eqref{eq:poissonmeasure}. Then, for any $p\geq 1$,  $C>0$, resolution $q>1$ and depth $L\in \mathbb{N}$ it holds that
\begin{align*}
    \mathbb{E}\left[ \KR_{p,C}\left( \hat{\mu}_{t,s},\mu\right)\right]\leq
    \mathcal{E}_{p,\mathcal{X},\mu}^{\mathrm{Poi}}(C,q,L)^{1/p}\begin{cases}
   \left(2(1-s)\mathbb{M}(\mu)+\frac{s}{\sqrt{t}}\sum_{x\in \X}\sqrt{\mu(x)} \right)^{\frac{1}{p}},  \\[0.1cm]  \hspace{2cm} \text{if } C\leq \left(2h_{q,L}(L) \vee {\min_{x\neq x^\prime}d(x,x^\prime)}\right), \\[0.4cm]
  \left( \frac{1}{st}\mathbb{M}(\mu)+\frac{1-s}{s}\sum_{x\in \X}\mu(x)^2 \right)^{\frac{1}{2p}}, \\[0.1cm]  \hspace{2cm} \text{else}. 
    \end{cases}
\end{align*}
For
\[
A_{q,p,L,\X}(l):=\mathrm{diam}(\mathcal{X})^p2^{p-1}\left(q^{-Lp}\lvert  \X \rvert^{\frac{1}{2}}+\left(\frac{q}{q-1}\right)^p \sum_{j=l}^{L}q^{p-jp}\lvert Q_j \rvert^{\frac{1}{2}}\right),
\]
the constant is equal to 
\begin{align*}
\mathcal{E}_{p,\mathcal{X},\mu}^{\mathrm{Poi}}(C,q,L)=\begin{cases}
\left(\frac{C^p}{2}-2^{p-1}\left( \frac{q-q^{-L}}{q-1}\mathrm{diam}(\mathcal{X})\right)^p \right)
+A_{q,p,L,\X}(1),\\[0.1cm] \hspace{2cm}\text{if }  C\geq 2h_{q,L}(0),\\[0.4cm]
A_{q,p,L,\X}(l),\\[0.1cm] \hspace{2cm}\text{if }  2h_{q,L}(l)\leq C< 2h_{q,L}(l-1),\\[0.4cm]
\frac{C^p}{2},\\[0.1cm] \hspace{2cm} \text{if }  C\leq \left(2h_{q,L}(L) \vee {\min_{x\neq x^\prime}d(x,x^\prime)}\right),
\end{cases}
\end{align*}
Furthermore, for $p=1$ the factor $\frac{q}{(q-1)}$ in $A_{q,1,L,\X}(l)$ can be removed for all $l=1,\dots,L$.
\end{theorem}

The constant $\mathcal{E}_{p,\mathcal{X},\mu}^{\mathrm{Poi}}(C,q,L)$ is reminiscent of the constants for similar deviation bounds for optimal transport between finitely supported measures \citep{boissard2014mean,sommerfeld2019optimal}. However, in the case of UOT one finds an interesting case distinction into roughly three cases depending on the relation between the penalty parameter $C$ of the $(p,C)$-KRD and the resolution $q$ and depth $L$ of the tree approximation. The different constants arise from the fact that $C$ controls the maximal range at which transport occurs in an UOT plan. In particular, if $d(x,x^\prime)>C^p$, then for any UOT plan $\pi$ it holds $\pi(x,x^\prime)=0$. If $C$ is sufficiently large ($C \geq 2h(0)$), i.e., larger than the diameter of $\X$, then $\mathcal{E}_{p,\mathcal{X},\mu}^{\mathrm{Poi}}(C,q,L)$ coincides with the deviation bounds for usual optimal transport, however, there is an additional summand arising from the necessary estimation of the true total mass $\mathbb{M}(\mu)$ (see also Lemma \ref{lem:treeapproximation}). For sufficiently small $C$, e.g., $C<\min_{x\neq x^\prime} d(x,x^\prime)$, the $(p,C)$-KRD is proportional to the TV distance, hence we obtain a constant which is oblivious to the geometry of the ground space (see \Cref{thm:TVbound}). For an intermediate value of $C$, the UOT problem on the ultra-metric tree $\mathcal{T}$ decomposes into smaller problems on subtrees of $\mathcal{T}$ (for details see the proof of \eqref{eq:KRultrametric} in \cite{heinemann2022kantorovich}) depending on $C$. The expected $(p,C)$-KRD error then depends on the size of these subtrees and the mass estimation error inherent to the total mass on these subtrees. 
\begin{remark}\label{rem:infimum}
Since the deviation bound holds for any resolution $q>1$ and depth $L\in\mathbb{N}$ one can optimize and equivalently state upper bounds in terms of the infimum over those parameters. When the dependence on $q$ or $L$ is omitted, it is assumed that the infimum over those parameters has been taken, i.e.
\begin{align*}
    \mathcal{E}_{p,\mathcal{X},\mu}^{\mathrm{Poi}}(C)=\inf_{L\in \mathbb{N},q>1} \mathcal{E}_{p,\mathcal{X},\mu}^{\mathrm{Poi}}(C,q,L).
\end{align*}
\end{remark}
From the reverse triangle inequality we immediately obtain the following corollary from \Cref{thm:samplingboundKR}.

\begin{corollary}[Expected deviation of empirical KRD]\label{cor:estimation_KRD}
Let $(\mathcal{X},d)$ be a finite metric space and $\mu,\nu\in\mathcal{M}_+(\mathcal{X})$. Let $\hat{\mu}_{t,s},\hat{\nu}_{t,s}$ be the estimator from \eqref{eq:poissonmeasure} for each of these measures, respectively. Then, for any $p\geq 1$, resolution $q>1$ and depth $L\in \mathbb{N}$ it holds that
\begin{align*}
    &\mathbb{E}\left[ \lvert \KR_{p,C}\left( \hat{\mu}_{t,s},\hat{\nu}_{t,s}\right) -  \KR_{p,C}\left( {\mu},{\nu}\right)\rvert \right]\\
       \leq \,
        &\mathcal{E}_{p,\mathcal{X},\mu}^{\mathrm{Poi}}(C,q,L)^{1/p}
\begin{cases}
            \left(2(1-s)\mathbb{M}(\mu+\nu)+\frac{s}{\sqrt{t}}\sum_{x\in \X}(\sqrt{\mu(x)} + \sqrt{\nu(x)}\right)^{\frac{1}{p}} ,\\[0.1cm]  \hspace{2cm} \text{if } C< \min_{x\neq x^\prime\in\X}d(x,x^\prime), \\
            \left( \frac{1}{st}\mathbb{M}(\mu+\nu)+\frac{1-s}{s}\sum_{x\in \X}(\mu(x)^2 + \nu(x)^2) \right)^{\frac{1}{p}}, \\[0.1cm]  \hspace{2cm} \text{else}.
    \end{cases}
\end{align*}
\end{corollary}

The deviation bound established in \Cref{cor:estimation_KRD} is sharp in the sense that they are optimal up to constants for the regime where $\mu = \nu$ (see \Cref{subsec:rate_optimality}). Moreover, even if $\mu$ and $\nu$ differ but are permitted to be arbitrarily close, the convergence rates are sharp. However, if the measures are strictly separated, the rates are suboptimal, an observation which was previously made in the context of empirical optimal transport \citep{chizat2020, manole2024sharp}. The subsequent deviation bound provides a refinement for the empirical KRD in this regime and is based on the stability bound via duality (Lemma \ref{lem:stability_dual}).

\begin{theorem}[Expected deviation of empirical KRD under separation]\label{thm:estimation_KRD_separate}
Consider the same setting as in Corollary \ref{cor:estimation_KRD} and assume additionally that $\KR_{p,C}(\mu, \nu)\geq C \delta >0$. Then it follows that 
\begin{align*}
    &\mathbb{E}\left[ \lvert \KR_{p,C}\left( \hat{\mu}_{t,s},\hat{\nu}_{t,s}\right) -  \KR_{p,C}\left( {\mu},{\nu}\right)\rvert \right]\\
       \leq \, 
    & \frac{\delta^{1-p} C}{2}\begin{cases}
       4 \left(2(1-s)\mathbb{M}(\mu+\nu)+\frac{s}{\sqrt{t}}\sum_{x\in \X}(\sqrt{\mu(x)} + \sqrt{\nu(x)}) \right)^{\frac{1}{2}},\\[0.1cm]  \hspace{2cm} \text{if } C> \min_{\substack{x\in\mathrm{supp}(\mu)\\ x' \in \mathrm{supp}(\nu)}}d(x,x^\prime), \\
        \left( \frac{1}{st}\mathbb{M}(\mu+\nu)+\frac{1-s}{s}\sum_{x\in \X}(\mu(x)^2 + \nu(x)^2) \right)^{\frac{1}{2}},\\[0.1cm]  \hspace{2cm} 
        \text{else}.
    \end{cases}
\end{align*}
\end{theorem}
\begin{proof}
The result follows from Lemma \ref{lem:stability_dual} by plugging in the bounds of $\TV(\hat{\mu}_{t,s}, \mu)$ and $\mathbb E[|\mathbb{M}(\hat{\mu}_{t,s}) - \mathbb{M}(\mu)|]$, obtained in \Cref{thm:TVbound} and in \eqref{eq:mass_convergence}.
\end{proof}

In the formulation of the statement, we impose the condition that $\KR_{p,C}(\mu, \nu) \geq C \delta>0$. This condition is motivated from the fact that $\KR_{p,C}(\mu, \nu) = 2^{-1/p} C \textup{TV}(\mu, \nu)^{1/p}$ for $C< \min_{x\neq x'}d(x,x')$. In particular, this way, we obtain a linear scaling in $C$ in the upper bound of \Cref{thm:estimation_KRD_separate}, which captures the correct scaling (see Section \ref{subsec:rate_optimality}).

\begin{remark}To compare the established convergence results in Corollary \ref{cor:estimation_KRD} and Theorem \ref{thm:estimation_KRD_separate}, set $s = 1$ and note that if $\mu$ is close to $\nu$, the convergence rate of the empirical KRD is of order $t^{-1/2p}$, whereas under a strict separation constraint the rate is instead of order $t^{-1/2}$. 
  This is a consequence of the fact that the $p$-th root functional $x \mapsto x^{1/p}$ is Lipschitz-continuous at the origin if and only if $p = 1$. 
\end{remark}

\subsection{Rate Optimality}\label{subsec:rate_optimality}
In the following, we provide some intuition and consequences of the deviation bound provided in \Cref{thm:samplingboundKR}. The term $\phi(s,t)$ must necessarily contain a sum of terms depending on $s$ and $t$, respectively. In particular, neither choosing $s=1$ nor letting $t$ go to infinity for $s<1$, would yield a zero error. For any fixed $t >0$ and $s=1$ the expected $(p,C)$-KRD error is clearly non-zero as the mass of the measures at each location is in general not estimated correctly. Similarly, for any fixed $s<1$ letting $t\rightarrow \infty$ can not yield an expected $(p,C)$-KRD error of zero as on average $(1-s)\lvert \X \rvert>0$ support points of $\mu$ are not observed. However, for $s=1$ the error vanishes for $t\rightarrow \infty$, as we observe all support points of $\mu$ and then the strong law of large numbers guarantees the convergence of the weights at each location. It remains to verify whether the rate in $t$ is optimal. For this, fix $s=1$ and observe that
\begin{align}\label{eq:naivetv}
  \min\{C,\min_{x\neq x^\prime}d(x,x^\prime)\}^p \TV(\mu,\nu)\leq  \KR_{p,C}^p(\mu,\nu)\leq C^p \TV(\mu,\nu).
\end{align}
To show that the rate $t^{-\frac{1}{2}}$ in \Cref{thm:samplingboundKR} is sharp, we prove that $t^{\frac{1}{2}} \TV(\mu,\hat{\mu}_{t,1})$ converges in distribution for $t \to \infty$ to a non-degenerate distribution. By combining the central limit theorem for Poisson random variables in conjunction with the continuous mapping theorem, it follows for $t \to \infty$ that 
\begin{align*}
  t^{\frac{1}{2}}\TV(\mu,\hat{\mu}_{t,1})  = t^{\frac{1}{2}} \sum_{x\in \XC} | \mu(x) - \hat{\mu}_{t,1}(x)| = t^{-\frac{1}{2}} \sum_{x\in \XC} | t \mu(x) - P_{x, t}|  \xrightarrow{\mathcal{D}} \sum_{x \in \XC} \sqrt{\mu(x)} |Z_x|,
\end{align*}
where $P_{x,t} \sim Poi(\mu(x)t)$ and $Z_x\sim \mathcal{N}(0,1)$ with $x\in \XC$ are jointly independent random variables. Hence, by the Portemanteau lemma for the function $x \mapsto x^{1/p}$, we conclude that 
\begin{align*}
 \liminf_{t \to \infty}\EE[ t^{\frac{1}{2p}}\TV(\mu,\hat{\mu}_{t,1})^{\frac{1}{p}}]  \geq \;& \EE\Big[\Big(\sum_{x \in \XC} \sqrt{\mu(x)} |Z_x|\Big)^{\frac{1}{p}}\Big] > 0,
\end{align*}
In conjunction with \eqref{eq:naivetv}, the expectation $\mathbb{E}\left[\KR_{p,C}(\mu,\hat{\mu}_{t,1}) \right]$ is decreasing at least with order $Ct^{-\frac{1}{2p}}$ and the rate in $t$ of \Cref{thm:samplingboundKR} is thus sharp.

To illustrate sharpness of the convergence statement \Cref{thm:estimation_KRD_separate} in the regime of different measures, take $\nu = 0$, which yields $\hat \nu_t = \nu = 0$ for all $t>0$, and note that for every measure $\mu$ on $\XC$ it follows that $$\KR_{p,C}^p(\mu,\nu) = \frac{C^p}{2}\mathbb{M}(\mu).$$ 
Then, by the central limit theorem for Poisson random variables in conjunction with the delta method for $x \mapsto x^{1/p}$ and the continuous mapping theorem for the absolute value function, it follows if $\mathbb{M}(\mu)>0$ for $t \to \infty$ that 
\begin{align*}
  t^{\frac{1}{2}}\left|\KR_{p,C}(\hat\mu_{t,1},\hat\nu_{t,1}) - \KR_{p,C}(\mu,\nu) \right| =  \frac{t^{\frac{1}{2}}C}{2^{1/p}} \left|\mathbb{M}(\mu)^{1/p} - \mathbb{M}(\hat\mu_{t,1})^{1/p} \right| \xrightarrow{\mathcal{D}} \frac{C \mathbb{M}(\mu)^{-\frac{1}{2}+\frac{1}{p}}}{2^{1/p}} |Z| %
\end{align*}
for $Z \sim \mathcal{N}(0,1)$. Again invoking the Portemanteau theorem thus yields that 
\begin{align*}
  \liminf_{t \to \infty}\EE\left[ t^{\frac{1}{2}}\left|\KR_{p,C}(\hat\mu_{t,1},\hat\nu_{t,1}) - \KR_{p,C}(\mu,\nu) \right|\right] \geq  \frac{C}{\sqrt{\pi}} \left(\frac{\mathbb{M}(\mu)}{2}\right)^{-\frac{1}{2}+\frac{1}{p}},
\end{align*}
which shows that the convergence rate of order $Ct^{-1/2}$ in \Cref{thm:estimation_KRD_separate} is sharp. 

\subsection{Explicit Bounds for Euclidean Spaces}\label{sec:explicitB}
While the constants in the previous theorem are valid on arbitrary metric spaces, more explicit bounds can be derived for many practical applications. Thus, we assume that for the $\varepsilon$-covering number of $\XC$, there exists a constant $A > 0$ such that  %
\begin{equation} \label{eq:Xcovering_ass}
    \mathcal{N}(\XC,\varepsilon \cdot \mathrm{diam}(\XC)) \leq \min(A \varepsilon^{-\alpha}, |\XC|)  \quad \text{for all } \varepsilon \in (0,1].
\end{equation}
This assumption covers, for instance, the  setting when $\X$ is a finite subset of an $\alpha$-dimensional submanifold of $\RR^D$ and with $d$ chosen as the Euclidean distance $d_2$. Notably, if the domain is the unit ball $\{x\in \RR^D |\, \|x\|\leq 1\}$, then $\alpha = D$ and $A = 2$ are viable choices \citep[Corollary 4.2.11]{vershynin2018high}.  
We fix $q=2$ for simplicity. %
By repeating the argument within this framework, we can compute explicit upper bounds on the constants in \Cref{thm:samplingboundKR}. For $\alpha<2p$ and $L\to \infty$, it holds %
\begin{align*}
    \mathcal{E}_{p,\XC,\mu}^{\mathrm{Poi}}(C)\leq \begin{cases}
        \frac{C^p}{2} - 2^{2p-1} \mathrm{diam}(\XC)^p + A^{1/2}\mathrm{diam}(\XC)^{p} 2^{3p-1} \frac{2^{\alpha/2-p}}{1 - 2^{\alpha/2-p}}, \\[0.1cm] \hspace{2cm} \text{if } C\geq 2h_{L}(0),\\[0.4cm]
        A^{1/2} 2^{3p-1} \mathrm{diam}(\XC)^{p} \frac{2^{(\alpha/2-p)l}}{1 - 2^{\alpha/2-p}}, \\[0.1cm] \hspace{2cm} \text{if } 2h_{L}(l)\leq C< 2h_{L}(l-1),\\[0.4cm]
        \frac{C^p}{2}, \\[0.1cm] \hspace{2cm}\text{if } C\leq \left( 2h_{L}(L) \wedge d_{\min}\right).
    \end{cases}
\end{align*}
For $\alpha = 2p$ we put $L = \lfloor \frac{1}{\alpha}\log_2 (|\XC|) \rfloor$ and get
\begin{align*}
    \mathcal{E}_{p,\XC,\mu}^{\mathrm{Poi}}(C)\leq \begin{cases}
        \frac{C^p}{2} - 2^{p-1} \left(2 - |\XC|^{-1/\alpha} \right)^p\mathrm{diam}(\XC)^p \\[0.1cm]
        \ + 2^{3p-1} \mathrm{diam}(\XC)^{p} \left(2^{-p} + A^{1/2} \frac{1}{\alpha} \log_2 (|\XC|) \right), \\[0.1cm] \hspace{2cm} \text{if } C\geq 2h_{L}(0),\\[0.4cm]
        2^{3p-1} \mathrm{diam}(\XC)^{p}  \left(2^{-p} + A^{1/2}  \left( \frac{1}{\alpha} \log_2 (|\XC|) - l\right)\right), \\[0.1cm] \hspace{2cm} \text{if } 2h_{L}(l)\leq C< 2h_{L}(l-1),\\[0.4cm]
        \frac{C^p}{2}, \\[0.1cm] \hspace{2cm}\text{if } C\leq \left( 2h_{L}(L) \wedge {d_{\min}}\right).
    \end{cases}
\end{align*}
For $\alpha>2p$ and $L=\lfloor \frac{1}{\alpha}\log_2(|\XC|) \rfloor$, it holds
\begin{align*}
    \mathcal{E}_{p,\XC,\mu}^{\mathrm{Poi}}(C)\leq \begin{cases}
        \frac{C^p}{2} - 2^{p-1} \left(2 - |\XC|^{-1/\alpha} \right)^p\mathrm{diam}(\XC)^p \\
        \ +  2^{3p-1} \mathrm{diam}(\XC)^{p} |\XC|^{1/2 - p/\alpha} \left(2^{-p} + A^{1/2}  \frac{2^{\alpha/2-2p}}{2^{\alpha/2-p}-1} \right), \\[0.1cm] \hspace{2cm} \text{if } C\geq 2h_{L}(0),\\[0.4cm]
        2^{3p-1}\mathrm{diam}(\XC)^{p} \left(2^{-p}|\XC|^{1/2-p/\alpha} \right. \\
        \ + \left.A^{1/2} \frac{2^{\alpha/2-p}}{2^{\alpha/2-p} - 1} \left(|\XC|^{1/2-p/\alpha} - 2^{(\alpha/2-p)(l-1)} \right)\right), \\[0.1cm] \hspace{2cm} \text{if } 2h_{L}(l)\leq C< 2h_{L}(l-1),\\[0.4cm]
        \frac{C^p}{2}, \\[0.1cm] \hspace{2cm}\text{if } C\leq \left( 2h_{L}(L) \wedge {d_{\min}}\right).
    \end{cases}
\end{align*}
These upper bounds depend on the penalty $C$ as well as the parameters $\alpha$ and $A$ which dictates the behavior of the covering number $\mathcal{N}(\XC,\varepsilon)$ of $\XC$ at difference scales $\varepsilon$. %

If $\alpha<2p$, then there is no dependence on $|\XC|$ and the convergence of approximation error of the empirical measure is independent of its support size. If $\alpha=2p$, then $|\XC|$ enters through a logarithmic term. If $\alpha>2p$, then the dependence becomes polynomial in $|\XC|$. These phase transitions for the dependence on the support size align with those for empirical (balanced) OT for an arbitrary finite subset of $\RR^D$ \citep{sommerfeld2019optimal}, corresponding to $\alpha = D$. Our results also conceptually align with rate results for the empirical Wasserstein distance \citep{fournier2015rate, weed2019sharp}, namely that for $\alpha < 2p$ parametric convergence rates in the sample size and independent of the domain manifest. Moreover, for $\alpha>2p$ some polynomial dependency in $|\XC|$ has to manifest, since otherwise it would imply generic parametric convergence rates in the sample size for the empirical Wasserstein distance for $\alpha> 2p$, which would contradict well-known lower bounds \citep{singh2018minimax,weed2019sharp}. We conjecture the dependency in $|\XC|$ for the different regimes to be sharp. %

A novelty for the UOT setting is the additional dependency on the different scales of $C$. This is explained by the previously discussed control of $C$ on the maximal distance at which transport occurs in an optimal plan. The height function is again used to specify the scale induced by a particular choice of the parameter $C$. Notably, the dependence on $\lvert \X \rvert$ does not change on most scales of $C$. There is an exception, however, for sufficiently small values of $C$, where the $(p,C)$-KRD is equal to a scaled total variation distance. Thus, these bounds are completely oblivious to the geometry of $\X$ in $\Y$, though they scale as $\lvert \X \rvert^{\frac{1}{2}}$ which is the same rate we obtain for the TV bound (recall Theorem \ref{thm:TVbound}). As a final observation, we note that for $C>2h_L(0)$, these bounds essentially recover analogue bounds for empirical optimal transport \citep{sommerfeld2019optimal}. However, for the $(p,C)$-KRD the bounds include an additional summand based on the estimation error for the measure's total mass intensity.

\begin{remark}
Assumption \eqref{eq:Xcovering_ass} on $\XC$ can be relaxed in terms of $\varepsilon$. Namely, for $C \leq 2h_{L}(L) \wedge d_{\min}$ the assumption is not required because in this case the UOT cost coincides with the TV distance and the complexity of $\XC$ does not play a role. For  $C\geq 2 h_L(0)$, identical upper bounds  in $t>0$ remain valid if \eqref{eq:Xcovering_ass} is satisfied for $\varepsilon > t^{-1/(2 \vee \alpha)}$. This reflects a multiscale behavior of the empirical plug-in estimator previously observed for empirical OT \citep{weed2019sharp}: if the finitely supported measure is concentrated on an $\varepsilon$-fattened low-dimensional domain, then for small $t$ the constant from the low-dimensional setting will manifest, but for  $t\to \infty$ the constant degrades due to the fattening. %
\end{remark}

\section{Empirical Unbalanced Optimal Transport Plans}\label{sec:UOT_plans}

In this section we derive novel quantitative convergence statements for empirical UOT plans. 
We rely on a novel stability result for the balanced setting, which we also include as \Cref{thm:stability_OTplan} as it might be interesting on its own; the proof is stated in \Cref{app:stab_OTplan_proofs}.
The rest of the omitted proofs can be found in \Cref{app:UOTplans}. 

To set notation, we denote the collection of UOT plans between measures $\mu, \nu\in \msrX$ for parameters $p\geq 1$ and $C\geq 0$ by
\begin{align}\label{eq:UOTplans_unrestricted}
  {\overline{\mathbf{P}}}^*_{p,C}(\mu, \nu) \coloneqq \argmin_{\pi\in\Pi_{\leq}(\mu, \nu)}\sum_{x,x'\in \X}d^p(x,x')\pi(x,x') + C^p\left(\frac{\mathbb{M}(\mu) + \mathbb{M}(\nu)}{2} -  \mathbb{M}(\pi)\right).
\end{align}
Moreover, we additionally define the domain $\mathcal{D}(C)\coloneqq \{ (x,x') \in \XC\times \XC | d(x,x') < C\}$ and consider the set of restricted UOT plans
\begin{align}\label{eq:UOTplans}
  {\mathbf{P}}^*_{p,C}(\mu, \nu) \coloneqq \left\{ \pi \cdot \mathds{1}(\cdot\in \mathcal{D}(C))  \,\Big|\, \pi \in {\overline{\mathbf{P}}}^*_{p,C}(\mu, \nu) \right\}.
\end{align}
Lemma \ref{lem:UOTplan_restriction_connection} in Appendix \ref{app:UOTplans} shows that ${\mathbf{P}}^*_{p,C}(\mu, \nu)$ is a subset of ${\overline{\mathbf{P}}}^*_{p,C}(\mu, \nu)$, and that every  sub-coupling in $\Pi_{\leq}(\mu, \nu)$ which is the sum of an element in ${\overline{\mathbf{P}}}^*_{p,C}(\mu, \nu)$ and a non-negative measure supported on $\{ (x,x') \in \XC\times \XC \,|\, d(x,x') = C\}$ is contained in ${\overline{\mathbf{P}}}^*_{p,C}(\mu, \nu)$. 
 Based on this insight, we restrict our attention to the set of restricted UOT plans for empirical and population measures, 
\begin{align*}
  {\mathbf{P}}^*_{p,C}(\hat\mu_{t,s},\hat\nu_{t,s}) %
  \quad \text{ and } \quad 
  {\mathbf{P}}^*_{p,C}(\mu, \nu),%
\end{align*}
and quantify the accuracy in estimating the latter in terms of the former via the Hausdorff distance induced by the total variation norm. 
Crucial for our analysis is the following novel stability bound of the UOT plans. %
\begin{theorem}[Stability bound for UOT plans]
  \label{thm:UOTplan_Stability}
Let $(\X,d)$ be a finite metric space and $\mu^1,\mu^2,\nu^1,\nu^2\in\msrX$. Then, for any $p\geq 1$ and $C\geq 0$ it holds that 
\begin{align*}
& \mathcal{H}_\TV\left({\mathbf{P}}^*_{p,C}(\mu^1,\nu^1),{\mathbf{P}}^*_{p,C}(\mu^2,\nu^2)\right)\\
\leq \;& 4\,(|\X| +1)\left(\TV(\mu^1,\mu^2) + |\mathbb{M}(\mu^1) - \mathbb{M}(\mu^2)| +\TV(\nu^1,\nu^2) + |\mathbb{M}(\nu^1) - \mathbb{M}(\nu^2)|\right),
\end{align*}
where $\mathcal{H}_\TV$ denotes the Hausdorff distance induced by the total variation norm.
\end{theorem}

The proof is based on relating the collection of (restricted) UOT plans to associated OT plans for suitably augmented measures and a stability bound for OT plans stated in \Cref{thm:stability_OTplan}. The latter follows from a general stability bound for optimal solutions of linear program theory based on \cite{li1994sharp}. A remarkable property of the above stability bound is that it does not depend on parameters $p$ or $C$. In fact, the assertion also remains valid for any other cost function besides of $c = d^p$, and might be of independent interest. This stability bound asserts at the main result of this section. 

\begin{theorem}[Expected deviation of empirical UOT plans]\label{thm:UOTplan_convergence}
  Let $(\mathcal{X},d)$ be a finite metric space and $\mu, \nu\in\mathcal{M}_+(\mathcal{X})$. Let $\hat{\mu}_{t,s}, \hat{\nu}_{t,s}$ be the estimator from \eqref{eq:poissonmeasure}. Then, for any $p\geq 1$ and $C\geq 0$ it follows that 
  \begin{align*}
    & \mathbb{E}\left[\sup_{p\geq 1, C>0}\mathcal{H}_\TV\left({\mathbf{P}}^*_{p,C}(\hat\mu_{t,s},\hat\nu_{t,s}),{\mathbf{P}}^*_{p,C}(\mu,\nu)\right)\right] \\
    \leq \;&  4\,(|\X|+1) \Bigg( 2(1-s)\mathbb{M}(\mu + \nu)+\frac{s}{\sqrt{t}}\sum_{x\in \X}\left[\sqrt{\mu(x)} + \sqrt{\nu(x)}\right] \\ & \quad \quad \quad \quad\;\;\; + \sqrt{\frac{1}{st} \mathbb{M}(\mu+  \nu)+\frac{1-s}{s} \sum_{x\in \X} (\mu(x)^2  + \nu(x)^2)  }  \Bigg).
  \end{align*}
\end{theorem}
\begin{proof}
The assertion follows by combining the stability bound from \Cref{thm:UOTplan_Stability} above with the bound on the total variation distance between the estimators $\hat{\mu}_{t,s}, \hat{\nu}_{t,s}$ and $\mu, \nu$ (Theorem \ref{thm:TVbound}) and on the absolute deviation between their masses \eqref{eq:mass_convergence}.
\end{proof}

\begin{remark}[Computation]
Thanks to the connection between UOT and balanced OT (see Appendix \ref{app:liftOT}), various polynomial time procedures exist to compute the UOT plan between  plug-in estimators $\hat \mu_{t,s}$ and $\hat \nu_{t,s}$. For exact computation, the auction algorithm \citep{bertsimas1997introduction} or the network simplex flow algorithm \citep{luenberger1984linear, bonneel2011displacement} can be used which admit a computational complexity of order $\mathcal{O}(N^3\log(N))$ where $N$ is the number of support points. Moreover, to approximate the UOT plan up to some precision $\varepsilon>0$, the Sinkhorn algorithm for the entropy regularized OT problem \citep{cuturi2013sinkhorn, peyre2019computational} provides a computationally effective method which scales with order $\mathcal{O}(N^2 \log(N)/\varepsilon^2)$, see \cite{altschuler2017near,dvurechensky2018computational} as well as \cite{weed2018explicit} for an explicit analysis on the difference between unregularized and entropy regularized OT cost. Notably, incorporating additionally the structure of the UOT plan, any such algorithm can be improved by restricting the measures on the domain $\mathcal{D}(C)$. 
\end{remark}

For our stability bound of the UOT plan we establish a novel stability bound for the vanilla OT plan between measures with equal mass which might be of independent interest. To formalize these results, we employ the following notation.   Let $\mu, \nu \in \mathcal{M}_+(\X)$ be measures with identical mass $\mathbb{M}(\mu) = \mathbb{M}(\nu)$, let $c\colon \XC\times \XC \to [0, \infty)$ be a cost function and denote the respective collection of OT plans as \begin{align}\label{eq:balancedOT_problem}
  \tilde{\mathbf{P}}_{c}^*(\mu, \nu) \coloneqq \argmin_{\pi\in \Pi_{=}(\mu, \nu)} \sum_{x,x'\in \X} c(x,x')\pi(x,x').
\end{align}
Note that the assumption of identical masses ensures that the collection of transport plans is always non-empty and compact, hence $\tilde{\mathbf{P}}_{c}^*(\mu, \nu)$ is always non-empty. 
\begin{theorem}[Stability bound for balanced OT plans]\label{thm:stability_OTplan}
  Let $\X$ be a finite discrete space and $\mu_1, \nu_1, \mu_2, \nu_2\in \mathcal{M}_+(\X)$ such that $\mathbb{M}(\mu_1) = \mathbb{M}(\nu_1)$ and $\mathbb{M}(\mu_2) = \mathbb{M}(\nu_2)$. Then, for every cost function $c\colon \XC\times\XC\to [0,\infty)$ it holds that 
  \begin{align*}
    & \mathcal{H}_\TV\left({\tilde{\mathbf{P}}}^*_{c}(\mu_1,\nu_1),{\tilde{\mathbf{P}}}^*_{c}(\mu_2,\nu_2)\right) \leq  4\,|\X| \left(\TV(\mu_1,\mu_2) +\TV(\nu_1,\nu_2)\right).
  \end{align*}
\end{theorem}
The proof is based on linear program theory and relies on a stability bound by \cite{li1994sharp}.

This quantitative error bound represents, to the best of our knowledge, the first of its kind for finite discrete domains. Moreover, thanks to the generality of our stability bound, we are not limited to the setting $p = 2$, which has been the main subject of analysis in Euclidean settings for continuous settings \citep{gigli2011holder,  deb2021, hutter2021minimax, delalande2023quantitative, manole2021plugin} as well as semi-discrete settings  \citep{bansil2022quantitative,divol2024tight}. Notably, these former results are concerned with the OT map, whereas our Theorem \ref{thm:stability_OTplan} is concerned with the OT plan. It remains open whether similar deviation bound for the OT plan can also be established for continuous settings. Since the total variation norm metrizes strong convergence, we believe that similar results are will not be valid for empirical plug-in estimators on continuous domains but can be achieved when choosing a weaker loss and smooth plug-in estimators. We leave this for future research. 

\begin{remark}
Let us address some aspects of our stability bound from \Cref{thm:stability_OTplan}. \begin{enumerate}
		\item The key insight of \Cref{thm:stability_OTplan} is the fact that they behave fairly stable under small perturbation of the marginal measures with respect to $TV$-norm. Remarkably, the underlying cost function does not affect the upper bound, which may seem surprising given that the OT plan crucially depends on the cost function (see,  e.g., \cite{villani2008optimal}).
		\item For $\mu^1 = \mu^2 = \delta_{x}$   and arbitrary measures $\nu^1, \nu^2\in \mathcal{M}_+(\XC)$  it follows that ${\tilde{\mathbf{P}}}^*_{c}(\mu_i,\nu_i) = \{\mu_i\otimes \nu_i\}$. Therefore,
        $$\mathcal{H}_\TV\left({\tilde{\mathbf{P}}}^*_{c}(\mu_1,\nu_1),{\tilde{\mathbf{P}}}^*_{c}(\mu_2,\nu_2)\right) = \TV(\mu_1\otimes \nu_1, \mu_2\otimes \nu_2) = \TV(\nu_1, \nu_2),$$
		and thus the dependency in \Cref{thm:stability_OTplan} with respect to the marginal measures is sharp. 
		\item The upper bound only scales linearly in the number of elements of the domain $\XC$. If all measures $\mu^1,\nu^1, \mu^2, \nu^2$ are concentrated on a subdomain $\XC' \subseteq \XC$, the constant could be replaced by $|\XC'|$.  
		 Obtaining the sharp dependency in the number of support points remains an interesting aspect for future work.  
		 \item Based on distributional limits for the empirical OT plan between discrete measures \citep{klatt2020limit, liu2023asymptotic} it follows that the deviation bound in \Cref{thm:UOTplan_convergence} is sharp in $s$ and $t$ up to multiplicative constants which depend on $\XC$, $\mu$ and $\nu$. 
		\item When perturbing the cost function a similar stability bound for OT plans can generally not be expected. As an  example, take identical probability measures $\mu=\nu$ which are not concentrated on a single point and consider $c_a(x,y) = a\cdot  d(x,y)$ as the cost function where $a\in [0,\infty)$ and $d$ is a metric on $\XC$. Then, for  $a>0$ it follows that  $\tilde{\mathbf{P}}^*_{c_a}(\mu, \nu ) = \{ (\textup{Id}, \textup{Id})_{\#}\mu \}$ whereas for $a = 0$ it holds $\tilde{\mathbf{P}}^*_{c_0}(\mu, \nu ) = \Pi(\mu, \nu)$. This confirms that the set of OT plans is discontinuous with respect to the cost function since 
		\begin{align*}
 \lim_{a \searrow 0 } \mathcal{H}_\TV\left(\tilde{\mathbf{P}}^*_{c_a}(\mu, \nu ),\tilde{\mathbf{P}}^*_{c_0}(\mu, \nu )\right)>0 \quad \text{ while } \quad \lim_{a \searrow 0 }\|c_a - c_0\|_\infty = 0.
\end{align*}
	\end{enumerate}
\end{remark}

\section{Empirical Kantorovich-Rubinstein Barycenters}\label{sec:KR_barycenter}
 Consider measures $\mu^1,\ldots,\mu^J\in\msrX$ which we replace by $\hat{\mu}^1_{t_1,s_1},\ldots,\hat{\mu}^J_{t_J,s_J}\in\msrX$ as defined in \eqref{eq:poissonmeasure}. We again focus on the Poisson model and treat the remaining two models in Appendices \ref{sec:multi} and ref{sec:ber}. The previous upper bound on the Kantorovich-Rubinstein distance in \Cref{thm:samplingboundKR} between a measure and its empirical version enables a bound on the mean absolute deviation of $(p,C)$-barycenters in terms of their $p$-Fr\'echet functional $F_{p,C}(\mu)=\frac{1}{J}\sum_{i=1}^J \text{KR}_{p,C}^p(\mu^i,\mu)$ from \eqref{eq:KRbarycenter}. We denote 
\begin{align*}
    \mu^\star\in \argmin_{\mu\in\mathcal{M}_+(\Y)} \frac{1}{J} \sum_{i=1}^J \KR_{p,C}^p(\mu^i,\mu),\quad \hat{\mu}^\star\in \argmin_{\mu\in\mathcal{M}_+(\Y)} \frac{1}{J} \sum_{i=1}^J \KR_{p,C}^p(\hat{\mu}_{t_i,s_i}^i,\mu)
\end{align*}
and measure the accuracy of approximation of $\mu^\star$ by $\hat{\mu}^\star$ in terms of their mean absolute $p$-Fr\'echet deviation. The omitted proofs for this section are detailed in \Cref{app:barycenters_proofs}.

\begin{theorem}[Expected deviation of Fr\'echet error]\label{thm:frechetboundPoi}
Let $\mu^1,\ldots,\mu^J\in\msrX$ and denote $\mathcal{X}_i=\mathrm{supp}(\mu^i)$ for $i=1,\dots ,J$. Consider random estimators $\hat{\mu}_{t_1,s_1}^1,\ldots,\hat{\mu}_{t_J,s_J}^J\in\msrX$ derived from \eqref{eq:poissonmeasure}. Then,
\begin{align*}
    \mathbb{E}\left[\lvert F_{p,C}(\hat{\mu}^\star)-F_{p,C}(\mu^\star) \rvert \right] \leq \frac{2p C^{p-1}}{J}\easysum{i=1}{J} \mathcal{E}_{1,\X_i,\mu^i}^{\mathrm{Poi}}(C)\phi(t_i,s_i),
\end{align*}
where $\phi$ is given by
\begin{align*}
    \phi(t,s)=\begin{cases}
   \left(2(1-s)\mathbb{M}(\mu)+\frac{s}{\sqrt{t}}\sum_{x\in \X}\sqrt{\mu(x)} \right), &\text{if } C\leq \min_{x\neq x^\prime}d(x,x^\prime), \\
  \left( \frac{1}{st}\mathbb{M}(\mu)+\frac{1-s}{s}\sum_{x\in \X}\mu(x)^2 \right)^{\frac{1}{2}}, &\text{else}. \\
    \end{cases}
\end{align*}
\end{theorem}

A more elaborate statement gives control over the set of empirical $(p,C)$-barycenters itself. This involves a related linear program that is presented in detail in \Cref{app:sec_LP}.
\begin{theorem}[Expected deviation of barycenters]\label{thm:kr_boundPoi}
Let $\mu^1,\ldots,\mu^J\in\msrX$ and denote $\mathcal{X}_i=\mathrm{supp}(\mu^i)$ for $i=1,\dots ,J$. Consider random estimators $\hat{\mu}_{t_1,s_1}^1,\ldots,\hat{\mu}_{t_J,s_J}^J\in\msrX$ derived from \eqref{eq:poissonmeasure}. Let $\mathbf{B}^\star$ be the set of $(p,C)$-barycenters of $\mu^1,\dots,\mu^J$ and $\hat{\mathbf{B}}^\star$ the set of $(p,C)$-barycenters of $\hat{\mu}^1_{t_1,s_1},\ldots,\hat{\mu}^J_{t_J,s_J}$. Then, for $p\geq 1$ it holds that
\begin{align*}
    \mathbb{E}\left[ \sup_{\hat{\mu}^\star \in \hat{\mathbf{B}}^\star} \inf_{\mu^\star \in \mathbf{B}^\star} \KR_{p,C}^p(\hat{\mu}^\star,\mu^\star) \right]
    \leq\frac{p C^{p-1} }{V_{P}J}\easysum{i=1}{J} \mathcal{E}_{1,\X_i,\mu^i}^{\mathrm{Poi}}(C)\phi(t_i,s_i),
\end{align*}
where $\phi$ is defined as in \Cref{thm:frechetboundPoi}. The constant $V_P$ is strictly positive and given by
\begin{align*}
V_{P}\coloneqq V_{P}(\mu^1,\dots,\mu^J):=(J+1)\diam(\mathcal{X})^{-p}\underset{v \in V\backslash V^\star}{\min} \ \frac{c^Tv-f^*}{d_1(v,\mathcal{M})},
\end{align*}
where $V$ is the set of feasible vertices from the linear program in \Cref{app:sec_LP}, $V^\star$ is the subset of optimal vertices, $c$ is the cost vector of the program, $f^\star$ is the optimal value, $\mathcal{M}$ is the set of minimizers of the linear program \eqref{eq:lpbary} and $d_1(x,\mathcal{M})=\inf_{y\in \mathcal{M}} \lVert x-y \rVert_1$. %
\end{theorem}

\begin{remark}
Let us comment on our convergence result for barycenters. 
\begin{enumerate}
	\item \Cref{thm:kr_boundPoi} implies that every empirical barycenter tends for $\min_{i \in \{1, \dots, J\}}t_i \to \infty$ and $\min_{i \in \{1, \dots, J\}}s_i \to 1$ with parametric rates towards a suitable population barycenter. This is practically relevant because for discrete domains, non-uniqueness of empirical and population barycenters generally cannot be ruled out. Insofar, our result guarantees that every estimator admits satisfactory convergence properties. 
	\item As an extension of \Cref{thm:kr_boundPoi} we deem it worthwhile to extend the analysis  to the Hausdorff distance $\mathcal{H}_{\KR_{p,C}}$ induced by the $(p,C)$-KRD, i.e., to establish convergence statements for empirical barycenters  $\hat{\mathbf{B}}^\star$ to population barycenters ${\mathbf{B}}^\star$, 
    \begin{align*}  \mathcal{H}_{\KR_{p,C}}^p\left(\hat{\mathbf{B}}^\star, {\mathbf{B}}^\star\right) = 
    \max\left(\sup_{\hat{\mu}^\star \in \hat{\mathbf{B}}^\star} \inf_{\mu^\star \in \mathbf{B}^\star} \KR_{p,C}^p(\mu^\star,\hat{\mu}^\star), \sup_{\mu^\star \in \mathbf{B}^\star} \inf_{\hat{\mu}^\star \in \hat{\mathbf{B}}^\star} \KR_{p,C}^p(\mu^\star,\hat{\mu}^\star) \right) .
\end{align*}
Such a convergence statement in  the Hausdorff distance would assert, in addition to point 1., that for every population barycenter there exists an empirical barycenter which approximates it. Arguing as in \Cref{app:barycenters_proofs}, one can see that \begin{align*}
    \hat V_P\sup_{\mu^\star \in \mathbf{B}^\star} \inf_{\hat{\mu}^\star \in \hat{\mathbf{B}}^\star} \KR_{p,C}^p(\mu^\star,\hat{\mu}^\star) \leq  \lvert \hat F_{p,C}(\hat{\mu}^\star)- \hat F_{p,C}(\mu^\star) \rvert,
\end{align*}
where $ \hat V_P$ is defined analogously to $V_P$ but with each $\mu^i$ replaced by the estimator $\hat\mu^i_{t_i, s_i}$. Controlling $ \hat V_P$ to infer quantitative convergence results, however, seems rather challenging, and we leave a refined analysis of the underlying constant for future research. 
\end{enumerate}
\end{remark}

\begin{remark}\label{rmk:BarycenterSharp}
For $J=1$ and any $p\geq 1, C>0$ the $(p,C)$-barycenter of $\mu^1$ is just $\mu^1$. Thus, the optimal value of the Fr\'echet functional is zero, and it holds 
\[
    F_{p,C}(\hat{\mu}^*)-F_{p,C}(\mu^*)=\KR_{p,C}^p(\mu^1,\hat{\mu}^1_{t,s}).
\]
Consequently, it also holds
\[
    \sup_{\hat{\mu}^\star \in \hat{\mathbf{B}}^\star} \inf_{\mu^\star \in \mathbf{B}^\star} \KR_{p,C}^p(\mu^*,\hat{\mu}^*)=\KR_{p,C}^p(\mu^1,\hat{\mu}^1_{t,s}).
\]
Thus, the rate for the convergence of the $(p,C)$-barycenter of the empirical measures, can in general not be faster than the convergence rate of a single estimator. In particular, the rates in $t$ in \Cref{thm:frechetboundPoi} and \Cref{thm:kr_boundPoi} are sharp.
\end{remark}

\begin{remark}
  While this work is only concerned with plug-in estimators for $(p,C)$-barycenters between general measures, there are also various other notions of barycenters between general measures \cite{chizat2018scaling,friesecke2021barycenters}. Exploring the statistical properties of these alternative barycenters presents an interesting venue for future research.
\end{remark}

\section{Application: Randomized Computation with Statistical Guarantees}\label{sec:applications}
 \label{sec:rand_comp}
In this section we discuss how our derived bounds enable randomized computations of UOT quantities with statistical guarantees. If the size of the population measures is computationally infeasible, then empirical versions of these measures can be used as a proxy for the population distance, plan and barycenter. Our bounds allow tuning the problem size (hence computational time) against the accuracy of the approximation. While  all three models allow this approximation approach, we exemplify this for the multinomial model. In this resampling scenario, the assumption of known total intensities amounts to have access to the full data set.
Here, the sample size $N$ provides a strict upper bound on the computational complexity of a given approximation. 
This is though not the case for the other two models, where such bounds can only be obtained by involving subsampling scheme instead of a resampling one (subsampling refers to sampling without replacement, while resampling refers to sampling with replacement).  We note however that subsampling approach is typically outperformed by the resampling one.

\begin{figure}[b!]
  \centering
  \begin{tabular}{ccc}
        \includegraphics[width=0.31\textwidth]{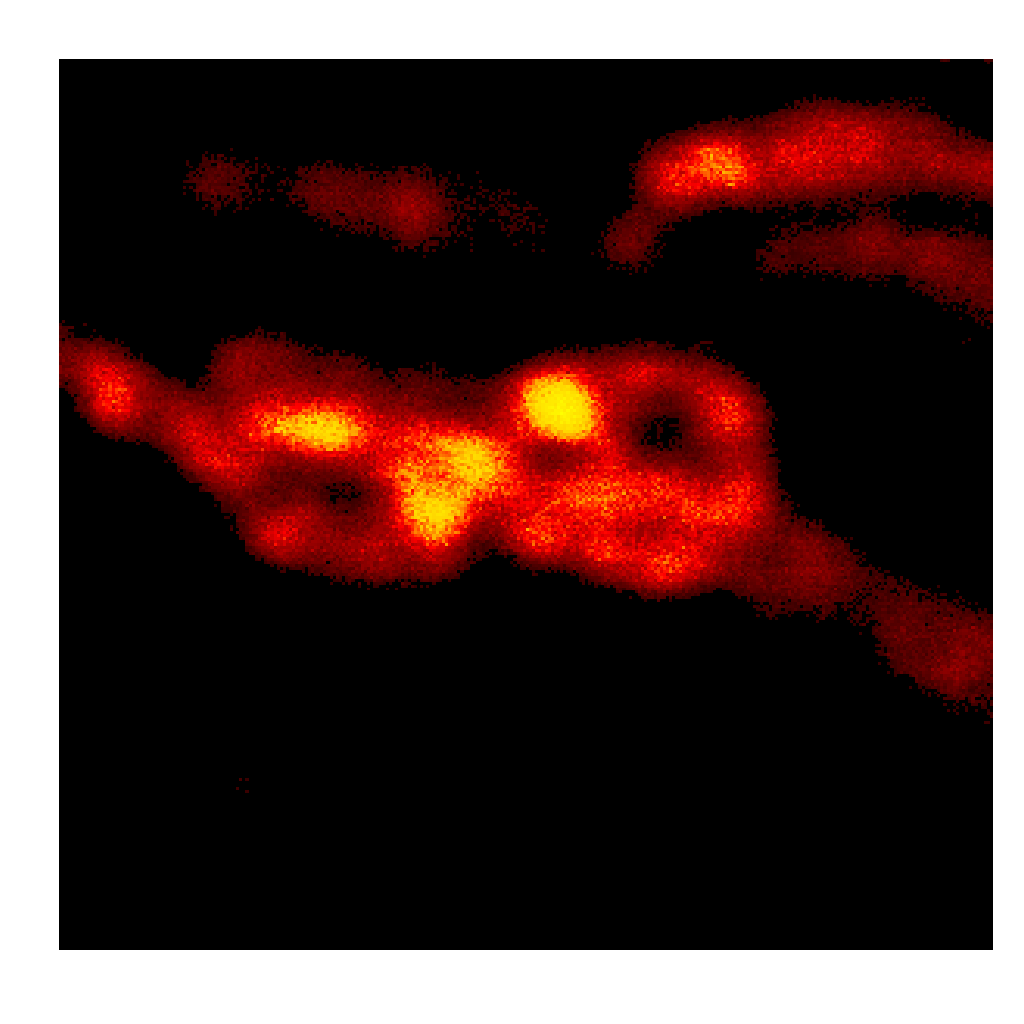} &     
        \includegraphics[width=0.31\textwidth]{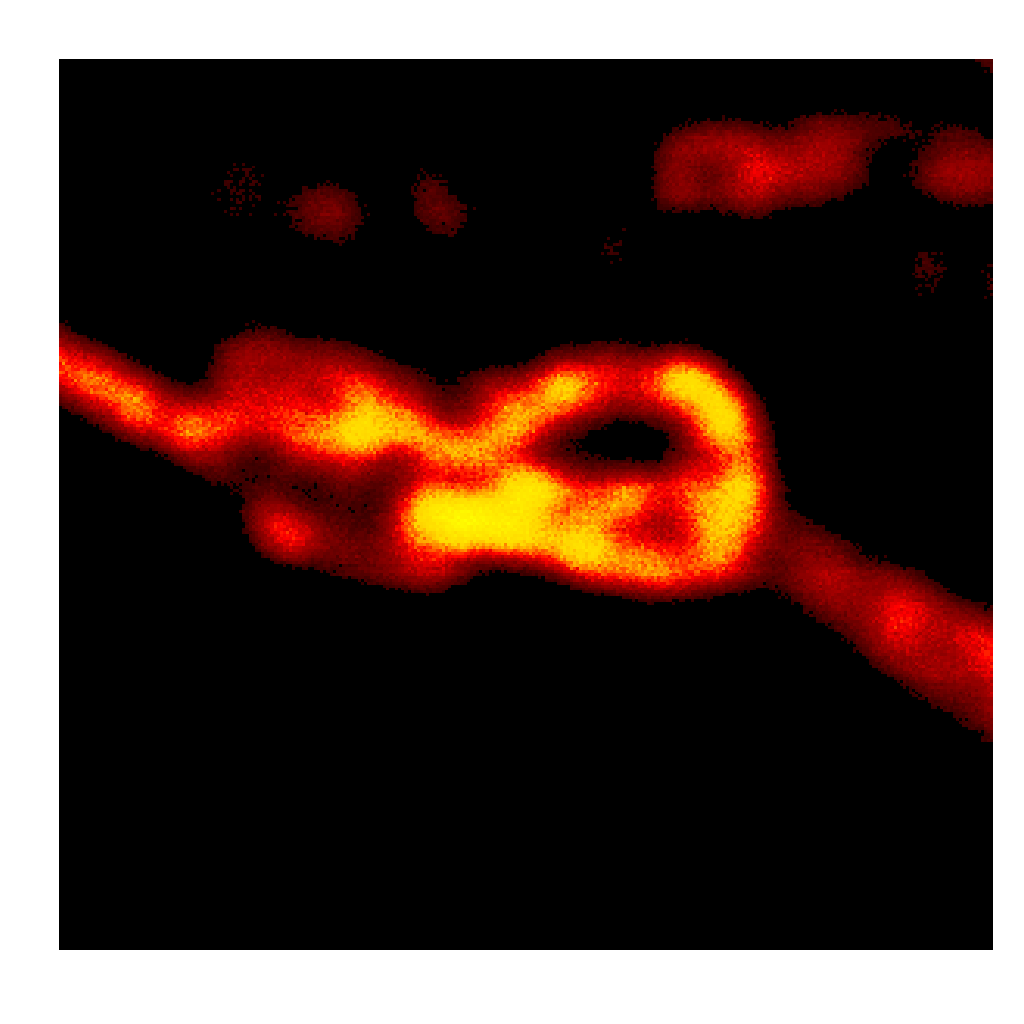} &
      \includegraphics[width=0.31\textwidth]{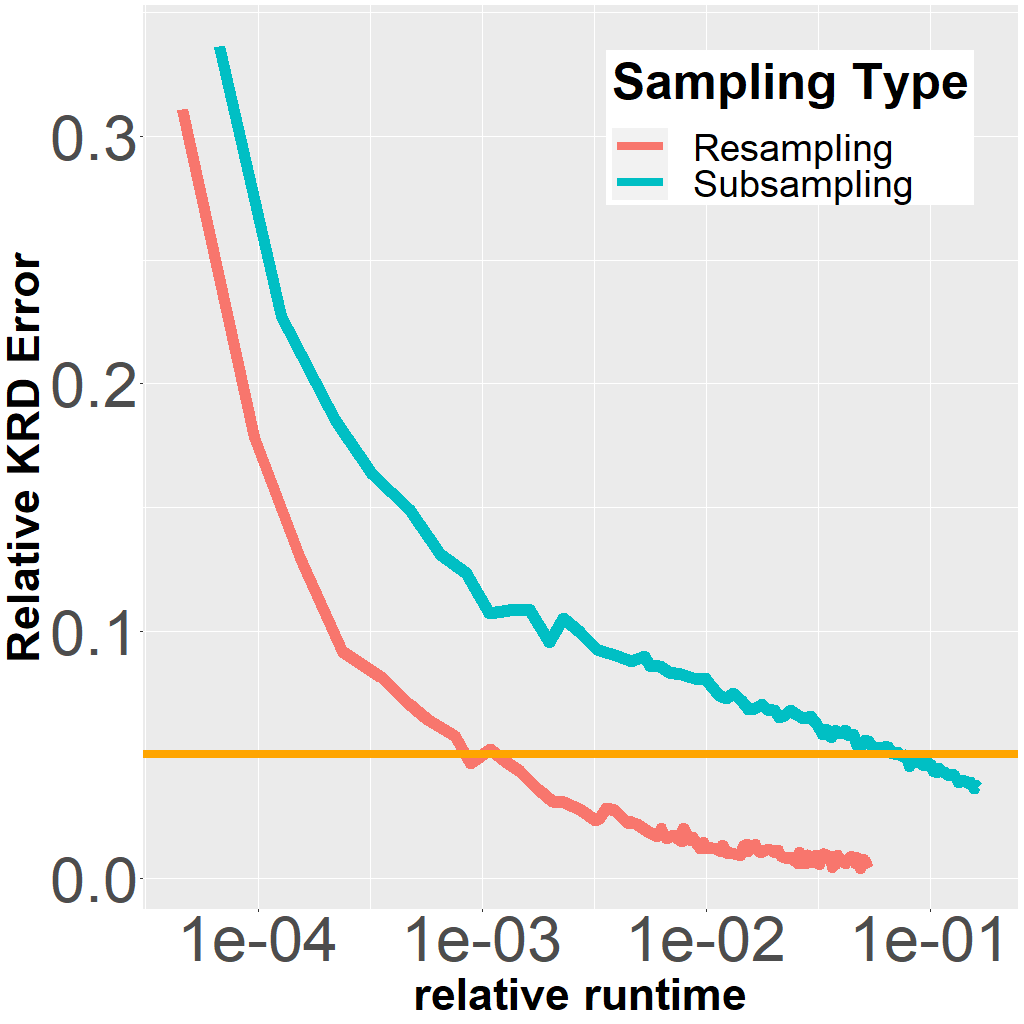} 
  \end{tabular}
  \caption{\textbf{Left:} An excerpt of size $300\times 300$ from the STED microscopy data of adult Human Dermal Fibroblasts labelled at MIC60 (a mitochondrial inner membrane complex, see \cite{tameling2021colocalization}). \textbf{Center:} The same type of data labelled at TOM20 (translocase of the outer mitochondrial membrane), see \cite{tameling2021colocalization}. \textbf{Right:} The expected relative $(2,0.1)$-KRD error curves obtained for the two images on the left and in the center from re- and subsampling for specified runtimes (computed with the CRAN package \emph{WSGeometry}). The orange line indicates an error level of $5\%$. The respective sample sizes are between $100$ and $9000$ for both approaches, while the original images have about $23000$ non-zero pixels.}
    \label{fig:re_vs_sub_time}
  \end{figure}

Indeed, for the subsampling scheme we replace the i.i.d.\ samples $X_1,\dots,X_N$ from $\mu$ by ones drawn without replacement. A natural choice of estimator is 
\begin{align}\label{eq:subsamp}
    \tilde {\mu}_{N}=\textstyle \frac{\mathbb{M}(\mu)}{\sum_{i=1}^N \mu(X_i)}\sum_{i=1}^N \mu(X_i)\delta_{X_i},
\end{align}
where the mass at each drawn location $x\in \X$ is proportional to the mass of the population measure at $x$ and the total mass intensity is rescaled to the known, true total intensity. This estimator is, due to the sampling without replacement, guaranteed to have $N$ support points, which yields close control on the required runtime for a given approximation. In recent years, this approach has become popular within the machine learning community where it is referred to as mini-batch OT \citep{fatras2021minibatch,nguyen2021transportation}. An illustration of the potential runtime advantages using the suggested randomized methods is displayed in \Cref{fig:re_vs_sub_time}. For this example, the resampling approach provides an expected relative KRD error of about $5\%$ while achieving a speedup of about a factor of $1000$ compared to the original runtime, while the subsampling approach requires nearly $10\%$ of the original runtime to achieve the same accuracy. A more detailed comparison of the empirical performance of the re- and subsampling model is found in \Cref{sec:realdata}. Though, we note that, in the considered data examples, the resampling approach consistently performed better than the subsampling one. We also study the convergence properties of the empirical measure and barycenter with respect to the $(p,C)$-KRD for all three described models in extended simulation studies on a wide range of synthetic datasets, again further in \Cref{sec:sims}. %

\section{Simulations}\label{sec:sims}
In this section we investigate empirically the decay in the expected error for the Poisson model for measures within $\X\subset [0,1]^2$. For the $(p,C)$-KRD we consider two measures $\mu,\nu \in \msrX$ and the \emph{relative $(p,C)$-KRD error}\footnote{We define $0/0:=0$.\label{fn0}}
\begin{align}\label{eq:relKR}
    \mathbb{E}\left[\left\lvert \frac{\KR_{p,C}(\hat{\mu}_{t,s},\hat{\nu}_{t,s})-\KR_{p,C}(\mu,\nu)}{\KR_{p,C}(\mu,\nu)} \right\rvert\right].
\end{align}
For the setting of barycenters we consider the \emph{relative $(p,C)$-Fr\'echet error}\cref{fn0}
\begin{align}\label{eq:relFre}
    \mathbb{E}\left[ \frac{F_{p,C}(\hat{\mu}^*)-F_{p,C}(\mu^*)}{F_{p,C}(\mu^*)} \right].
\end{align}
In both cases, the relative error allows for easier comparisons between models than the absolute error. In particular, since $\KR^p_{p,C}(\mu, \nu) = \frac{C^p}{2} \mathrm{TV}(\mu, \nu)$ \cite[Theorem 2]{heinemann2022kantorovich} for $C\searrow 0$ and $\mu \neq \nu$, it follows that the relative $(p,C)$-KRD error enables an easier comparison of the estimation error among different choices of $C$. Additionally, for the $(p,C)$-barycenter the relative $(p,C)$-Fr\'echet error in \eqref{eq:relFre} is readily available from simulations, while numerically considering the quantity in \Cref{thm:kr_boundPoi} is difficult, as it requires all optimal solutions instead of a single one. All computations of the KRD and the $(p,C)$-barycenter in this section are performed using the methods available in the CRAN package \href{https://github.com/cran/WSGeometry}{\emph{WSGeometry}}.

\subsection{Synthetic Datasets}
We consider multiple types of measures for our simulations. Below we describe four types; Appendix \ref{app:additionalSim} contains four additional types and the respective Poisson simulations are provided in \Cref{sec:addfigpoi}. Analogous simulations for Multinomial Sampling and Bernoulli Sampling are detailed in Appendices \ref{app:sim_mult} and \ref{app:sim_ber}, respectively. 

Let us fix some notation. Let $J\in \mathbb{N}$ be the number of measures generated. Let $U[0,1]^2$ denote the uniform distribution and let $\mathrm{Poi}(\lambda)$ denote a Poisson distribution with intensity $\lambda$. In all settings considered below, the measures are of the form
\[
\mu^i=\sum_{k=1}^{K_i}w^i_k\delta_{l^i_k}
\]
for some weights $w^i_k$, locations $l^i_k$ and $K_i\in \mathbb{N}$. If $K_i=K_j$ for all $i,j=1,\dots ,J$, then we omit the index and denote the number of points by $K$. Note that all measures have been constructed to have their support included in $[0,1]^2$.

\subsubsection*{Nested Ellipses (NE), see \Cref{fig:ellipsedata} (a)}
Let $G_1,\dots ,G_J\sim U\{1,2,3,4,5\}$ and let $K_i=MG_i$ for $M\in \mathbb{N}$. Set all $w_k^i$ equal to $1$ for each $1\leq k \leq K_i$ for $i=1,\dots,J$. Let $t_1,\dots ,t_M$ be a discretization of $[0,2\pi]$. Let $U^i_1,\dots,U^i_K,V^i_1,\dots,V^i_K\sim U[0.2,1]$. For $1\leq i\leq J$, set
\[
l^i_{M(j_i-1)+k}=0.5(1+3^{-j}(U_{M(j_i-1)+k}\sin(t_k),V_{M(j_i-1)+k}\cos(t_k))^T), \quad j_i=0,\dots ,G_i.
\]
\subsubsection*{Clustered Nested Ellipses (NEC), see \Cref{fig:ellipsedata} (b)}
Let $G^c_1,\dots ,G^c_J\sim \mathrm{Poi}(\lambda_c)$ for $c=1,\dots ,5$. Take $\lambda_3=2$ and set $\lambda_c=1$ else. Let $K_i=M\sum_{c=1}^{5}G^c_i$. Set $w_k^i$ equal to $1$ for $1\leq k\leq K_i$ for $i=1,\dots,J$. Let $t_1,\dots ,t_M$ be a discretization of $[0,2\pi]$. Choose $U^i_1,\dots,U^i_K,V^i_1,\dots,V^i_K\sim U[0.2,1]$. Let $\alpha=(2,12,12,22,12)^T$ and $\beta=(12,2,12,12,22)^T$. Set for $c=1,\dots,5$ and $j_i=0,\dots ,G_i^c$
\[
l^i_{M\left(\sum_{r=1}^{c-1}G^r_i\right)+M(j_i-1)+k}=\frac{1}{24}((3^{-j}U_{M(j_i-1)+k}\sin(t_k)+\alpha_c,3^{-j}V_{M(j_i-1)+k}\cos(t_k)+\beta_c))^T,
\]
where we use the convention that a sum is zero if its last index is smaller than its first one.
\begin{figure}[H]
  \centering
  \subfloat[][]{\includegraphics[width=0.45\linewidth]{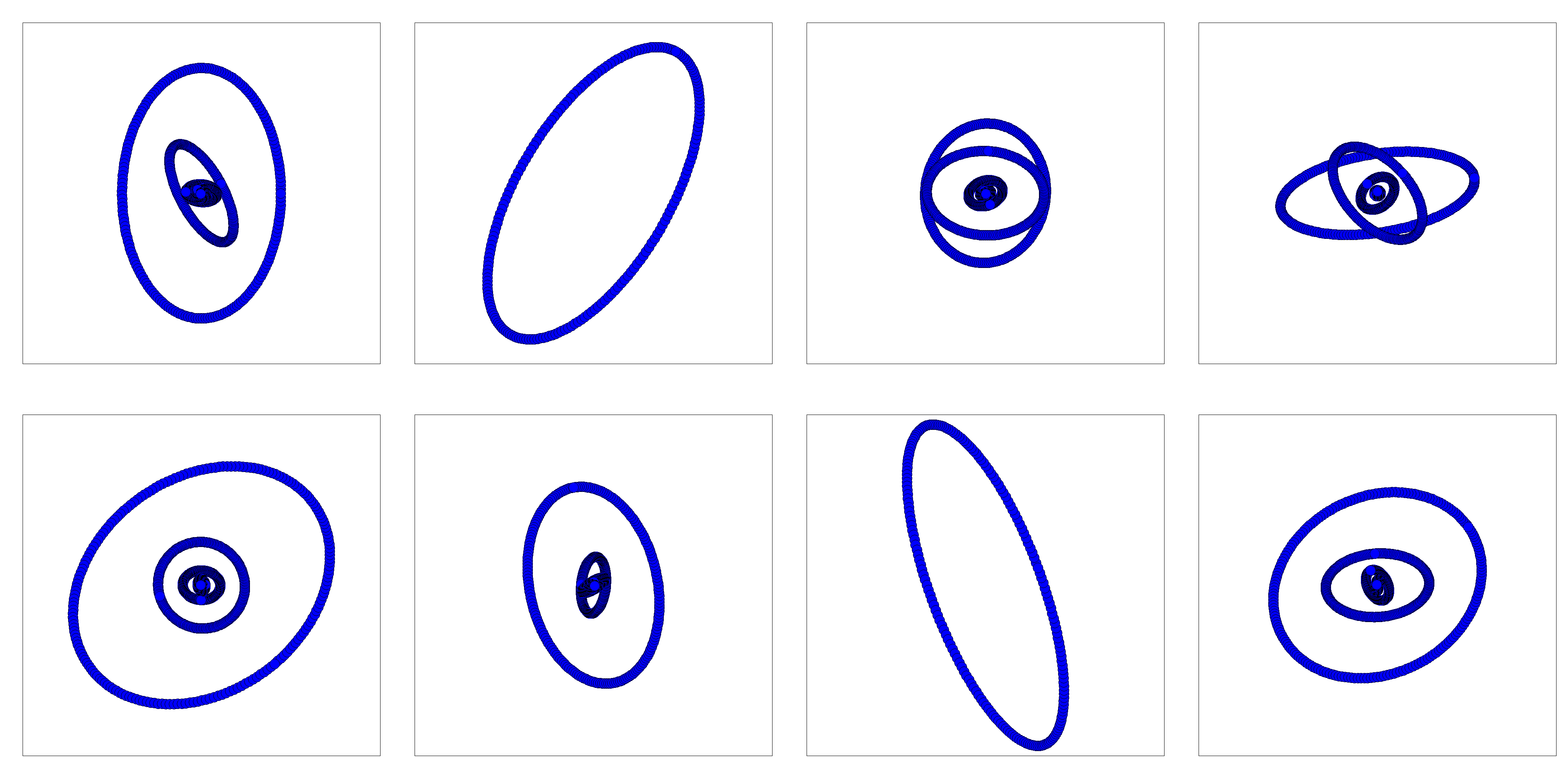}}%
  \qquad
  \subfloat[][]{\includegraphics[width=0.45\linewidth]{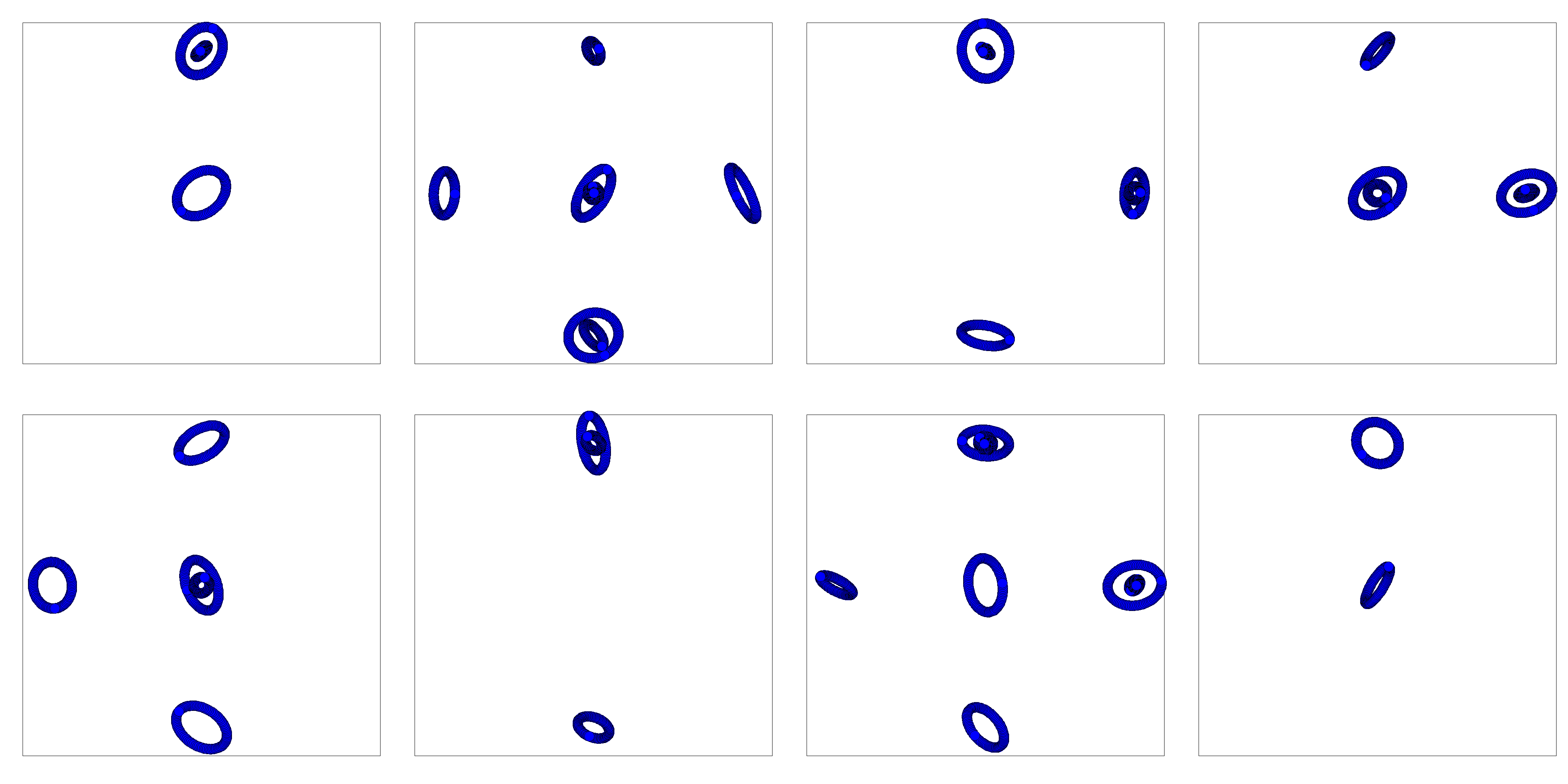}}%
  \caption{\textbf{(a)} An example of $J=8$ measures from the \emph{NE} dataset with $M=200$. \textbf{(b)} An example of $J=8$ measures from the \emph{NEC} dataset with $M=95$.}%
  \label{fig:ellipsedata}
\end{figure}

\subsubsection*{Poisson Intensities on Uniform Positions (PI), see \Cref{fig:poissondata} (a)}
Set $K=M$ for $M\in \mathbb{N}$ and let $w^i_1,\dots ,w^i_K\sim \mathrm{Poi}(\lambda)$ and some intensity $\lambda>0$ and $l^i_1,\dots,l^i_K \sim U[0,1]^2$ for $1\leq i\leq J$.
\subsubsection*{Poisson Intensities on a Grid (PIG), see \Cref{fig:poissondata} (b)}
Set $K=M^2$ for $M\in \mathbb{N}$ and let $w^i_1,\dots ,w^i_{M^2}\sim \mathrm{Poi}(\lambda)$ and $l^i_1,\dots,l^i_{M^2}$ be the locations of equidistant $M\times M$ grid points in $[0,1]^2$ for $1\leq i\leq J$.
\begin{figure}[H]
  \centering
  \subfloat[][]{\includegraphics[width=0.45\linewidth]{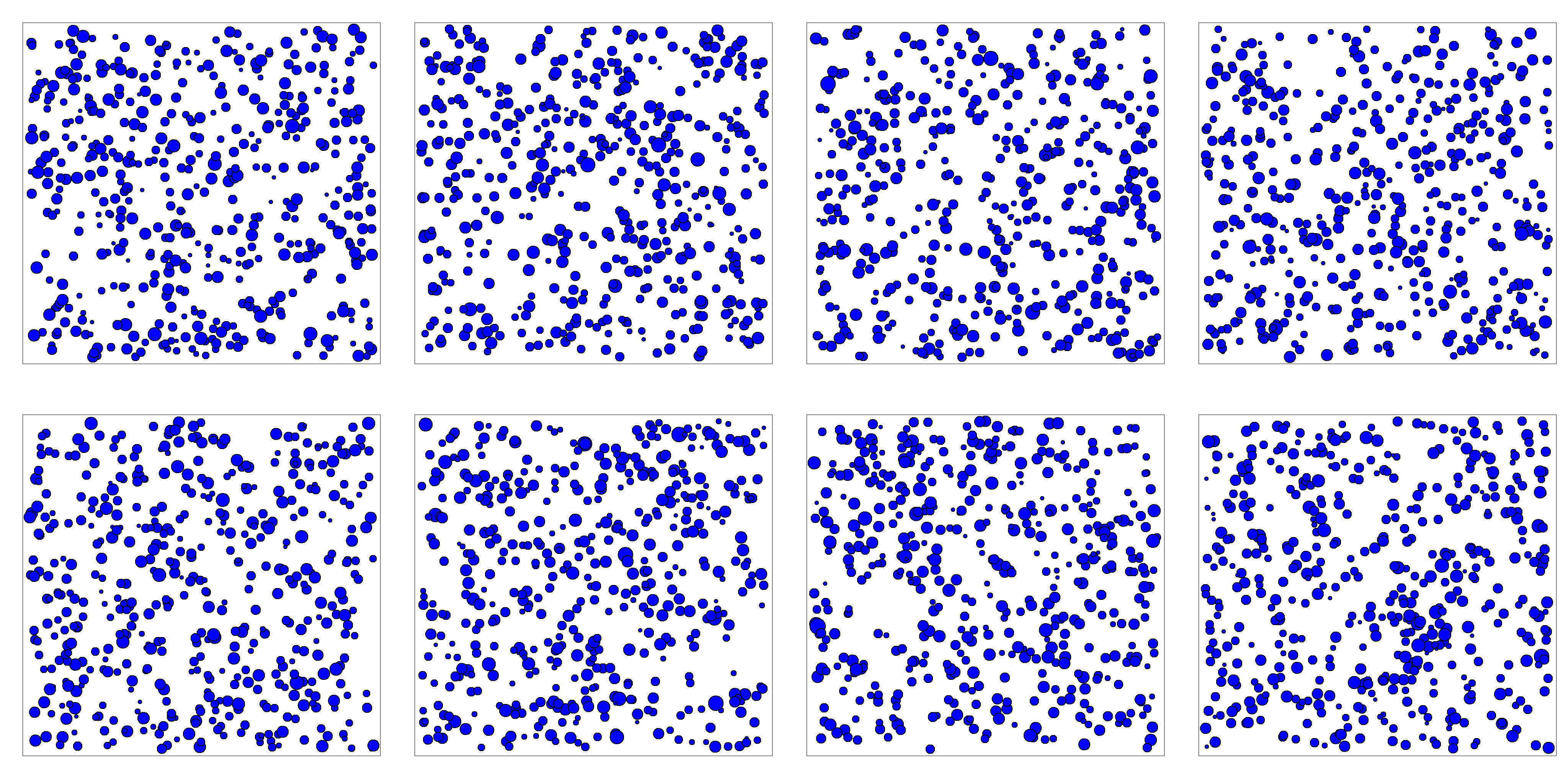}}%
  \qquad
  \subfloat[][]{\includegraphics[width=0.45\linewidth]{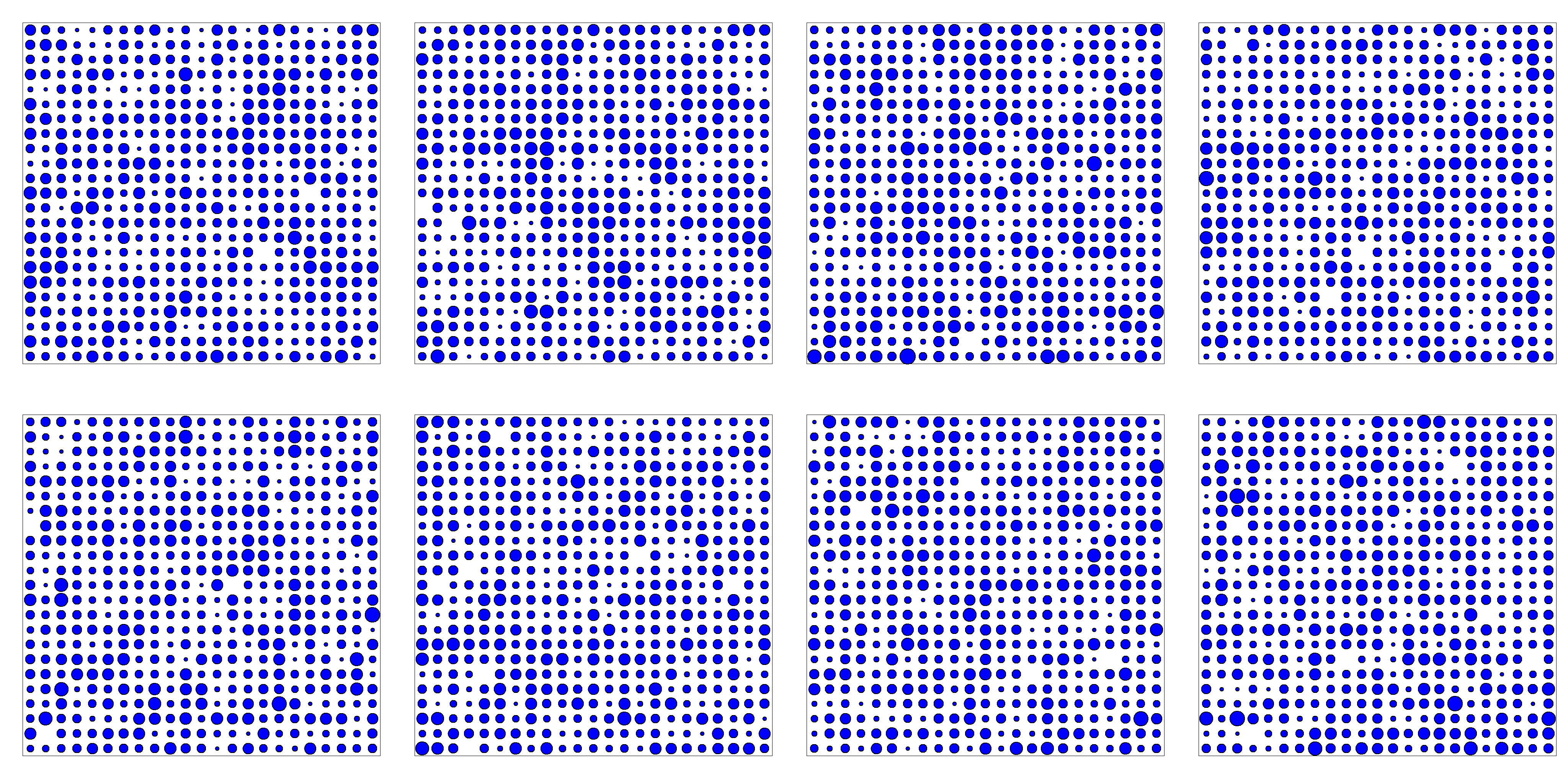}}%
  \caption{\textbf{(a)} An example of $J=8$ measures from the \emph{PI} dataset with $M=500$ and $\lambda=5$. \textbf{(b)} An example of $J=8$ measures from the \emph{PIG} dataset with $M=23$ and $\lambda=5$.}%
  \label{fig:poissondata}
\end{figure}

\subsection{Simulation Results for the $\mathbf{(2,C)}$-Kantorovich-Rubinstein Distance}

In the following, we discuss the results from our simulation studies for the Poisson model for the $(2,C)$-KRD between two measures within one of the eight classes of measures introduced above. %
For the error of the NE class in \Cref{fig:dist_poiNE} the error is decreasing in $s$ and $t$, but increasing in $C$. Both of these behaviors are in line with the bound in \Cref{thm:samplingboundKR}. The decrease of the error for increasing $s$ and $t$ is immediately clear from our theoretical results. The increase of error for increasing $C$ is based on the fact that in the Poisson model the population total intensities of $\mu$ and $\nu$ are unknown and have to be estimated from the data. The $(p,C)$-KRD penalizes mass deviation with a factor scaling with $C$, so naturally for increasing $C$, the errors in the estimation of the true difference of masses yields an increase in the expected relative $(p,C)$-KRD error.

\begin{figure}[b!]
  \centering
\begin{tabular}{cc}
      \includegraphics[width=0.45\textwidth]{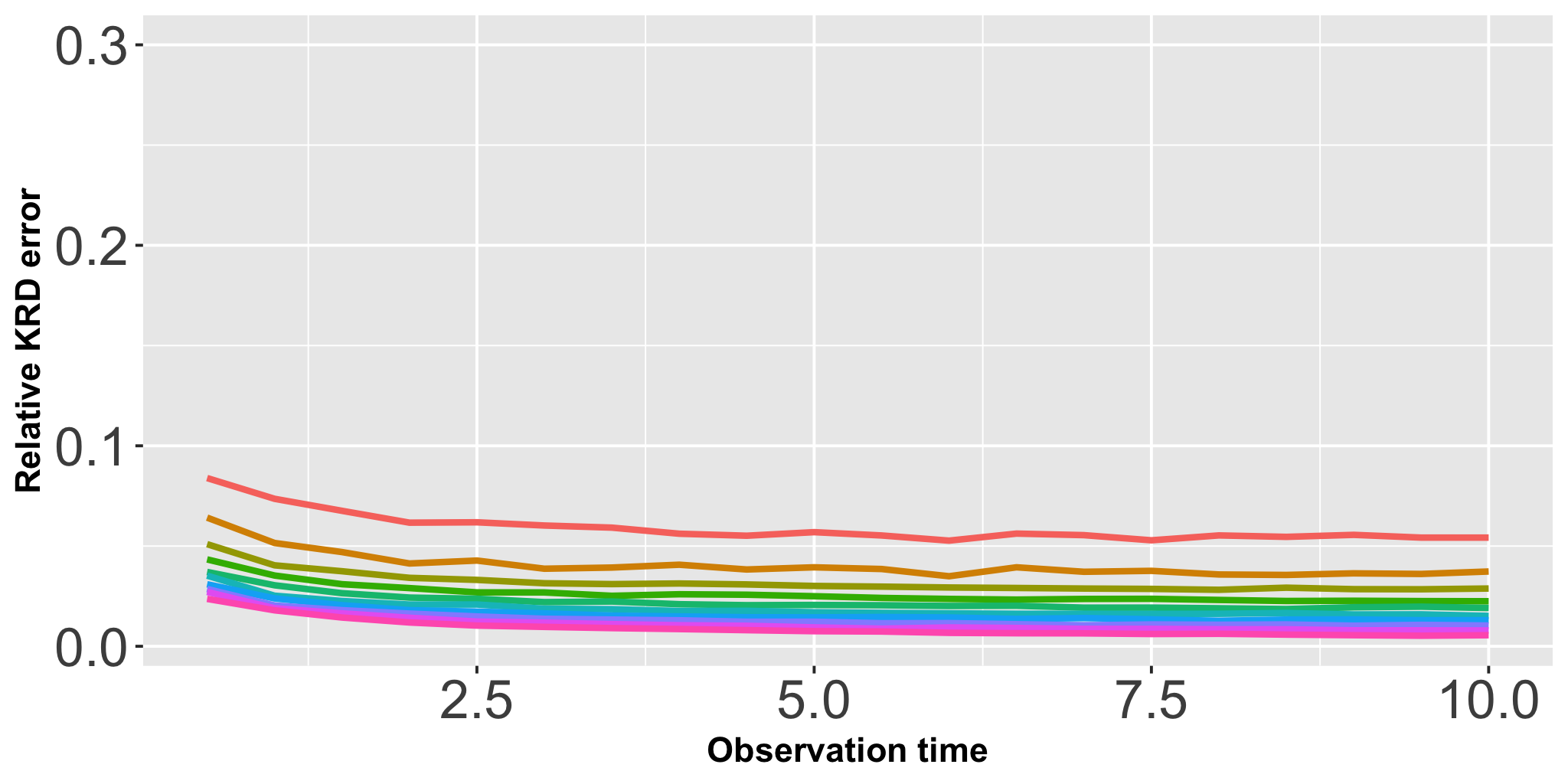} &       \includegraphics[width=0.45\textwidth]{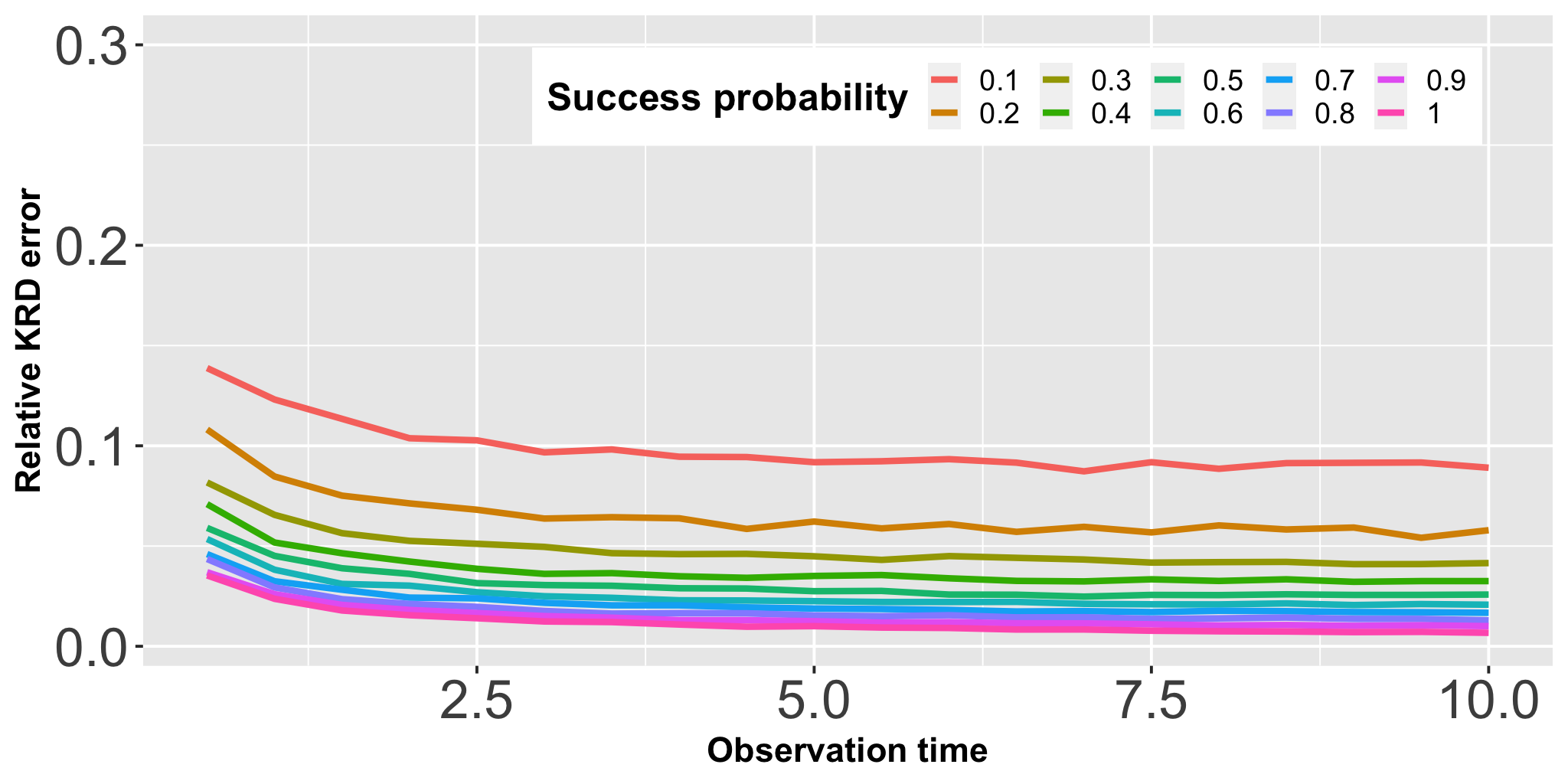} \\
    \includegraphics[width=0.45\textwidth]{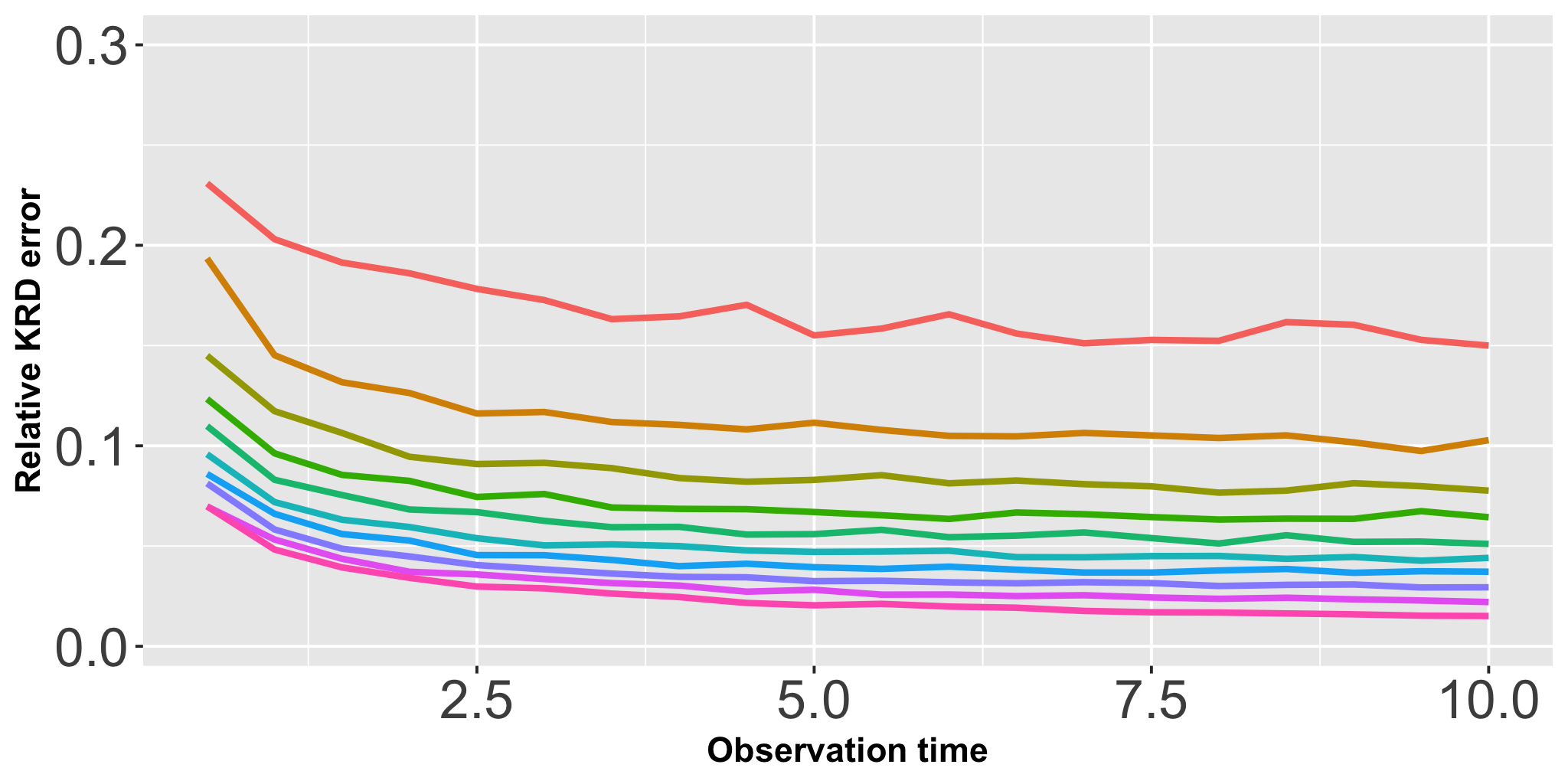} &       \includegraphics[width=0.45\textwidth]{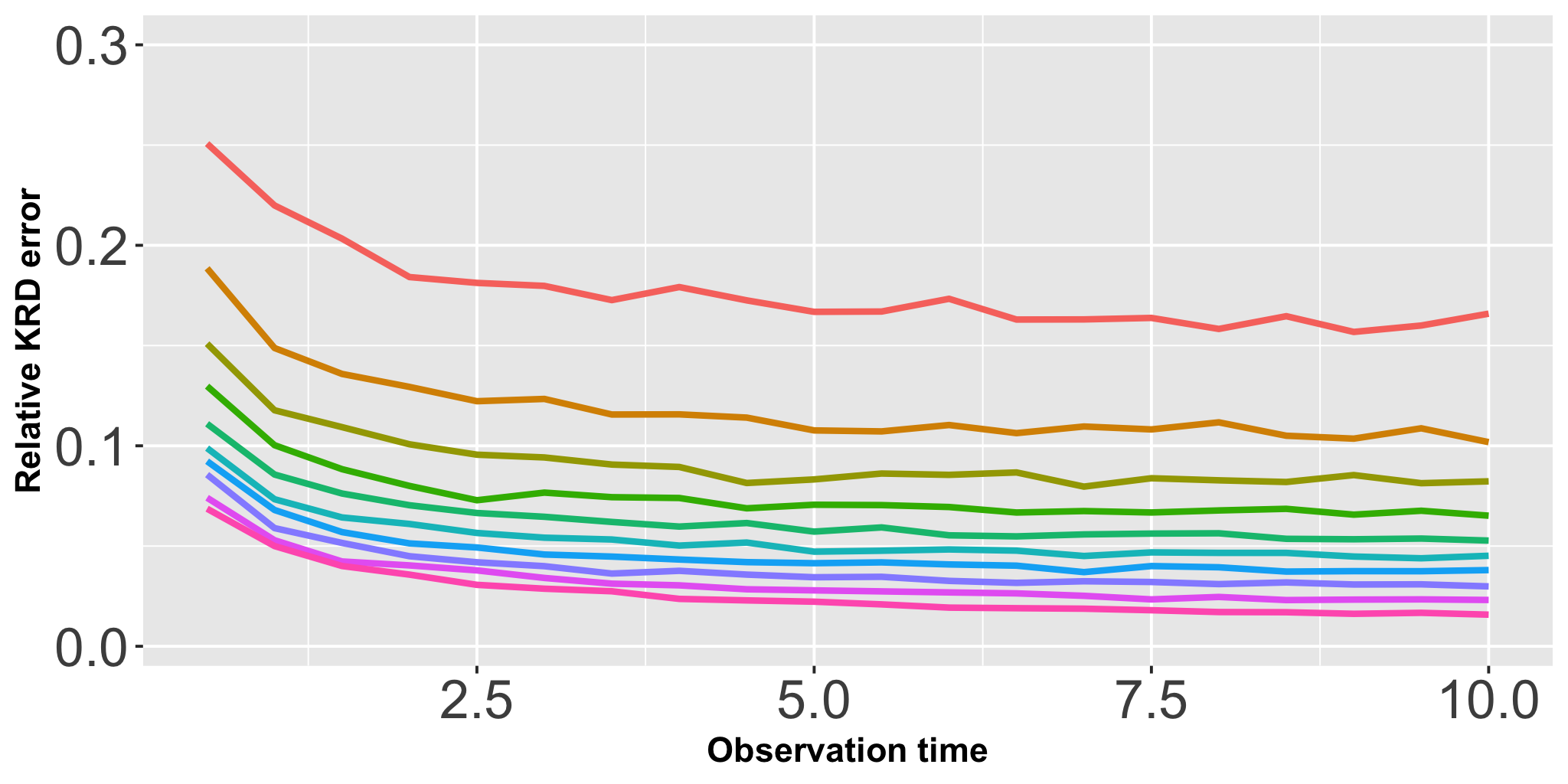} \\
\end{tabular}
\caption{Expected relative $(2,C)$-KRD error for two measures in the Poisson sampling model for the NE class with $M= 100$ and different success probabilities $s$. For each pair of success probability $s$ and observation time $t$ the expectation is estimated from $1000$ independent runs. %
From top-left to bottom-right we have $C=0.01,0.1,1,10$, respectively.}
  \label{fig:dist_poiNE}
\end{figure}

\begin{figure}[h!]
  \centering
\begin{tabular}{cc}
      \includegraphics[width=0.45\textwidth]{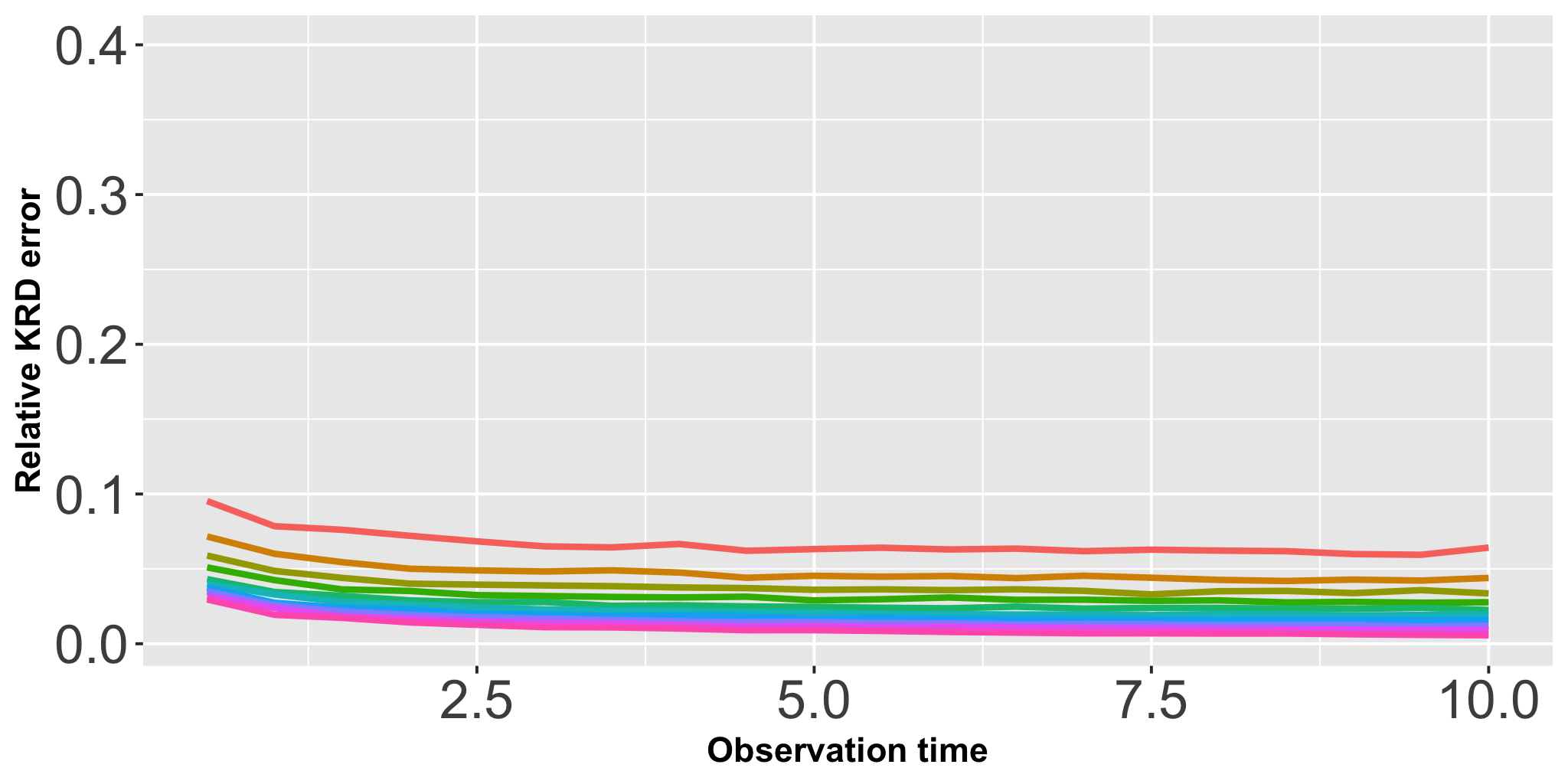} &       \includegraphics[width=0.45\textwidth]{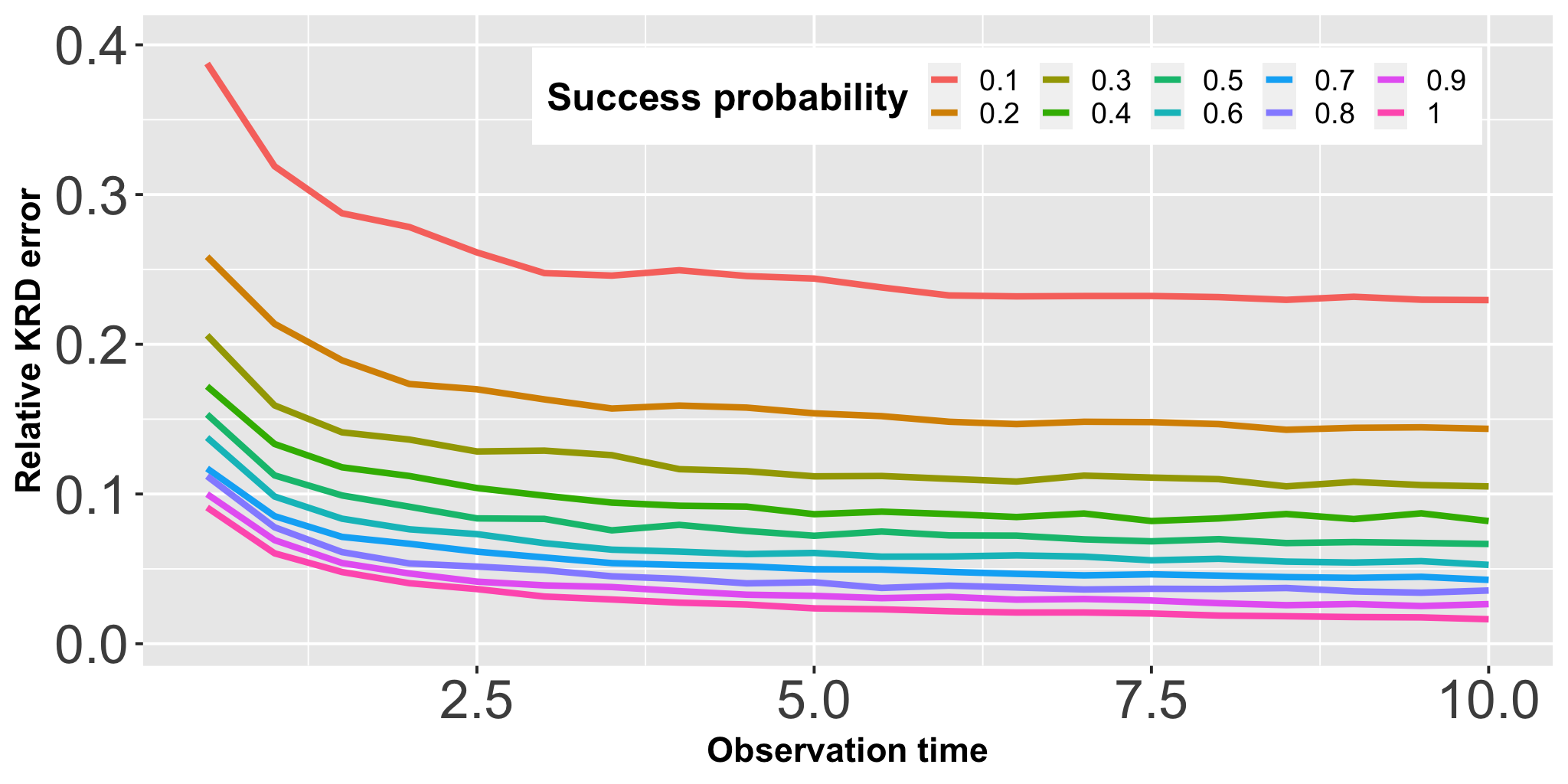} \\
    \includegraphics[width=0.45\textwidth]{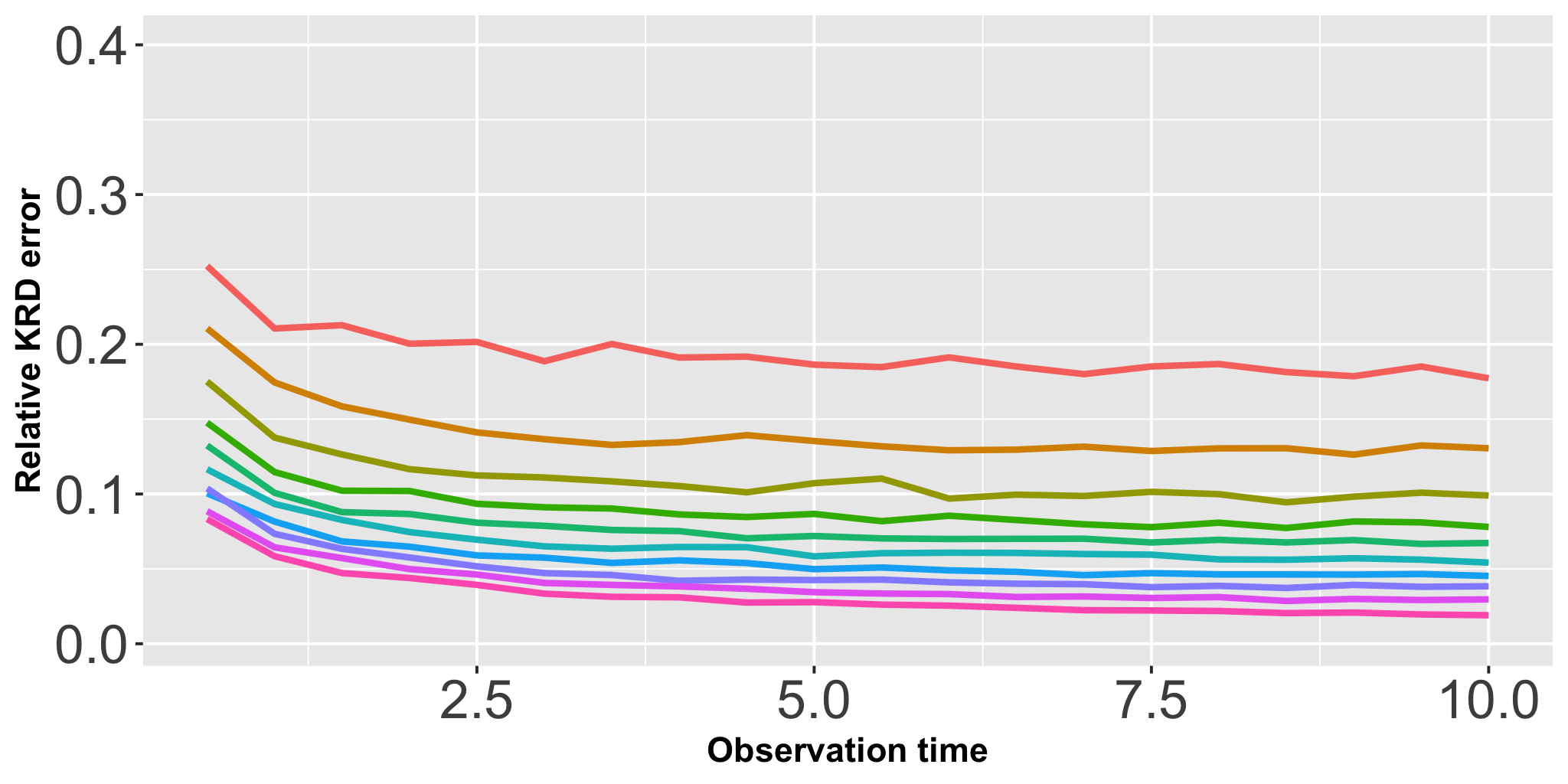} &       \includegraphics[width=0.45\textwidth]{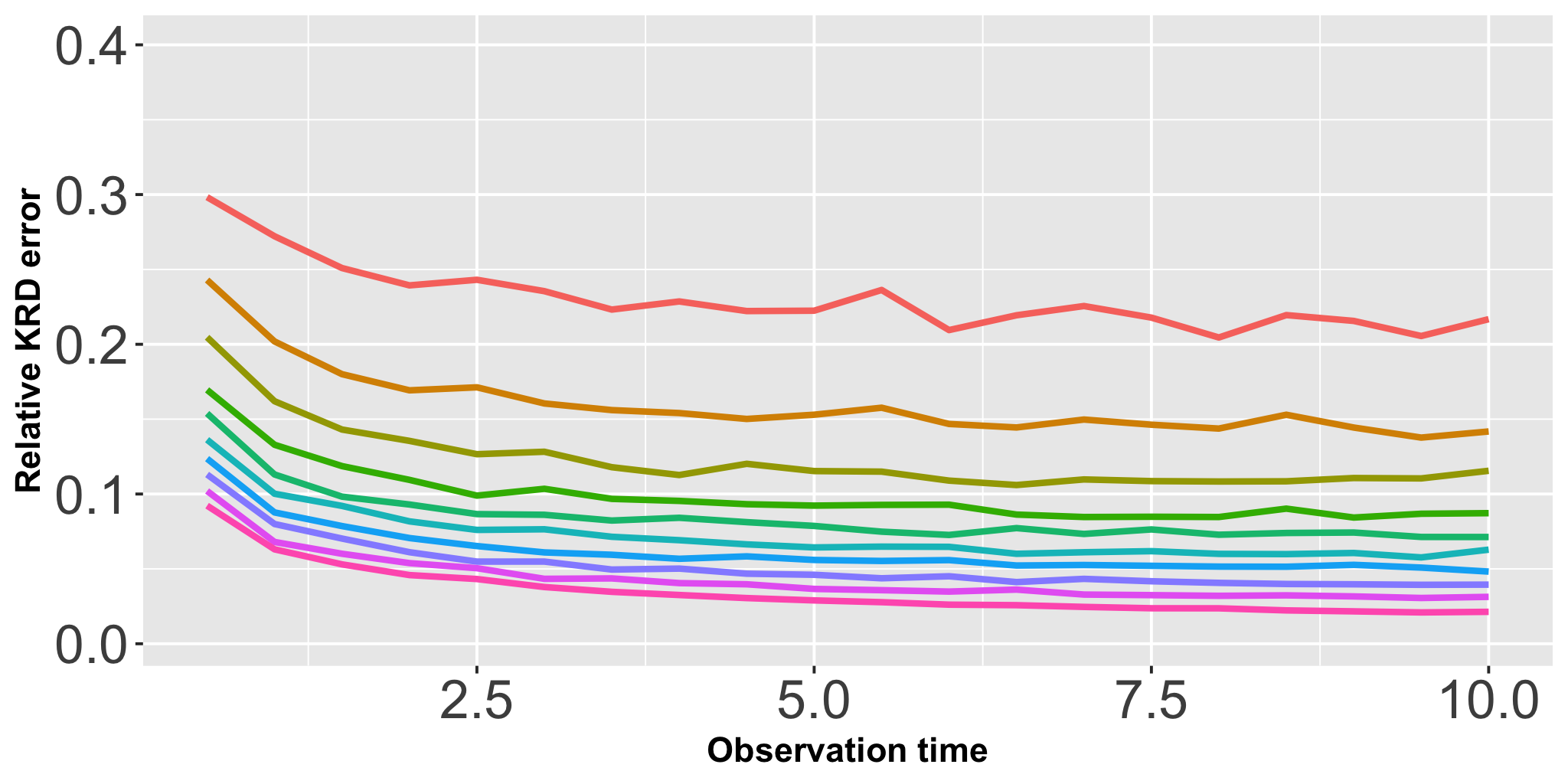} \\
\end{tabular}
\caption{As in \Cref{fig:dist_poiNE}, but for the NEC class with $M=75$.}
  \label{fig:dist_poiNEC}
\end{figure}

\begin{figure}[h!]
  \centering
  \begin{tabular}{cc}
        \includegraphics[width=0.45\textwidth]{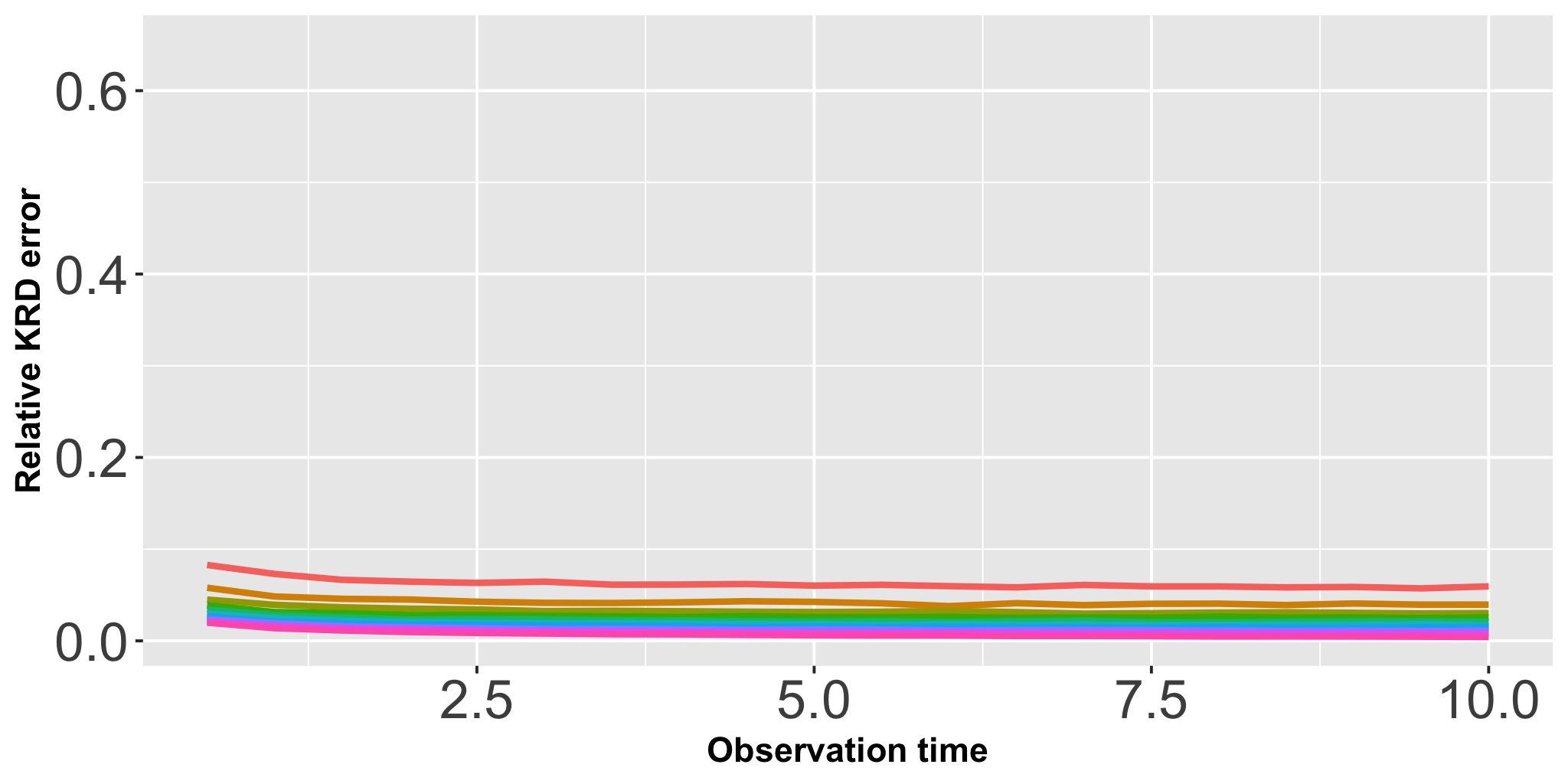} &       \includegraphics[width=0.45\textwidth]{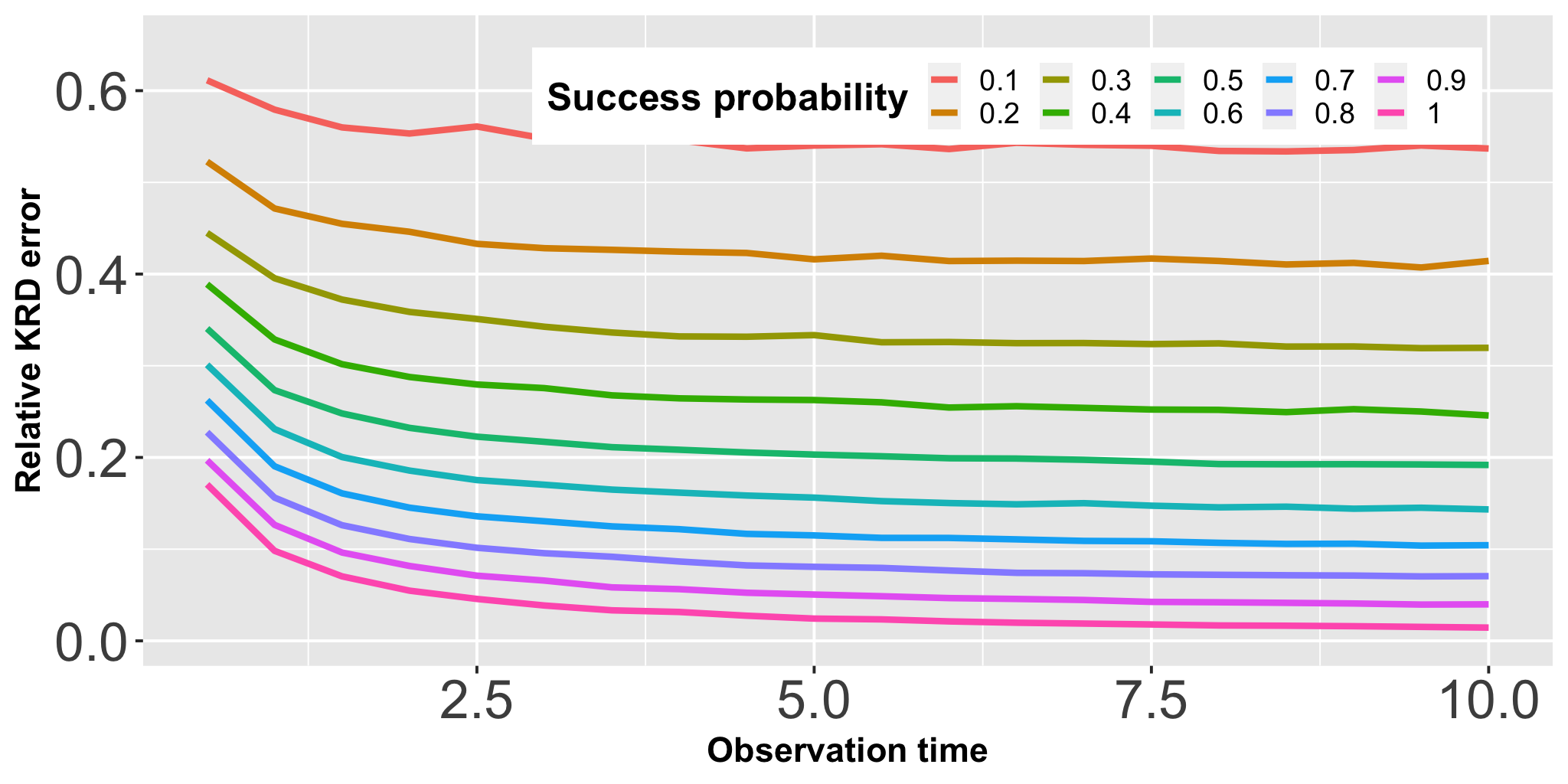} \\
      \includegraphics[width=0.45\textwidth]{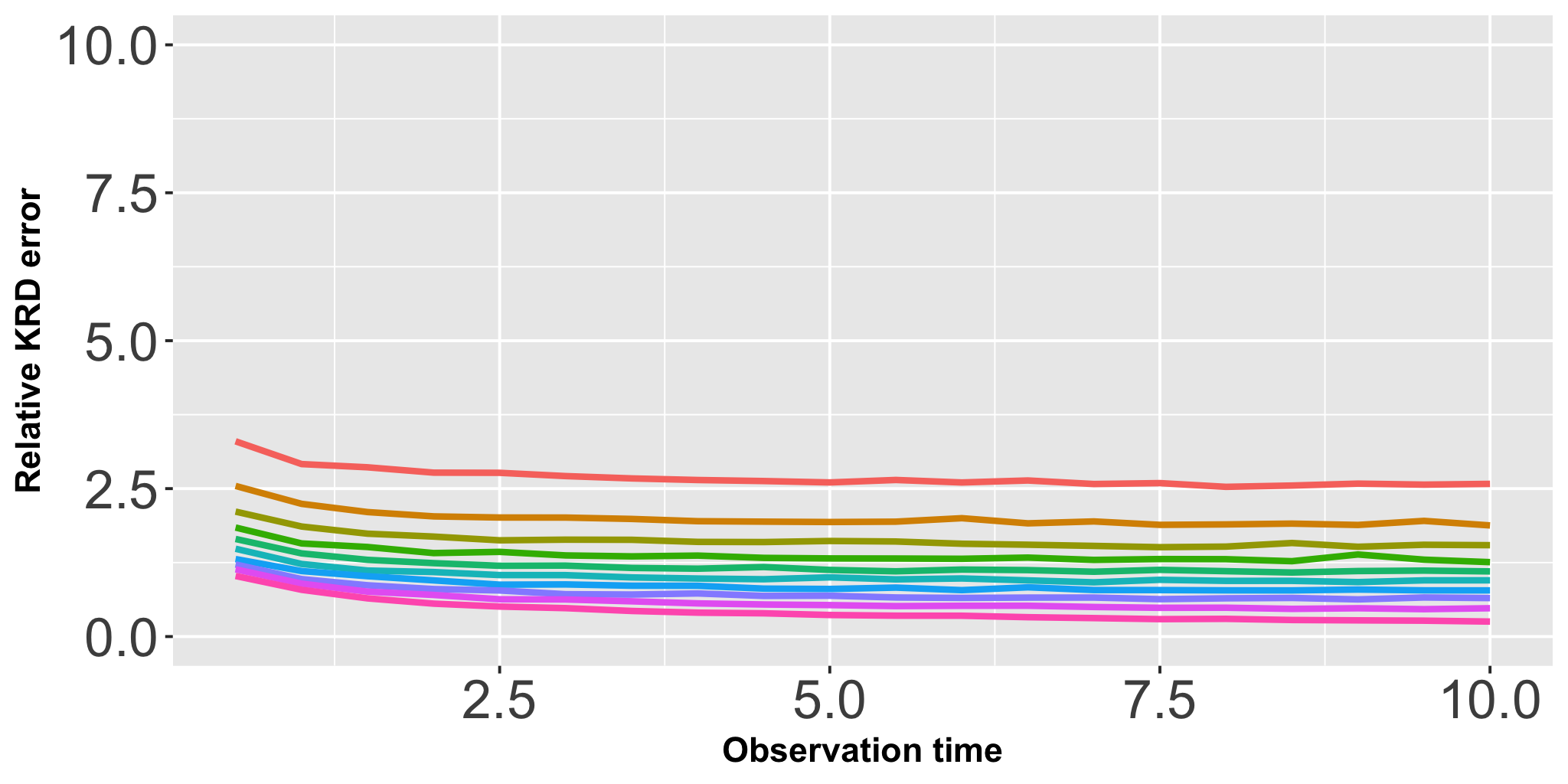} &       \includegraphics[width=0.45\textwidth]{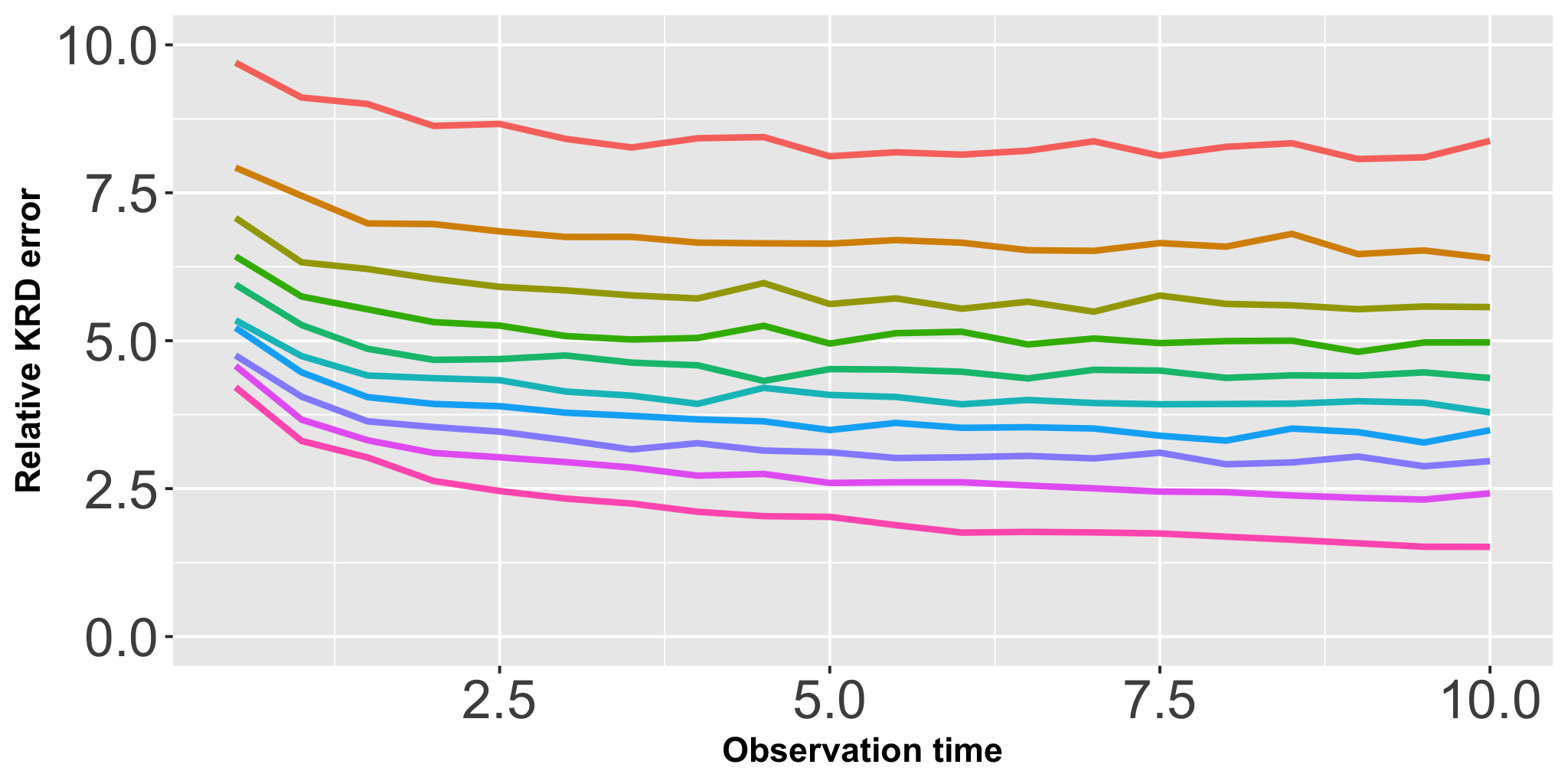} \\
  \end{tabular}
  \caption{As in \Cref{fig:dist_poiNE}, but for the PI class with $M=450$.}
    \label{fig:dist_poiPI}
  \end{figure}

\begin{figure}[h!]
  \centering
\begin{tabular}{cc}
      \includegraphics[width=0.45\textwidth]{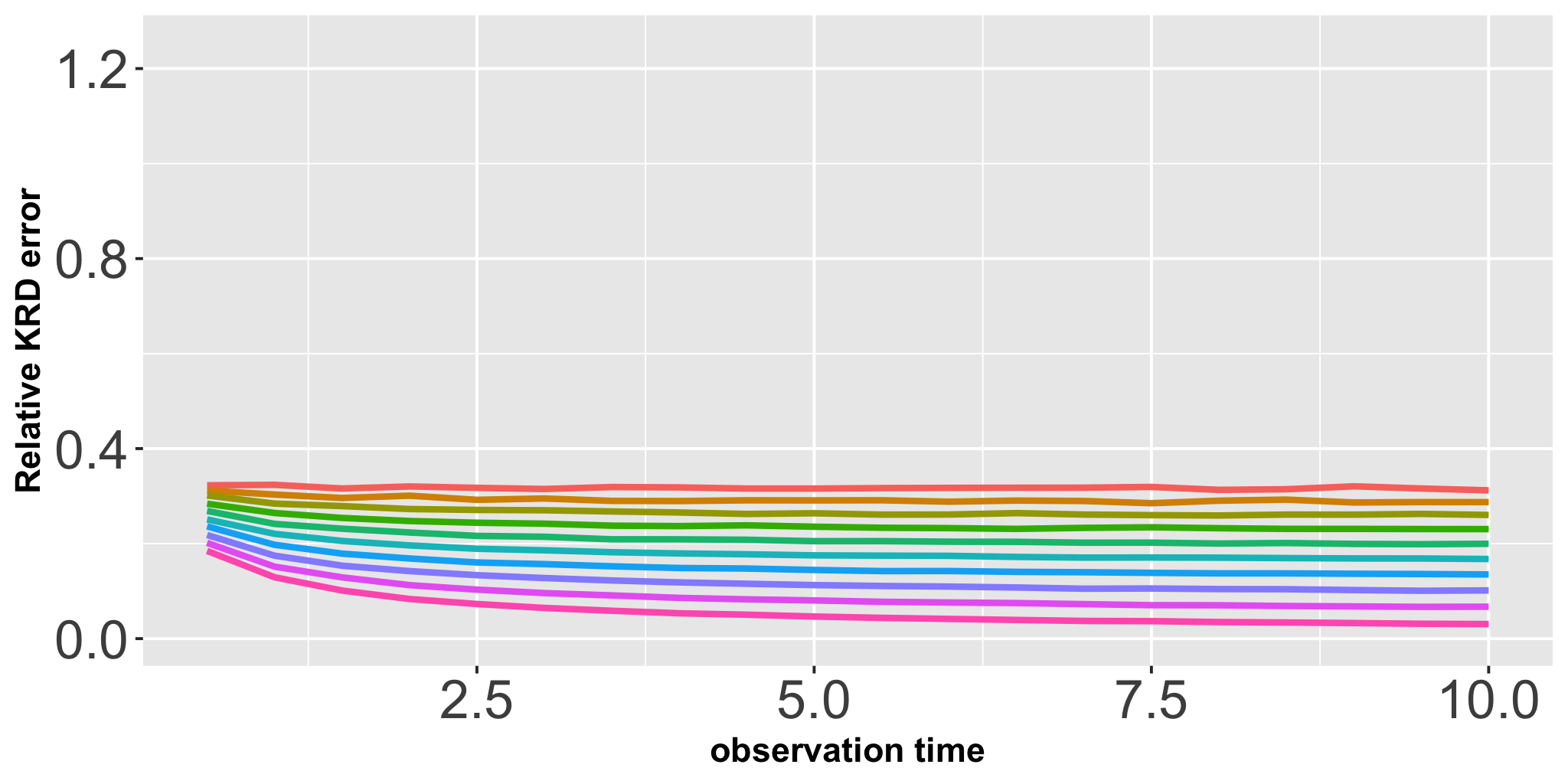} &       \includegraphics[width=0.45\textwidth]{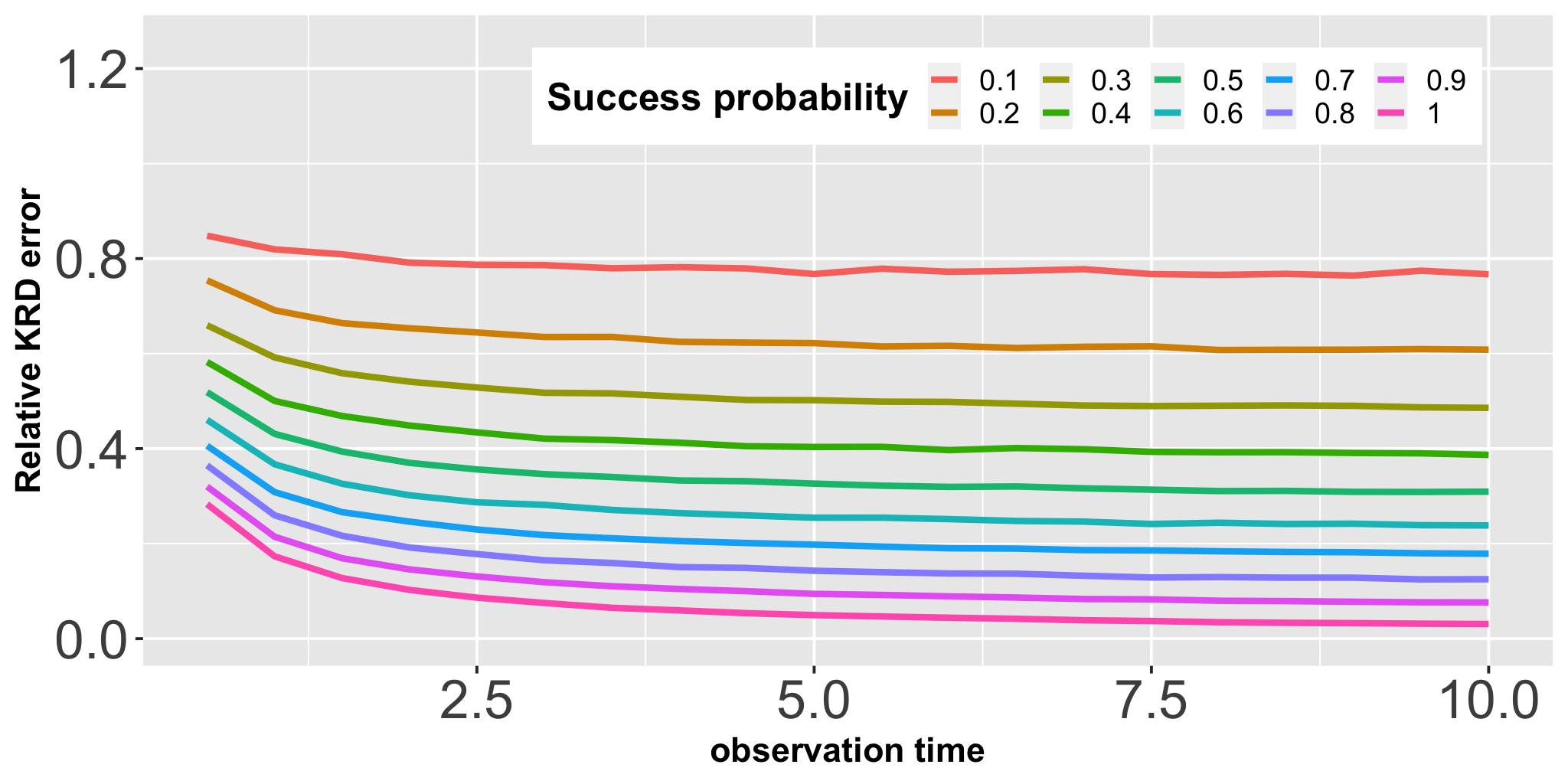} \\
    \includegraphics[width=0.45\textwidth]{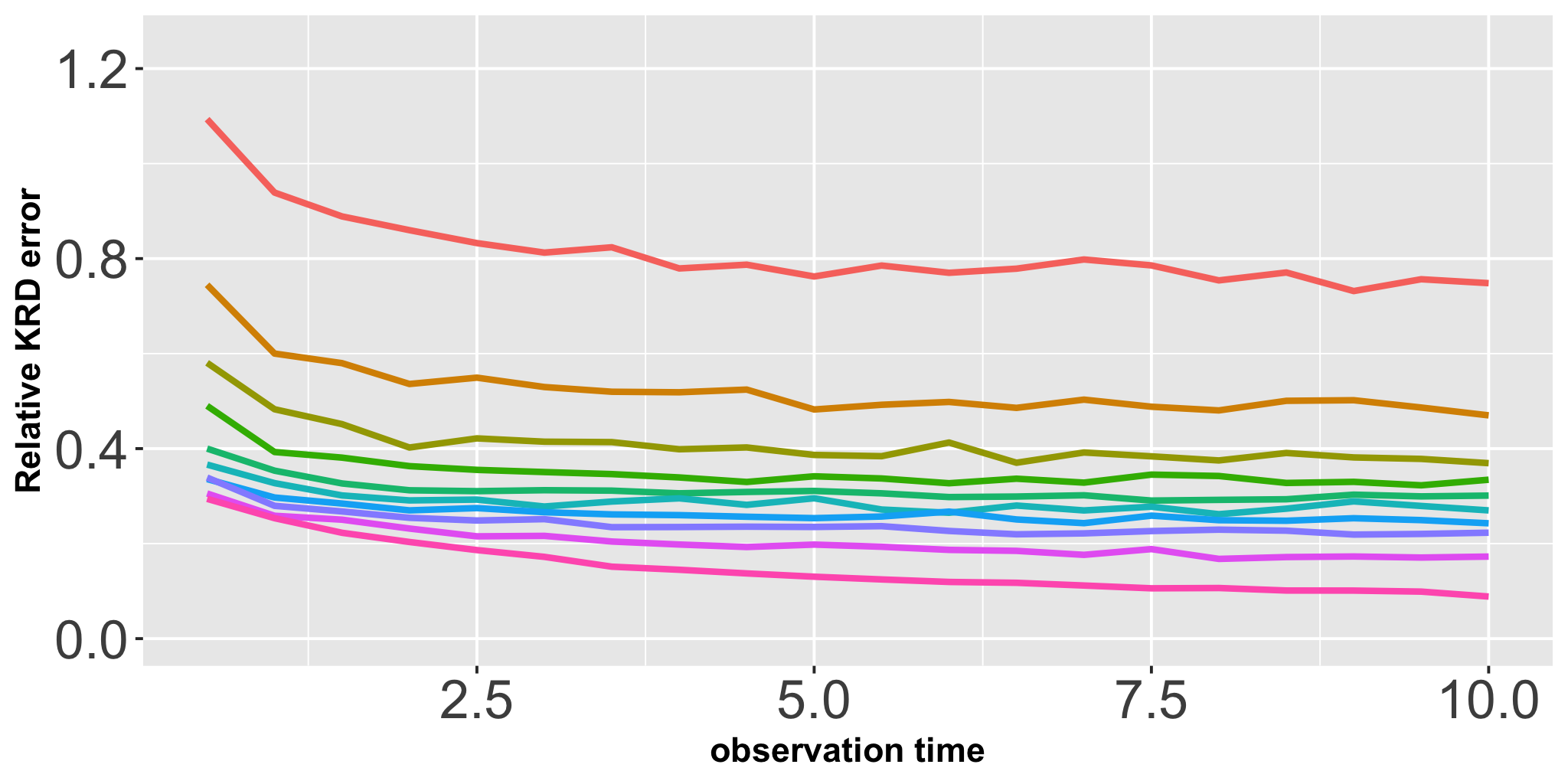} &       \includegraphics[width=0.45\textwidth]{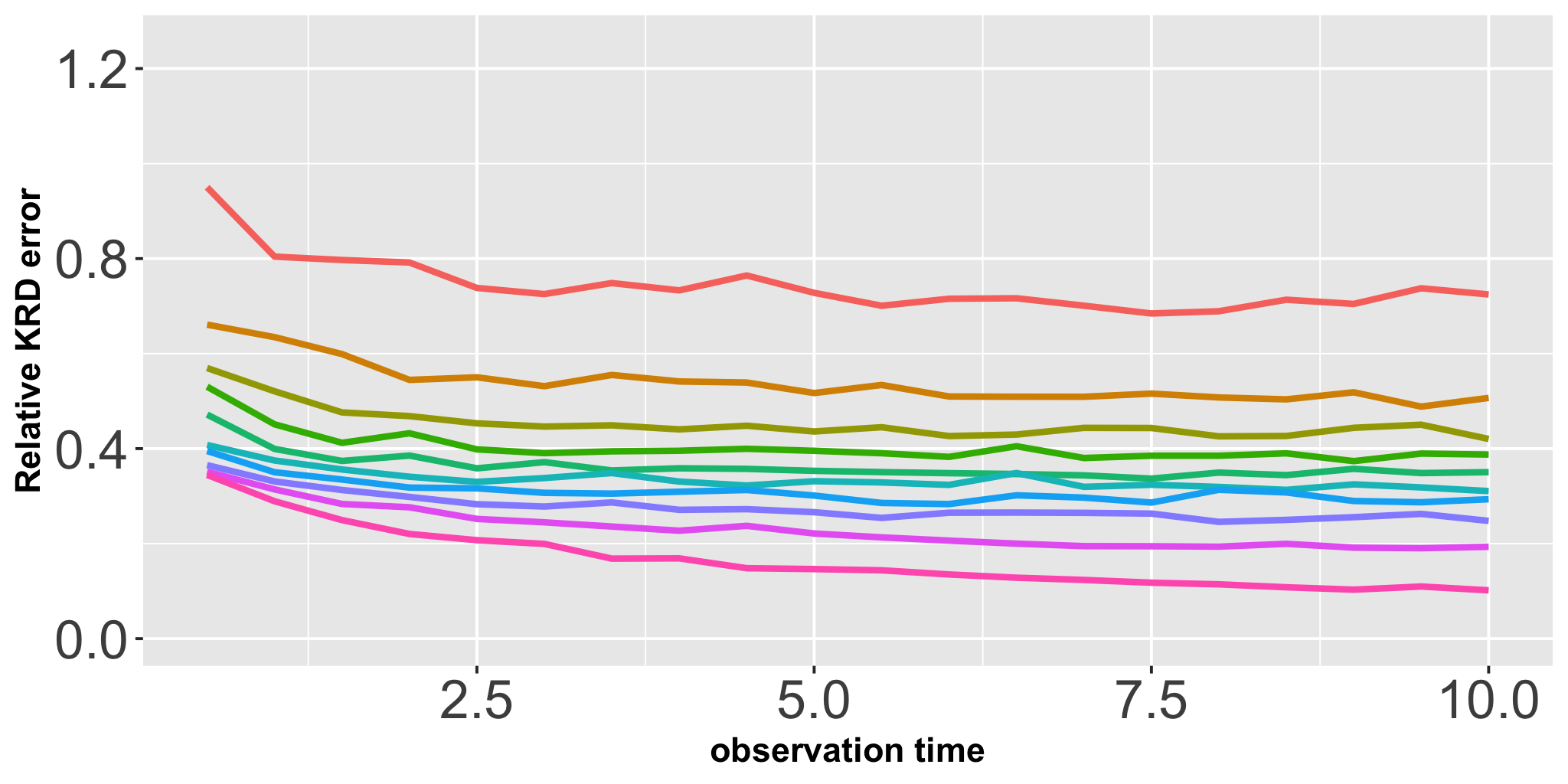} \\
\end{tabular}
\caption{As in \Cref{fig:dist_poiNE}, but for the PIG class with $M=22$.}
  \label{fig:dist_poiPIG}
\end{figure}

Notably, while the decrease in $s$ and $t$ is similar for the error of the NEC class in \Cref{fig:dist_poiNEC}, the error is no longer increasing in $C$. Instead, the errors increase from $C=0.01$ to $C=0.1$, but then decrease again for $C=1$. Afterwards they increase again at $C=10$. This difference in behavior is explained by the cluster structure of the measure in NEC. There is still the general trend of increasing error for increasing $C$, as present in the NE class, but now there is an additional change in behavior based on the fact if transport occurs within clusters or between clusters. From $C=0.1$ to $C=1$, we pass the size of the clusters and the distance between the clusters. Thus, $C=1$ is the first value in our simulation for which inter-cluster transports can occur. This causes a decrease in error, as the impact of the estimation of the total mass intensity of a measure within one cluster is decreased. After this point the usual increase in error for increasing $C$ due to the estimation of the total mass intensity occurs. 

For the $(p,C)$-KRD error in the PI class in \Cref{fig:dist_poiPI} this effect is particularly strong. The error increases on average  about two orders of magnitude from $C=0.01$ to $C=10$. This is explained by the fact that the total mass intensity in this class is significantly larger than for the classes NE and NEC, where each location in the support of the measures has mass one. This also causes an increase of the variance of the mass of the empirical measures at each location, which causes a faster increase of error for increasing $C$ at all scales of $C$. \\

Meanwhile, the $(p,C)$-KRD error in the PIG class in \Cref{fig:dist_poiPIG} increases from $C = 0.01$ to $C= 0.1$ but does not increase further as $C$ increases. This is likely a consequence of the homogeneous structure of the measures in the PIG class, which cause the UOT plan to transport mass along small scales and since the total mass is well concentrated. 

\subsection{Simulation Results for the $\mathbf{(2,C)}$-Barycenter}

It remains to discuss the results from our simulation studies for the Poisson model for the $(2,C)$-barycenter between sets of measures within one of the classes of measures introduced above. %
We restrict our analysis to the values of $C=0.1,1,10$, since for $C=0.01$ the $(p,C)$-KRD is close to the TV distance. In particular, for all classes except PIG, where all measures share the same support grid, the $(2,C)$-barycenter will be close or identical to the zero measure, since the measures in the other classes are almost surely disjoint. Additionally, if the barycenter of the population measures is the zero measure, any empirical barycenter has mass zero as well. Thus, there is little merit in simulating the barycenters in these cases. For the class PIG the barycenters are essentially TV-barycenters for small $C$ which removes any geometrically interesting features from the barycenter. Finally, for extremely small values of $C$ the $(p,C)$-barycenter computations tend to become numerically unstable due to either involving values close to machine accuracy and UOT plans for these values of $C$ often being close to the zero measure. Hence, empirical simulations of the expected relative Fr\'echet error would also be less reliable in this regime of values for $C$. In summary, empirical analysis of the properties of the $(p,C)$-barycenter for values of $C$ which are several orders of magnitude smaller than the diameter of $\Y$ is inadvisable. 

\begin{figure}[t!]
  \begin{tabular}{ccc}
    \includegraphics[width=0.31\textwidth]{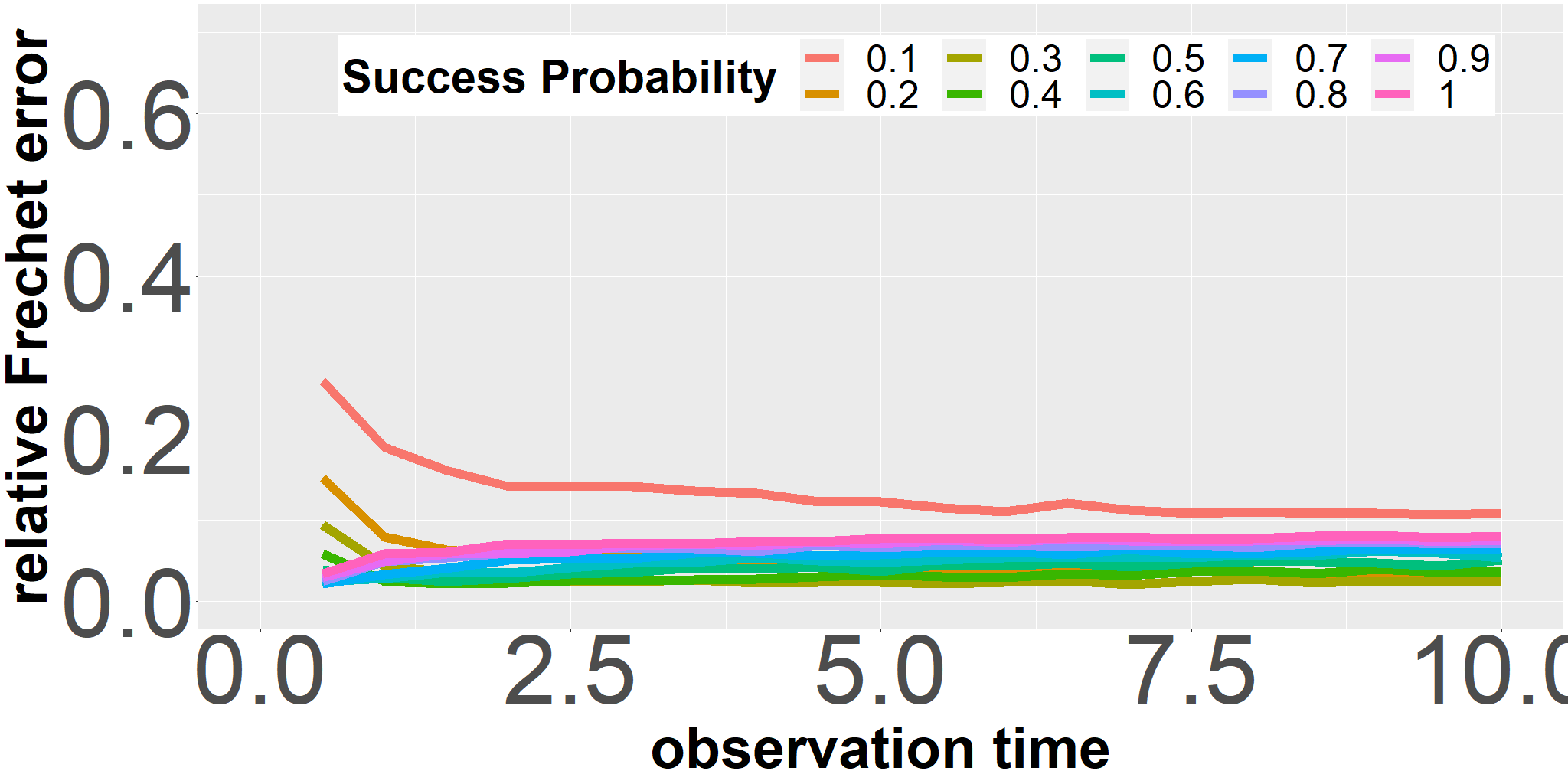} &
      \includegraphics[width=0.31\textwidth]{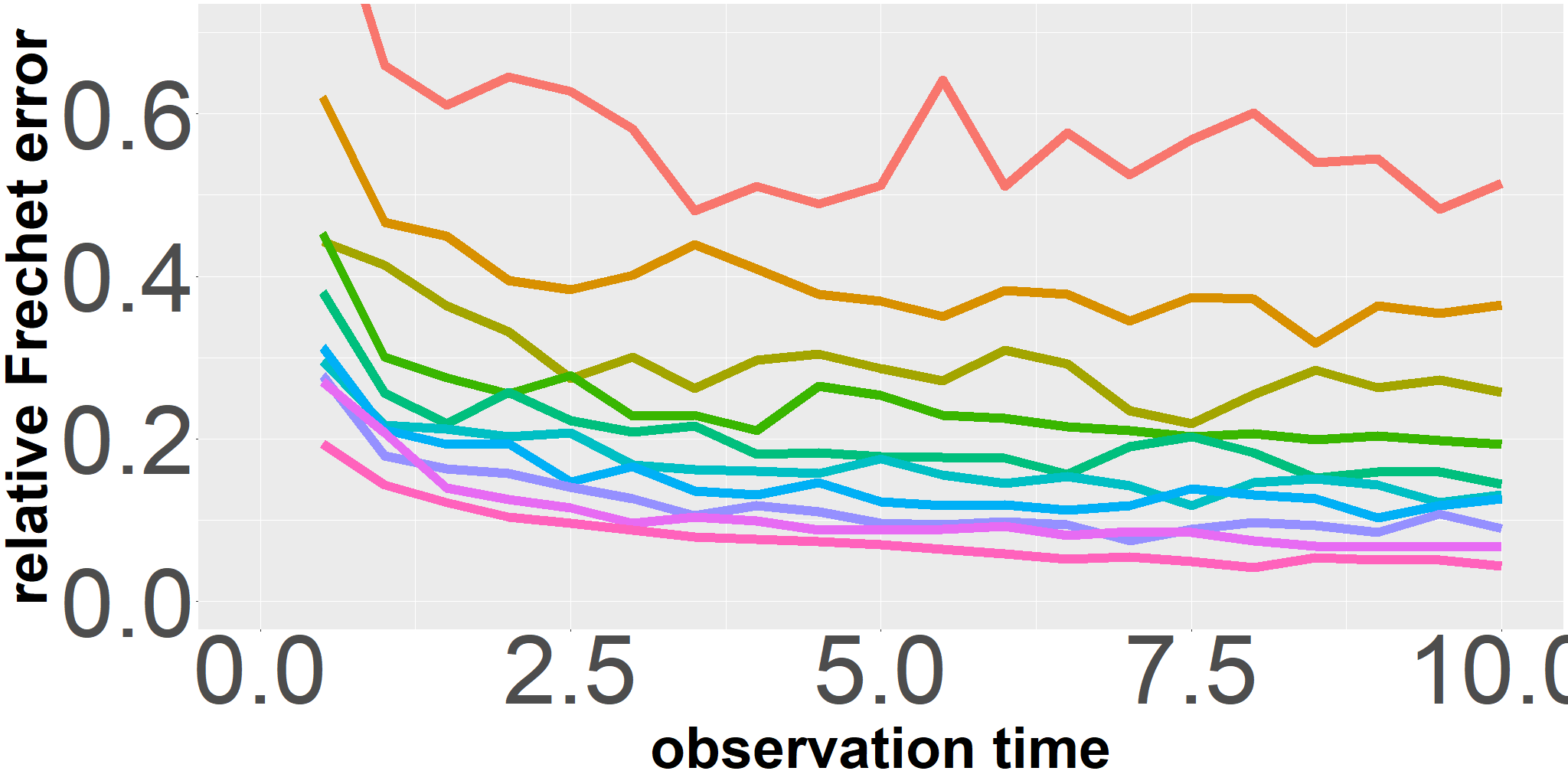} &       \includegraphics[width=0.31\textwidth]{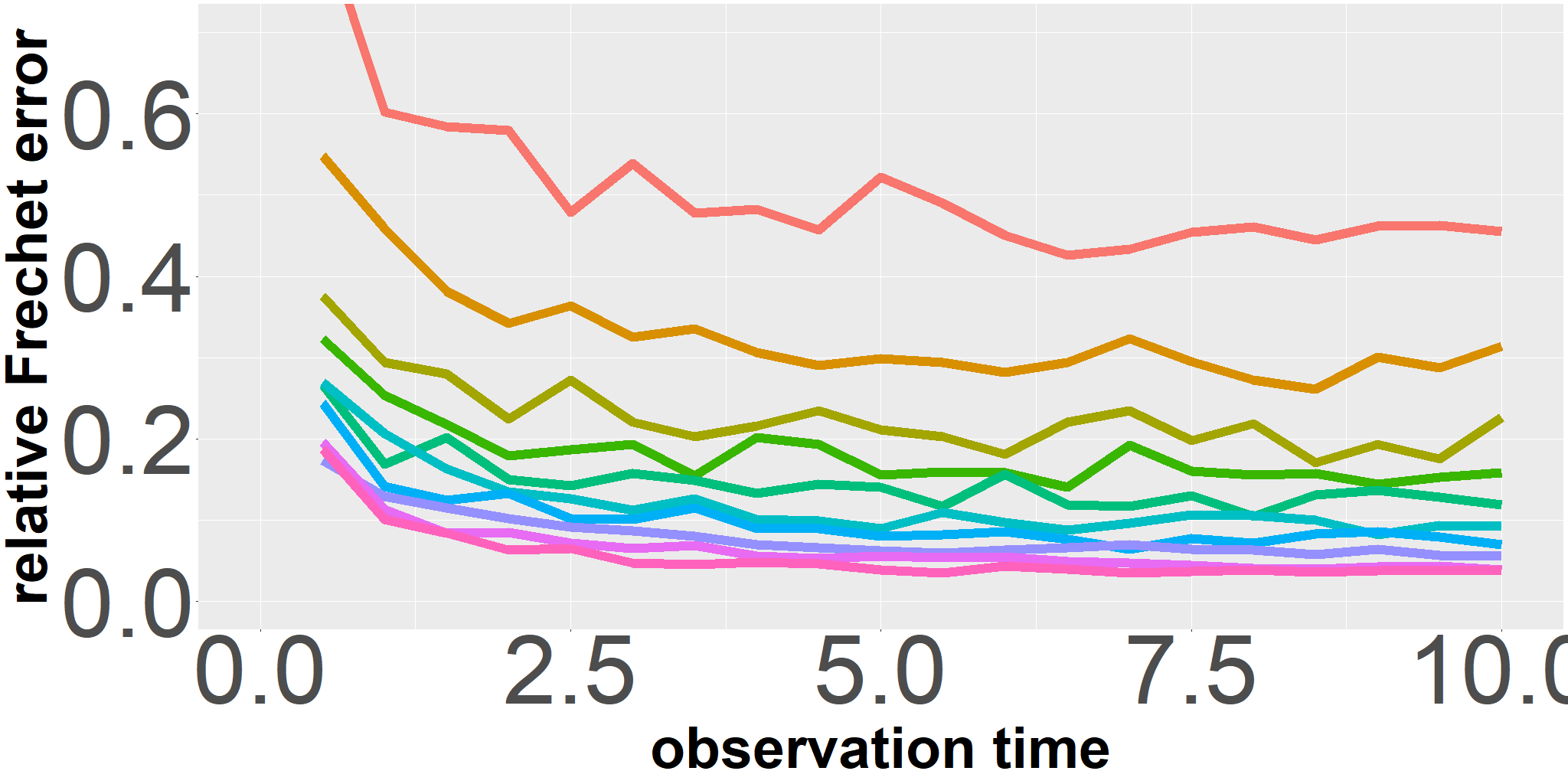} 
  \end{tabular}
  \caption{Expected relative Fr\'echet error for the $(2,C)$-barycenter for $J=5$ measures from the NE class in the Poisson sampling model with different success probabilities $s$. For each pair of success probability $s$ and observation time $t$ the expectation is estimated from $100$ independent runs. Set $M=100$. From left to right we have $C=0.1,1,10$, respectively.}
    \label{fig:bary_poiNE}
  \end{figure}

\begin{figure}[t!]
\begin{tabular}{ccc}
    \includegraphics[width=0.31\textwidth]{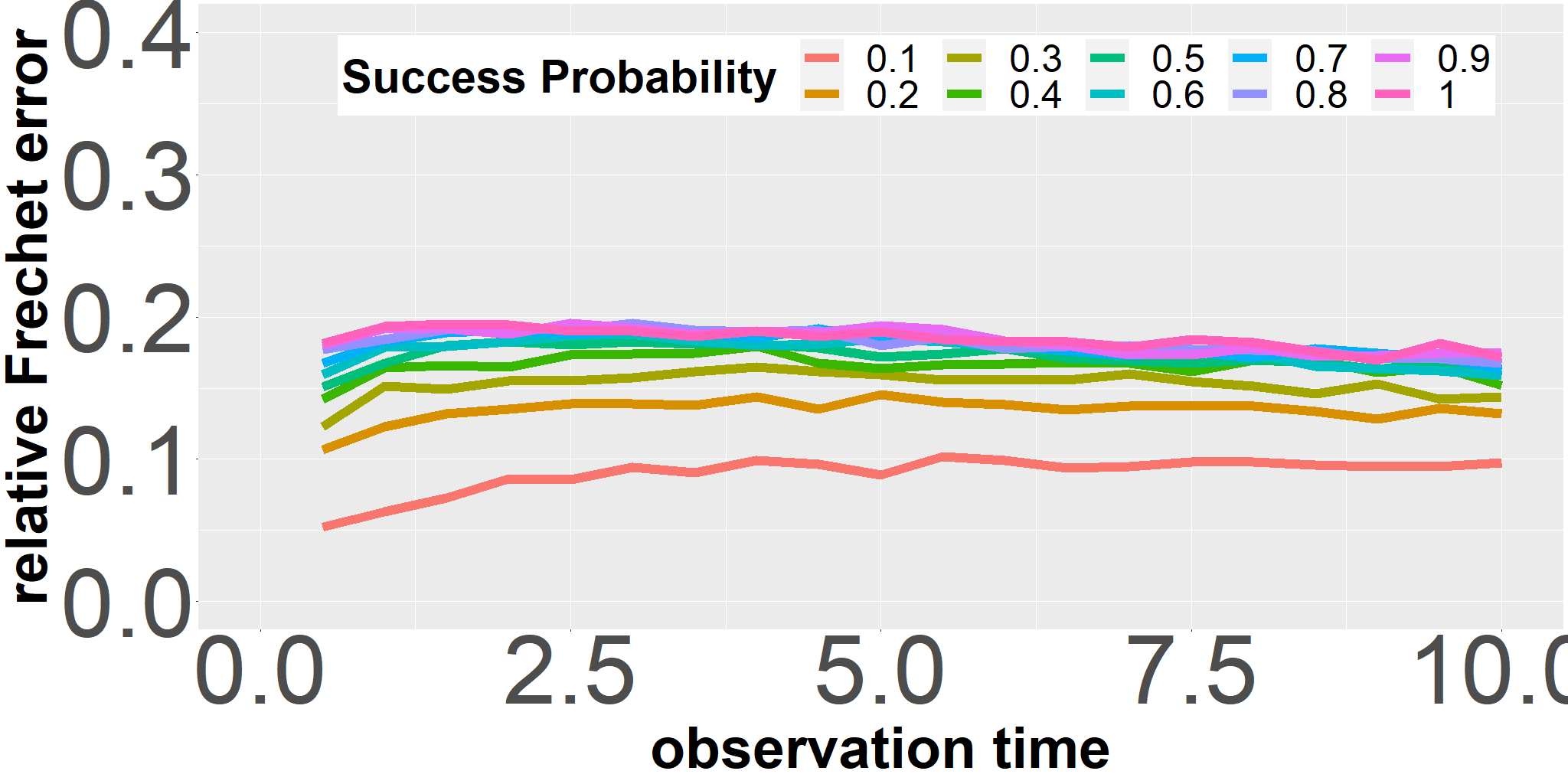} &
    \includegraphics[width=0.31\textwidth]{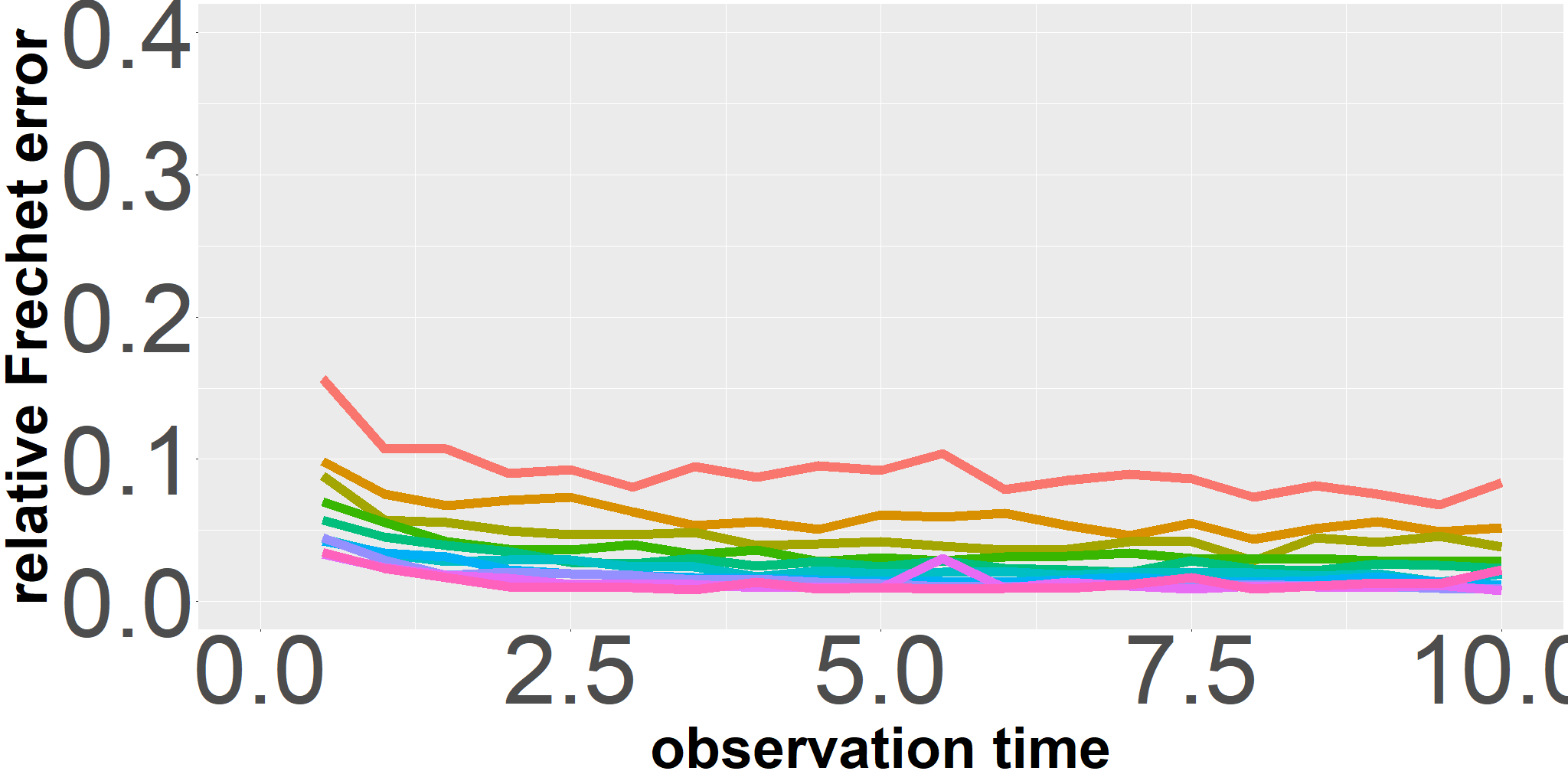} &       \includegraphics[width=0.31\textwidth]{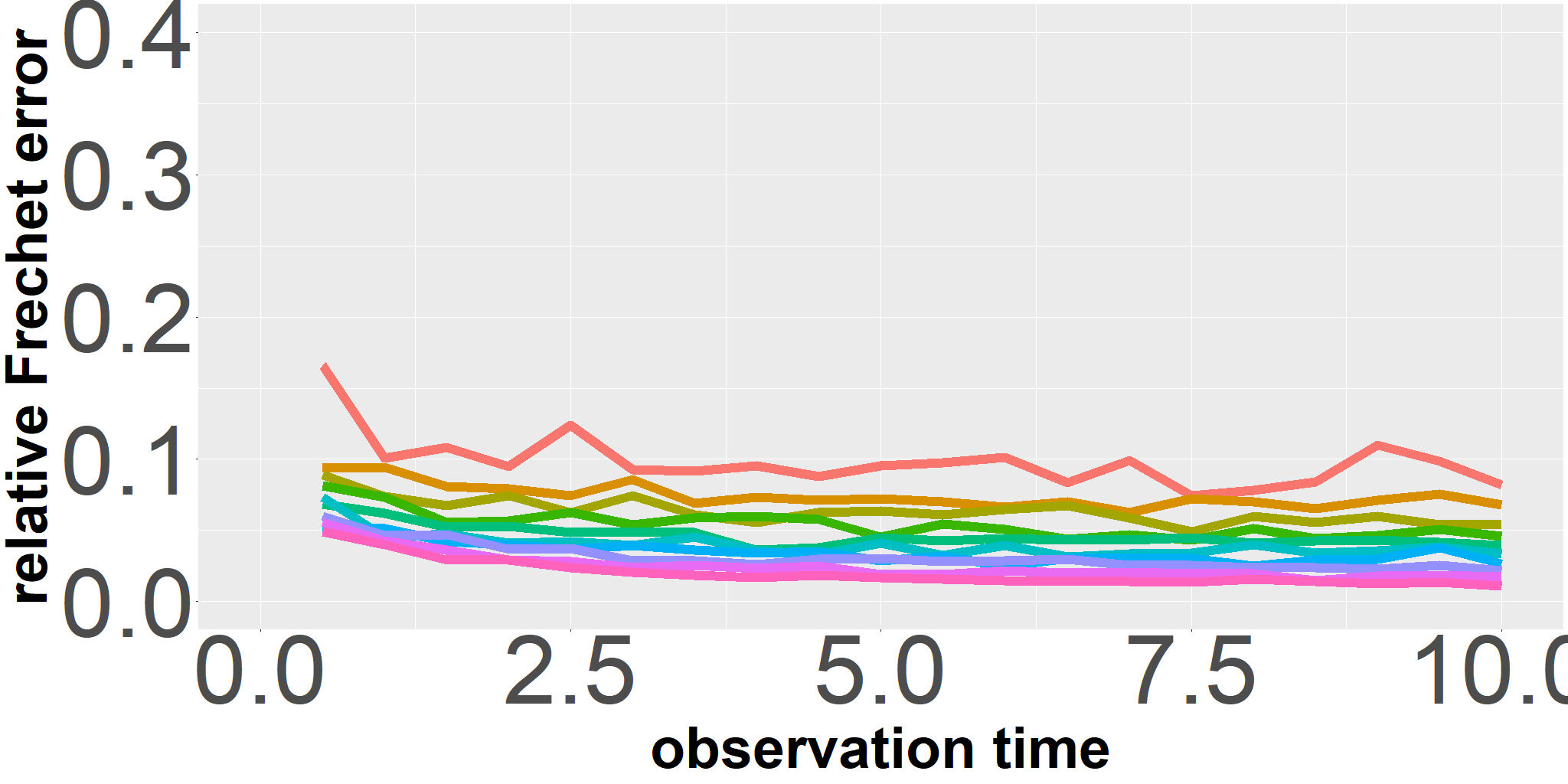} 
\end{tabular}
\caption{As in \Cref{fig:bary_poiNE}, but for the NEC class with $M=75$.}
  \label{fig:bary_poiNEC}
\end{figure}
For the relative Fr\'echet errors we observe significant changes in behavior compared to the relative $(p,C)$-KRD before. Considering the error for the NE class in \Cref{fig:bary_poiNE}, we note that for $C=0.1$ the behavior in $s$ and $t$ is different from the empirical $(p,C)$-KRD. Namely, for fixed $s$, the error is in general not strictly decreasing in $t$ and vice versa for fixed $t$, the error is not always strictly decreasing in $s$. This is an interesting effect arising for small values of the product $st$. A point $y\in \Y$ can only be a support point of a $(p,C)$-barycenter if it is in the intersection of at least $J/2$ balls of size $C^p/2$ around support points of different measures (compare the construction of the centroid set in \eqref{eq:KRCentroid}). Now, for small $s$ and $t$ many support points of the population barycenter are not included in the support of the empirical one, since centroid set of the empirical measures is significantly smaller than the population level one. In particular, this can create situations where an increase in $s$ or $t$ on average adds support points to empirical barycenter, which cause the relative error to increase, since placing mass zero at this location, for small $C$, is actually better than placing a potentially larger mass (since we assumed $(ts)^{-1}$ to be relatively small) at this location. Thus, while asymptotically, the rate in \Cref{thm:frechetboundPoi} is optimal, for certain, sufficiently small, values of $s$, $t$ and $C$, the behavior of the relative Fr\'echet error might be counter-intuitive. For $C=0.1$ and $C=1$, the errors behave quite similarly to the $(p,C)$-KRD setting, though there is essentially no increase in error going from $C=1$ to $C=10$. This is explained by two points. First, the location of the $(p,C)$-barycenter tends to be more centered within the support of the measures (all measures are support on subsets of the unit square), so little transport between the barycenter and the $\mu^i$ occurs at a distance larger than one. Second, the key factor for the increasing error for increasing $C$ in the $(p,C)$-KRD case is the estimation error for the total mass intensities. However, for sufficiently large $C$ the mass of the $(p,C)$-barycenter is the median of the total masses of the $\mu^i$. Since, this quantity is significantly more stable under estimation than the individual total mass intensities, it is to be expected that the mass estimation has little effect on the relative Fr\'echet error. For the error of the NEC class in \Cref{fig:bary_poiNEC} the results are similar to the NE case. We observe similar effects on the dependence of $s$ and $t$ for $C=0.1$ and for $C=1$ and $C=10$, the errors look extremely similar. One notable distinction is the fact that from $C=0.1$ to $C=1$ the errors decrease on average. As before for the NEC class in the $(p,C)$-KRD setting, this can be explained by its cluster structure and $C=1$ being the first value for which inter-cluster transport becomes possible in an UOT plan. This is therefore also the first value of $C$ which allows the $(p,C)$-barycenter to have mass between clusters. Finally, for the error of the PI class in \Cref{fig:bary_poiPI} the value of $C$ only has a minimal effect on the resulting errors. Notably and contrary to the two prior classes, we do not encounter any additional effects for $C=0.1$. This is explained by the in general higher mass intensities of the measures in the PI class, which make the previously described effects due to low values of $s$ and $t$ less likely. Additionally, these measures do not possess any geometrical structures in their support, which could impact the behavior on different scales. There is again little increase in error for increasing $C$, which is in stark contrast to the PI class in the $(p,C)$-KRD setting, where the error increased by multiple orders of magnitude. This is another strong indicator, that the Fr\'echet error is significantly more stable under $C$, due to the stability of total mass intensity of the empirical barycenter opposed to the total mass intensity of the individual measures.

\begin{figure}[t!]
  \begin{tabular}{ccc}
       \includegraphics[width=0.31\textwidth]{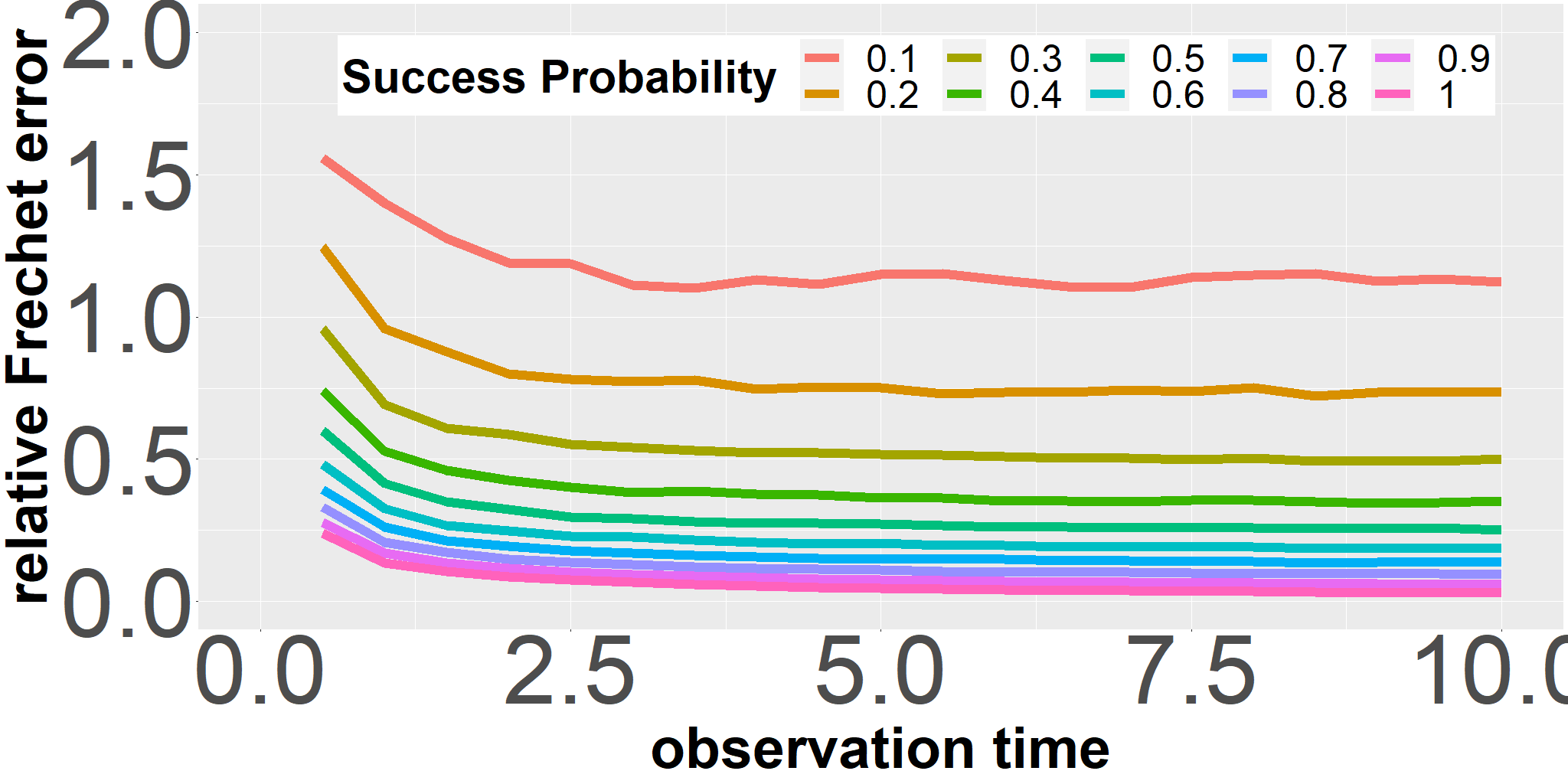} &
      \includegraphics[width=0.31\textwidth]{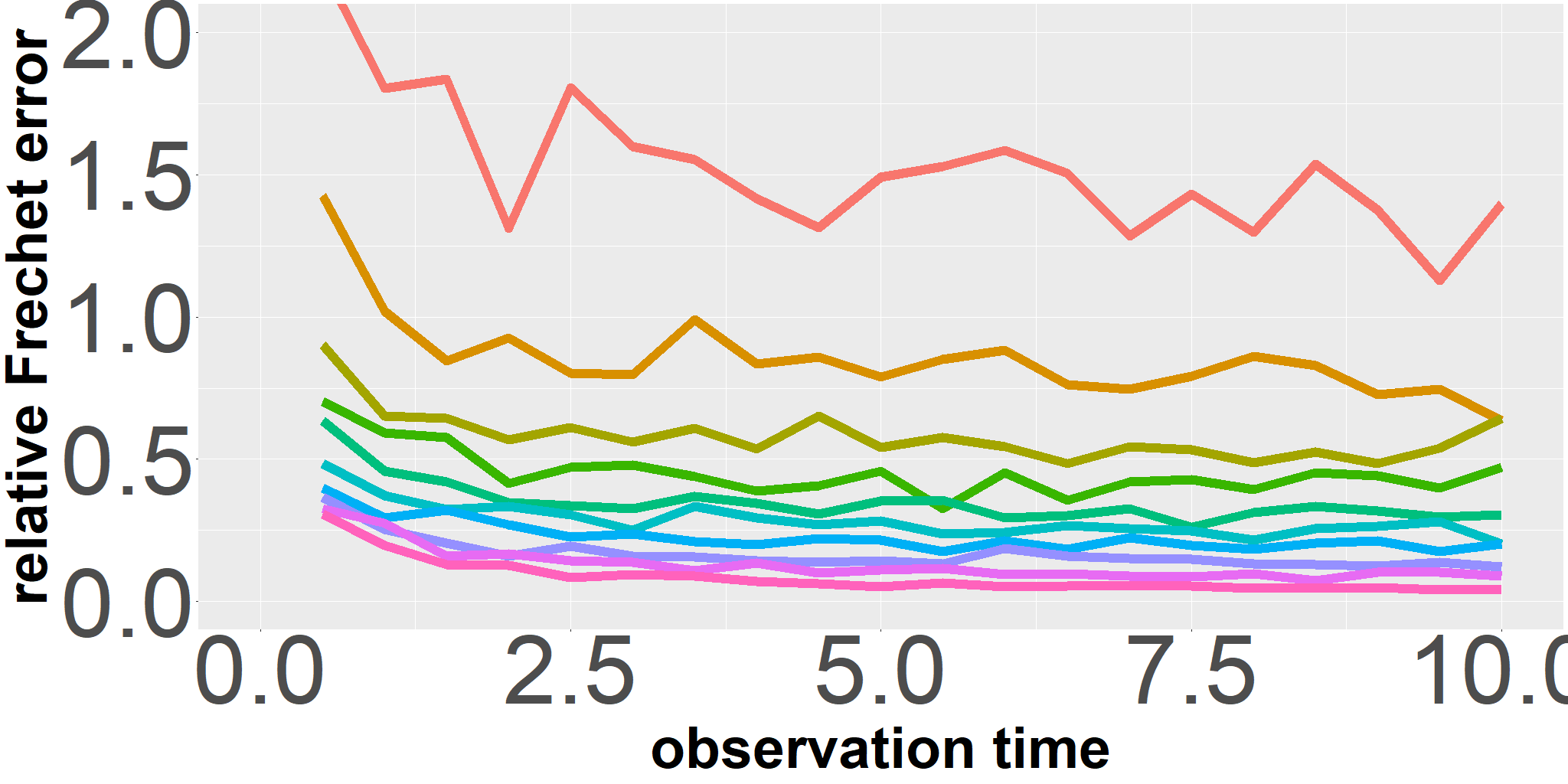} &       \includegraphics[width=0.31\textwidth]{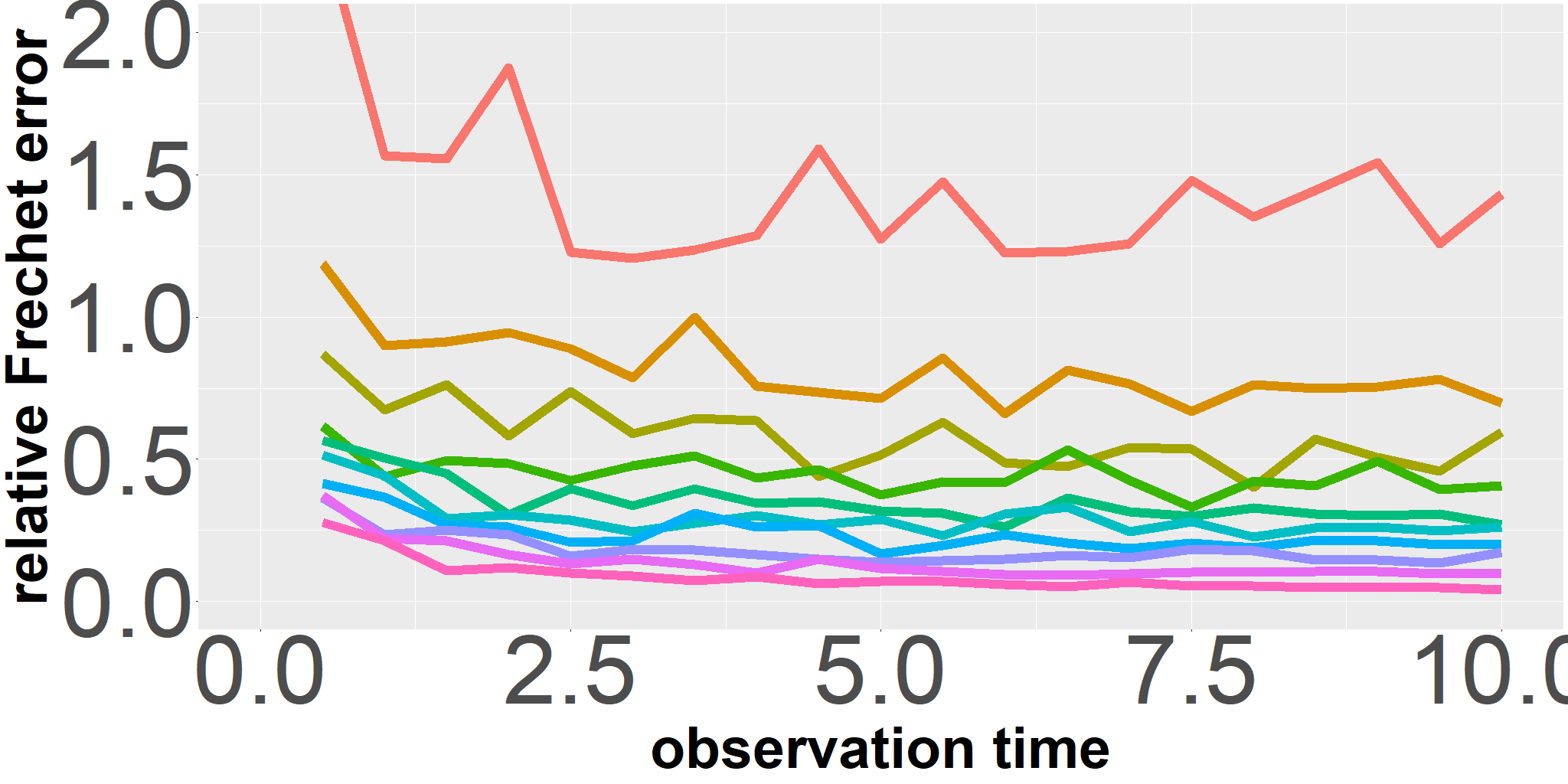} 
  \end{tabular}
  \caption{As in \Cref{fig:bary_poiNE}, but for the PI class with $M = 450$.}
    \label{fig:bary_poiPI}
  \end{figure}
  \begin{figure}[t!]
    \begin{tabular}{ccc}
    \includegraphics[width=0.3\textwidth]{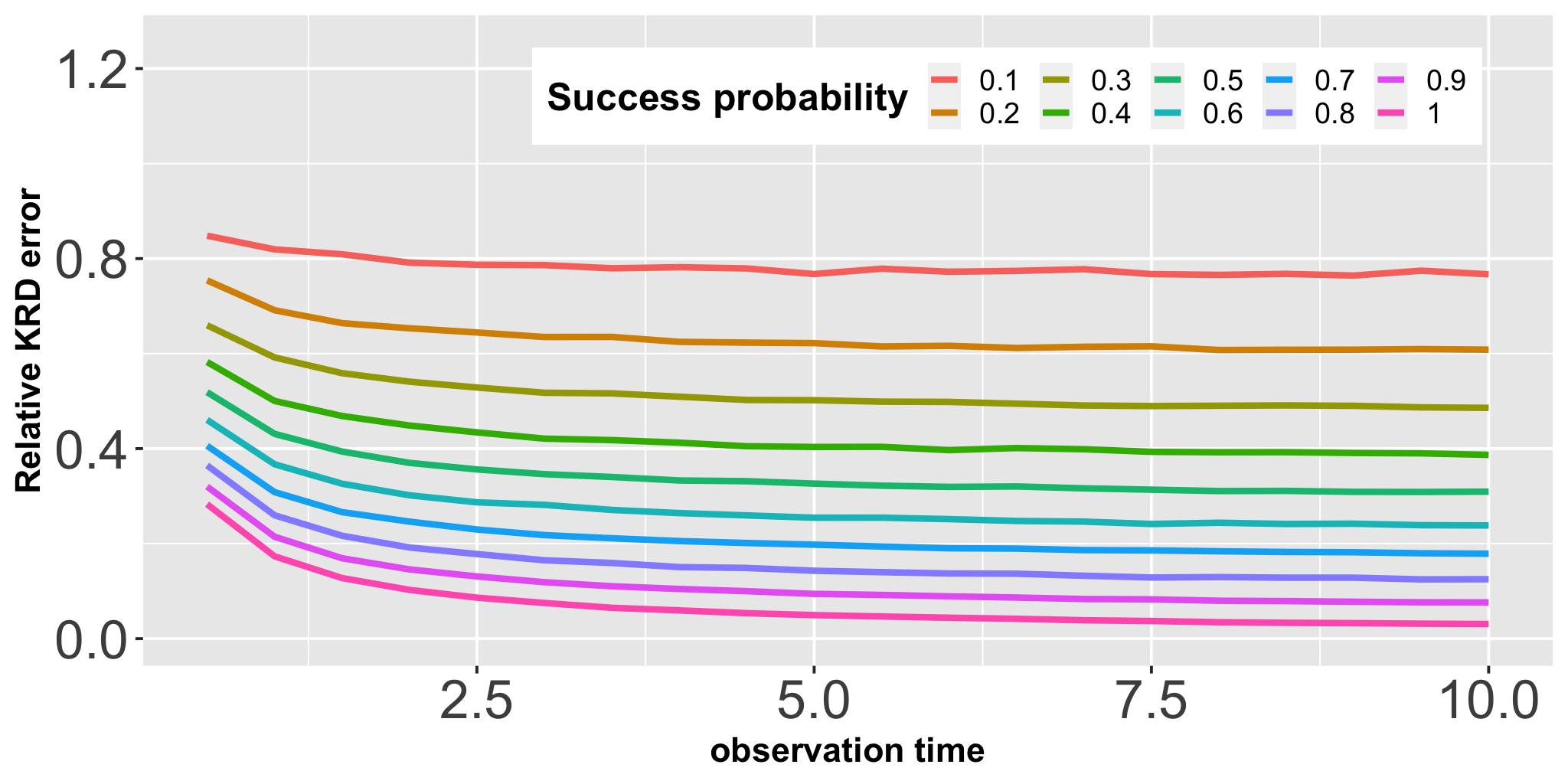} &
        \includegraphics[width=0.3\textwidth]{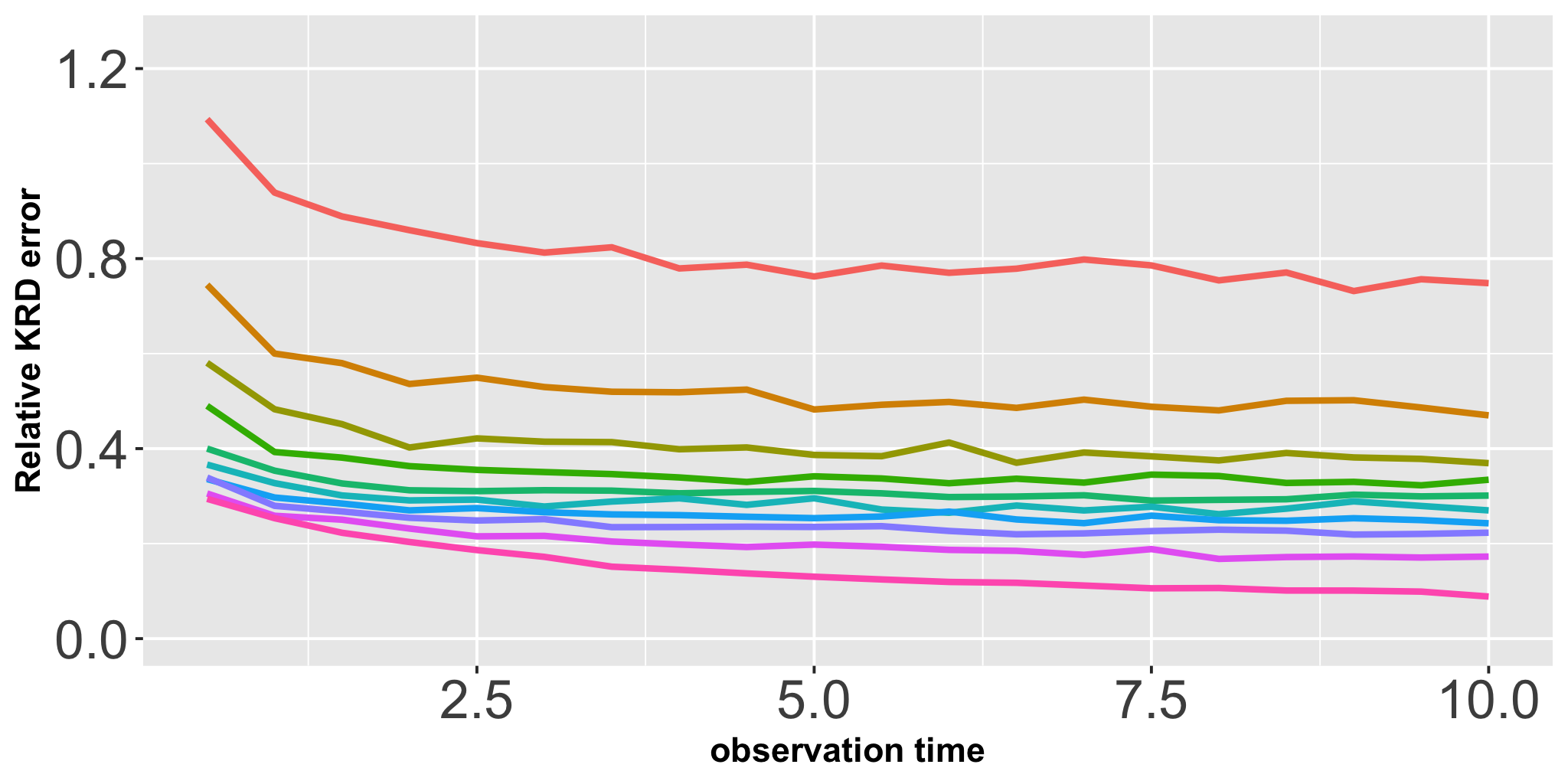} &       \includegraphics[width=0.3\textwidth]{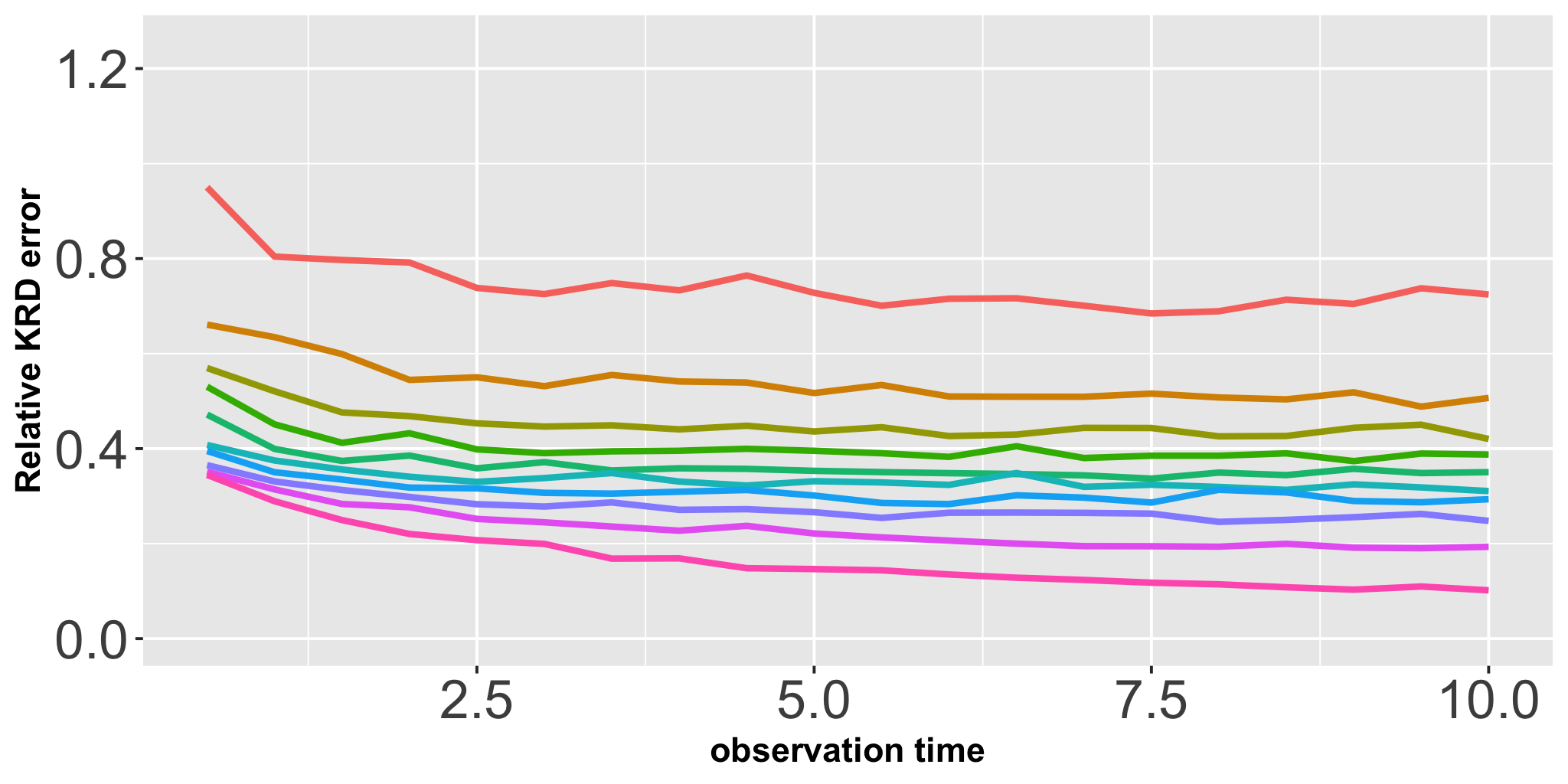} 
    \end{tabular}
    \caption{As in \Cref{fig:bary_poiNE}, but for the PIG class with $M=22$.}
      \label{fig:bary_poiPIG}
    \end{figure}

\subsection{Real Data Example}\label{sec:realdata}
In \Cref{fig:realdata} we consider the $(2,0.1)$-KRD between images which are an excerpt from STED microscopy of adult human dermal fibroblast cells (for the full dataset see \cite{tameling2021colocalization}). Figures \ref{fig:realdata}(a), (b) and Figures \ref{fig:realdata}(c), (d) are visually similar, as they correspond to measurements taken based on two different markers (one at the inner mitochondrial membrane and one at the outer) in the same cells. The $(2,0.1)$-KRD captures this: the pairwise distance between the measures are smallest for these pairs of images. Utilizing UOT on these type of datasets is a potential way of quantifying dissimilarity between the respective measures and extending OT based dissimilarity analysis to measures of unequal total intensity (the total mass intensities in these examples lies between $6200$ and $9500$).

We further want to use this dataset to illustrate the performance of the randomized computational approach for the $(p,C)$-KRD based on the multinomial model (recall \eqref{eq:multinomialmeasure}). The $300\times 300$ images here are specifically chosen such that the true distances can still be computed which allows to compare the expected error of the empirical $(p,C)$-KRD for given sample sizes on this data set. We compare the results obtained from the resampling approach (i.e., the estimator from \eqref{eq:multinomialmeasure}) considered in the multinomial model to the subsampling approach (i.e., the estimator from \eqref{eq:subsamp}) obtained by sampling without replacement from the measures instead. In these simulations the maximum sample size is about $1/5$ of the support sizes. This corresponds to a runtime of about $2.5\%$ of the original problem size. While it is clear by construction that for sufficiently large sample sizes, subsampling yields a smaller error than the resampling (as the error approaches zero if the sample size approaches the support size), for smaller sample sizes the resampling can have a better performance. It yields a relative error below $5\%$ at less than $10\%$ of the original support size in all considered instances. This approximation can be achieved in around $0.5\%$ of the original runtime. The subsampling approach does not reach this level of accuracy for the considered sample sizes. Thus, these simulations suggest that randomized computations based on the multinomial model allow for high accuracy approximations of the $(p,C)$-KRD in real data applications at a significantly lower computational cost than the original problem and that for small sample sizes there are scenarios where the resampling approach yields significantly better performance than the subsampling one. 

\begin{figure}[t!]
  \centering
  \subfloat[][]{\includegraphics[width=0.25\linewidth]{graphics/IM2.png}} 
    \subfloat[][]{\includegraphics[width=0.25\linewidth]{graphics/IM1.png}} 
      \subfloat[][]{\includegraphics[width=0.25\linewidth]{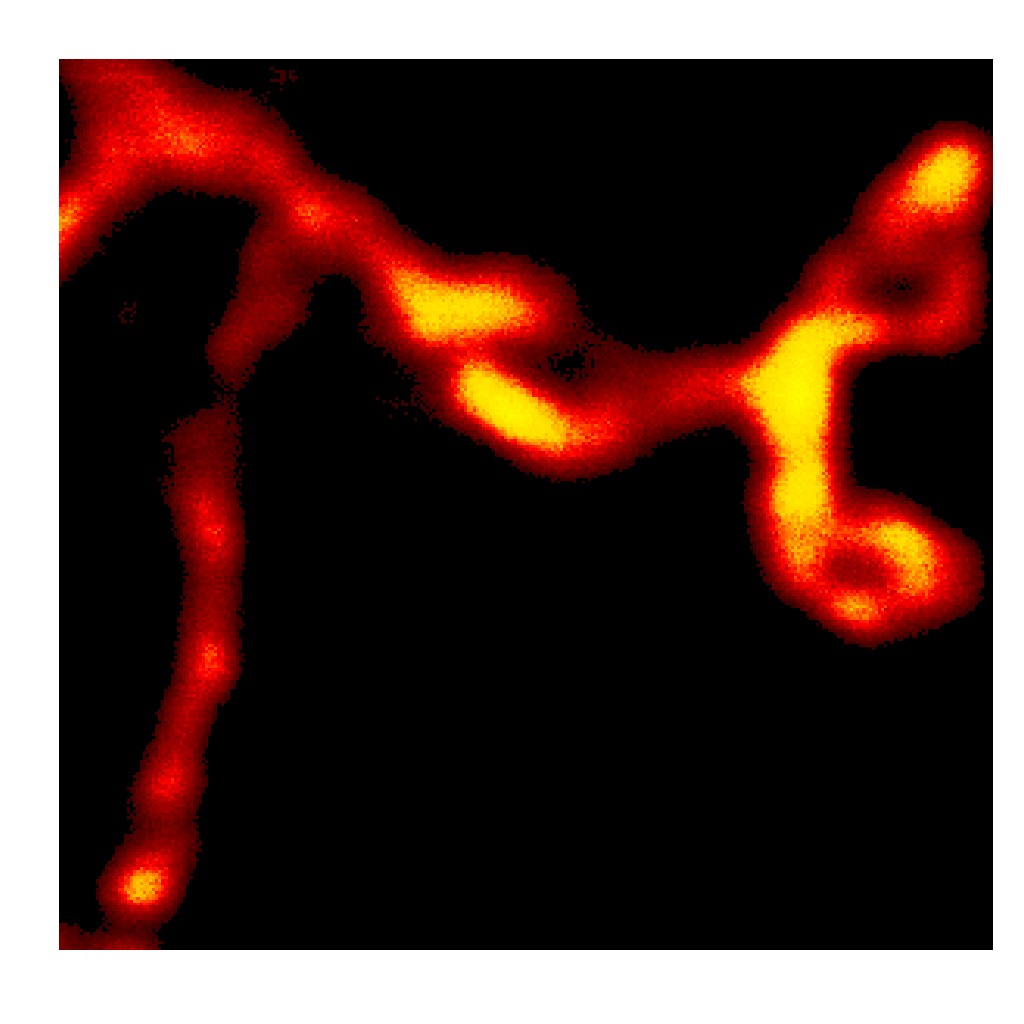}} 
  \subfloat[][]{\includegraphics[width=0.25\linewidth]{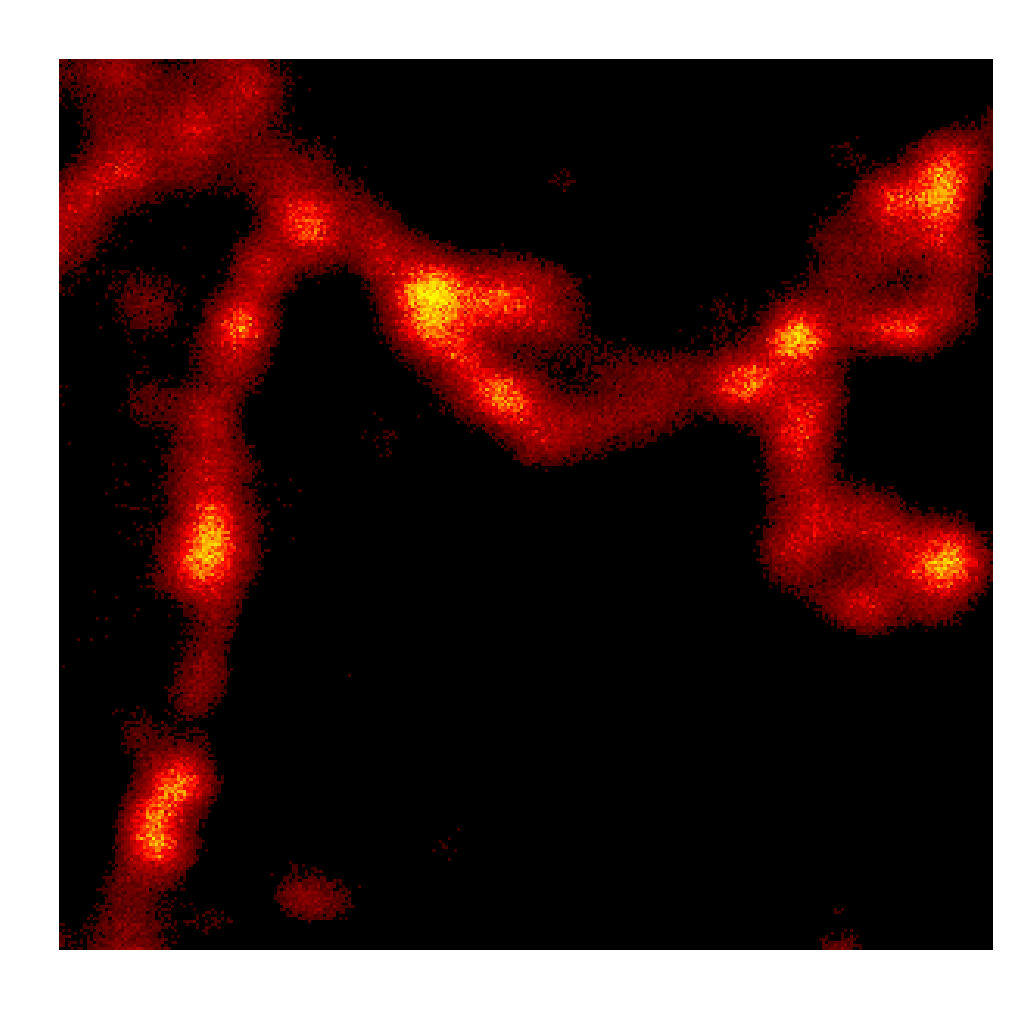}} \hspace{0.01\linewidth}
      \subfloat[][]{\includegraphics[width=0.33\linewidth]{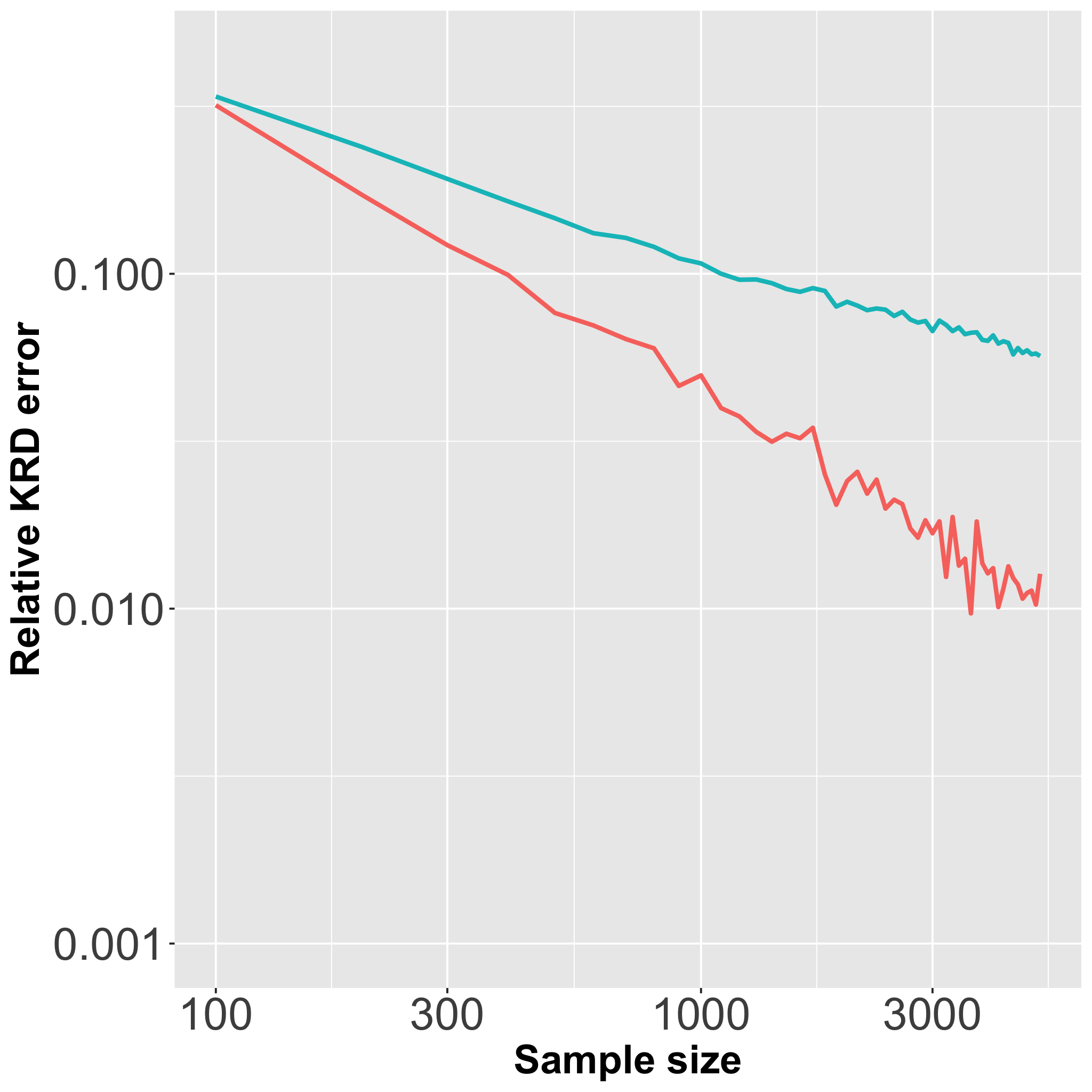}} 
    \subfloat[][]{\includegraphics[width=0.33\linewidth]{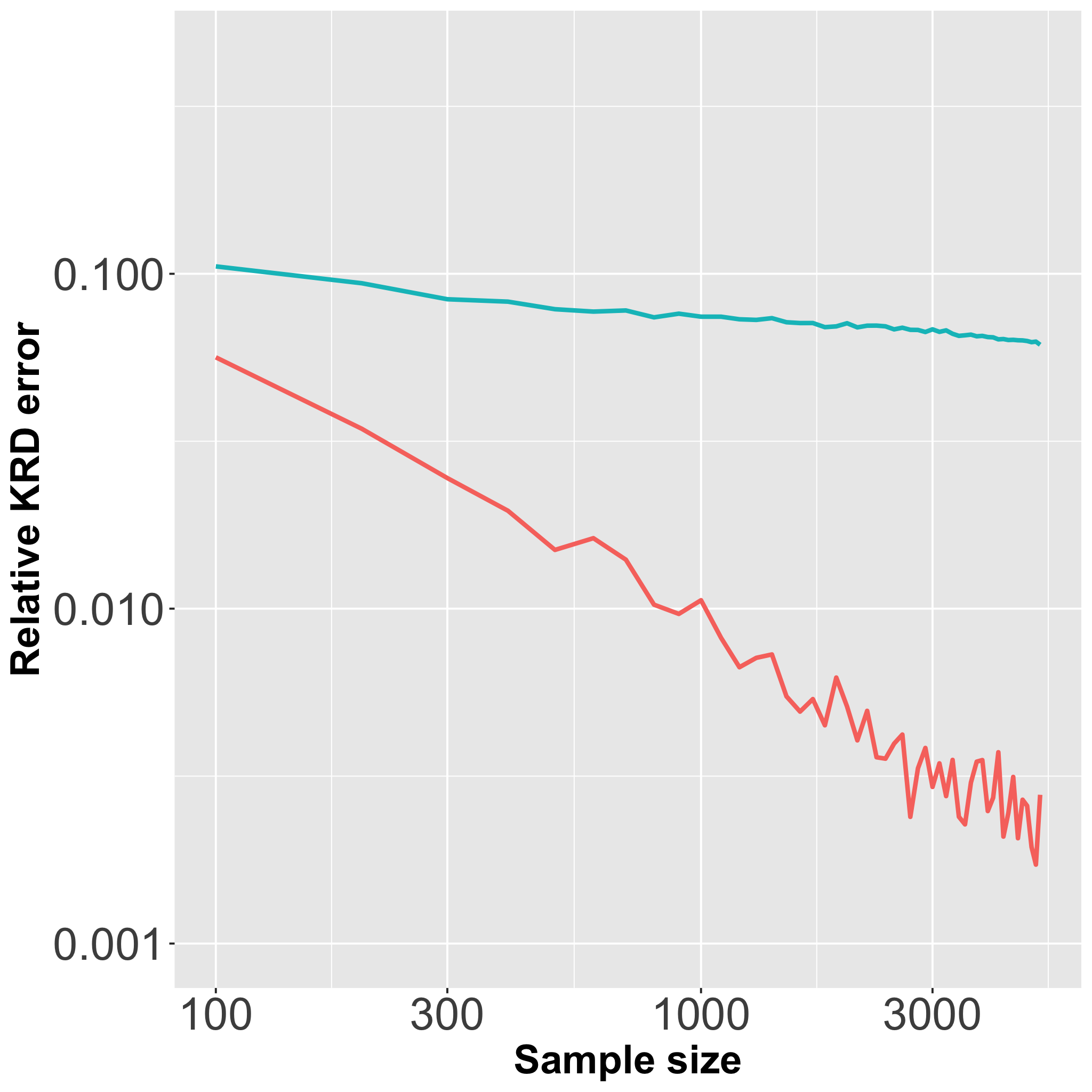}} 
      \subfloat[][]{\includegraphics[width=0.33\linewidth]{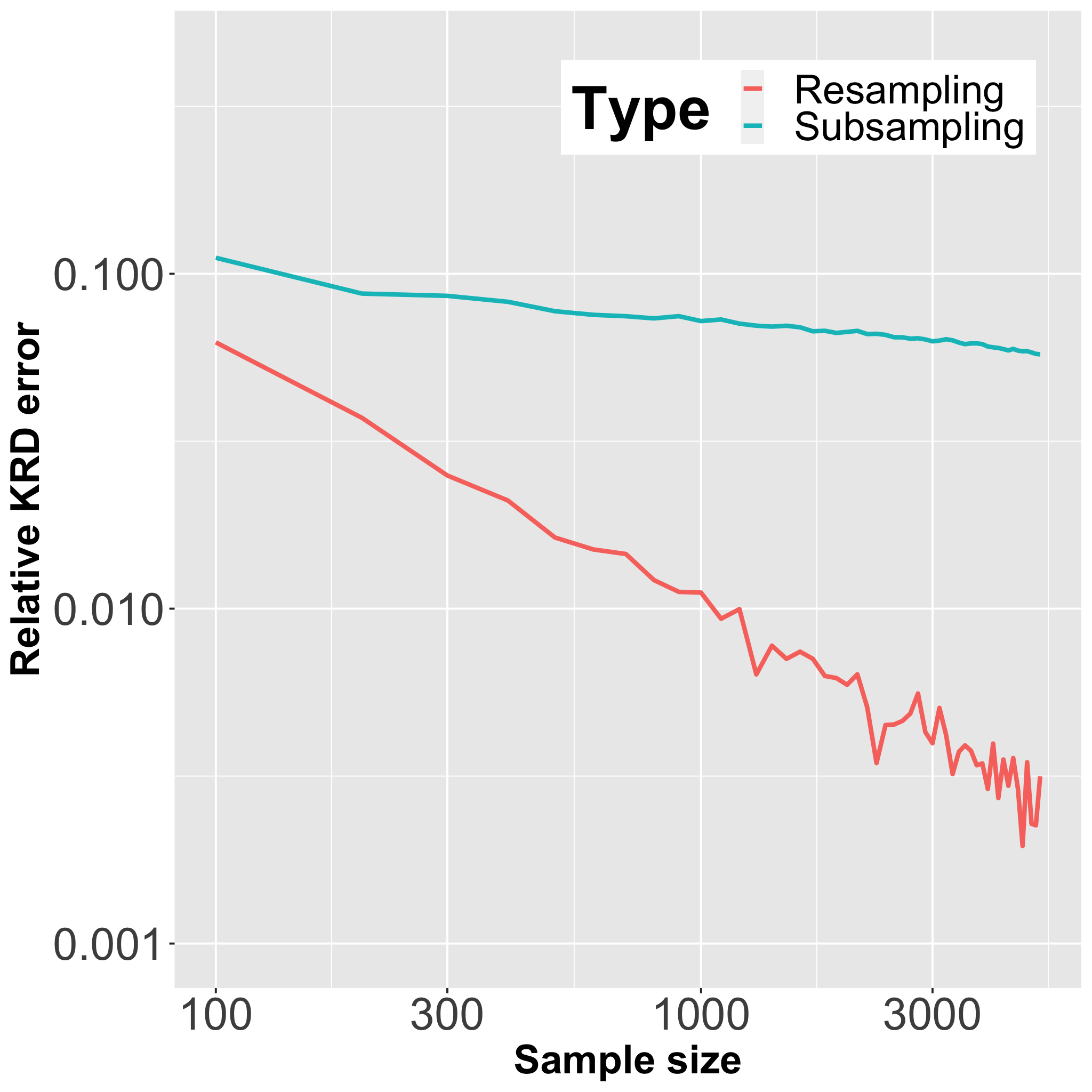}} \hspace{0.01\linewidth}
        \subfloat[][]{\includegraphics[width=0.33\linewidth]{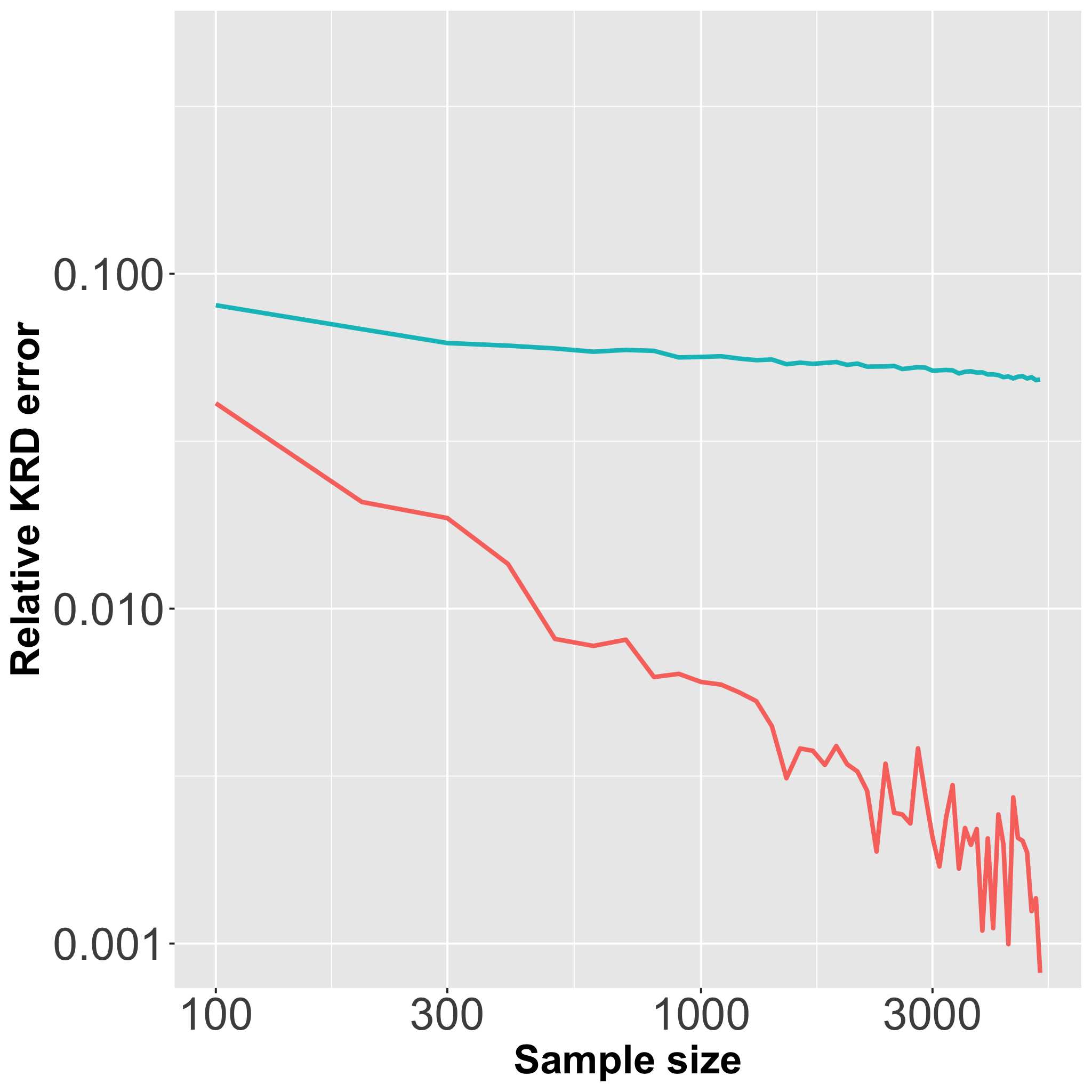}} 
    \subfloat[][]{\includegraphics[width=0.33\linewidth]{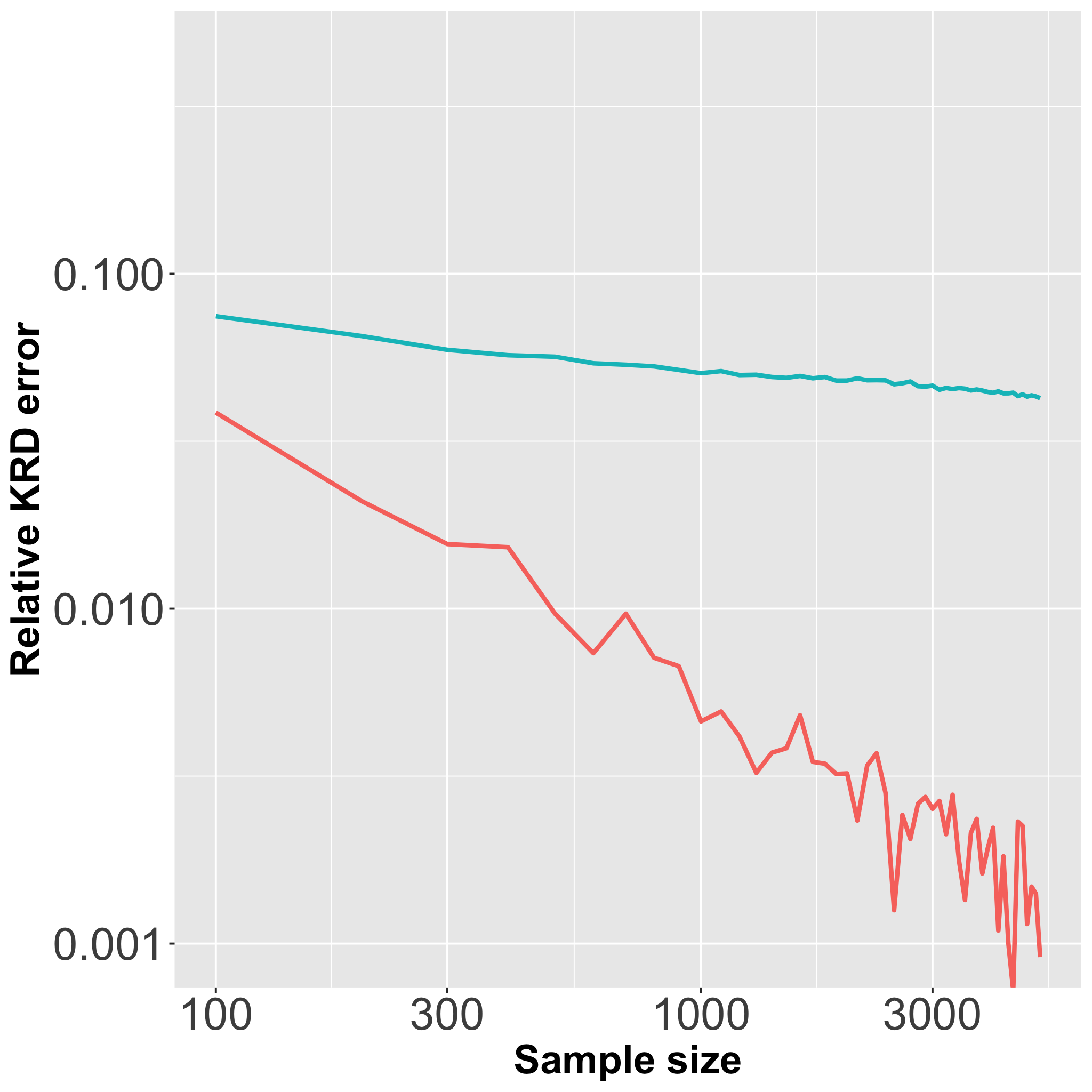}} 
      \subfloat[][]{\includegraphics[width=0.33\linewidth]{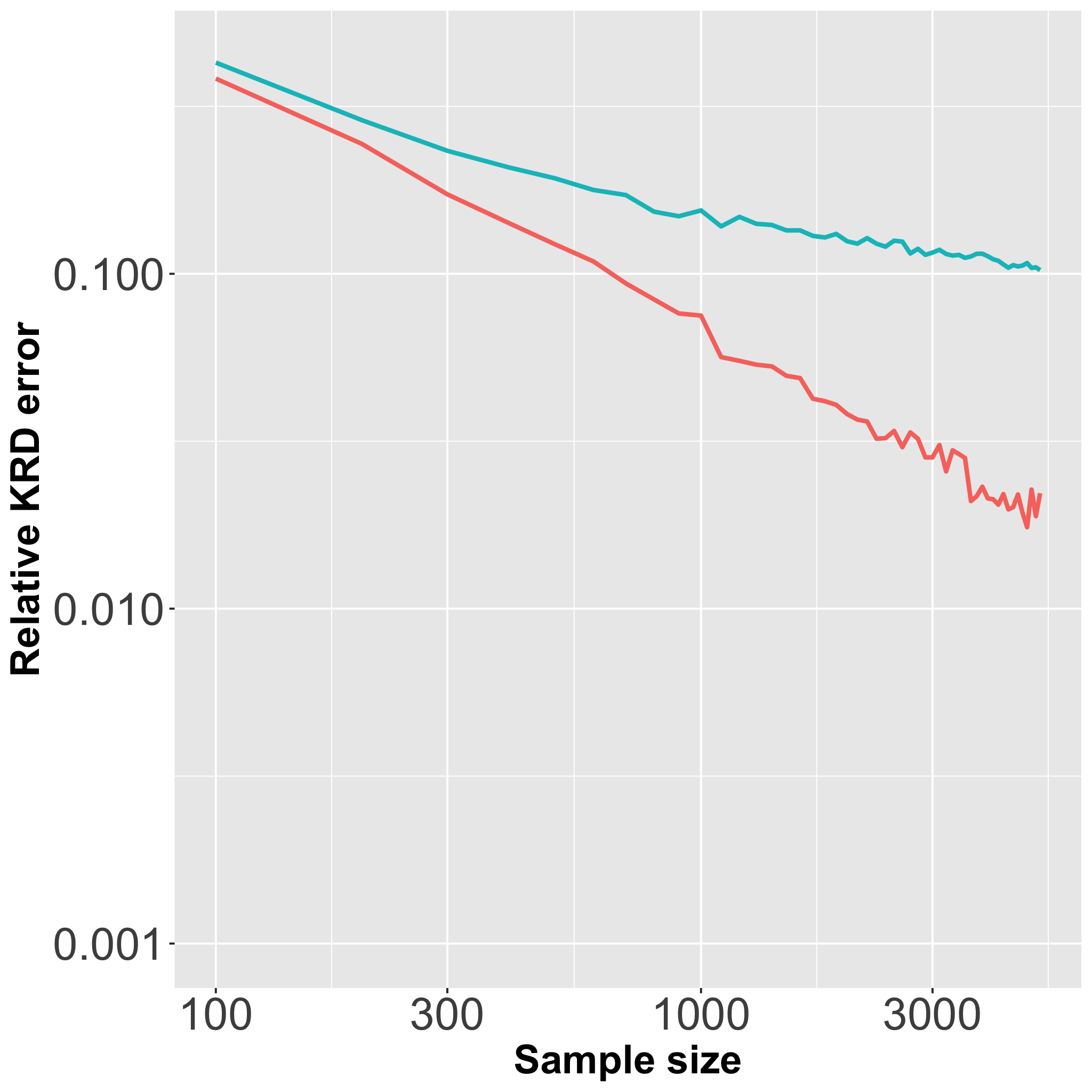}} 
  \caption{\textbf{(a)-(d):} Excerpts of size $300\times 300$ from the STED microscopy data of adult Human Dermal Fibroblasts in \cite{tameling2021colocalization}. The images have on average about $25000$ non-zero pixels. (a) and (c) have been labelled at MIC60 (a mitochondrial inner membrane complex); (b) and (d) have been labelled at TOM20 (translocase of the outer mitochondrial membrane). \textbf{(e)-(j):} Log-log-plots of the relative error for the empirical $(2,0.1)$-KRD (obtained from resampling and subsampling) between the four filament structures in (a)-(d) considered as measures in $[0,1]^2$. \textbf{(e)} Between (a) and (b). \textbf{(f)} Between (a) and (c). \textbf{(g)} Between (a) and (d). \textbf{(h)} Between (b) and (c). \textbf{(i)} Between (b) and (d). \textbf{(j)} Between (c) and (d).}%
  \label{fig:realdata}
\end{figure}

\section*{Acknowledgments}
S. Hundrieser, F.\  Heinemann,  M.\ Klatt, and M. Struleva gratefully acknowledge support from the DFG Research Training Group 2088 \textit{Discovering structure in complex data: Statistics meets optimization and inverse problems}. A.\ Munk gratefully acknowledges support from the DFG CRC 1456 \textit{Mathematics of the Experiment A04, C06}, DFG RU 5381 \textit{Mathematical Statistics in the information age -- Statistical efficiency and computational tractability}, and the DFG Cluster of Excellence 2067 MBExC \textit{Multiscale bioimaging--from molecular machines to networks of excitable cells}.

% \bibliography{kr_stats}{}

\begin{appendix}
  \section{Bounds for the Multinomial Model}\label{sec:multi}
  In this section we provide analogue results  to the convergence statement for the expected deviation of the estimator in KRD and TV (Theorems \ref{thm:TVbound} and \ref{thm:samplingboundKR}), the empirical UOT plan (\Cref{thm:UOTplan_convergence}), and  empirical barycenter (Theorems \ref{thm:frechetboundPoi} and \ref{thm:kr_boundPoi}) but for the multinomial model in \eqref{eq:multinomialmeasure}. We do not explicitly state the convergence statements for the empirical KRD (\Cref{cor:estimation_KRD} and \Cref{thm:estimation_KRD_separate}) as they can be immediately derived from our bound in \eqref{eq:stability bound_intro}.
  The proofs for all subsequent results only differ in the way the respective expectations are bounded, so whenever suitable, we only state relevant differences in the proofs.  Notably, in Appendix \ref{app:sim_mult} we detail some simulations for several classes of measures which showcase that our theoretical results are realized.

  \begin{theorem}[Expected deviation of estimator in TV and KRD]\label{thm:TVboundMult}
  Let $(\mathcal{X},d)$ be a finite metric space and $\mu\in\mathcal{M}_+(\mathcal{X})$ with total mass $\mathbb{M}(\mu)$. Let $\hat{\mu}_N$ be the estimator from \eqref{eq:multinomialmeasure}. Then, for any $p\geq 1$ it holds that 
  \begin{align*}
  \mathbb{E}\left[\KR^p_{p,C}\left(\hat{\mu}_N,\mu\right)\right] \leq C^p  \mathbb{E}\left[\TV\left(\hat{\mu}_N,\mu\right)\right] \leq  \left(C^p\sqrt{\mathbb{M}(\mu)}\sum_{x\in \X}\sqrt{\mu(x)} \right)N^{-\frac{1}{2}}.
  \end{align*}
  \end{theorem}
  \begin{proof}Following the proof of  Theorem~\ref{thm:TVbound}, it suffices to bound the TV norm in expectation,
  \begin{align*}
      \mathbb{E}\left[\TV\left(\hat{\mu}_N,\mu\right)\right] &= \sum_{x\in\mathcal{X}}\mathbb{E}\left[ \right\vert \hat{\mu}_N(x)-\mu(x)\left\vert\right]=\frac{\mathbb{M}(\mu)}{N}\sum_{x\in\mathcal{X}} \mathbb{E}\Big[ \Big\vert\sum_{i=1}^N \mathbbm{1}\{X_i=x\}-N\frac{\mu(x)}{\mathbb{M}(\mu)}\Big\vert\Big]\\
      &\leq \frac{\mathbb{M}(\mu)}{N}\sum_{x\in\mathcal{X}} \sqrt{N\frac{\mu(x)}{\mathbb{M}(\mu)}\left( 1-\frac{\mu(x)}{\mathbb{M}(\mu)}\right)}\leq N^{-\frac{1}{2}}\sqrt{\mathbb{M}(\mu)}\sum_{x\in\mathcal{X}}\sqrt{\mu(x)},
  \end{align*}
  where the inequality follows from the fact the $X_i\sim Ber\left(\mu(x)/\mathbb{M}(\mu)\right)$ for $i=1,\dots,N$.
  \end{proof}
  \begin{theorem}[Expected deviation of estimator in KRD]\label{thm:samplingboundKRMult}
  Let $(\mathcal{X},d)$ be a finite metric space and $\mu\in\mathcal{M}_+(\mathcal{X})$ with total mass $\mathbb{M}(\mu)$. Let $\hat{\mu}_N$ be the estimator from \eqref{eq:multinomialmeasure}. Then, for any $p\geq 1$, resolution $q>1$ and depth $L\in \mathbb{N}$ it holds that
  \begin{align*}
         \mathbb{E}\left[ \KR_{p,C}\left( \hat{\mu}_N,\mu\right)\right]\leq \mathcal{E}_{p,\mathcal{X},\mu}^{\mathrm{Mult}}(C)^{1/p}N^{-\frac{1}{2p}}.
  \end{align*}
  For
  \[
      A_{q,p,L,\X}(l):=\mathrm{diam}(\mathcal{X})^p2^{p-1}\left(q^{-Lp}\lvert  \X \rvert^{\frac{1}{2}}+\left(\frac{q}{q-1}\right)^p \sum_{j=l}^{L}q^{p-jp}\lvert Q_j \rvert^{\frac{1}{2}}\right),
  \]
  the constant is equal to 
  \begin{align*}
      &\mathcal{E}_{p,\mathcal{X},\mu}^{\mathrm{Mult}}(C,q,L)=
      &\begin{cases}
        \mathbb{M}(\mu) A_{q,p,L,\X}(1),\, &\text{if }C\geq 2h_{q,L}(0),\\[3ex]
        \mathbb{M}(\mu) A_{q,p,L,\X}(l),\, &\text{if }2h_{q,L}(l)\leq C< 2h_{q,L}(l-1),\\[3ex]
          \frac{C^p}{2}\sqrt{\mathbb{M}(\mu)}\sum_{x\in \X}\sqrt{\mu(x)}, \, &\text{if } C\leq \left(2h_{q,L}(L) \vee {\min_{x\neq x^\prime}d(x,x^\prime)}\right).
      \end{cases}
  \end{align*}
  Furthermore, for $p=1$ the factor $\frac{q}{(q-1)}$ in $A_{q,1,L,\X}(a,b,l)$ can be removed. Denote
  \[
      \mathcal{E}_{p,\mathcal{X},\mu}^{\mathrm{Mult}}(C):=\inf_{L\in \mathbb{N},q>1} \mathcal{E}_{p,\mathcal{X},\mu}^{\mathrm{Mult}}(C,q,L).
  \]
  \end{theorem}
  \begin{proof}
  The proof of this result only differs from the proof of \Cref{thm:samplingboundKR} by the upper bounds on the relevant expectations. By definition $\mathbb{E}\left[\lvert \mathbb{M}(\hat{\mu}_N)-\mathbb{M}(\mu)\rvert\right]=0$. Furthermore, scaling the expectation by total mass 
  \begin{align*}
  \mathbb{E}\left[\left\vert \hat{\mu}_N^L(\mathcal{C}(x))-\mu^L(\mathcal{C}(x))\right\vert\right]=\mathbb{M}(\mu) \mathbb{E}\left[\left\vert \frac{\hat{\mu}_N^L(\mathcal{C}(x))}{\mathbb{M}(\mu)}-\frac{\mu^L(\mathcal{C}(x))}{\mathbb{M}(\mu)}\right\vert\right],
  \end{align*}
  we notice that $\frac{\hat{\mu}_N^L(\mathcal{C}(x))}{\mathbb{M}(\mu)} \overset{D}{=} \frac{1}{N} \sum_{i=1}^N X_i(x)$,
  where $X_1(x),\ldots,X_N(x)\stackrel{\text{i.i.d.}}{\sim} \mathrm{Ber}\left(a(x)\right)$ with $a(x)\coloneqq \frac{\mu^L(\mathcal{C}(x))}{\mathbb{M}(\mu)}$. Consequently, it holds that 
  \begin{align*}
  \sum_{x\in Q_l}\mathbb{E}\left[\left\vert \hat{\mu}_N^L(\mathcal{C}(x))-\mu^L(\mathcal{C}(x))\right\vert\right]&=\mathbb{M}(\mu)\sum_{x\in Q_l}\mathbb{E}\left[\left\vert\frac{1}{N}\sum_{i=1}^N X_i(x)-a(x)\right\vert\right]\\
  &\leq \mathbb{M}(\mu) \sum_{x\in Q_l} \sqrt{\frac{a(x)(1-a(x))}{N}}\\
  &\leq \mathbb{M}(\mu)\sqrt{\frac{\left\vert Q_l\right\vert}{N}}.
  \end{align*}
  \end{proof}
  Notably, compared to $\mathcal{E}_{p,\mathcal{X},\mu}^{\mathrm{Poi}}(C,q,L)$, the constant $\mathcal{E}_{p,\mathcal{X},\mu}^{\mathrm{Mult}}(C,q,L)$ misses an additional summand for large $C$. This summand corresponds to the estimation error of the total mass intensity of $\hat{\mu}_N$ which is zero by assumptions of the model.
  \begin{remark}
  If $C>\diam(\X)$ and $\mathbb{M}(\mu)=\mathbb{M}(\nu)$, the $(p,C)$-KRD between $\mu$ and $\nu$ is equal to the Wasserstein distance between these two measures.  In particular, for $C>2h_{q,L}(0)$ we recover the respective deviation bounds  by \cite{sommerfeld2019optimal} for the measure estimator  Wasserstein distance. In particular,  for this setting, it holds in the multinomial model for all $N\in \mathbb{N}$ that $\mathbb{M}(\hat{\mu}_N)=\mathbb{M}(\mu)$ which implies for $C>\diam(\X)$ that $\KR_{p,C}(\hat \mu_N, \mu) = W_p(\hat \mu_N, \mu)$. Since for the latter term the parametric rate $N^{-\frac{1}{2p}}$  is already known to be optimal \citep{sommerfeld2019optimal}, our rate in $N$ is sharp.
  \end{remark}
  
  \begin{theorem}[Expected deviation of empirical UOT plans]\label{thm:UOTplan_convergence_mult}
  Let $(\mathcal{X},d)$ be a finite metric space and $\mu, \nu\in\mathcal{M}_+(\mathcal{X})$. Let $\hat{\mu}_N, \hat{\nu}_N$ be their respective  estimators from \eqref{eq:multinomialmeasure}. Then, for any $p\geq 1$ and $C\geq 0$ it follows that 
  \begin{align*}
      & \mathbb{E}\left[\mathcal{H}_\TV\left({\mathbf{P}}^*_{p,C}(\hat{\mu}_N,\hat{\nu}_N),{\mathbf{P}}^*_{p,C}(\mu,\nu)\right)\right] \\
      &\leq 4(|\XC|+1)N^{-1/2} \left( \sqrt{\mathbb{M}(\mu)} \sum_{x \in \XC} \sqrt{\mu(x)} + \sqrt{\mathbb{M}(\nu)} \sum_{x \in \XC} \sqrt{\nu(x)}\right) \\
      &\leq 4(|\XC|+1)^{3/2}N^{-1/2}\mathbb{M}(\mu+\nu).
    \end{align*}
  \end{theorem}
  \begin{proof}
  The assertion follows by combining our stability bound (Theorem \ref{thm:UOTplan_Stability}) with convergence of the estimators $\hat{\mu}_N, \hat{\nu}_N$ to $\mu, \nu$ with respect to the total variation norm (Theorem \ref{thm:TVboundMult}) and the fact that by definition $\mathbb{E}\left[\lvert \mathbb{M}(\hat{\mu}_N)-\mathbb{M}(\mu)\rvert\right]= \mathbb{E}\left[\lvert \mathbb{M}(\hat{\nu}_N)-\mathbb{M}(\nu)\rvert\right]=0$. The last inequality is due to Cauchy–Schwarz.
  \end{proof}
  \begin{remark}
  The convergence rate in $N$ matches with the distributional limit obtained by \cite{klatt2020limit} and \cite{liu2023asymptotic} and is therefore sharp in $N$. In particular, our result explicitly quantifies how the number of support points affect the convergence rate. %
  \end{remark}
  
  \begin{theorem}[Expected deviation of Fr\'echet error]\label{thm:frechetboundMult}
  Let $\mu^1,\ldots,\mu^J\in\msrX$ and denote $\mathcal{X}_i=\mathrm{supp}(\mu^i)$ for $i=1,\dots ,J$. Consider random estimators $\hat{\mu}^1_{N_i},\dots,\hat{\mu}^J_{N_J}\in\msrX$ derived from \eqref{eq:multinomialmeasure} and based on sample size $N_1,\dots,N_J$, respectively. Then it holds for any barycenter $\mu^*$ of the population measures and any barycenter $\hat{\mu}^*$ of the estimators,
  \begin{align*}
  \mathbb{E}\left[\lvert F_{p,C}(\hat{\mu}^\star) - F_{p,C}(\mu^\star) \rvert \right] \leq \frac{2p C^{p-1}}{J}\easysum{i=1}{J} \mathcal{E}_{1,\X_i,\mu^i}^{\mathrm{Mult}}(C)N_i^{-\frac{1}{2}}.
  \end{align*}
  \end{theorem}
  \begin{theorem}[Expected deviation of empirical barycenters]\label{thm:kr_boundMult}
  Let $\mu^1,\ldots,\mu^J\in\msrX$ and denote $\mathcal{X}_i=\mathrm{supp}(\mu^i)$ for $i=1,\dots ,J$. Consider random estimators $\hat{\mu}^1_{N_i},\dots,\hat{\mu}^J_{N_J}\in\msrX$ derived from \eqref{eq:multinomialmeasure} and based on sample size $N_1,\dots,N_J$, respectively. Let $\mathbf{B}^\star$ be the set of $(p,C)$-barycenters of $\mu^1,\dots,\mu^J$ and $\hat{\mathbf{B}}^\star$ the set of $(p,C)$-barycenters of $\hat{\mu}^1_{N_i},\dots,\hat{\mu}^J_{N_J}$. Then, for $p\geq 1$ it holds that
  \begin{align*}
      \mathbb{E}\left[ \sup_{\hat{\mu}^\star \in \hat{\mathbf{B}}^\star} \inf_{\mu^\star \in \mathbf{B}^\star} \KR_{p,C}^p(\hat{\mu}^\star,\mu^\star) \right]
      \leq\frac{p C^{p-1} }{V_{P}J}\easysum{i=1}{J}\mathcal{E}_{1,\X_i,\mu^i}^{\mathrm{Mult}}(C)N_i^{-\frac{1}{2}},
  \end{align*}
  where the constant $V_P$ is defined as in \Cref{thm:kr_boundPoi}.
  \end{theorem}
  The proofs of \Cref{thm:frechetboundMult} and \Cref{thm:kr_boundMult} are deferred to \Cref{app:barycenters_proofs}.
  \begin{remark}
  By the same argument as for the Poisson model (\Cref{rmk:BarycenterSharp}), the convergence rate for the empirical barycenter does not to zero faster than for a single measure. Thus, the $N^{-1/2}$ rate is sharp.
  \end{remark}

  In the following we, derive explicit bounds for $\mathcal{E}_{p,\XC,\mu}^{\mathrm{Mult}}(C)$ for under the structural Assumption \eqref{eq:Xcovering_ass}. To this end, we follow the arguments in \Cref{sec:explicitB}. For the regime,  for $\alpha < 2p$ and $L \to \infty$ we obtain, 
  \begin{align*}
      \mathcal{E}_{p,\XC,\mu}^{\mathrm{Mult}}(C)\leq 
      \begin{cases}
        \mathbb{M}(\mu)A^{1/2}\mathrm{diam}(\XC)^{p} 2^{3p-1} \frac{2^{\alpha/2-p}}{1 - 2^{\alpha/2-p}}, \\[0.1cm] \hspace{2cm} \text{if }C\geq 2h_{L}(0),\\
        \mathbb{M}(\mu)A^{1/2} \mathrm{diam}(\XC)^{p} 2^{3p-1} \frac{2^{(\alpha/2-p)l}}{1 - 2^{\alpha/2-p}}, \\[0.1cm] \hspace{2cm} \text{if }2h_{L}(l)\leq C< 2h_{L}(l-1),\\
          \frac{C^p}{2}\sqrt{\mathbb{M}(\mu)}\sum_{x\in \XC}\sqrt{\mu(x)}, \\[0.1cm] \hspace{2cm} \text{if } C\leq \left(2h_{L}(L) \vee {\min_{x\neq x'}d(x,x')}\right).
      \end{cases}
  \end{align*}
  For $\alpha = 2p$ and $L = \lfloor \frac{1}{\alpha}\log_2 (|\XC|) \rfloor$ we get
  \begin{align*}
      \mathcal{E}_{p,\XC,\mu}^{\mathrm{Mult}}(C)\leq 
      \begin{cases}
        \mathbb{M}(\mu) \mathrm{diam}(\XC)^{p} 2^{3p-1}\left(2^{-p} + A^{1/2} \frac{1}{\alpha} \log_2 |\XC| \right), \\[0.1cm] \hspace{2cm} \text{if }C\geq 2h_{L}(0),\\
        \mathbb{M}(\mu) \mathrm{diam}(\XC)^{p} 2^{3p-1} (2^{-p} + A^{1/2}  \left( \frac{1}{\alpha} \log_2 |\XC| - l\right)), \\[0.1cm] \hspace{2cm} \text{if }2h_{L}(l)\leq C< 2h_{L}(l-1),\\
          \frac{C^p}{2}\sqrt{\mathbb{M}(\mu)}\sum_{x\in \XC}\sqrt{\mu(x)}, \, \\[0.1cm] \hspace{2cm} \text{if } C\leq \left(2h_{L}(L) \vee {\min_{x\neq x'}d(x,x')}\right).
      \end{cases}
  \end{align*}
  Finally, for $\alpha>2p$ and $L=\lfloor \frac{1}{\alpha}\log_2(|\XC|) \rfloor$, it holds
  \begin{align*}
      \mathcal{E}_{p,\XC,\mu}^{\mathrm{Mult}}(C)\leq 
      \begin{cases}
        \mathbb{M}(\mu) \mathrm{diam}(\XC)^{p} 2^{3p-1} |\XC|^{1/2 - p/\alpha}  \left(2^{-p} + A^{1/2}  \frac{2^{\alpha/2-2p}}{2^{\alpha/2-p}-1} \right), \\[0.1cm] \hspace{2cm} \text{if }C\geq 2h_{L}(0),\\
        \mathbb{M}(\mu) \mathrm{diam}(\XC)^{p} 2^{3p-1} \times \\
          \ (2^{-p}|\XC|^{1/2-p/\alpha} + A^{1/2} \frac{2^{\alpha/2-p}}{2^{\alpha/2-p} - 1} \left(|\XC|^{1/2-p/\alpha} - 2^{(\alpha/2-p)(l-1)} \right)), \\[0.1cm] \hspace{2cm} \text{if }2h_{L}(l)\leq C< 2h_{L}(l-1),\\
          \frac{C^p}{2}\sqrt{\mathbb{M}(\mu)}\sum_{x\in \XC}\sqrt{\mu(x)}, \\[0.1cm] \hspace{2cm}\text{if } C\leq \left(2h_{L}(L) \vee {\min_{x\neq x'}d(x,x')}\right).
      \end{cases}
  \end{align*}
  We stress that while these constants do not include the additional term for the estimation of the total mass intensity, their dependency on $\lvert \X \rvert$ is identical to that of the upper bounds on $\mathcal{E}_{p,\mathcal{X},\mu}^{\text{Pois}}(C)$. In particular, the phase transitions still occur depending on whether $\alpha$ is larger than $2p$, smaller than $2p$ or equal to it.

  \section{Bounds for the Bernoulli Model}\label{sec:ber}
  In this section we provide results analogue to previous subsection for the estimator in the Bernoulli model in \eqref{eq:bermeasure}. As before, since the proofs only differ in the way expectations are bounded, we only state relevant differences in the proofs. Simulations, corroborating our theoretical findings are provided in Appendix \ref{app:sim_ber}. 
  \begin{theorem}[Expected deviation of estimator in TV and KRD]\label{thm:TVboundBer}
  Let $(\mathcal{X},d)$ be a finite metric space and $\mu\in\mathcal{M}_+(\mathcal{X})$ with $\mu(x)\in \{ 0,1 \}$ for $x\in \X$. Let $\hat{\mu}_{s_\X}$ be the measure in \eqref{eq:bermeasure}. Then, for any $p\geq 1$ it holds that
  \begin{align*}
      \mathbb{E}\left[ \KR^p_{p,C}\left(\hat{\mu}_{s_\X},\mu\right)\right]\leq  C^p \mathbb{E}\left[\TV\left(\hat{\mu}_{s_\X},\mu\right)\right]\leq 2 C^p\sum_{x\in \X}\left(1-s_x\right).
  \end{align*}
  \end{theorem}
  \begin{proof}
  This proof is identical to the proof of  Theorem~\ref{thm:TVbound} except for the bound on the expectation. For this, note that
  \begin{align*}
      \mathbb{E}\left[\TV\left(\hat{\mu}_{s_\X},\mu\right)\right]= \sum_{x\in\mathcal{X}} \mathbb{E}\left[ \left\vert \frac{1}{s_x}B_x-\mu(x)\right\vert\right]&\leq \sum_{x\in\supp(\mu)} (1-s_x)+s_x(\frac{1}{s_x}-1)
      \leq 2\sum_{x\in\mathcal{X}}(1-s_x),
  \end{align*}
  with $B_x\sim Ber(s_x \mu(x))$ for $s_x\in (0,1]$ for all $x\in \X$.
  \end{proof}
  \begin{theorem}[Expected deviation of estimator in KRD]\label{thm:samplingboundKRBer}
  Let $(\mathcal{X},d)$ be a finite metric space and $\mu\in\mathcal{M}_+(\mathcal{X})$ with $\mu(x)\in \{ 0,1 \}$ for $x\in \X$. Let $\hat{\mu}_{s_\X}$ be the measure in \eqref{eq:bermeasure}. Then, for any $p\geq 1$, resolution $q>1$ and depth $L\in \mathbb{N}$ it holds that
  \begin{align*}
         \mathbb{E}\left[ \KR_{p,C}\left( \hat{\mu}_{s_\X},\mu\right)\right]\leq \mathcal{E}_{p,\mathcal{X},\mu}^{\mathrm{Ber}}(C,q,L)^{1/p}\begin{cases}
     \left(2\sum_{x\in \X} (1-s_x) \right)^{\frac{1}{p}} , \\[0.1cm] \hspace{2cm}\text{if } C\leq \left(2h_{q,L}(L) \vee {\min_{x\neq x^\prime}d(x,x^\prime)}\right) \\[0.4cm]
    \left( \sum_{x\in \X}\frac{1-s_x}{s_x} \right)^{\frac{1}{2p}}, \\[0.1cm] \hspace{2cm} \text{else}. \\
      \end{cases}
  \end{align*}
  The constant $\mathcal{E}_{p,\mathcal{X},\mu}^{\mathrm{Ber}}(C,q,L)$ is equal to $\mathcal{E}_{p,\mathcal{X},\mu}^{\mathrm{Poi}}(C,q,L)$ for all $C>0$, $q>1$ and $L\in \mathbb{N}$. We denote 
  \[
  \mathcal{E}_{p,\mathcal{X},\mu}^{\mathrm{Ber}}(C):=\inf_{L\in \mathbb{N},q>1} \mathcal{E}_{p,\mathcal{X},\mu}^{\mathrm{Ber}}(C,q,L).
  \]
  \end{theorem}
  \begin{proof}
  The proof of this result only differs from the proof of \Cref{thm:samplingboundKR} in the upper bounds on the relevant expectations. Recall the estimator $\hat{\mu}_{s_\X}$ from \eqref{eq:bermeasure} and let $B_x\sim \mathrm{Ber}(s_x \mu(x))$ for $s_x\in (0,1]$ for all $x\in  \X$. It holds that
  \begin{align*}
      \sum_{x\in Q_l}\mathbb{E}\left[\left\vert \hat{\mu}_{s_\X}^L(\mathcal{C}(x))-\mu^L(\mathcal{C}(x))\right\vert\right]&=  \sum_{x\in Q_l}\mathbb{E}\left[\left\vert \sum_{y\in \mathcal{C}(x)}\frac{B_y}{s_y}-\sum_{y\in \mathcal{C}(x)}\mu^L(y)\right\vert\right]\\
      &\leq  \sum_{x\in Q_l} \sqrt{\text{Var}\left( \sum_{y\in \mathcal{C}(x)}\frac{B_y}{s_y} \right)}
      = \sum_{x\in Q_l}  \sqrt{\sum_{y\in \mathcal{C}(x)}s_y^{-2}\text{Var}(B_y)}\\
      &= \sum_{x\in Q_l}\sqrt{\sum_{y\in \mathcal{C}(x)} \frac{(1- \mu(y)s_y)\mu(y)}{s_y}}
      \leq \sqrt{\lvert Q_l \rvert} \sqrt{\sum_{x\in \X} \frac{1-s_x}{s_x}}.
  \end{align*}
  The total mass can be bounded analogously as
  \begin{align} \label{eq:massconvBer}
     \mathbb{E}\left[ \lvert \hat{\mu}^L_{s_\X}(\X)-\mu^L(\X) \rvert \right]\leq \sqrt{\sum_{x\in\X}\frac{1-s_x}{s_x}}.
  \end{align}
  \end{proof}
  Since the constants for the deviation bounds for this model coincide with those for the Poisson model we refer to the previous discussion on their properties. 
  \begin{remark}
  Consider $s_\X$ such that $s_x=s$ for some $s\in(0,1]$ and all $x \in \X$. Note, that for sufficiently small $C$ the upper bound is an equality, since the $(p,C)$-KRD in this setting is proportional to the TV distance and that distance has a closed form solution here. For larger $C$, the expectation in the proof of \Cref{thm:samplingboundKRBer} amounts to bounding the mean absolute deviation of a binomial distribution. This has a closed form solution which scales as the standard deviation of the respective binomial for $s$ not too close to $0$ or $1$ \citep{berend2013sharp}. Hence, in this context the upper bound on the mean absolute deviation in the proof is sharp. So based on the presented approach for the deviation bounds, the upper bound is non-improvable.
  \end{remark}
  
  To state the result analogous to \Cref{thm:UOTplan_Stability},  we consider $\mu$ and $\nu\in\mathcal{M}_+(\XC)$ with respective supports $\XC_\mu := \mathrm{supp}(\mu) = \{x_1, \dots, x_M\}$ and success probabilities $s_x \in [0, 1]$ for $x \in \XC_\mu$, and $\XC_\nu := 
   \mathrm{supp}(\nu) = \{x'_1, \dots, x'_{M'}\}$  with success probabilities $s_{x'} \in [0, 1]$ for $x' \in \XC_\nu$. 
  \begin{theorem}[Expected deviation of empirical UOT plans]\label{thm:UOTplan_convergence_ber}
  Let $(\mathcal{X},d)$ be a finite metric space and $\mu, \nu\in\mathcal{M}_+(\mathcal{X})$.
  Let $\hat{\mu}_{s_{\XC_\mu}}, \hat{\nu}_{s_{\XC_\nu}}$ be their respective  estimators from \eqref{eq:bermeasure}. Then, for any $p\geq 1$ and $C\geq 0$ it follows that 
  \begin{align*}
      & \mathbb{E}\left[\mathcal{H}_\TV\left({\mathbf{P}}^*_{p,C}(\hat{\mu}_{s_{\XC_\mu}}, \hat{\nu}_{s_{\XC_\nu}}),{\mathbf{P}}^*_{p,C}(\mu,\nu)\right)\right] \\
      &\leq 4(|\XC|+1) \left(2 \sum_{x \in \XC_\mu}\left(1-s_x\right) + 2 \sum_{x' \in \XC_\nu}\left(1-s_{x'}\right) + \sqrt{\sum_{x \in \XC_\mu} \frac{1-s_x}{s_x}} + \sqrt{\sum_{x' \in \XC_\nu} \frac{1-s_{x'}}{s_{x'}}} \right)\\
      &\leq 4(|\XC|+1) \left(2(\mathbb{M}(\mu+\nu)-2) + \sqrt{\sum_{x \in \XC_\mu} \frac{1-s_x}{s_x}} + \sqrt{\sum_{x' \in \XC_\nu} \frac{1-s_{x'}}{s_{x'}}} \right)
    \end{align*}
  \end{theorem}
  \begin{proof}
  The assertion follows by combining our stability bound (Theorem \ref{thm:UOTplan_Stability}) with convergence of the estimators $\hat{\mu}_{s_{\XC_\mu}}, \hat{\nu}_{s_{\XC_\nu}}$ to $\mu, \nu$ with respect to the total variation norm and in terms of their masses \eqref{eq:massconvBer}.
  \end{proof}

  \begin{theorem}[Expected deviation of Fr\'echet error]\label{thm:frechetboundBer}
  Let $\mu^1,\ldots,\mu^J\in\msrX$ and denote $\mathcal{X}_i=\mathrm{supp}(\mu^i)$ for $i=1,\dots ,J$. Consider (random) estimators $\hat{\mu}^1_{s_{\X_1}},\dots,\hat{\mu}^J_{s_{\X_J}}\in\msrX$ derived from \eqref{eq:bermeasure}. Then,
  \begin{align*}
  \mathbb{E}\left[\lvert F_{p,C}(\hat{\mu}^\star) -F_{p,C}(\mu^\star) \rvert \right] \leq \frac{2p C^{p-1}}{J}\easysum{i=1}{J} \mathcal{E}_{1,\X_i,\mu^i}^{\mathrm{Ber}}(C)\psi(s_{\X_i}),
  \end{align*}
  where $\psi$ is given by
  \begin{align*}
     \psi(s_\X)= \begin{cases}
     (2\sum_{x\in \X} (1-s_x), & \text{if }C\leq \min_{x\neq x^\prime}d(x,x^\prime) \\
    \left( \sum_{x\in \X}\frac{1-s_x}{s_x} \right)^{\frac{1}{2}}, &\text{else}. \\
      \end{cases}
  \end{align*}
  \end{theorem}
  \begin{theorem}[Expected deviation of empirical barycenters]\label{thm:kr_boundBer}
  Let $\mu^1,\ldots,\mu^J\in\msrX$ and denote $\mathcal{X}_i=\mathrm{supp}(\mu^i)$ for $i=1,\dots ,J$. Consider (random) estimators $\hat{\mu}^1_{s_{\X_1}},\dots,\hat{\mu}^J_{s_{\X_J}}\in\msrX$ derived from \eqref{eq:bermeasure}. Let $\mathbf{B}^\star$ be the set of $(p,C)$-barycenters of $\mu^1,\dots,\mu^J$ and $\hat{\mathbf{B}}^\star$ the set of $(p,C)$-barycenters of $\hat{\mu}^1_{s_{\X_1}},\dots,\hat{\mu}^J_{s_{\X_J}}$. Then, for $p\geq 1$ it holds that
  \begin{align*}
      \mathbb{E}\left[ \sup_{\hat{\mu}^\star \in \hat{\mathbf{B}}^\star} \inf_{\mu^\star \in \mathbf{B}^\star} \KR_{p,C}^p(\hat{\mu}^\star,\mu^\star) \right]
      \leq\frac{p C^{p-1} }{V_{P}J}\easysum{i=1}{J} \mathcal{E}_{1,\X_i,\mu^i}^{\mathrm{Ber}}(C)\psi(s_{\X_i}),
  \end{align*}
  where $\psi$ is defined as in \Cref{thm:frechetboundBer} and $V_P$ is defined as in \Cref{thm:frechetboundPoi}.
  \end{theorem}
  The proofs of \Cref{thm:frechetboundBer} and \Cref{thm:kr_boundBer} are deferred to \Cref{app:barycenters_proofs}.

  \section{A Lift to the Balanced Optimal Transport Problem}\label{app:liftOT}
  A key tool in establishing properties of the $(p,C)$-KRD and the $(p,C)$-barycenter is the lift of these problems to the space of probability measures by augmenting the space $\X$ with a dummy point having a fixed distance to all points in $\X$, see \cite[Section 3.1]{heinemann2022kantorovich} for details. For a fixed parameter $C>0$, consider a dummy point $\dum$ and define the augmented space $\tilde{\X}\coloneqq \X\cup \{\dum\}$ with metric cost
  \begin{align}\label{eq:augmetric}
  \tilde{d}^p_C(x,x^\prime)=
  \begin{cases}
  d^p(x,x^\prime)\wedge C^p,\, &\text{if }x,x^\prime\in \mathcal{X},\\[1ex]
  \frac{C^p}{2},\, &\text{if }x\in\mathcal{X},\,x^\prime=\dum,\\[1ex]
  \frac{C^p}{2},\, &\text{if }x=\dum,\,x^\prime\in\mathcal{X},\\[1ex]
  0,\, &\text{if }x=x^\prime=\dum.
  \end{cases}
  \end{align}
  Consider the subset $\mathcal{M}_+^B(\X)\coloneqq \left\lbrace \mu \in \msrX\, \mid\, \mathbb{M}(\mu)\leq B\right\rbrace\subset \mathcal{M}_+(\mathcal{X})$ of non-negative measures whose total mass is bounded by $B$. Setting $\tilde{\mu}\coloneqq\mu + (B-\mathbb{M}(\mu))\delta_\dum$, any measure $\mu \in \mathcal{M}_+^B(\mathcal{X})$ defines an \emph{augmented measure} $\tilde{\mu}$ on $\X$ such that $\mathbb{M}(\tilde{\mu})=B$. For any $\mu,\nu\in \msrX$ and their augmented versions $\tilde{\mu},\tilde{\nu}\in \msrX$ it holds \begin{align}\label{eq:liftOT}
      \KR^p_{C,p}(\mu,\nu)=\tilde{\mathrm{OT}}_{\tilde{d}^p_C}(\tilde{\mu},\tilde{\nu}).
  \end{align}
  Here, $\tilde{\mathrm{OT}}_{\tilde{d}^p_C}$ denotes the OT cost defined for measures $\mu,\nu$ on $(\tilde{\X},\tilde{d})$ with $\mathbb{M}(\mu)=\mathbb{M}(\nu)$ as 
  \[
      \tilde{\mathrm{OT}}_{\tilde{d}^p_C}(\mu,\nu)\coloneqq \min_{\pi\in \Pi_{=}(\mu,\nu)} \sum\limits_{x,x^\prime \in \X}\tilde{d}^p_C(x,x^\prime)\pi(x,x^\prime),
  \]
  where the set of couplings $\Pi_{=}(\mu,\nu)$ is the set $\Pi_{\leq}(\mu,\nu)$ with inequalities replaced by equalities. 
  
  Similarly, the $(p,C)$-barycenter problem can be augmented. For this, let $\tilde{\Y}\coloneqq\Y\cup\{\dum\}$ endowed with the metric $\tilde{d}_C$ in \eqref{eq:augmetric} (replace $\X$ by $\Y$ and recall that $\X\subset\Y$) and augment the measures $\mu^1,\ldots,\mu^J$ to $\tilde{\mu}^1,\ldots,\tilde{\mu}^J$ where $\tilde{\mu}_i=\mu^i+\sum_{j\neq i} \mathbb{M}(\mu^i) \delta_\dum$ for $1\leq i\leq J$. In particular, it holds $\mathbb{M}(\tilde{\mu}_i)=\sum_{i=1}^J \mathbb{M}(\mu^i)$ and the \emph{augmented $p$-Fr\'echet functional} is defined as
  \begin{align*}
      \tilde{F}_{p,C}(\mu):=\frac{1}{J}\sum_{i=1}^{J}\tilde{\mathrm{OT}}_{\tilde{d}^p_C}(\tilde{\mu}_i,\mu).
  \end{align*}
  Any minimizer of $\tilde{F}_{p,C}$ is referred to as augmented $(p,C)$-barycenter.
  
  \subsubsection*{LP-Formulation for the $\mathbf{(p,C)}$-Barycenter}
  \label{app:sec_LP}
  According to \cite[Lemma 3.2]{heinemann2022kantorovich}, the augmented $(p,C)$-barycenter problem can be rewritten as a linear program based on the centroid set $\tilde{\mathcal{C}}_{KR}(J,p,C)= \mathcal{C}_{KR}(J,p,C) \cup \{\dum \}$ (recall \eqref{eq:KRCentroid} for the definition of $\mathcal{C}_{KR}(J,p,C)$) of the augmented measures. This yields
  \begin{alignat}{3}\label{eq:lpbary}
      \underset{\pi^{(1)},\dots,\pi^{(J)},a}{\min} \quad & \frac{1}{J} \easysum{i=1}{J}\easysum{j=1}{\lvert \tilde{\mathcal{C}}_{\KR}(J,p,C) \rvert}&&\easysum{k=1}{M_i}\pi^{(i)}_{jk}c_{jk}^i \nonumber \\
  \text{s.t.} \quad \easysum{k=1}{M_i}\pi^{(i)}_{jk}&=a_j, &&\forall \ i=1,\dots ,J, \forall j=1,\dots,\lvert \tilde{\mathcal{C}}_{\KR}(J,p,C) \rvert, \nonumber \\
  \easysum{j=1}{\lvert \tilde{\mathcal{C}}_{\KR}(J,p,C) \rvert}\pi^{(i)}_{jk}&=b^i_k, &&\forall \ i=1,\dots ,J , \forall k=1,\dots ,M_i, \\
  \pi^{(i)}_{jk}&\geq 0,  &&\forall i=1,\dots ,J,  \forall j=1,\dots ,\lvert \tilde{\mathcal{C}}_{\KR}(J,p,C)\rvert , \nonumber \\& &&\forall k=1,\dots ,M_i,\nonumber 
  \end{alignat}
  where $M_i=\lvert \tilde{\mathcal{X}}_i \rvert$ for each $1\leq i\leq J$ is the cardinality of the support of the augmented measure $\tilde{\mu}_i$. Here, $c^i_{jk}$ denotes the distance between the $j$-th point of $\lvert \tilde{\mathcal{C}}_{\KR}(J,p,C) \rvert$ and the $k$-th point in the support of $\tilde{\mu}_i$, while $b^i$ is the vector of masses corresponding to $\tilde{\mu}_i$.

  \section{Auxiliary Results and Deferred Proofs}\label{app:proofs}

    \subsection{Proof of Stability Bound of KRD via Dual Formulation} \label{app:stability_proofs}
    
    \begin{proof}[Proof of Lemma \ref{lem:stability_dual}]
    We first observe that the second assertion follows from $|x-y| \leq x^{1-p}\left|x^p-y^p\right|$ for all $x, y \geq 0, p \geq 1$. To prove \eqref{eq:KRD_duality_stabilityBound}, 
    note that if $\mu^i = 0$ or $\nu^i = 0$ for each $i \in \{1, 2\}$, then by $\Pi_\leq(\mu^i,\nu^i)= \{0\}$ it follows that 
    \begin{align}\label{eq:MassRepresentation}
      \KR_{p,C}(\mu^i, \nu^i) = \frac{C^p}{2}  \left(\mathbb{M}(\mu^i) + \mathbb{M}(\nu^i)\right),
    \end{align}
     and the assertion follows by triangle inequality, 
    \begin{align*}
      |\KR_{p,C}^p(\mu^1, \nu^1) - \KR_{p,C}^p(\mu^2, \nu^2)| 
      = \;&  \frac{C^p}{2}\Big|\mathbb{M}(\mu^1) + \mathbb{M}(\nu^1)- \mathbb{M}(\mu^2) - \mathbb{M}(\nu^2)\Big|\\
      \leq \;& \frac{C^p}{2}\left(\left|\mathbb{M}(\mu^1) - \mathbb{M}(\mu^2)\right| + \left|\mathbb{M}(\nu^1) - \mathbb{M}(\nu^2)\right|\right).
    \end{align*}
    Moreover, if $C\leq \min\{\dist(\mu^1, \nu^1), \dist(\mu^2, \nu^2)\}$, then by \citet[Lemma 2.1]{heinemann2022kantorovich} the UOT plan between $\mu^i$ and $\nu^i$ is to not transport anything, and delete the excess mass at cost $C^p/2$, which yields \eqref{eq:MassRepresentation} and again inequality \eqref{eq:KRD_duality_stabilityBound} follows by triangle inequality. 
    
    For the remaining case, we augment $\mu^i$,  $\nu^i$ to $\tilde \mu^i$, $\tilde \nu^i$ on $\tilde \XC = \XC \cup \{\mathfrak{d}\}$ as spelled out in \Cref{app:liftOT} with total mass $B \coloneqq \max(\mathbb{M}(\mu^i),\mathbb{M}(\nu^i))$. Then, by the representation of the UOT cost in terms of the augmented OT cost (see Equation \eqref{eq:liftOT}), it follows that 
    \begin{equation*}
        \KR_{p,C}^p(\mu^i, \nu^i) = \tilde{\mathrm{OT}}_{\tilde d^p_C}(\tilde{\mu}^i, \tilde{\nu}^i) = \max_{\substack{f, g: \tilde \XC \to \RR \\ f(x)+g(x') \leq \tilde d^p_C(x,x')}} \sum_{x \in \tilde \XC} f(x) \tilde \mu^i(x)+\sum_{x' \in \tilde \XC} g(x') \tilde \nu^i(x').
    \end{equation*}
    By \citet[Remark 1.13]{villani2003topics} there always exist optimal potentials $f$ and $g$ such that they are bounded by $\|\tilde d^p_C\|_\infty \leq C^p$ in absolute value. %
    Hence, upon denoting such pairs of optimal potentials for $\tilde{\mathrm{OT}}_{p,C}^p(\mu^i, \nu^i)$ by $(f_i, g_i)$, it follows that
    \begin{align*}
        \KR_{p,C}^p(\mu^1, \nu^1) - \KR_{p,C}^p(\mu^2, \nu^2) 
        \leq\;& \sum_{x \in \XC} f_1(x) (\tilde \mu^1 - \tilde \mu^2)(x)+\sum_{x' \in \tilde \XC} g_1(x') (\tilde \nu^1 - \tilde \nu^2)(x')\\
        \leq\;& C^p\left(\sum_{x\in \tilde \XC} |(\tilde \mu^1 - \tilde \mu^2)(x)| + \sum_{x\in \tilde \XC} |(\tilde \mu^1 - \tilde \mu^2)(x)| \right)\\
        =\;& C^p\left( \TV(\tilde \mu^1, \tilde \mu^2) + \TV(\tilde \nu^1, \tilde \nu^2)\right)  \\
        \leq\;& 2C^p\left( \TV( \mu^1, \mu^2) + \TV(\nu^1, \nu^2)\right),
    \end{align*}
    where the last inequality follows from Lemma \ref{thm:TVbound_augmentation}. Exchanging $(\mu^1, \nu^1)$ with $(\mu^2, \nu^2)$ yields a corresponding lower bound and finishes the proof. 
    \end{proof}

    \subsection{Proof of Statistical Deviation Bound for Empirical KRD via TV-norm and Tree Approximation} \label{app:TV_bound_proofs}
    As a preparatory step to prove \Cref{thm:samplingboundKR}, we treat the significantly simpler case of an empirical deviation bound with respect to the total variation distance.
    
    \begin{proof}[Proof of \Cref{thm:TVbound}]
    Let $\mu,\nu\in \msrX$. By retaining all common mass between $\mu$ and $\nu$ at place and delete (resp. create) excess mass (resp. deficient mass) we obtain a feasible solution for \eqref{eq:krdistance} with objective value in terms of a total variation distance between $\mu$ and $\nu$. Thus, it holds
    \begin{align*}
        \KR^p_{p,C}\left(\mu,\nu\right)\leq C^p \TV\left(\mu,\nu\right).
    \end{align*}
    In particular, this inequality is satisfied for $\nu=\hat{\mu}_{t,s}$. Taking expectations yields
    \begin{align*}
        \mathbb{E}\left[\TV\left(\hat{\mu}_{t,s},\mu\right)\right]&=\frac{1}{st}\sum_{x\in \X}\mathbb{E}\left[ \lvert P_xB_x - st\mu(x) \rvert \right] \\
        &=\frac{1}{st}\sum_{x\in \X}s\mathbb{E}\left[ \lvert P_x - st\mu(x) \rvert \right]+(1-s)st\mu(x)\\
        &\leq \frac{1}{st}\sum_{x\in \X}s(1-s)\mathbb{E}\left[ P_x \right] +s^2\mathbb{E}\left[  \lvert P_x-t\mu(x) \rvert \right]+(1-s)st\mu(x)\\
        &\leq \frac{1}{st}\sum_{x\in \X} 2s(1-s)t\mu_x +s^2\sqrt{t}\sqrt{\mu(x)}\\
        &=2(1-s)\mathbb{M}(\mu)+\frac{s}{\sqrt{t}}\sum_{x\in \X} \sqrt{\mu(x)}.
    \end{align*}
    \end{proof}

    With Lemma~\ref{lem:treeapproximation} and Theorem~\ref{thm:TVbound} at our disposal we are able to prove \Cref{thm:samplingboundKR}.

    \begin{proof}[Proof of \Cref{thm:samplingboundKR}]
    Let $\hat{\mu}_{t,s}$ be the estimator from \eqref{eq:poissonmeasure}. We fix $p=1$ and detail the case $p>1$ at the end of the proof. Suppose first that $C\leq \min_{x\neq x^\prime} d(x,x^\prime)$. According to \cite[Theorem 2.2 (ii)]{heinemann2022kantorovich} it holds that 
    \begin{align*}
    \mathbb{E}\left[\KR_{1,C}(\hat{\mu}_{t,s},\mu)\right]=\frac{C}{2}\mathbb{E}\left[\sum_{x\in \mathcal{X}} \left\vert \hat{\mu}_{t,s}-\mu(x) \right\vert\right]=\frac{C}{2}\mathbb{E}\left[\TV(\hat{\mu}_{t,s},\mu)\right].
    \end{align*}
    This yields the total variation bounds (see Theorem~\ref{thm:TVbound}). Next, consider the tree approximation as outlined in \Cref{app:treeapproximation} and construct an ultrametric tree $\mathcal{T}$ such that $\text{KR}_{1,C}(\hat{\mu}_{t,s},\mu)\leq \text{KR}_{d_\mathcal{T},C}\left(\hat{\mu}_{t,s}^L,\mu^L\right).$
    Applying Lemma~\ref{lem:treeapproximation} for $p=1$ where by definition the difference of height function is equal to 
    \begin{align*}
    h_{q,L}(j-1)-h_{q,L}(j)=\frac{\mathrm{diam}(\mathcal{X})}{q-1} \left( q^{2-j}-q^{1-j}\right)=\mathrm{diam}(\mathcal{X})q^{1-j}
    \end{align*}
    and yields the upper bound
    \begin{align*}
        &\mathbb{E}\left[\KR_{1,C}(\hat{\mu}_{t,s},\mu)\right]\leq \mathbb{E}\left[\KR_{d_\mathcal{T},C}\left(\hat{\mu}_{t,s}^L,\mu^L\right)\right]\\[3ex]
        &= \begin{cases}
            \left( \frac{C}{2}-h_{q,L}(0)\right)\mathbb{E}\left[\lvert \mathbb{M}(\hat{\mu}_{t,s})-\mathbb{M}(\mu)\rvert\right]\\
            +\mathrm{diam}(\mathcal{X})\sum\limits_{j=1}^{L+1}q^{1-j}\sum\limits_{x\in Q_j}\mathbb{E}\left[\left\vert \hat{\mu}_{t,s}^L(\mathcal{C}(x))-\mu^L(\mathcal{C}(x))\right\vert\right], &\text{if } C\geq 2h_{q,L}(0) \\[3ex]
            \mathrm{diam}(\mathcal{X})\sum\limits_{j=l}^{L+1}q^{1-j}\sum\limits_{x^\prime\in Q_j} \mathbb{E}\left[\left\vert \hat{\mu}_{t,s}^L(\mathcal{C}(x^\prime))-\mu^L(\mathcal{C}(x^\prime))\right\vert\right], \,  &\text{if } 2h_{q,L}(l)\leq C< 2h_{q,L}(l-1),\\[3ex]
            \frac{C^p}{2} \mathbb{E}\left[ \TV(\hat{\mu}_{t,s},\mu)\right],\,  &\text{if } C\leq \left( 2h_{q,L}(L) \vee {\min_{x\neq x^\prime}d(x,x^\prime)}\right). \\
        \end{cases}
    \end{align*}
    For the estimator from \eqref{eq:poissonmeasure} with $B_x\sim \mathrm{Ber}(s)$ and $P_x\sim \mathrm{Poi}(t\mu(x))$ for all $x \in \X$, it holds 
    \begin{align*}
        &\sum_{x\in Q_l}\mathbb{E}\left[\left\vert \hat{\mu}_{t,s}^L(\mathcal{C}(x))-\mu^L(\mathcal{C}(x))\right\vert\right]\\&= \sum_{x\in Q_l}\
        \frac{1}{st}\mathbb{E}\left[ \left\lvert \sum_{y\in \mathcal{C}(x)} P_yB_y - st\sum_{y\in \mathcal{C}(x)} \mu(y) \right\rvert \right] \\
        &\leq \sum_{x\in Q_l}\ \frac{1}{st} \sqrt{\text{Var}\left(\sum_{y\in \mathcal{C}(x)} P_yB_y\right)}
        = \sum_{x\in Q_l}\ \frac{1}{st} \sqrt{\sum_{y\in \mathcal{C}(x)}\text{Var}(P_yB_y)}\\
        &=\sum_{x\in Q_l}\ \frac{1}{st} \sqrt{\sum_{y\in \mathcal{C}(x)} s(1-s)t\mu(y)+s(1-s)\mu(y)^2+t\mu(y)s^2}\\
        &= \sum_{x\in Q_l}\ \sqrt{ \frac{1-s}{st}\sum_{y\in \mathcal{C}(x)}\mu(y)+ \frac{1-s}{s}\sum_{y\in \mathcal{C}(x)}\mu(y)^2+\frac{1}{t}\sum_{y\in \mathcal{C}(x)}\mu(y)}\\
            &=\sum_{x\in Q_l} \sqrt{ \frac{1}{st}\mu\left(\mathcal{C}(x)\right)+ \frac{1-s}{s}\sum_{y\in \mathcal{C}(x)}\mu(y)^2}\\
            &\leq \sqrt{\lvert Q_l \rvert}\sqrt{\frac{1}{st} \sum_{x\in Q_l}\mu(\mathcal{C}(x))+\frac{1-s}{s} \sum_{x\in Q_l} \sum_{y \in \mathcal{C}(x)}\mu(y)^2   }\\
            &= \sqrt{\lvert Q_l \rvert}\sqrt{\frac{1}{st} \mathbb{M}(\mu)+\frac{1-s}{s} \sum_{x\in \X} \mu(x)^2   }
    \end{align*}
    Following an analogous computation one bounds the estimation error for the total mass intensity as
    \begin{align}
        \mathbb{E}\left[ \lvert \mathbb{M}(\hat{\mu}_{t,s})-\mathbb{M}(\mu) \rvert \right] \leq \sqrt{\text{Var}(\mathbb{M}(\hat{\mu}_{t,s}))}\leq \sqrt{\frac{1}{st} \mathbb{M}(\mu)+\frac{1-s}{s} \sum_{x\in \X} \mu(x)^2   }. \label{eq:mass_convergence}
    \end{align}
    Applying both of these bounds to the previous upper bound on the $(p,C)$-KRD in  Lemma~\ref{lem:treeapproximation} yields the claim. \\
    For $p>1$, we first observe again that if $C\leq \min_{x\neq x^\prime} d(x,x^\prime)$ then according to \cite[Theorem 2.2]{heinemann2022kantorovich} it holds that
    \begin{align*}
        \mathbb{E}\left[\KR^p_{p,C}(\hat{\mu}_{t,s},\mu)\right]=\frac{C^p}{2}\mathbb{E}\left[\TV(\hat{\mu}_{t,s},\mu)\right]
    \end{align*}
    which yields the total variation bounds. For more general $C$, we  repeat the previous calculations with the upper bounds on the difference of height function $h_{q,L}(j-1)^p-h_{q,L}(j)^p\leq \mathrm{diam}(\mathcal{X})^p \left(\frac{q}{q-1}\right)^p q^{p-jp}$.
    Since $h_{q,L}(L+1)=0$ we also have $$h_{q,L}(L)^p-h_{q,L}(L+1)^p= \mathrm{diam}(\mathcal{X})^pq^{-Lp}.$$ The expectations are bounded identically as before. Finally, using Jensen's inequality,
    \[
        \mathbb{E}\left[\KR_{p,C}(\mu,\hat{\mu}_{t,s}) \right]\leq \left( \mathbb{E}\left[\KR_{p,C}^p(\mu,\hat{\mu}_{t,s}) \right]\right)^{\frac{1}{p}},
    \]
    finishes the proof.
    \end{proof}

    \begin{remark} \label{rem:krdp}
    Omitting Jensen's inequality in the last step of the proof of \Cref{thm:samplingboundKR} implies the slightly stronger result 
    \begin{align*}
        \mathbb{E}\left[ \KR^p_{p,C}\left( \hat{\mu}_{t,s},\mu\right)\right] \leq
        \mathcal{E}_{p,\mathcal{X},\mu}^{\mathrm{Poi}}(C,q,L)\begin{cases}
       \left(2(1-s)\mathbb{M}(\mu)+\frac{s}{\sqrt{t}}\sum_{x\in \X}\sqrt{\mu(x)} \right) , \\[0.1cm] \hspace{2cm} \text{if } C\leq \left(2h_{q,L}(L) \vee {\min_{x\neq x^\prime}d(x,x^\prime)}\right), \\[0.4cm]
      \left( \frac{1}{st}\mathbb{M}(\mu)+\frac{1-s}{s}\sum_{x\in \X}\mu(x)^2 \right)^{\frac{1}{2}}, \\[0.1cm] \hspace{2cm} \text{else}. \\
        \end{cases}
    \end{align*}
    \end{remark}
  
    \subsection{Proof of Stability Bound for Balanced Optimal Transport Plan} \label{app:stab_OTplan_proofs}
    \begin{proof}[Proof of \Cref{thm:stability_OTplan}]
    The proof is divided into several steps. In the first step we cast the optimal transport problem as an appropriate linear program for which we can utilize the Lipschitz stability bound by \cite{li1994sharp}. In the second step we analyze the optimization problem which underlies the definition of respective Lipschitz constant. Finally, in the third step we quantify the Lipschitz constant. 
    
     \emph{Step 1. Reduction to linear program. }
    The collection of OT plans $\tilde{\mathbf{P}}_{c}^*(\mu, \nu)$ for measures $\mu$ and $\nu$ for cost function $c$ can be interpreted as the solution of the linear program
    \begin{align}
      \min_{\pi \in \mathcal{M}_+(\X\times \X)} \sum_{x,x'\in \X}c(x,x') \pi(x,x')  \text{ subject to } \begin{cases}\pi(x,x') \geq 0 \text{ for all } x,x'\in \X,\\ \sum_{x'\in \X}\pi(x,x')= \mu(x) \text{ for all } x\in \X,\\ \sum_{x\in \X}\pi(x,x')= \nu(x') \text{ for all } x'\in \X.
    \end{cases}\label{eq:OT_problem}
    \end{align}
    According to \cite[Theorem, p. 148]{luenberger1984linear}, one of the summation constraints is redundant and by dropping one of them all remaining one are become linearly independent. We drop the constraint for $\nu(x_{|\XC|})$ and the description of the feasible set reduces to 
    \begin{align*}
     \begin{cases}\pi(x,x') \geq 0 \text{ for all } x,x'\in \X,\\ \sum_{x'\in \X}\pi(x,x')= \mu(x) \text{ for all } x\in \X,\\ \sum_{x\in \X}\pi(x,x')= \nu(x') \text{ for all } x'\in \X\backslash\{x_{|\X|}\}.
    \end{cases}  
    \end{align*}
    
    Following the notation of \cite{li1994sharp}, by enumerating the elements of $\X$ as $x_1,\ldots,x_{|\X|}$, we 
    identify $\pi$ as a vector in $\RR^{n}$ with $n = |\X|^2$ with entries $\pi_{i+|\X|(j-1)} = \pi(x_i,x_j)$ and the marginal measures $\mu, \nu$ as a vector in $\RR^{|\XC|}$ with entries $\mu_i= \mu(x_i)$ and $\nu_i = \nu(x_i)$. The constraints of the optimization problem  \eqref{eq:OT_problem} can be rewritten as $A\pi\leq b$ and $C \pi = d$ for suitable matrices $A\in \RR^{n\times n}$ and $C\in \mathbb{R}^{k\times n}$ vectors $b\in \RR^{n}$ and $d \in\RR^{k}$ with $k = 2|\X|-1$. Herein, the matrix $A$ and the vector $b$ are given by 
    \begin{align*}
      A = - I_{n}
       \quad \text{ and } \quad b = 0.
    \end{align*}
    whereas the matrix $C$ and the vector $d$ are given by 
      \begin{align*}
      C = 
       \begin{pmatrix}
        I_{|\XC|} & I_{|\XC|} & \cdots & I_{|\XC|}  & I_{|\XC|}\\
        1_{|\XC|} & 0 & \cdots & 0& 0\\
        0 & 1_{|\XC|} & \cdots & 0 &0 \\
        \vdots & \vdots & \ddots & \vdots & \vdots \\
        0 & 0 & \cdots & 1_{|\XC|} & 0
      \end{pmatrix} \quad \text{ and } \quad d = d(\mu, \nu) =\begin{pmatrix}
          \mu(x_1)\\
        \vdots\\
        \mu(x_{|\X|})\\
        \nu(x_1)\\
        \vdots\\
        \nu(x_{|\X|-1})
      \end{pmatrix}.
    \end{align*}

    The objective function in \eqref{eq:OT_problem} can be written as $c^T\pi$ with $c\in \RR^{n}$ and $c_{i+|\X|(j-1)} = c(x_i,x_j)$. 
    With this notation, the collection of OT plans $\tilde{\mathbf{P}}^*_c$ is characterized as the solutions of the linear program 
    \begin{align*}
      \min_{\pi \in \RR^{n}} c^T\pi \text{ subject to } A\pi\leq b \text{ and } C \pi = d.
    \end{align*}
    
    Invoking Theorems 2.5 and 3.3 of \cite{li1994sharp} then yields the stability estimate 
    \begin{align*}
      \mathcal{H}_\TV\left(\tilde{\mathbf{P}}^*_{p,C}(\mu^1,\nu^1),\tilde{\mathbf{P}}^*_{p,C}(\mu^2,\nu^2)\right)& \leq \gamma \|d(\mu^1, \nu^1) - d^2(\mu^2, \nu^2)\|\\ & = \gamma \left(\TV(\mu^1,\mu^2)+\TV(\nu^1,\nu^2)\right),
    \end{align*}
    where $\gamma = \gamma(A,C)>0$ is defined as 
    \begin{align*}
      \gamma = & \sup_{(p,u,v)\in \RR^{|\XC|^2}\times \RR^{|\XC|}\times \RR^{|\XC|-1}} \|(p, u,v)\|_{\infty} \\
      &\text{ subject to } \quad \begin{cases} \|A^Tp+ C^T(u^T,v^T)^T\|_{\infty} \leq 1, \text{ and the rows of $A$ } \\ \text{corresponding to non-zero entries of $p$} & \\ \text{and the rows of $C$ are linearly independent.} & \end{cases}
    \end{align*}
    Hence, the assertion follows once we show that $\gamma \leq 4|\XC|$. 
    
    \emph{Step 2. Analysis of optimization problem for  definition of Lipschitz constant.}
    To show the bound, note for $(p, u,v) \in \RR^{|\XC|^2}\times \RR^{|\XC|}\times \RR^{|\XC|-1}$ that $(A^Tp+ C^T(u^T,v^T)^T)\in \RR^n$   has  entries \begin{align*}
    (A^Tp+ C^T(u^T,v^T)^T)_{i+|\X|(j-1)} = \begin{cases}
      -p_{i+|\X|(j-1)} + u_i + v_j, & \text{ if } j <|\XC|,  \\
       -p_{i+|\X|(j-1)} + u_i, & \text{ if } j = |\XC|
     \end{cases} \quad \text{ for }i,j\in |\XC|.
    \end{align*}
    To guarantee the constraint in the definition of $\gamma$, consider some slack variable $z\in \R^{n}$ with $\|z\|_\infty \leq 1$ and suppose that \begin{align}\label{eq:systemEQ_for_gamma}
      (A^Tp+ C^T(u^T,v^T)^T)_{i+|\X|(j-1)} = z_{i+|\X|(j-1)} \quad \text{ for }  i,j \in \{1, \dots, |\XC|\}.
    \end{align}
    As is, this system of equation is underdetermined, but by imposing that some entries of $p$ are equal to zero such that the corresponding rows of $A$ and the rows of $C$ are independent, the system either becomes uniquely solvable or unsolvable. The latter case can be ignored, as it does not contribute to the constraint set in the definition of $\gamma$. 
    
    Since the matrix $A$ has rank $n = |\XC|^2$ while $C$ has rank $k = 2|\XC|-1$, it follows that exactly $k$ entries of $p$ need to be set equal to zero for \eqref{eq:systemEQ_for_gamma} to admit a unique solution. 
    
    For $k$ such zero-entries $(i,j)$ of $p$ it follows that \eqref{eq:systemEQ_for_gamma} reduces to \begin{align}\label{eq:reducedSystemEQ_for_gamma}
      z_{i+ |\XC|(j-1)} = C^T(u^T,v^T)^T_{i+|\X|(j-1)} =  \begin{cases}
     u_i + v_j, & \text{ if } j < |\XC|,\\
     u_i, & \text{ if } j = |\XC|.
     \end{cases}
    \end{align}
    Now, given a pair of solutions for \eqref{eq:reducedSystemEQ_for_gamma}, we can recover the remaining entries of $p$ via 
     \begin{align}\label{eq:entriesOfP}
      p_{i + |\XC|(j-1)} = - z_{i + |\XC|(j-1)} + \begin{cases}
       u_i + v_j, & \text{ if } j < |\XC|,\\
       u_i, & \text{ if } j = |\XC|
     \end{cases}.
    \end{align}
    Hence, for there to be a unique solution $(p,u,v)$ of \eqref{eq:systemEQ_for_gamma}, it is necessary that $u$ and $v$ are uniquely determined by \eqref{eq:reducedSystemEQ_for_gamma}. 
    This can only happen if at least one entry $p_{i^* + |\XC|(|\XC|-1)}$ for $i^* \in \{1, \dots, |\XC|\}$ (i.e., for $j = |\XC|$) is equal to zero, since otherwise given a solution $(u,v)\in \RR^k$ we could construct an alternative solution via $(u+\delta 1_{k}, v -\delta 1_{k})\in \RR^k$ for $\delta>0$. 
    
    \emph{Step 3. Derivation of upper bound for Lipschitz constant} With the previous insight, we infer from \eqref{eq:reducedSystemEQ_for_gamma} that for every $i  \in \{1, \dots, |\XC|\}$ for which $p_{i + |\XC|(|\XC|-1)} =0$ that 
    \begin{align*}
        u_{i} = z_{i+ |\XC|(|\XC|-1)} \in [-1,1].
    \end{align*}
    Moreover, for some $j \in \{1, \dots, |\XC|-1\}$ such that $p_{i+ |\XC|(j-1)} = 0$ it follows that 
    \begin{align*}
      v_{j} = z_{i+ |\XC|(|\XC|-1)} - u_i\in [-2,2].
    \end{align*}
    Since the system of equations \eqref{eq:reducedSystemEQ_for_gamma} admits a unique solution, it follows that every entry $u_i$ and $v_j$ can be obtained by chain of equations. Since there are in total $2|\XC|-1$ equations and from each equation the bound in absolute value for $u_i$ or $v_j$ increases by $|z_{i+|\XC|(j-1)}|\leq 1$, it follows altogether that \begin{align*}
      \|(u^T, v^T)^T\|_\infty \leq 2 |\XC| -1 \leq 4 |\XC|. 
    \end{align*}
    Based on these insights, we infer from \eqref{eq:systemEQ_for_gamma} for all non-zero entries of $p$ via \eqref{eq:entriesOfP} that  $$|p_{i + |\XC|(j-1)}| \leq 4|X|-2 + 1 \leq 4|\XC|.$$
    This finally confirms the desired bound for $\gamma$ and proves the claim. 
    \end{proof}
    
    \subsection{Proof of Statistical Deviation Bound for Empirical Unbalanced Optimal Transport Plans}\label{app:UOTplans}
    
    \begin{lemma}[Relation between restricted and unrestricted UOT plans]
      \label{lem:UOTplan_restriction_connection}
      Let $(\X,d)$ be a finite space,  $\mu, \nu\in \msrX$. Then, for $p\geq 1$ and $C\geq 0$ the sets ${\overline{\mathbf{P}}}^*_{p,C}(\mu, \nu)$ defined in \eqref{eq:UOTplans_unrestricted} and ${\mathbf{P}}^*_{p,C}(\mu, \nu)$ defined in \eqref{eq:UOTplans} are related as follows. 
      \begin{enumerate}
        \item It holds that ${\mathbf{P}}^*_{p,C}(\mu, \nu)\subseteq  {\overline{\mathbf{P}}}^*_{p,C}(\mu, \nu)$, i.e., the set ${\mathbf{P}}^*_{p,C}(\mu, \nu)$ consists of unbalanced optimal transport plans.
        \item Every unbalanced optimal transport plan $\pi\in {\overline{\mathbf{P}}}^*_{p,C}(\mu, \nu)$ can be represented as the sum of a restricted transport plan ${\mathbf{P}}^*_{p,C}(\mu, \nu)$ and a non-negative measure supported on $\{ (x,x') \in \XC\times \XC \,|\, d(x,x') = C\}$.
        \item Every sub-coupling $\pi\in \Pi_{\leq}(\mu, \nu)$ which is the sum of a restricted transport plan ${\mathbf{P}}^*_{p,C}(\mu, \nu)$  and a non-negative measure supported on $\{ (x,x') \in \XC\times \XC \,|\, d(x,x') = C\}$ is also an unbalanced optimal transport plan, i.e., $\pi \in {\overline{\mathbf{P}}}^*_{p,C}(\mu, \nu)$.
      \end{enumerate}
    \end{lemma}
    \begin{proof}
    For the first claim consider $ \pi \in {\mathbf{P}}^*_{p,C}(\mu, \nu)$ and let $\overline  \pi \in {\overline{\mathbf{P}}}^*_{p,C}(\mu, \nu)$ be such that $\pi = \overline \pi|_{\mathcal{D}(C)}$. Then it follows that 
    \begin{align*}
        &\sum_{x,x'\in \X}d^p(x,x')\pi(x,x') + C^p\left(\frac{\mathbb{M}(\mu) + \mathbb{M}(\nu)}{2} -  \mathbb{M}(\pi)\right) \\ 
        = &\sum_{x,x'\in \X}d^p(x,x') \pi(x,x') + C^p\mathbb{M}(\overline \pi - \pi)  + C^p\left(\frac{\mathbb{M}(\mu) + \mathbb{M}(\nu)}{2} -  \mathbb{M}(\overline \pi)\right)\\
        = &\sum_{x,x'\in \X}d^p(x,x') \overline \pi(x,x') + C^p\left(\frac{\mathbb{M}(\mu) + \mathbb{M}(\nu)}{2} -  \mathbb{M}(\overline \pi)\right)= \text{KR}_{p,C}(\mu, \nu),
      \end{align*}
    where the last equality follows from the fact that $\overline \pi \in {\overline{\mathbf{P}}}^*_{p,C}(\mu, \nu)$. Hence, $\pi$ is an UOT plan, which shows the first assertion. 
    
    For the second claim, we prove that every UOT plan $\overline \pi \in {\overline{\mathbf{P}}}^*_{p,C}(\mu, \nu)$ assigns no mass to the set $\mathcal{E}(C) \coloneqq \{(x,x') \in \X\times \X \mid d(x,x') > C\}$. To this end, consider the decomposition $\overline\pi = \overline\pi|_{\mathcal{E}(C)} + \overline\pi|_{\mathcal{E}(C)^c}$. Then, 
    \begin{align*}
        \KR_{p,C}(\mu, \nu) = & \sum_{x,x'\in \X}d^p(x,x')\overline\pi(x,x') + C^p\left(\frac{\mathbb{M}(\mu) + \mathbb{M}(\nu)}{2} -  \mathbb{M}(\overline\pi)\right) 
        \\
       \geq & \sum_{(x,x')\in \mathcal{E}(C)}d^p(x,x')\overline\pi(x,x') + C^p \mathbb{M}(\overline \pi|_{\mathcal{E}(C)^c}) + C^p\left(\frac{\mathbb{M}(\mu) + \mathbb{M}(\nu)}{2} -  \mathbb{M}(\overline\pi)\right) 
        \\ =& \sum_{x,x'\in \X}d^p(x,x')\overline\pi|_{\mathcal{E}(C)}(x,x') + C^p\left(\frac{\mathbb{M}(\mu) + \mathbb{M}(\nu)}{2} -  \mathbb{M}(\overline\pi|_{\mathcal{E}(C)})\right) \geq  \KR_{p,C}(\mu, \nu).
      \end{align*}
      This implies that $\overline\pi|_{\mathcal{E}(C)^c} = 0$, which shows the second assertion.
      
      Finally, for the third claim, consider $\pi \in {\overline{\mathbf{P}}}^*_{p,C}(\mu, \nu)$ and $\tilde \pi\in \msrX$ which is supported on $\{(x,x') \in \X\times \X \mid d(x,x') = C\}$, and assume that $\overline \pi \coloneqq \pi + \tilde \pi\in \Pi_{\leq}(\mu, \nu)$. Then,\begin{align*}
      \KR_{p,C}(\mu, \nu) &\leq \sum_{x,x'\in \X}d^p(x,x')\overline\pi(x,x') + C^p\left(\frac{\mathbb{M}(\mu) + \mathbb{M}(\nu)}{2} -  \mathbb{M}(\overline\pi)\right) \\
      &= \sum_{x,x'\in \X}d^p(x,x')\pi(x,x') + C^p\tilde\pi(x,x') + C^p\left(\frac{\mathbb{M}(\mu) + \mathbb{M}(\nu)}{2} -  \mathbb{M}(\pi) - \mathbb{M}(\tilde\pi)\right)\\
      &= \sum_{x,x'\in \X}d^p(x,x')\pi(x,x')  + C^p\left(\frac{\mathbb{M}(\mu) + \mathbb{M}(\nu)}{2} -  \mathbb{M}(\pi)\right) =  \KR_{p,C}(\mu, \nu),
    \end{align*}
    which asserts optimality of $\overline \pi$.
    \end{proof}
    
    \begin{lemma}\label{lem:connection_UOT_OT_plan}
    Let $(\X,d)$ be a finite metric space  and let $p\geq 1$ and $C\geq 0$. Then, the following assertions hold.
    \begin{enumerate}
        \item  For every pair of measures $\mu, \nu\in \msrX$ and their augmented counterparts $\tilde \mu, \tilde \nu \in \mathcal{M}^B_+(\tilde \XC) = \{\mu \in \mathcal{M}_{+}(\XC) \mid \mathbb{M}(\mu) \leq B\}$ for $B = \max(\mathbb{M}(\mu),\mathbb{M}(\nu))$ it holds that 
        \begin{align}\label{eq:OT_UOT_equality}
          {\mathbf{P}}^*_{p,C}(\mu, \nu) = \left\{ \tilde\pi \cdot \mathds{1}(\cdot\in \mathcal{D}(C)) \,\Big|\, \tilde\pi \in {\tilde {\mathbf{P}}}^*_{p,C}(\tilde \mu, \tilde \nu) \right\}.
        \end{align}
        \item For every two pairs of measures $\mu^1, \nu^1, \mu^2, \nu^2\in \msrX$ and corresponding augmented measures $\tilde \mu^i, \tilde \nu^i \in \mathcal{M}^{B_i}_+(\tilde\X)$ for $B_i = \max(\mathbb{M}(\mu^i),\mathbb{M}(\nu^i))$ and $i \in \{1,2\}$ it holds that 
        \begin{align*}
          \mathcal{H}_\TV\left({\mathbf{P}}^*_{p,C}(\mu^1, \nu^1), {\mathbf{P}}^*_{p,C}(\mu^2, \nu^2)\right) \leq \mathcal{H}_\TV\left({\tilde {\mathbf{P}}}^*_{p,C}(\tilde \mu^1,\tilde  \nu^1), {\tilde {\mathbf{P}}}^*_{p,C}(\tilde \mu^2, \tilde \nu^2)\right).
        \end{align*}
        \end{enumerate}
      \end{lemma}
      
      \begin{proof}
    For the first assertion assume without loss of generality that $\mathbb{M}(\mu) \geq \mathbb{M}(\nu)$ which implies $\tilde \mu = \mu$. Further, throughout the proof we will make use of the following equality which has been shown by \citet[p.\ 17]{heinemann2022kantorovich} and is discussed in \Cref{app:liftOT} of this manuscript,
    \begin{align}\label{eq:UOT_OT_equality_values}
          \text{KR}_{p,C}^p(\mu, \nu) = \text{UOT}_{d^p\wedge C^p,C}(\mu, \nu) = \tilde{\text{OT}}_{\tilde d^p_C}(\tilde \mu, \tilde \nu).
    \end{align}
        
    To show that the right-hand side of \eqref{eq:OT_UOT_equality} is included in the left-hand side, let $\tilde \pi \in {\tilde {\mathbf{P}}}^*_{p,C}(\tilde \mu, \tilde \nu)$ and define $\pi \coloneqq \tilde\pi \cdot \mathds{1}(\cdot \in \mathcal{D}(C))$ and $\pi^* \coloneqq \tilde \pi\cdot \mathds{1}(\cdot \in \X\times \X)$. 
    Note that $\tilde \pi(\{\mathfrak{d}\} \times \tilde \X) = \tilde \mu(\{\mathfrak{d}\}) = 0$. As a result, by the marginal constraints on $\tilde \pi$, $\tilde \pi(\tilde \X\times \X) = \tilde \nu(\X) = \nu(\X)$, it follows that $\mathbb{M}(\pi^*) = \mathbb{M}(\nu)$. From this it is evident that
    \begin{align*}
          & \sum_{x, x'\in \X} d^p(x,x') \pi(x,x') + C^p\left(\frac{\mathbb{M}(\mu) + \mathbb{M}(\nu)}{2} - \mathbb{M}(\pi)\right)\\
      = \;& \sum_{x, x'\in \X} d^p(x,x') \pi(x,x') + C^p \mathbb{M}( \pi^* - \pi)+ C^p\left(\frac{\mathbb{M}(\mu) + \mathbb{M}(\nu)}{2} - \mathbb{M}(\pi^*)\right)\\
      = \;& \sum_{x, x'\in \X} \tilde d^p_C(x,x')  \pi^*(x,x')  + \frac{C^p}{2}\left(\mathbb{M}(\mu) - \mathbb{M}(\nu) \right)\\
      =\;& \sum_{x, x'\in \X} \tilde d^p_C(x,x')  \tilde \pi(x,x')  +  \sum_{x \in \X}  \tilde d^p_C(x,\mathfrak{d})  \tilde \pi(x,\mathfrak{d})  \\
      =\; & \sum_{x, x'\in \tilde\X} \tilde d^p_C(x,x')  \tilde \pi(x,x') = \tilde{\text{OT}}_{\tilde d^p_C}(\tilde \mu, \tilde \nu) = \text{KR}_{p,C}^p(\mu,\nu),
    \end{align*}
    where the equality third to last is a consequence of $\tilde \pi(\mathfrak{d} \times \tilde \X) = \tilde \mu(\mathfrak{d}) = 0$, the second to last equality follows by optimality of $\tilde \pi$, and the final equality is a consequence of \eqref{eq:UOT_OT_equality_values}. In conjunction with Lemma \ref{lem:UOTplan_restriction_connection} this asserts that $\pi\in {\mathbf{P}}^*_{p,C}(\mu, \nu)$. 
      
    Conversely, to show that the left-hand side of \eqref{eq:OT_UOT_equality} is contained in the right-hand side, let $\pi \in {\mathbf{P}}^*_{p,C}(\mu, \nu)$. Further, define the measure $\check\pi\in \mathcal{M}_+(\XC\times \XC)$ for $x,x'\in \XC$ by
      \begin{align*}
        \check\pi(x,x') \coloneqq \begin{cases} 0,  & \text{ if } \mathbb{M}(\nu) = \mathbb{M}(\pi),\\
          \frac{\mu(x)-\pi(x, \XC)}{\mathbb{M}(\mu) - \mathbb{M}(\pi)} \left(\nu(x') - \pi(\XC,x')\right) &, \text{ if } \mathbb{M}(\nu) > \mathbb{M}(\pi),
        \end{cases}
      \end{align*}
      and further set $\pi^* \coloneqq \pi + \check \pi$. Then, it follows that $\pi^* \in \Pi_{\leq}(\mu, \nu)$ and $\mathbb{M}(\pi^*) = \mathbb{M}(\nu)$. Further, from optimality of $\pi$ we infer that
      \begin{align}
       \KR_{p,C}^p(\mu, \nu) =\;& \sum_{x, x'\in \X} d^p(x,x') \pi(x,x') + C^p\left(\frac{\mathbb{M}(\mu) + \mathbb{M}(\nu)}{2} - \mathbb{M}(\pi)\right)\notag\\
        =\; & \sum_{x, x'\in \X} (d^p(x,x')\wedge C^p) \pi(x,x') + C^p \mathbb{M}(\check\pi) + C^p\left(\frac{\mathbb{M}(\mu) + \mathbb{M}(\nu)}{2} - \mathbb{M}(\pi + \check \pi)\right)\notag\\
        \geq \; & \sum_{x, x'\in \X} (d^p(x,x')\wedge C^p) \pi^*(x,x') + C^p\left(\frac{\mathbb{M}(\mu) + \mathbb{M}(\nu)}{2} - \mathbb{M}(\pi^*)\right)\notag\\
        \geq \; &\mathrm{UOT}_{ d^p\wedge C^p, C}(\mu, \nu) = \KR_{p,C}^p(\mu, \nu),\label{eq:derivation:_UOT_OT}
      \end{align}
    where the last equality follows \eqref{eq:UOT_OT_equality_values} and thus asserts that all inequalities are qualities. This implies that $\check\pi$ is concentrated on $\{(x,x')\in \X \times \X \mid d(x,x') \geq C\}$ and, thus, $\pi^* \cdot \mathds{1}(\cdot\in \mathcal{D}(C)) = \pi$. Next, define $\tilde \pi \coloneqq \pi^* + \sum_{x\in \X}\delta_{(x,\mathfrak{d})} \left(\mu(x) - \pi^*(x,\XC)\right)$ which fulfills $\tilde \pi\in \Pi_{=}(\tilde \mu, \tilde \nu)$ and $\tilde \pi \cdot \mathds{1}(\cdot\in \mathcal{D}(C)) = \pi$. In particular, from the above derivation in \eqref{eq:derivation:_UOT_OT} we infer from $\mathbb{M}(\pi^*) = \mathbb{M}(\nu)$, $\mathbb{M}(\tilde \pi) = \mathbb{M}(\tilde \mu) = \mathbb{M}(\mu)$ and since $\tilde \pi(\mathfrak{d},\tilde \X) = \tilde \mu(\mathfrak{d}) = 0$ that
      \begin{align*}
        \KR_{p,C}^p(\mu, \nu) =\; &\sum_{x, x'\in \X}(d^p(x,x')\wedge C^p) \pi^*(x,x') + C^p\left(\frac{\mathbb{M}(\mu) + \mathbb{M}(\nu)}{2} - \mathbb{M}(\pi^*)\right)\\
        =\; & \sum_{x, x'\in \X}(d^p(x,x')\wedge C^p) \pi^*(x,x') +  \frac{C^p}{2}\left(\mathbb{M}(\mu)  - \mathbb{M}(\nu)\right)\\
        =\; & \sum_{x, x'\in \X'}\tilde d^p_C(x,x')  \pi^*(x,x') \geq \tilde{\mathrm{OT}}_{\tilde d^p_C}(\tilde \mu, \tilde \nu) = \KR_{p,C}^p(\mu, \nu).
      \end{align*}
    Here, the final equality is a consequence of \eqref{eq:UOT_OT_equality_values}, and altogether confirms that $\tilde \pi \in {\tilde {\mathbf{P}}}^*_{p,C}(\tilde \mu, \tilde \nu)$. This concludes the proof of the first assertion.   
      
    For the proof of the second assertion note from the considerations above that for every pair of UOT plans $\pi^i\in \mathbf{P}^*(\mu^i,\nu^i)$ with $i \in \{1,2\}$, every pair of corresponding OT plans $\tilde \pi^i \in {\tilde{\mathbf{P}}}^*(\mu^i, \nu^i)$ with $\tilde \pi^i \cdot \mathds{1}(\cdot \in \mathcal{D}(C)) = \pi^i$ fulfills the inequality 
      \begin{align*}
        \text{TV}(\pi^1, \pi^2) \leq \text{TV}(\tilde \pi^1, \tilde \pi^2).
      \end{align*}
    Based on this, the claim follows from the definition of the Hausdorff distance induced by the total variation norm.
    \end{proof}
    
    \begin{proof}[Proof of Theorem \ref{thm:UOTplan_Stability}] 
      For the proof we again rely on the lift to optimal transport as detailed in Appendix \ref{app:liftOT}. To this end, consider the augmented space $\tilde \X = \X \cup \{\mathfrak{d}\}$ for some dummy element $\mathfrak{d}$, let $B$ be a majorant for the masses of the measures $\mu^1,\mu^2,\nu^1,\nu^2$, and define for $i \in\{1,2\}$ the  augmented measures
    \begin{align*}
      \tilde \mu^i &\coloneqq \mu^i + \delta_{\mathfrak{d}}(\max(\mathbb{M}(\mu^i),\mathbb{M}(\nu^i)) - \mathbb{M}(\mu^i))\quad \text{ and } \quad\\  \tilde \nu^i &\coloneqq \nu^i + \delta_{\mathfrak{d}}(\max(\mathbb{M}(\mu^i),\mathbb{M}(\nu^i)) - \mathbb{M}(\nu^i)).
    \end{align*}
    Then, upon denoting the collection of OT plans between $\tilde \mu^i$ and $\tilde \nu^i$ with respect to the cost function $\tilde d^p_C$ by ${\tilde {\mathbf{P}}}^*_{p,C}(\tilde \mu^i,\tilde \nu^i)$, it follows by Lemma \ref{lem:connection_UOT_OT_plan} and Theorem \ref{thm:stability_OTplan} that 
    \begin{align*}
      \mathcal{H}_\TV\left({\mathbf{P}}^*_{p,C}(\mu^1,\nu^1),{\mathbf{P}}^*_{p,C}(\mu^2,\nu^2)\right)
      \leq &   \mathcal{H}_\TV\left({\tilde {\mathbf{P}}}^*_{p,C}(\tilde \mu^1,\tilde \nu^1),{\tilde {\mathbf{P}}}^*_{p,C}(\tilde \mu^2,\tilde \nu^2)\right)\\
      \leq \;& 4\,|\tilde \X| \left(\TV(\tilde \mu^1,\tilde \mu^2) + \TV(\tilde \nu^1,\tilde \nu^2)\right).
    \end{align*}
    The assertion now follows by $|\tilde \X| = |\X| + 1$ and the subsequent Lemma \ref{thm:TVbound_augmentation}.
    \end{proof}
    
    \begin{lemma}
    \label{thm:TVbound_augmentation}
      Let $\mu^i, \nu^i \in \msrX$ for $i \in \{1, 2\}$ be two non-negative measures and consider their augmented measures $\tilde \mu^i, \tilde \nu^i \in \mathcal{M}^B_+(\tilde \XC)$ to $\tilde \XC$ for $B = \max(\mathbb{M}(\mu^i),\mathbb{M}(\nu^i))$,  then 
      \begin{align*}
         & \TV(\tilde \mu^1,\tilde \mu^2)  + \TV(\tilde \nu^1,\tilde \nu^2) \\ \leq \; &  \TV(\mu^1,\mu^2)  + \TV(\nu^1,\nu^2) + |\mathbb{M}(\mu^1) - \mathbb{M}(\mu^2)| + | \mathbb{M}(\nu^1)- \mathbb{M}(\nu^2)|
         \\ \leq \;&  2\left(\TV(\mu^1,\mu^2)  + \TV(\nu^1,\nu^2)\right).
      \end{align*}
    \end{lemma}
    \begin{proof}
      We first observe that \begin{align*}
        & \TV(\tilde \mu^1,\tilde \mu^2)  + \TV(\tilde \nu^1,\tilde \nu^2)   \\ 
        =\;& \left(\sum_{x\in \X} |\mu^1(x) - \mu^2(x)| + |\nu^1(x) - \nu^2(x)|\right) + |\tilde \mu^1(\mathfrak{d}) - \tilde \mu^2(\mathfrak{d})|+ |\tilde \nu^1(\mathfrak{d}) - \tilde \nu^2(\mathfrak{d})|\\
        =\;&  \TV(\mu^1,\mu^2) +\TV(\nu^1,\nu^2) + |\tilde \mu^1(\mathfrak{d}) - \tilde \mu^2(\mathfrak{d})|+ |\tilde \nu^1(\mathfrak{d}) - \tilde \nu^2(\mathfrak{d})|.
      \end{align*}
      Now if $\mathbb{M}(\mu^1)\geq \mathbb{M}(\nu^1)$ and $\mathbb{M}(\mu^2)\geq \mathbb{M}(\nu^2)$, then 
      \begin{align*}
          |\tilde \mu^1(\mathfrak{d}) - \tilde \mu^2(\mathfrak{d})|+ |\tilde \nu^1(\mathfrak{d}) - \tilde \nu^2(\mathfrak{d})| 
          =\; & |\mathbb{M}(\mu^1) -  \mathbb{M}(\nu^1) - \mathbb{M}(\mu^2) + \mathbb{M}(\nu^2)|\\
      \leq \;& |\mathbb{M}(\mu^1) - \mathbb{M}(\mu^2)| + | \mathbb{M}(\nu^1)- \mathbb{M}(\nu^2)|,
      \end{align*}
      whereas if  $\mathbb{M}(\mu^1)\geq \mathbb{M}(\nu^1)$ and $\mathbb{M}(\mu^2)<\mathbb{M}(\nu^2)$, then 
      \begin{align*}
        |\tilde \mu^1(\mathfrak{d}) - \tilde \mu^2(\mathfrak{d})|+ |\tilde \nu^1(\mathfrak{d}) - \tilde \nu^2(\mathfrak{d})| 
       =\; & |\mathbb{M}(\nu^1) -  \mathbb{M}(\mu^1)| + |\mathbb{M}(\mu^2) - \mathbb{M}(\nu^2)|\\
       =\; & \mathbb{M}(\nu^1) -  \mathbb{M}(\mu^1) + \mathbb{M}(\mu^2) - \mathbb{M}(\nu^2)\\
       \leq \;& |\mathbb{M}(\mu^1) - \mathbb{M}(\mu^2)| + | \mathbb{M}(\nu^1)- \mathbb{M}(\nu^2)|.
       \end{align*}
       The remaining two cases can be treated analogously, asserting the first inequality of the claim. For the second inequality it suffices to note that the total variation norm provides an upper bound on the mass difference of measures
    \end{proof}

    \subsection{Proof of Statistical Deviation Bound for Empirical $\mathbf{(p,C)}$-Barycenters} \label{app:barycenters_proofs}
    \begin{proof}[Proof of \Cref{thm:frechetboundPoi}, \Cref{thm:frechetboundMult} and \Cref{thm:frechetboundBer}]
    Let $\hat{\mu}^1,\dots,\hat{\mu}^J$ be any of the three estimators from \eqref{eq:multinomialmeasure}, \eqref{eq:bermeasure} or \eqref{eq:poissonmeasure}. Further, for each $1\leq i \leq J$ let $\mathcal{E}_{1,\X_i,\mu^i}(C)$ be the corresponding constant in \Cref{thm:samplingboundKR}, \Cref{thm:samplingboundKRMult} or \Cref{thm:samplingboundKRBer} for $\mu^i$, respectively. Let $\theta_i$ denote the respective sampling parameter dependencies $N_i^{-1/2}$, $\phi(t_i,s_i)$ or $\psi(s_{\X_i})$, respectively. Involving the augmentation argument, due to the construction of the lifted problem, it holds for any $\mu^1,\mu^2,\mu^3 \in \msrY$ that
    \begin{align*}
        \lvert \KR_{p,C}^p(\mu^1,\mu^3)-\KR_{p,C}^p(\mu^2,\mu^3) \rvert &= \lvert \tilde{\mathrm{OT}}_p^p(\tilde{\mu}^1,\tilde{\mu}^3)-\tilde{\mathrm{OT}}_p^p(\tilde{\mu}^2,\tilde{\mu}^3) \rvert \\
        &\leq \mathrm{diam}(\ty)^{p-1} p \,\tilde{\mathrm{OT}}_1(\tilde{\mu}^1,\tilde{\mu}^2) \\
        &=C^{p-1} p\,\KR_{1,C}(\mu^1,\mu^2),
    \end{align*}
    where the inequality follows from \cite{sommerfeld2018inference}. Taking expectation and applying the previous display together with \Cref{thm:samplingboundKR} yields
    \begin{align*}
        \mathbb{E}\left[\lvert F_{p,C}(\mu)-\hat{F}_{p,C}(\mu)\rvert\right]
        &\leq \frac{1}{J}\easysum{i=1}{J}\mathbb{E}\left[\lvert \KR_{p,C}^p(\mu^i,\mu)-\KR_{p,C}^p(\hat{\mu}^i,\mu) \rvert \right] \\
        &\leq p\,C^{p-1} \frac{1}{J}\easysum{i=1}{J}\mathbb{E} \left[\KR_{1,C}(\mu^i,\hat{\mu}^i)\right]   \\
        &\leq p C^{p-1}\frac{1}{J}\easysum{i=1}{J} \mathcal{E}_{1,\X_i,\mu^i}(C)\theta_i.
    \end{align*}
    Let $\mu^\star$ and $\hat{\mu}^\star$ be minimizers of their respective $p$-Fr\'echet functional $F_{p,C}$ and $\hat{F}_{p,C}$. Then, it follows that
    \begin{align*}
    \mathbb{E}[\lvert F_{p,C}(\hat{\mu}^\star)-F_{p,C}(\mu^\star) \rvert ]
    &= \mathbb{E}\left[ F_{p,C}(\hat{\mu}^\star)-\hat{F}_{p,C}(\mu^\star)+\hat{F}_{p,C}(\mu^\star)-F_{p,C}(\mu^\star)  \right]\\ 
    &\leq \mathbb{E}\left[  F_{p,C}(\hat{\mu}^\star)-\hat{F}_{p,C}(\mu^\star) \right]+ \mathbb{E}  \left[\hat{F}_{p,C}(\mu^\star)-F_{p,C}(\mu^\star)  \right]\\ 
    &\leq\mathbb{E}\left[ F_{p,C}(\hat{\mu}^\star)-\hat{F}_{p,C}(\mu^\star)\right]+ p\, C^{p-1} \frac{1}{J}\easysum{i=1}{J} \mathcal{E}_{1,\X_i,\mu^i}(C)\theta_i\\
    &\leq\mathbb{E}\left[ F_{p,C}(\hat{\mu}^\star)-\hat{F}_{p,C}(\hat{\mu}^\star)\right]+ p\,C^{p-1} \frac{1}{J}\easysum{i=1}{J} \mathcal{E}_{1,\X_i,\mu^i}(C)\theta_i \\
    &\leq p\,C^{p-1} \frac{2}{J}\easysum{i=1}{J} \mathcal{E}_{1,\X_i,\mu^i}(C)\theta_i,
    \end{align*}
    where the fourth inequality follows from $\hat{\mu}^*$ being a minimizer of $\hat{F}_{p,C}$.
    \end{proof}
    
    \begin{proof}[Proof of \Cref{thm:kr_boundPoi}, \Cref{thm:kr_boundMult} and \Cref{thm:kr_boundBer}]
    Let $\hat{\mu}^1,\dots,\hat{\mu}^J$, $\mathcal{E}_{1,\X_i,\mu^i}(C)$ and $\theta_i$ for all $i=1,\dots,J$ as in the previous proof. Let $\mathbf{B}$ be the set of $(p,C)$-barycenters of the measures $\mu^1,\ldots,\mu^J$ and define $\tilde{\mathbf{B}}$ as the set of $\mathrm{OT}_p$-barycenters of the augmented measures $\tilde{\mu}^1,\ldots,\tilde{\mu}^J$. Similar, we denote $\hat{\mathbf{B}}$ the set of $(p,C)$-barycenters of the estimated measures $\hat{\mu}^1,\ldots,\hat{\mu}^J$ and let $\hat{\tilde{\mathbf{B}}}$ be the set of $p$-barycenters of their augmented versions. Define the lift of a measure $\mu\in\msrY$ to a measure $\tilde{\mu}\in\mathcal{M}(\tilde{\Y})$ by
    \begin{align*}
    \phi_{\mu^1,\ldots,\mu^J}(\mu)=\mu + \left( \sum_{i=1}^J \mathbb{M}(\mu^i)-\mathbb{M}(\mu)\right) \delta_{\dum}.
    \end{align*}
    If $\mu\in\mathbf{B}$ then it follows by \cite[Lemma 3.3]{heinemann2022kantorovich} that $\phi_{\mu^1,\ldots,\mu^J}(\mu)\in\tilde{\mathbf{B}}$. Conversely, for any $\tilde{\mu}\in\tilde{\mathbf{B}}$ it holds that $\phi_{\mu^1,\ldots,\mu^J}^{-1}(\tilde{\mu})\in\mathbf{B}$. We denote by $\phi(\mathbf{B})\coloneqq\{\phi_{\mu^1,\ldots,\mu^J}(\mu)\vert \mu \in \mathbf{B}\}$ and analogously $\phi^{-1}\Big(\tilde{\mathbf{B}}\Big)\coloneqq\{\phi_{\mu^1,\ldots,\mu^J}^{-1}(\tilde{\mu})\vert \tilde{\mu} \in \tilde{\mathbf{B}}\}$. With this we have
    \begin{align} \label{eq:wsthm_reform}
        \mathbb{E}\left[ \sup_{\hat{\mu}\in\hat{\mathbf{B}}}\inf_{\mu\in\mathbf{B}}\text{KR}_{p,C}^p(\mu,\hat{\mu})\right]&=\mathbb{E}\bigg[ \sup_{\hat{\mu}\in\phi^{-1}\Big(\hat{\tilde{\mathbf{B}}} \Big)}\inf_{\mu\in\phi^{-1}\Big(\tilde{\mathbf{B}} \Big)}\text{KR}_{p,C}^p(\mu,\hat{\mu})\bigg]\\
    &=\mathbb{E}\bigg[ \sup_{\hat{\mu}\in\phi^{-1}\Big(\hat{\tilde{\mathbf{B}}} \Big)}\inf_{\mu\in\phi^{-1}\Big(\tilde{\mathbf{B}} \Big)}\tilde{\mathrm{OT}}_{p}^p(\phi(\mu),\phi(\hat{\mu}))\bigg]\\ 
    &=\mathbb{E}\bigg[ \sup_{\hat{\mu}\in\hat{\tilde{\mathbf{B}}}}\inf_{\mu\in\tilde{\mathbf{B}}}\tilde{\mathrm{OT}}_p^p(\tilde{\mu},\hat{\tilde{\mu}})\bigg].\notag
    \end{align}
    We continue by recalling a slightly adapted version of Lemma $3.8$ in \cite{heinemann2022randomized}. Since we only apply this lemma to the augmented balanced OT problem, the proof remains unchanged and is therefore omitted.
    \begin{lemma}\label{lem:slope_bound}
    Let $\tilde{F}_{p,C}$ be the augmented Fr\'echet functional corresponding to $\tilde{\mu}^1,\dots ,\tilde{\mu}_N \in \mathcal{M}_+(\tilde{\mathcal{Y}})$. Then, for any $\tilde{\mu} \in \mathcal \mathcal{M}_+(\mathcal{C}_{\KR}(J,p,C)\cup \{ \dum \})$ with $\mathbb{M}(\tilde{\mu})=\sum_{i=1}^J \mathbb{M}(\mu^i)$ there exists a $\tilde{\mu}^* \in {\argmin}_{\tilde{\nu} } \tilde{F}_{p,C}(\tilde{\nu})$ such that
    \begin{align*}
    \tilde{F}_{p,C}(\tilde{\mu})-\tilde{F}_{p,C}(\tilde{\mu}^*) \geq 2V_{P} \tilde{\mathrm OT}^p_p (\tilde{\mu},\tilde{\mu}^*),
    \end{align*}
    where $V_{P}$ is the constant from \Cref{thm:kr_boundPoi}.
    \end{lemma}
    Invoking  Lemma~\ref{lem:slope_bound} with $\hat \mu\in \hat{\tilde{\mathbf{B}}}$ and applying \Cref{thm:frechetboundPoi} yields
    \begin{align*}
    \frac{2p}{J}\sum_{i=1}^J \mathcal{E}_{1,\X_i,\mu^i}(C)C^{p-1} \theta&\geq \mathbb{E}\left[F_{p,C}(\hat{\tilde{\mu}})-F_{p,C}(\tilde{\mu}) \right]\\
    &\geq \mathbb{E}\left[ 2V_p\sup_{\hat{\tilde{\mu}}\in\hat{\tilde{\mathbf{B}}}}\inf_{\tilde{\mu}\in\tilde{\mathbf{B}}}\tilde{\mathrm{OT}}^p_{p}\left(\tilde{\mu},\hat{\tilde{\mu}}\right)\right]=\mathbb{E}\left[ 2V_p\sup_{\hat{\mu}\in\hat{\mathbf{B}}}\inf_{\mu\in\mathbf{B}}\KR_{p,C}^p\left(\mu,\hat{\mu}\right)\right],
    \end{align*}
    where the equality follows from \eqref{eq:wsthm_reform} and hence
    \begin{align*}
    \frac{\frac{p}{J}\sum_{i=1}^J \mathcal{E}_{1,\X_i,\mu^i}(C) \ C^{p-1} }{V_{P}}\theta \geq \mathbb{E}\left[ \sup_{\hat{\mu}\in\hat{\mathbf{B}}}\inf_{\mu\in\mathbf{B}}\KR_{p,C}^p\left(\mu,\hat{\mu}\right)\right].
    \end{align*}
    \end{proof}
    
  \vspace{-1cm}
    \section{Additional Simulations}
  
  In this appendix we display additional simulations which showcase the relative error in estimating the $(2,C)$-KRD and the $(2,C)$-KR barycenter for the Poisson model as well as for the multinomial and Bernoulli model. 
  
  \subsection{Additional Synthetic Datasets}\label{app:additionalSim}
  This subsection introduces four additional classes of measures SPI, SPIC, and  NI, NIG. The first two are qualitatively similar to the classes NE and NEC, the latter two resemble weighted modifications of the classes PI and PIC.

  \subsubsection*{Spirals of varying Length (SPI), \Cref{fig:spiraldata} (a)}
  Let $a_i\sim U[2,4]$ and $b_i\sim U[3,6]$ for $i=1,\dots,J$. Let $K_i=\lceil b_iM \rceil$ and let $t_1,\dots,t_K$ be a discretization of $[0,b\pi]$. Set $w^i_k=1$ for $k=1,\dots,K_i$ and $i=1,\dots,J$. Set 
  \[
  l^i_k=a_i((t_k\sin(t_k)+64)/140,(t_k\cos(t_k)+70)/130)^T.
  \]
  \subsubsection*{Clustered Spirals (SPIC), see \Cref{fig:spiraldata} (b)}
  Let $a^c_i\sim U[2,4]$ and $b^c_i\sim U[3,6]$ for $i=1,\dots,J$, $c=1,\dots, 5$. Let $K^c_i=\lceil b_i^cM \rceil$ and let $t_1,\dots,t_K$ be a discretization of $[0,b\pi]$. Set $w^i_k=1$ for $k=1,\dots,K_i$ and $i=1,\dots,J$ and let $\alpha=(0,3,3,6,3)^T$ and $\beta=(3,0,3,3,6)^T$. Set 
  \[
  l^i_{\sum_{r=1}^{c-1}K^r_i+k}=(1/7)(((a_i^ct_k\sin(t_k)+64)/140)+\alpha_c,(a_i^ct_k\cos(t_k)+70)/130+\beta_c)^T,
  \]
  where we use the convention that a sum is zero if its last index is smaller than its first one. 
  \begin{figure}[h]
    \centering
    \subfloat[][]{\includegraphics[width=0.45\linewidth]{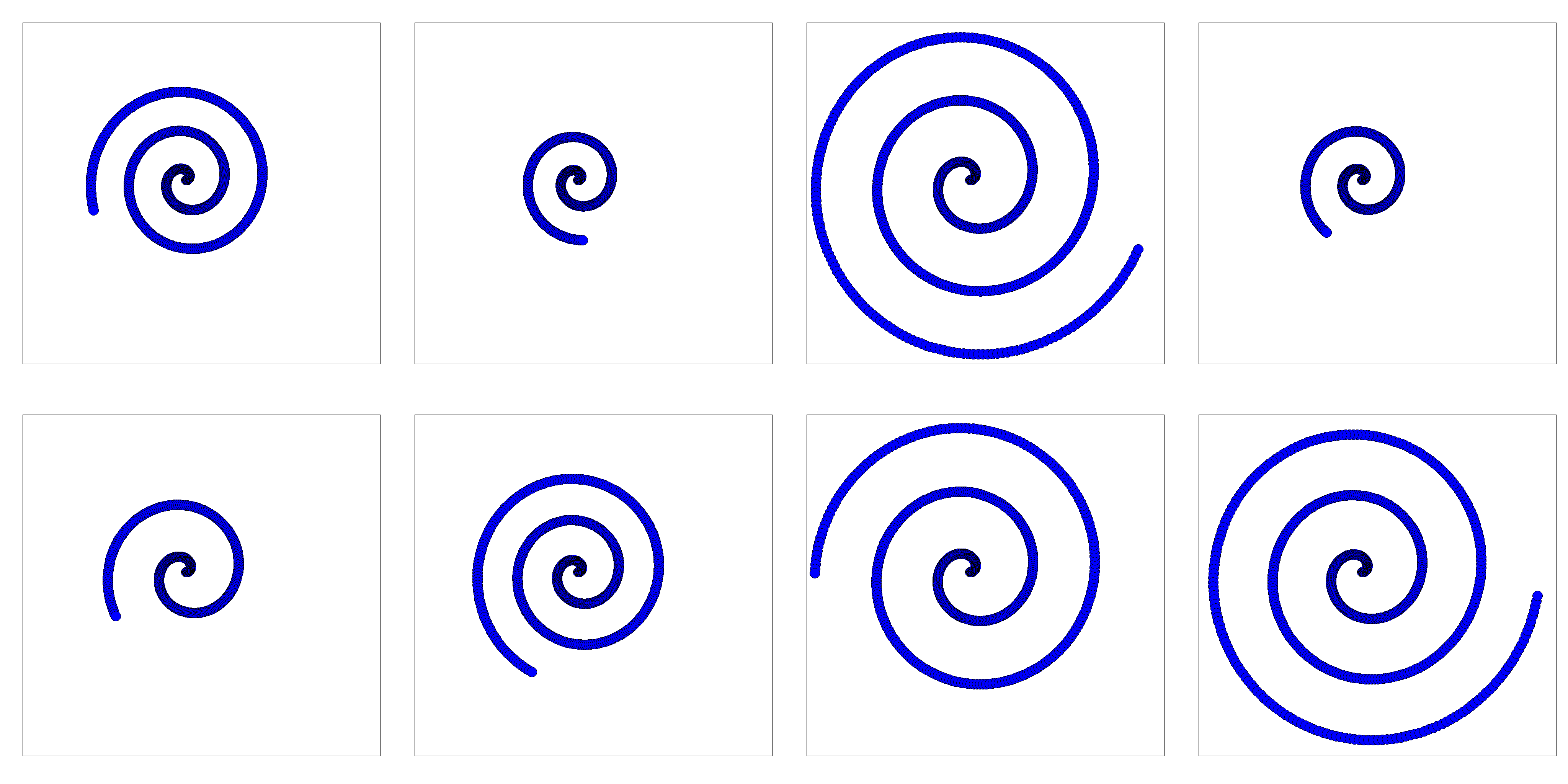}}%
    \qquad
    \subfloat[][]{\includegraphics[width=0.45\linewidth]{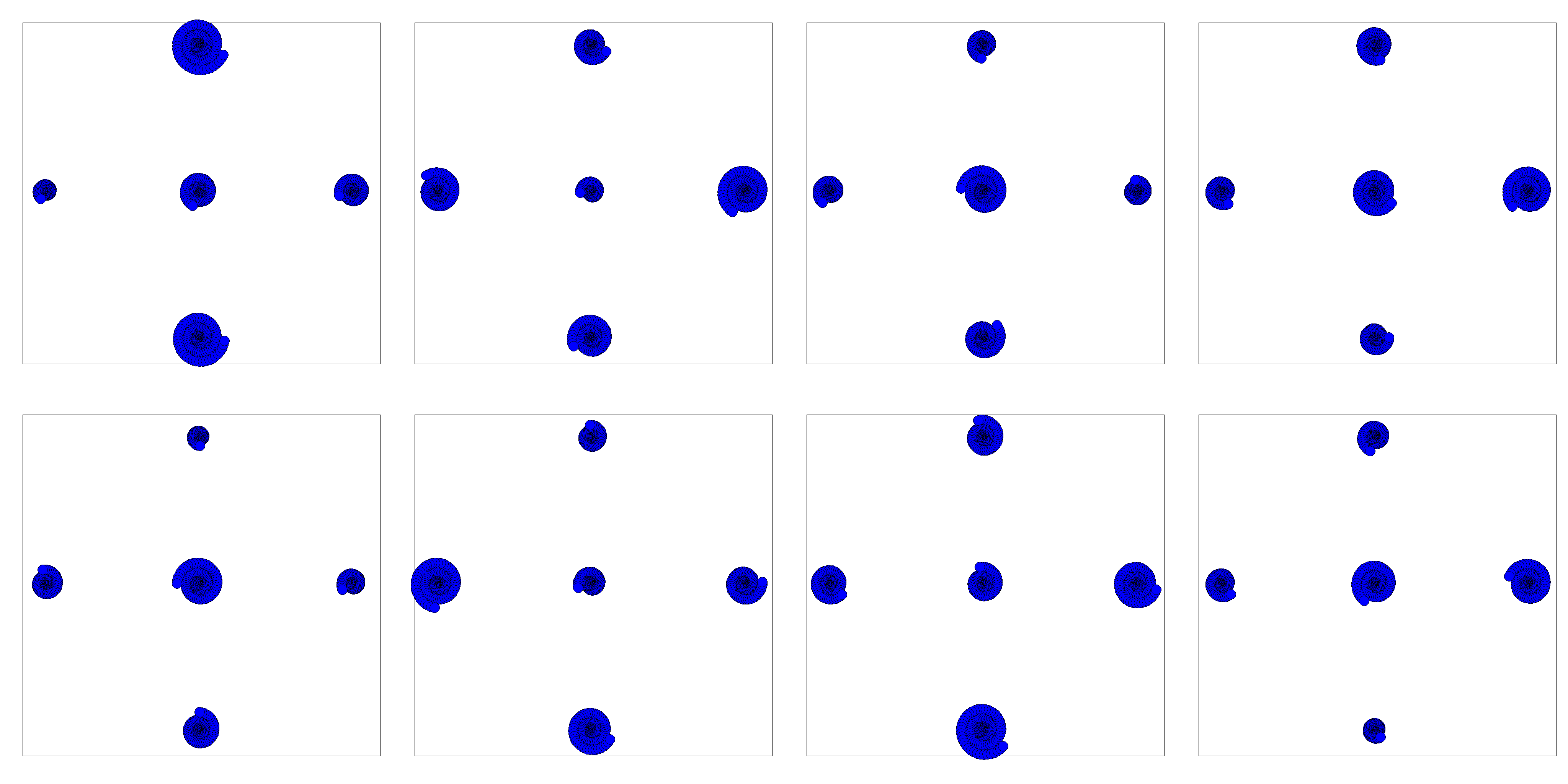}}%
    \caption{\textbf{(a)} An example of $J=8$ measures from the \emph{SPI} dataset with $M=110$. \textbf{(b)} An example of $J=8$ measures from the \emph{SPIC} dataset with $M=22$.}%
    \label{fig:spiraldata}
  \end{figure}
  \begin{figure}[b]
    \centering
    \subfloat[][]{\includegraphics[width=0.45\linewidth]{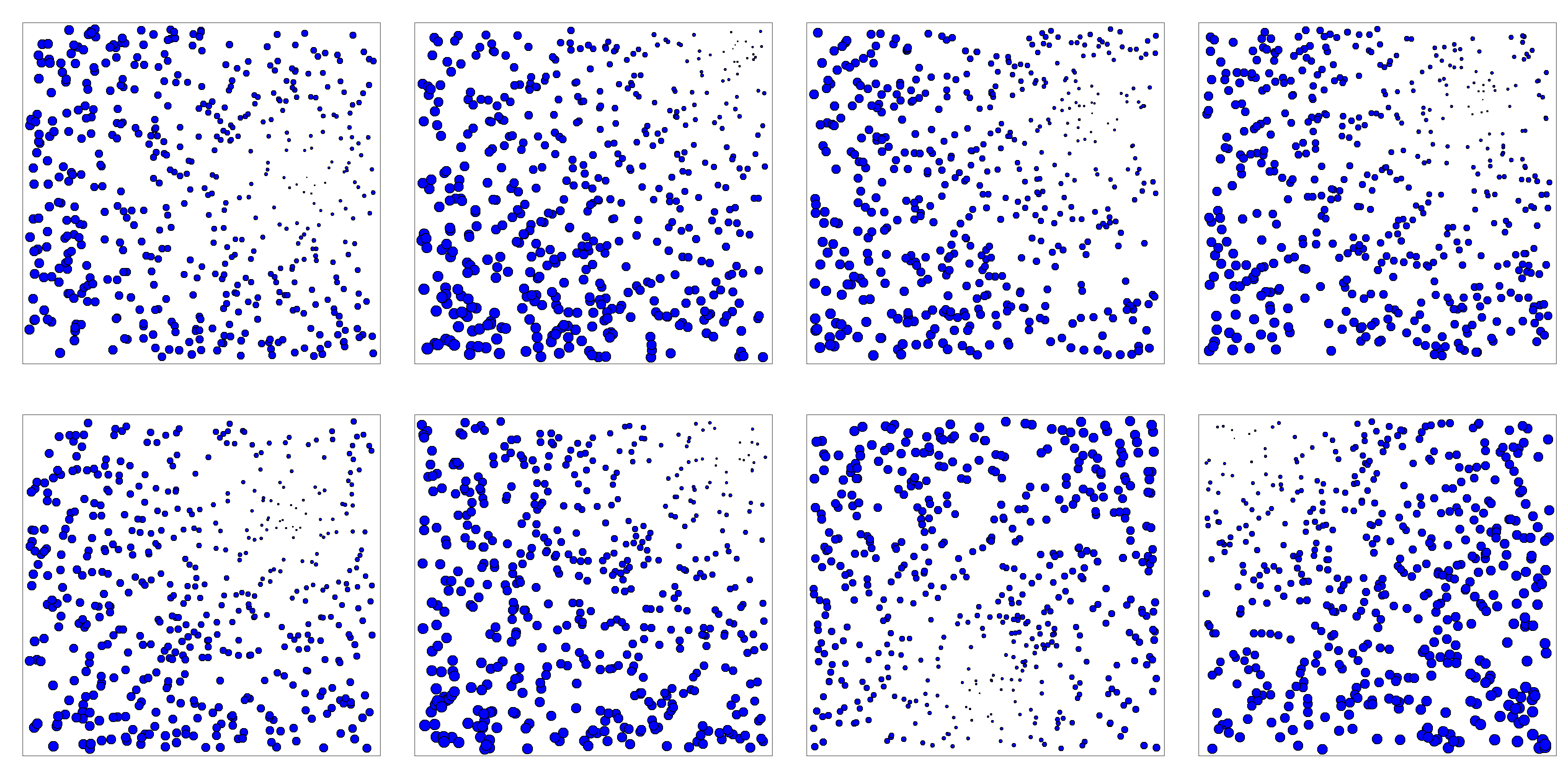}}%
    \qquad
    \subfloat[][]{\includegraphics[width=0.45\linewidth]{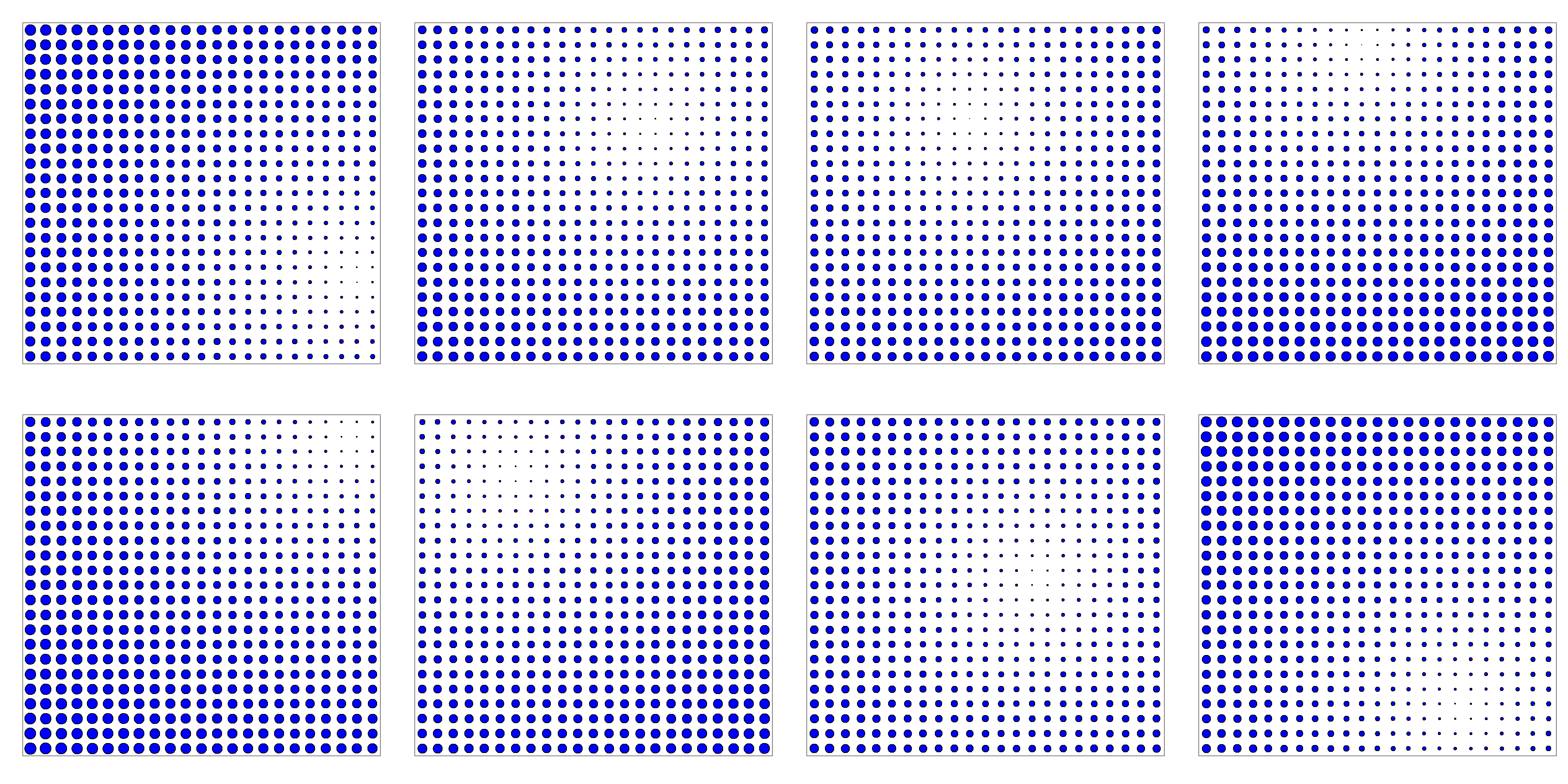}}%
    \caption{\textbf{(a)} An example of $J=8$ measures from the \emph{NI} dataset with $M=500$. \textbf{(b)} An example of $J=8$ measures from the \emph{NIG} dataset with $M=23^2$.}%
    \label{fig:normeddata}
  \end{figure}
  \subsubsection*{Norm-based Intensities on Uniform Positions (NI), see \Cref{fig:normeddata} (a)}
   Fix $J$ locations $l^1_0,\dots,l^J_0 \in [0,1]^2$. Let $l^i_1,\dots,l^i_K \sim U[0,1]^2$ and let $w^i_k=\lVert l^i_k-l_0^i \rVert_2$ for $1\leq i\leq J$. 
   \subsubsection*{Norm-based Intensities on a Grid (NIG), see \Cref{fig:normeddata} (b)}
   Let $K=M^2$ for $M\in \mathbb{N}$. Fix $J$ locations $l^1_0,\dots,l^J_0 \in [0,1]^2$ and let $l^i_1,\dots,l^i_{M^2}$ be the points of an equidistant $M\times M$ grid in $[0,1]^2$ for each $1\leq i\leq J$. Set $w^i_k=\lVert l^i_k-l_0^i \rVert_2$ for $1\leq i\leq J$.

  \subsection{Simulations for the Poisson Model}\label{sec:addfigpoi}

  This section details the simulation results for the additional classes of measures SPI, SPIC, NI, and NIG. Figures \ref{fig:dist_poiSPI} -- \ref{fig:dist_poiNIG} display the relative error of the empirical KRD in estimating the population KRD  across several choices for $C\in \{0.01, 0.1, 1, 10\}$. Moreover, Figures \ref{fig:bary_poiSPI} -- \ref{fig:bary_poiNIG} detail the relative Fr\'echet error in estimating the $(2,C)$-KR barycenter for the choices $C \in \{0.1, 1, 10\}$. 
  
  Structurally, the simulation results and corresponding conclusions for the SPI and SPIC classes are similar to those for the respective NE and NEC classes. Likewise, the insights about the NI and NIG classes are closely tied to those of the PI and PIG classes, respectively. 
  
  \begin{figure}[H]
    \centering
  \begin{tabular}{cc}
    \centering
        \includegraphics[width=0.37\textwidth]{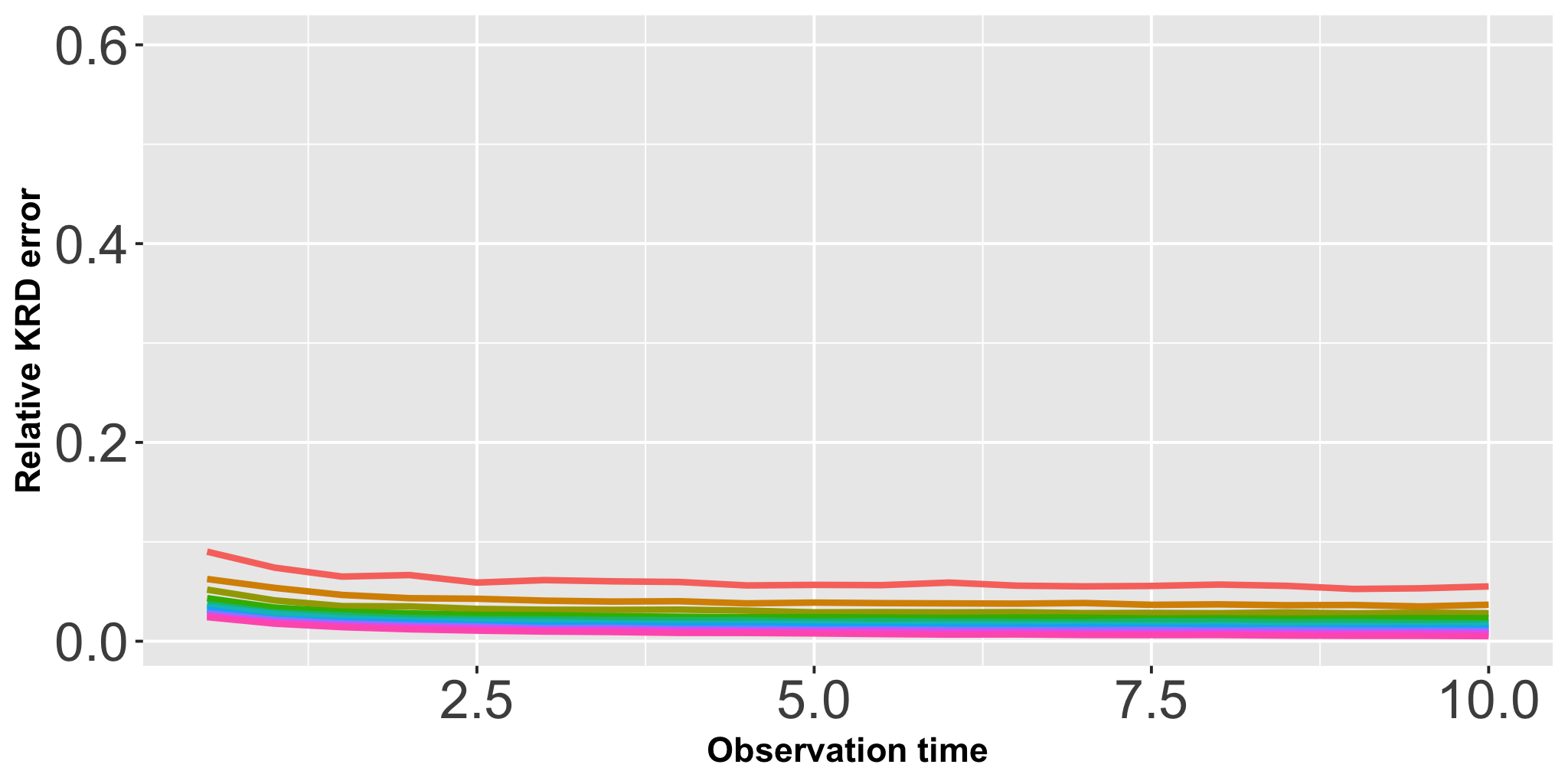} &       \includegraphics[width=0.37\textwidth]{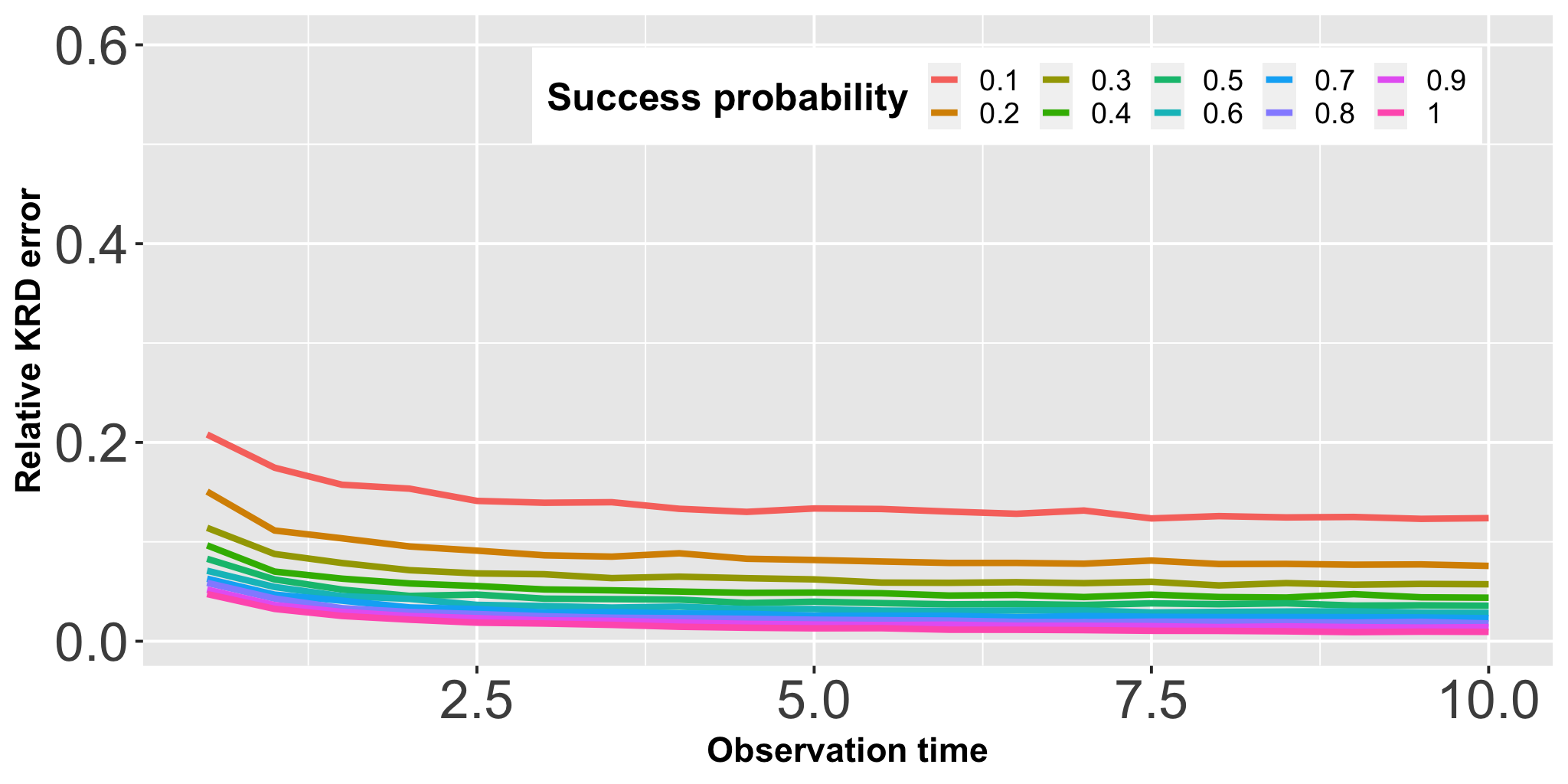} \\
      \includegraphics[width=0.37\textwidth]{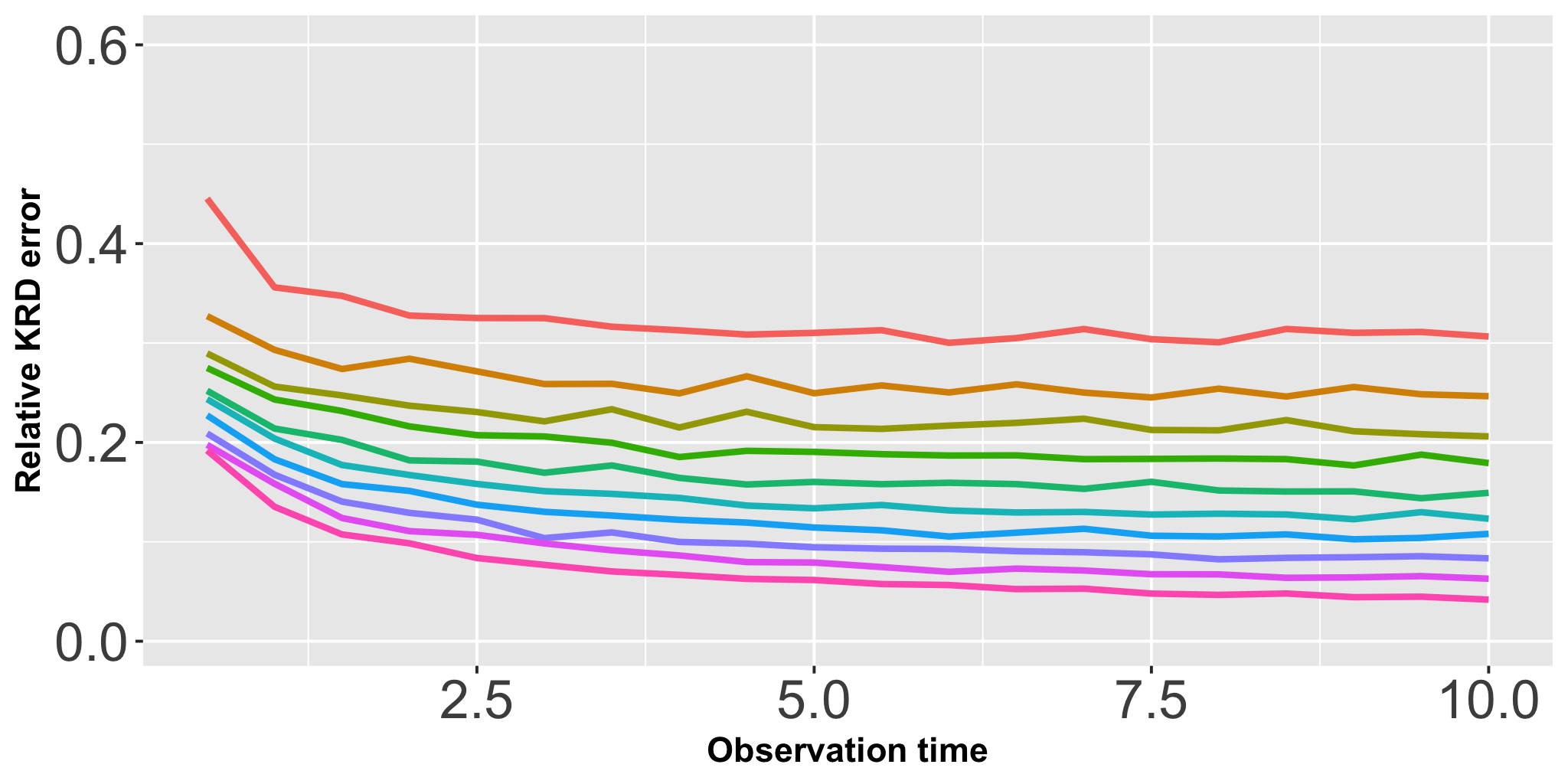} &       \includegraphics[width=0.37\textwidth]{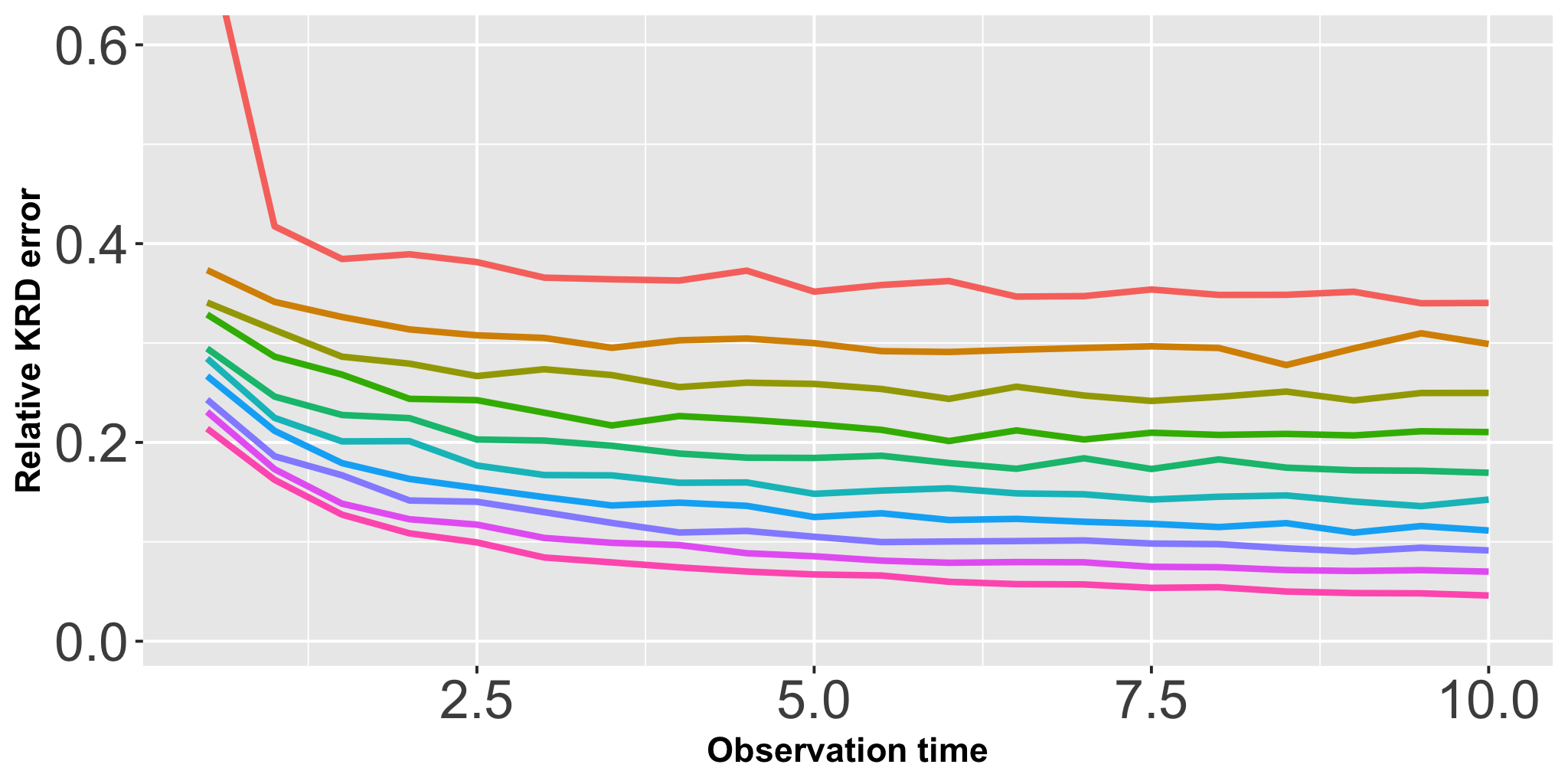} \\
  \end{tabular}
  \caption{As in \Cref{fig:dist_poiPI}, but for the SPI class and $M=65$.}
    \label{fig:dist_poiSPI}
  \end{figure}
  \begin{figure}[H]
    \centering
  \begin{tabular}{cc}
        \includegraphics[width=0.37\textwidth]{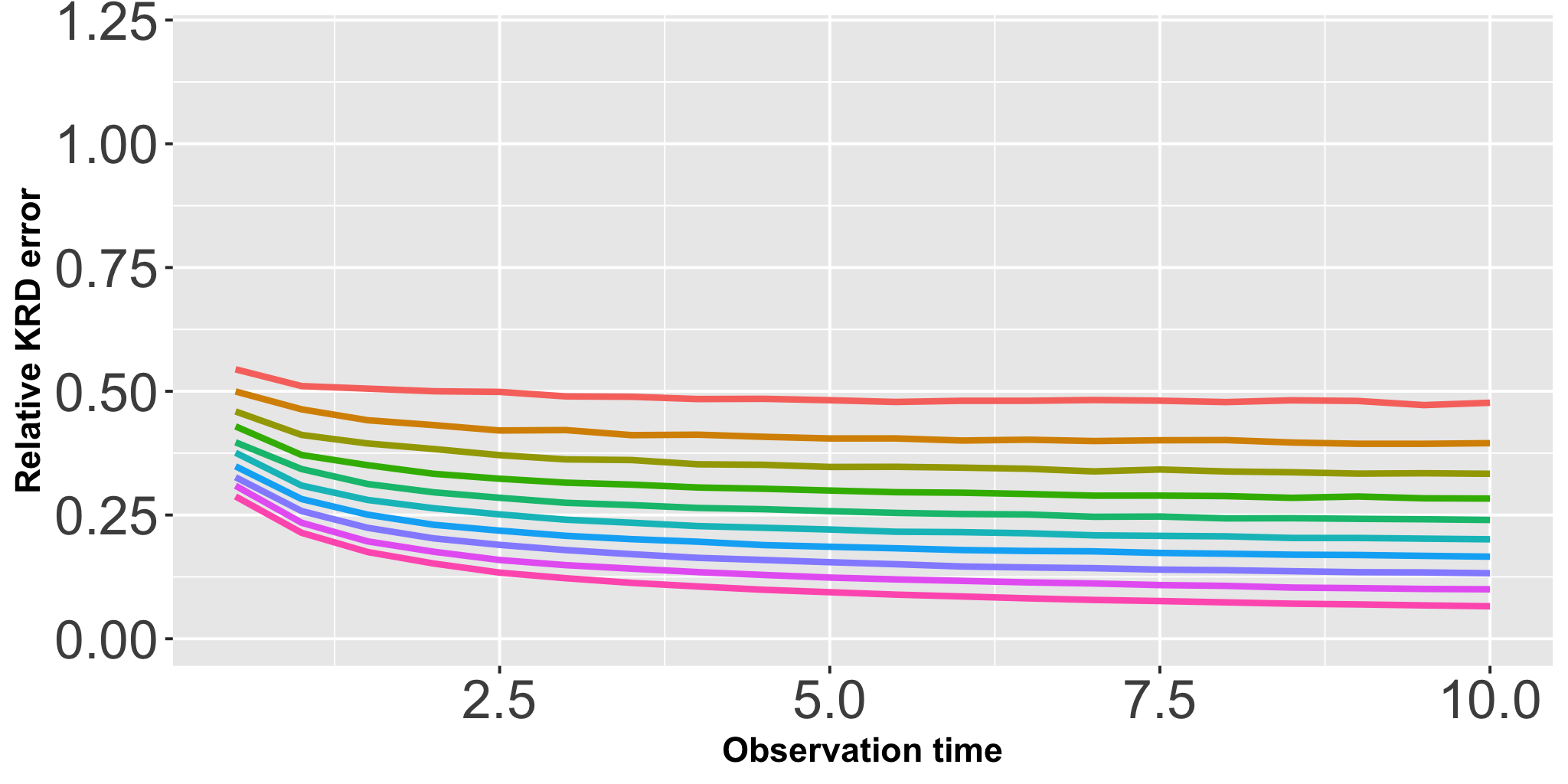} &       \includegraphics[width=0.37\textwidth]{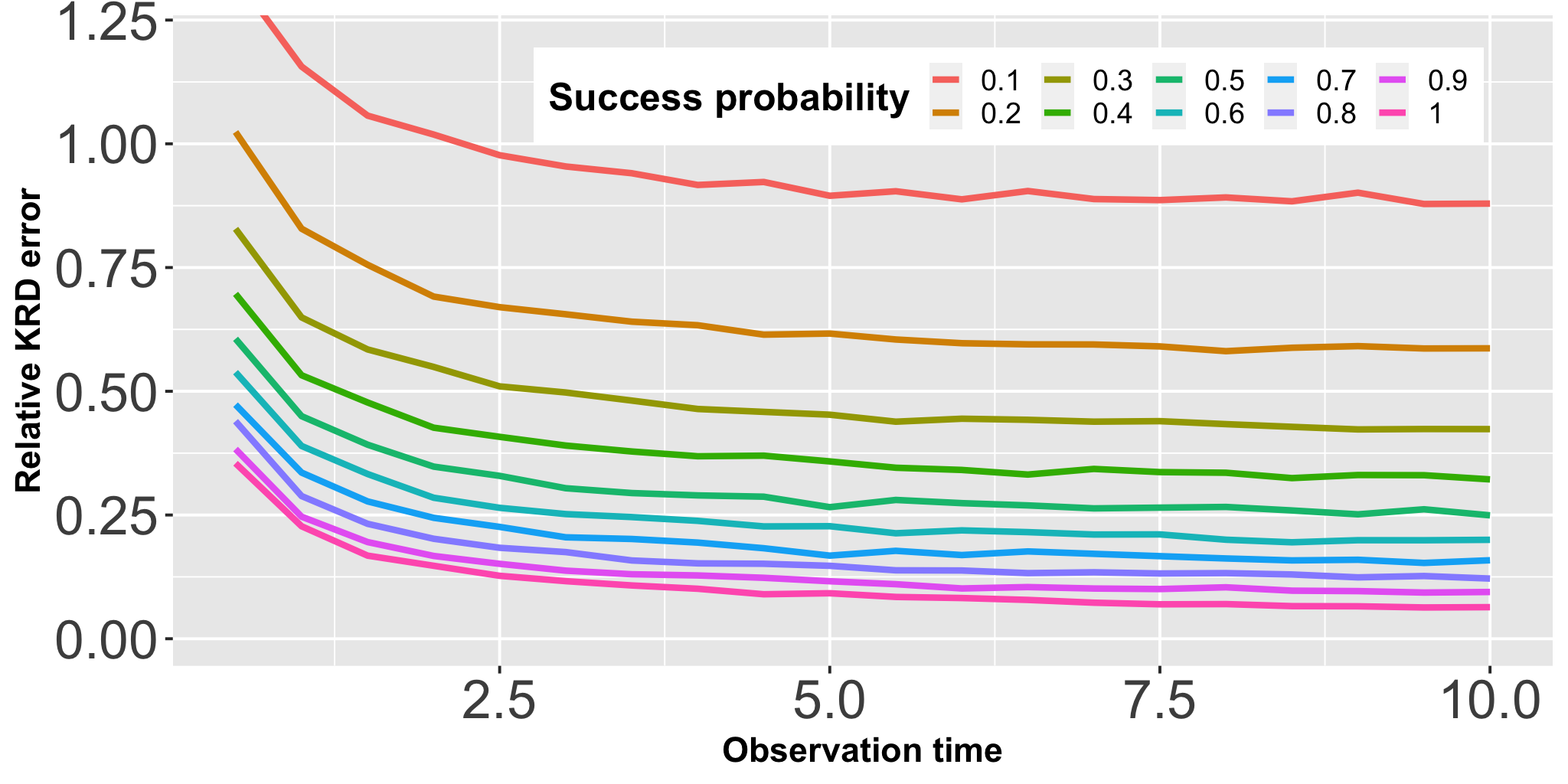} \\
      \includegraphics[width=0.37\textwidth]{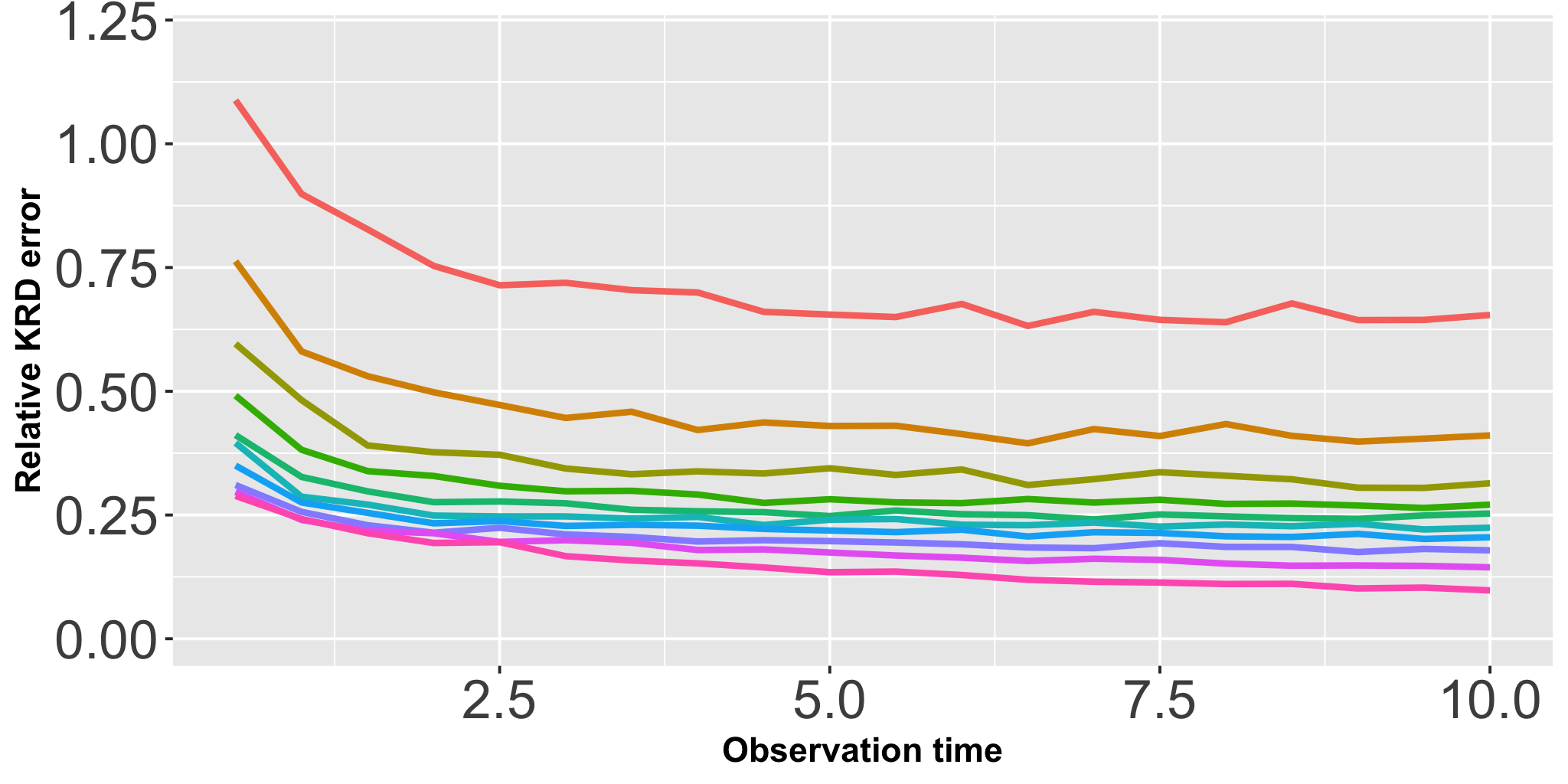} &       \includegraphics[width=0.37\textwidth]{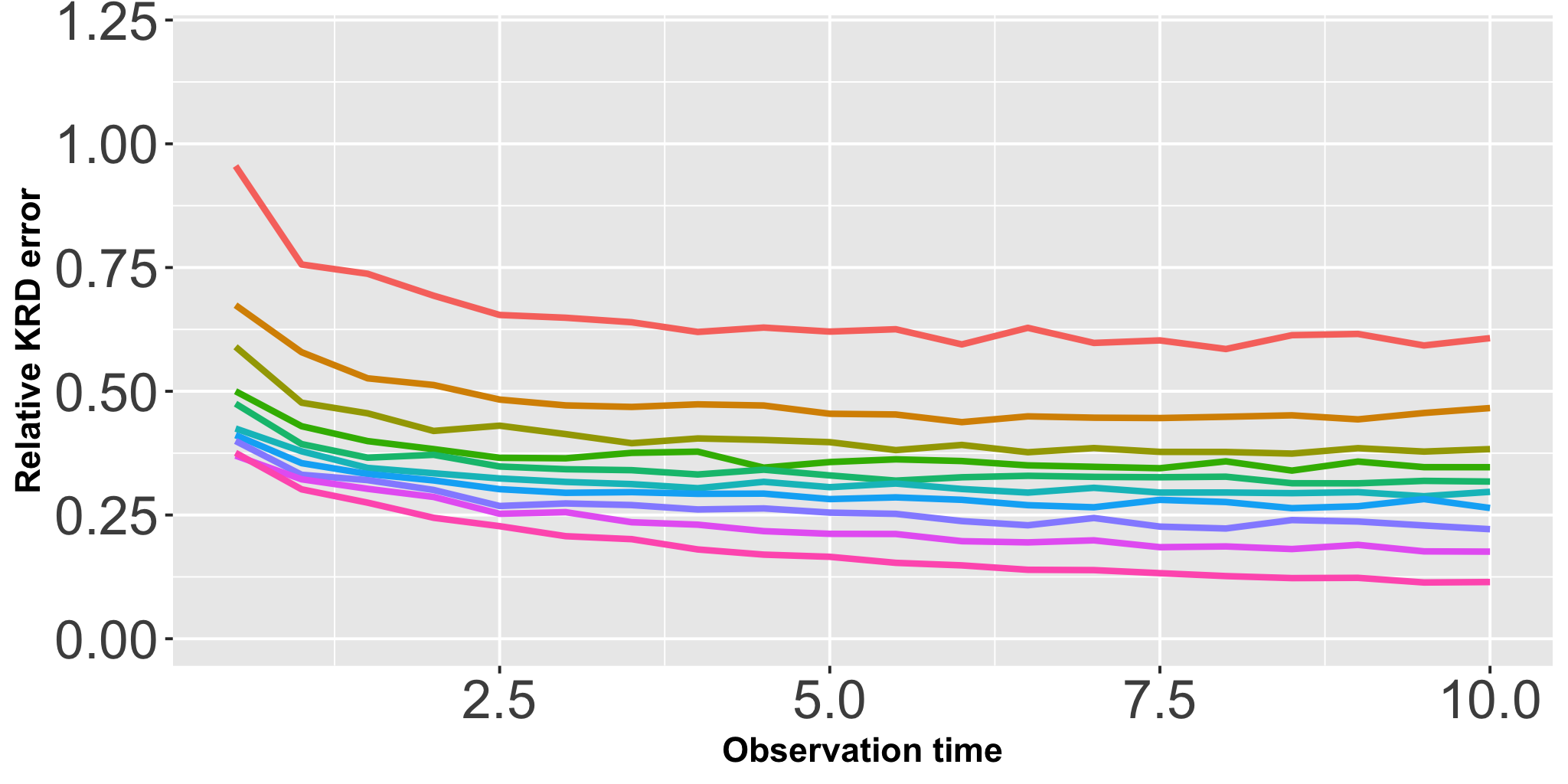} \\
  \end{tabular}
  \caption{As in \Cref{fig:dist_poiPI}, but for the SPIC class and $M=12$.}
    \label{fig:dist_poiSPIC}
  \end{figure}
  
  \begin{figure}[H]
    \centering
   \begin{tabular}{cc}
         \includegraphics[width=0.37\textwidth]{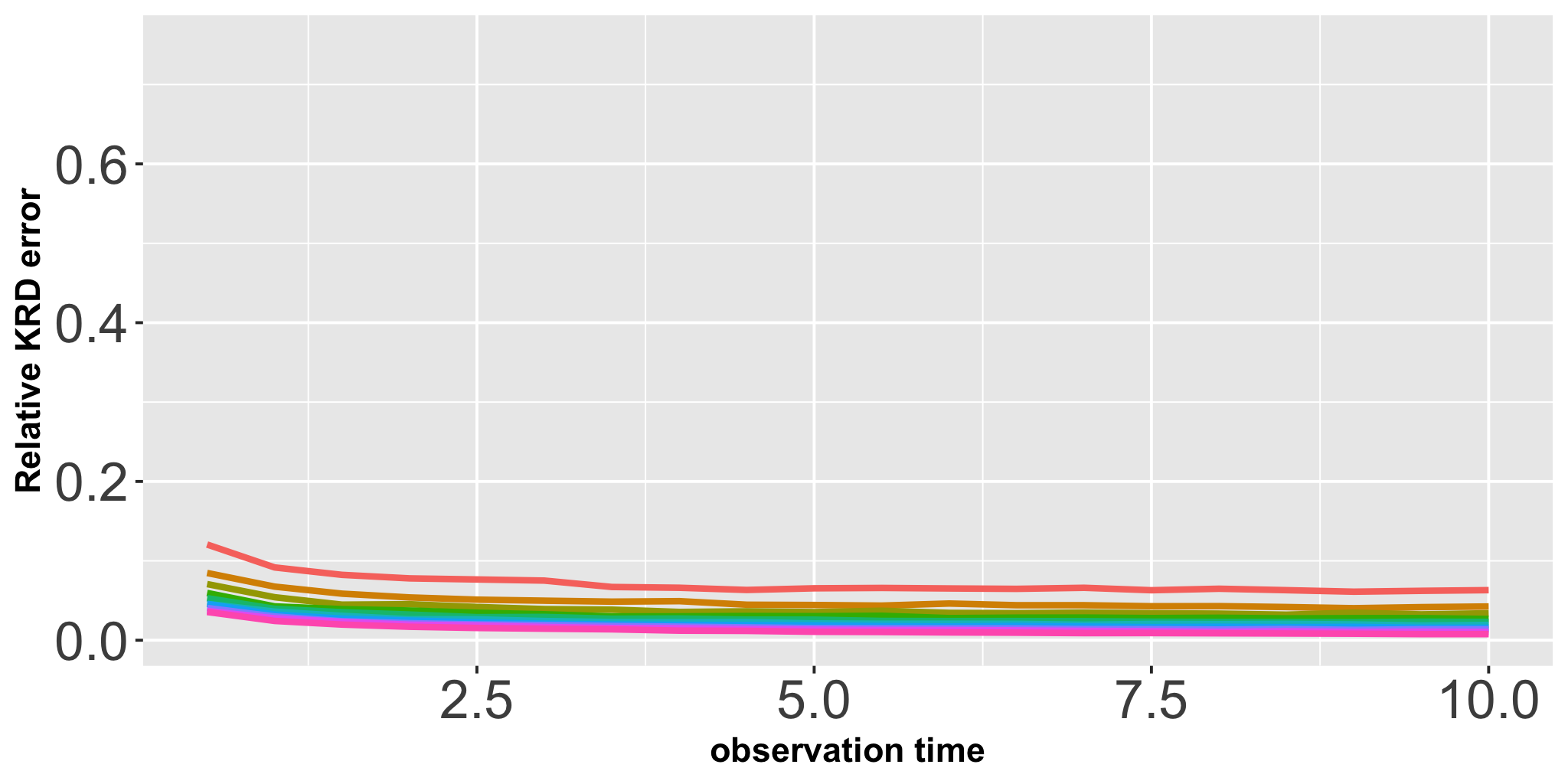} &       \includegraphics[width=0.37\textwidth]{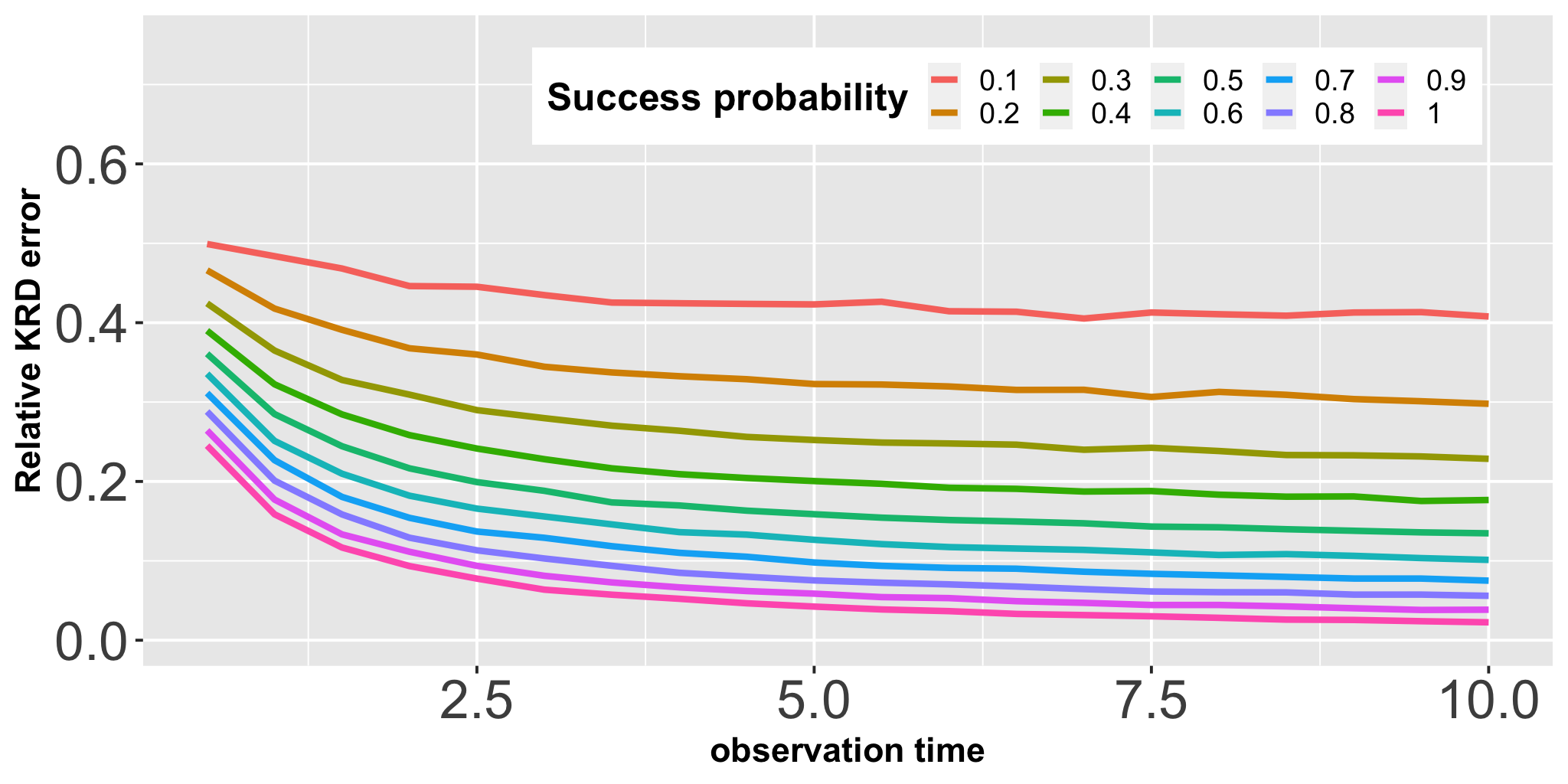} \\
       \includegraphics[width=0.37\textwidth]{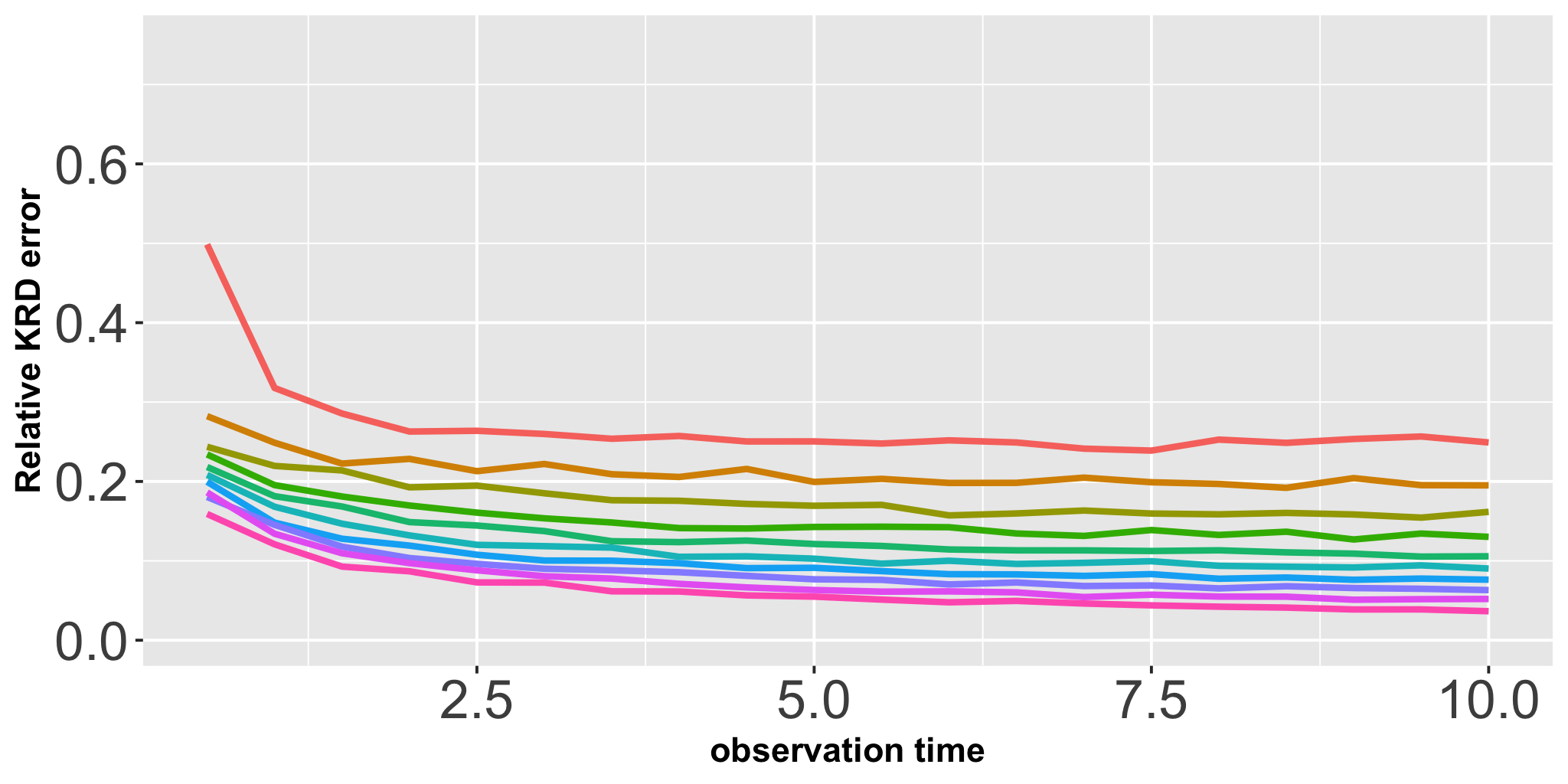} &       \includegraphics[width=0.37\textwidth]{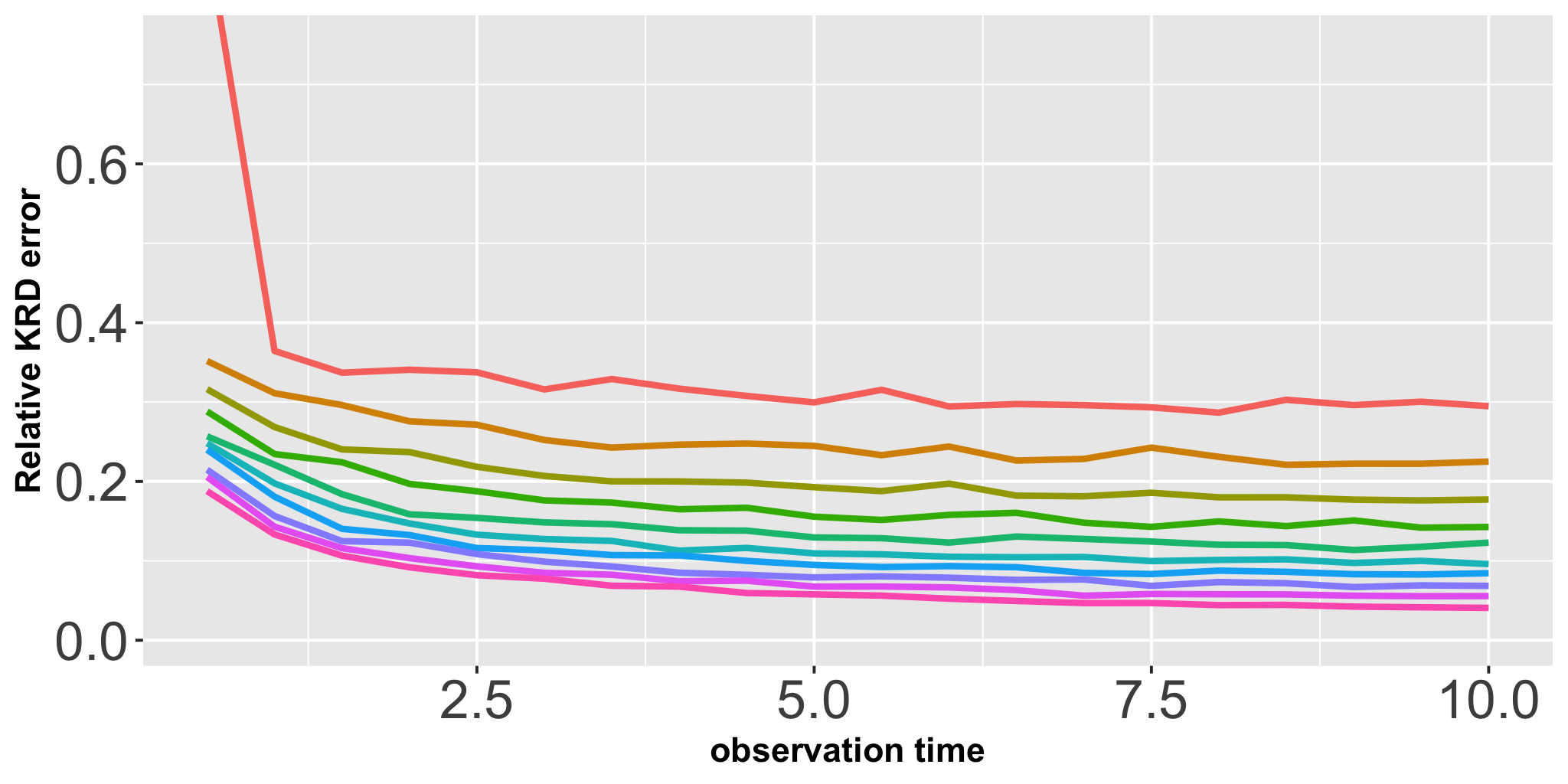} \\
   \end{tabular}
   \caption{As in \Cref{fig:dist_poiPI}, but for the NI class and $M=300$.}
     \label{fig:dist_poiNI}
   \end{figure}
   \begin{figure}[H]
    \centering
   \begin{tabular}{cc}
         \includegraphics[width=0.37\textwidth]{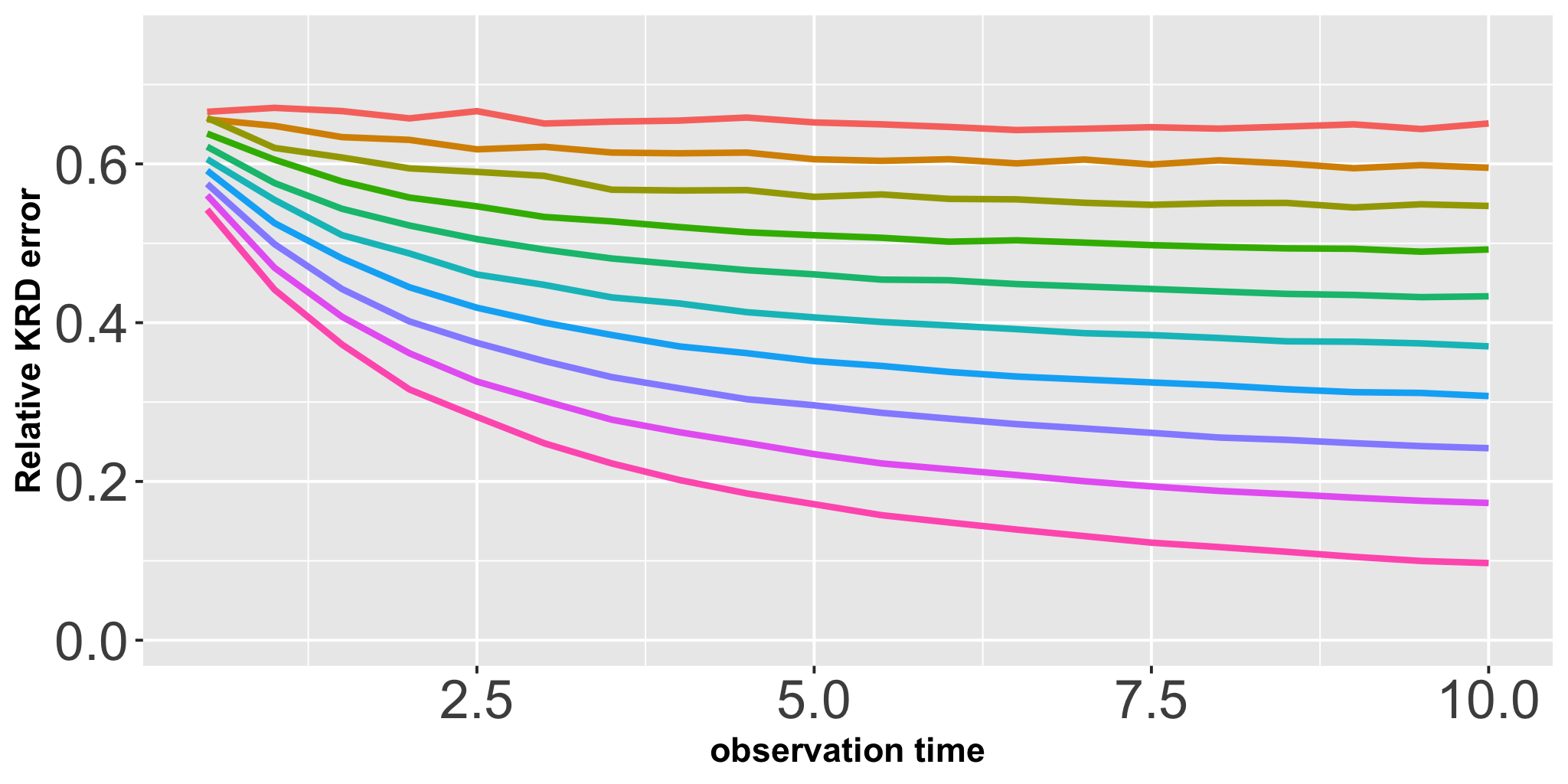} &       \includegraphics[width=0.37\textwidth]{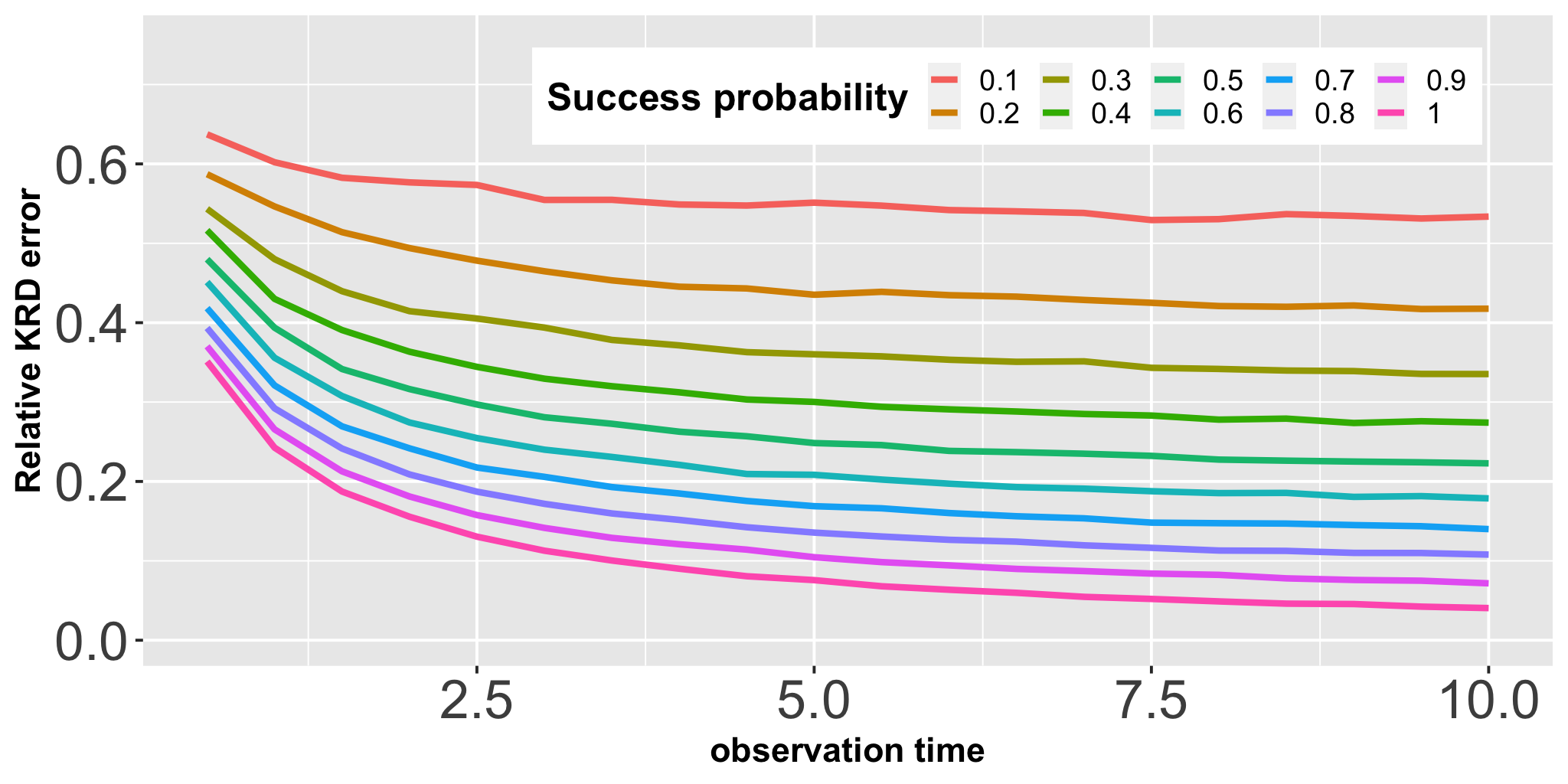} \\
       \includegraphics[width=0.37\textwidth]{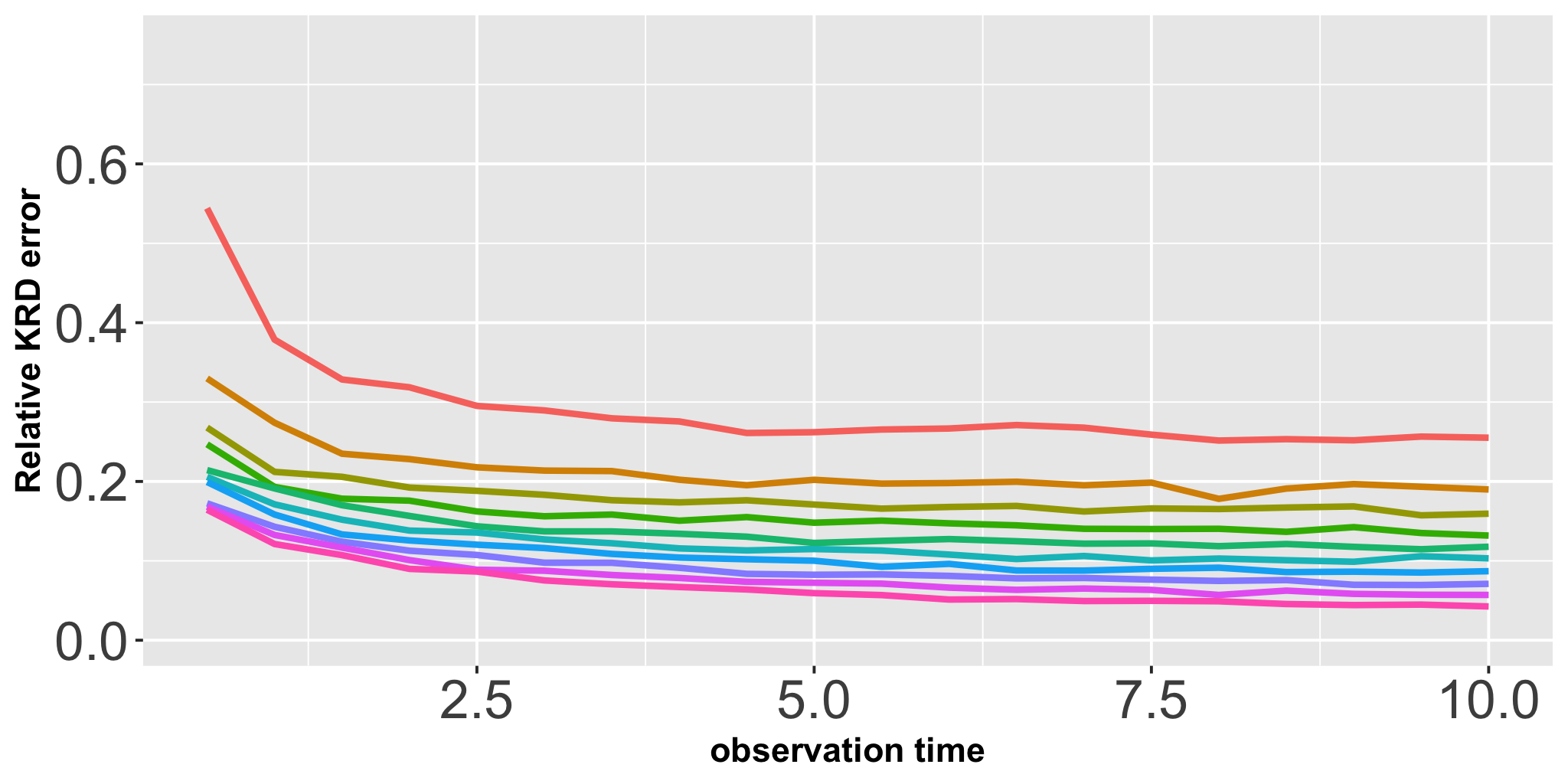} &       \includegraphics[width=0.37\textwidth]{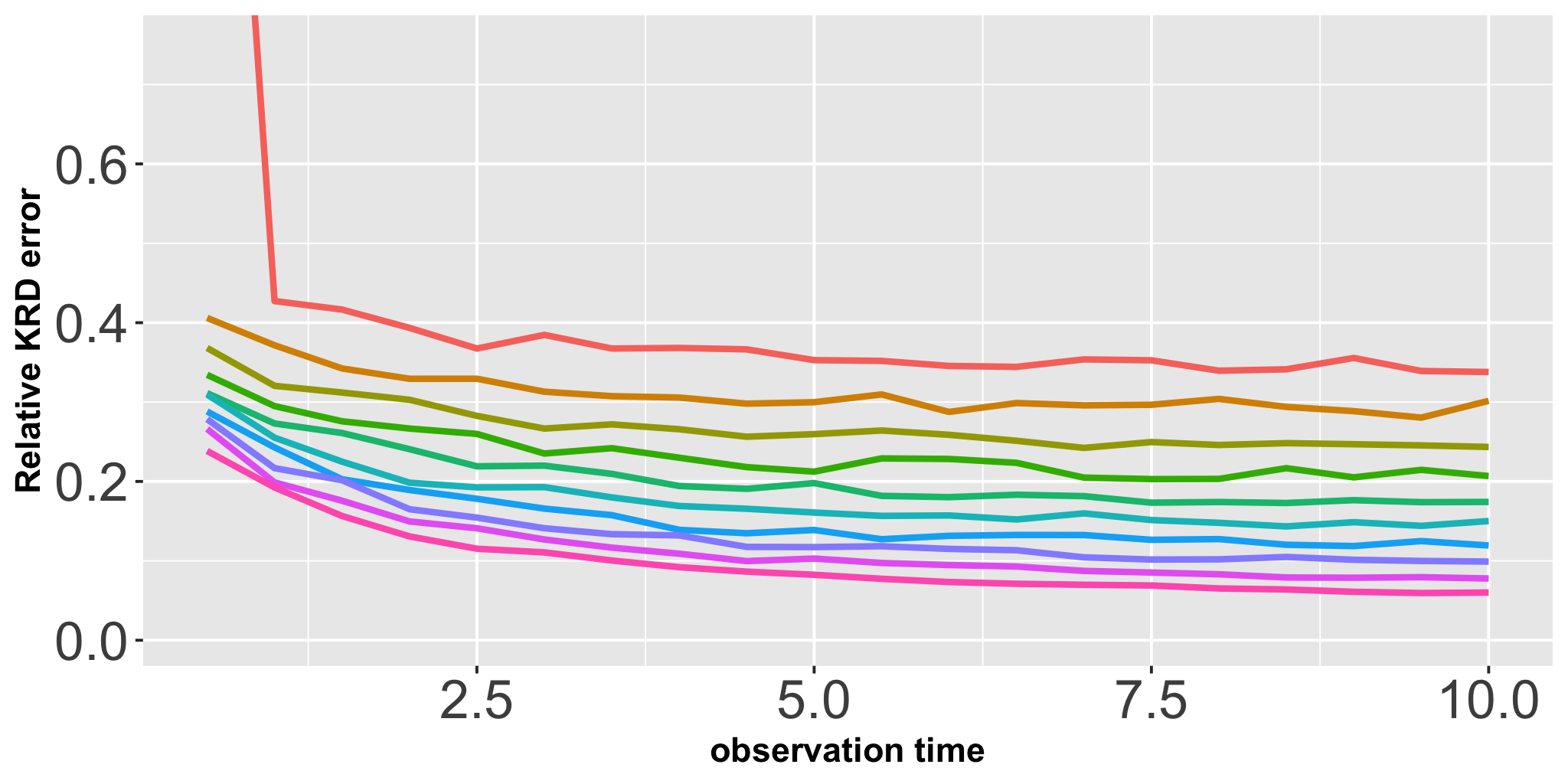} \\
   \end{tabular}
   \caption{As in \Cref{fig:dist_poiPI}, but for the NIG class and $M=17$.}
     \label{fig:dist_poiNIG}
   \end{figure}

  \begin{figure}[H]
  \begin{tabular}{ccc}
  \includegraphics[width=0.3\textwidth]{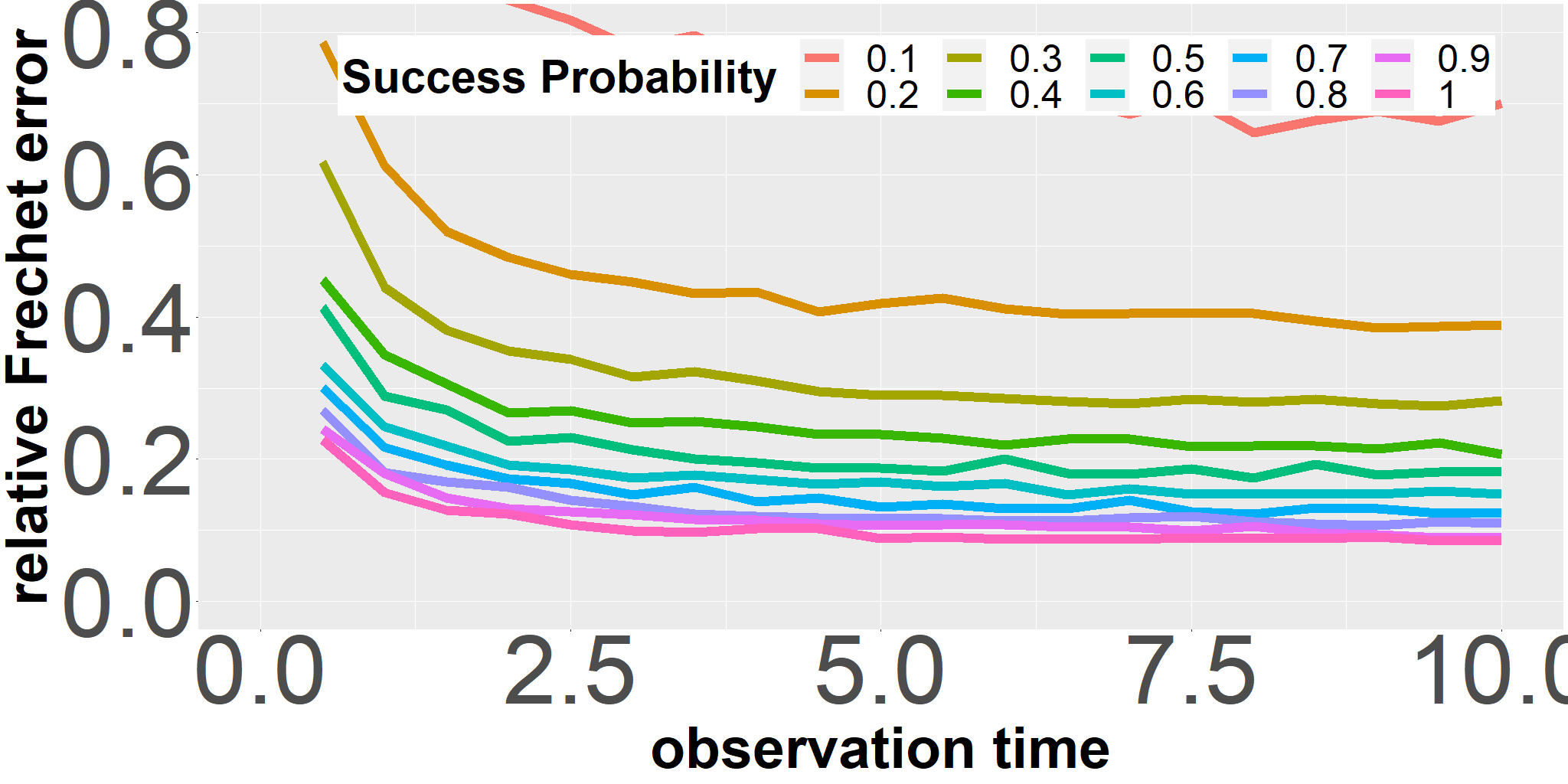} &
      \includegraphics[width=0.3\textwidth]{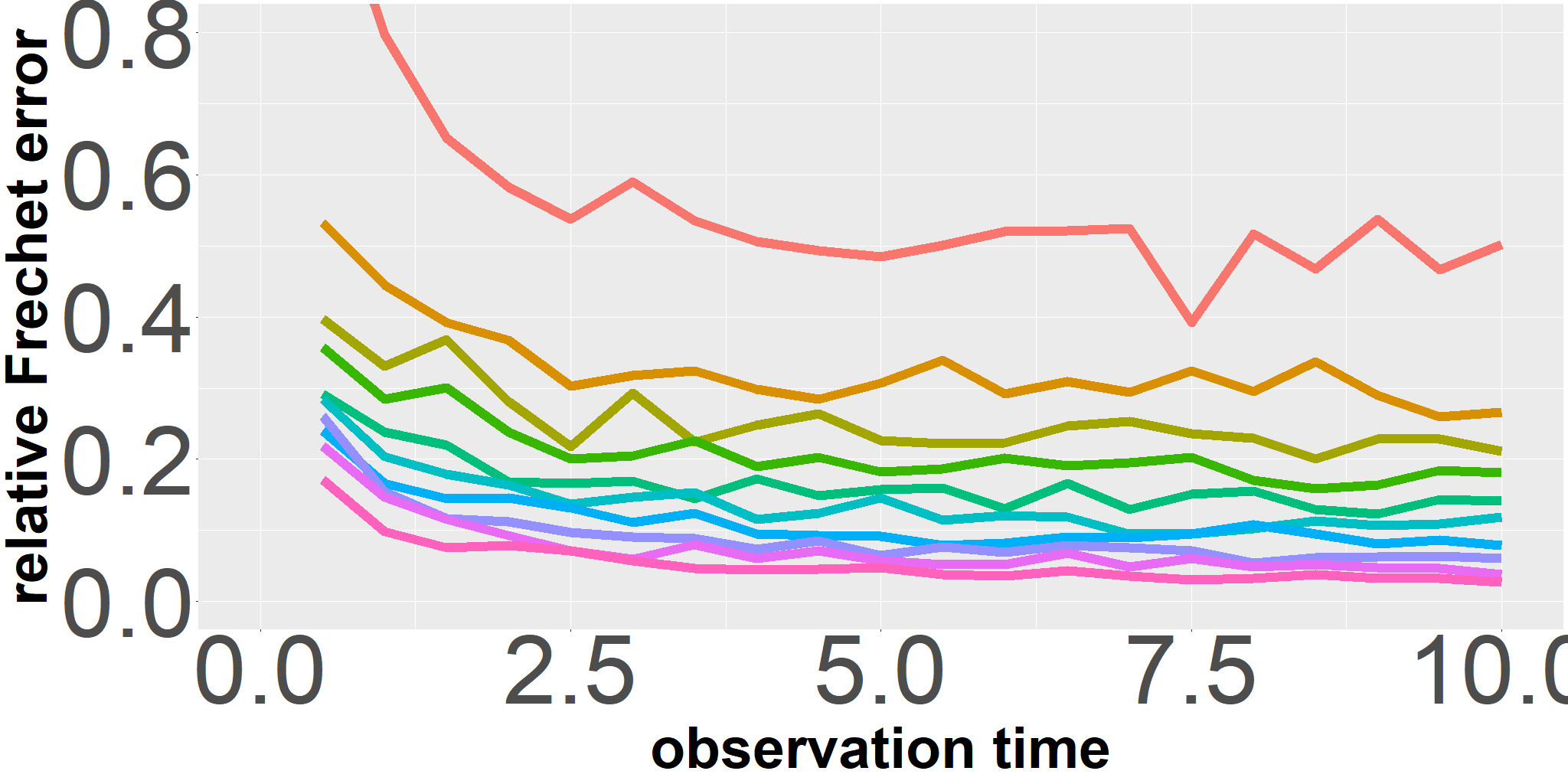} &       \includegraphics[width=0.3\textwidth]{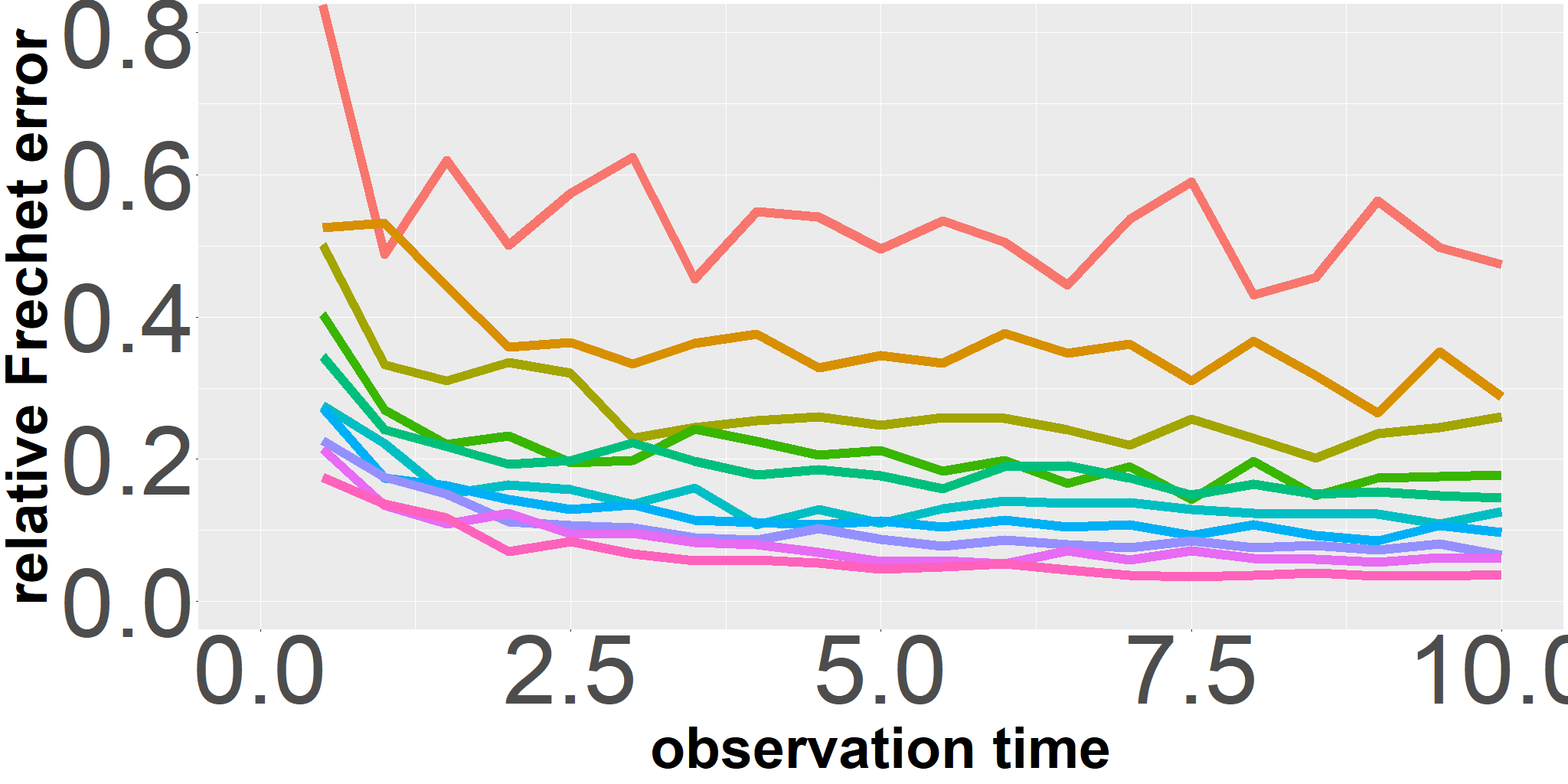} 
  \end{tabular}
  \caption{As in \Cref{fig:bary_poiPIG}, but for the SPI class and $M=65$.}
    \label{fig:bary_poiSPI}
  \end{figure}
  \begin{figure}[H]
  \begin{tabular}{ccc}
  \includegraphics[width=0.3\textwidth]{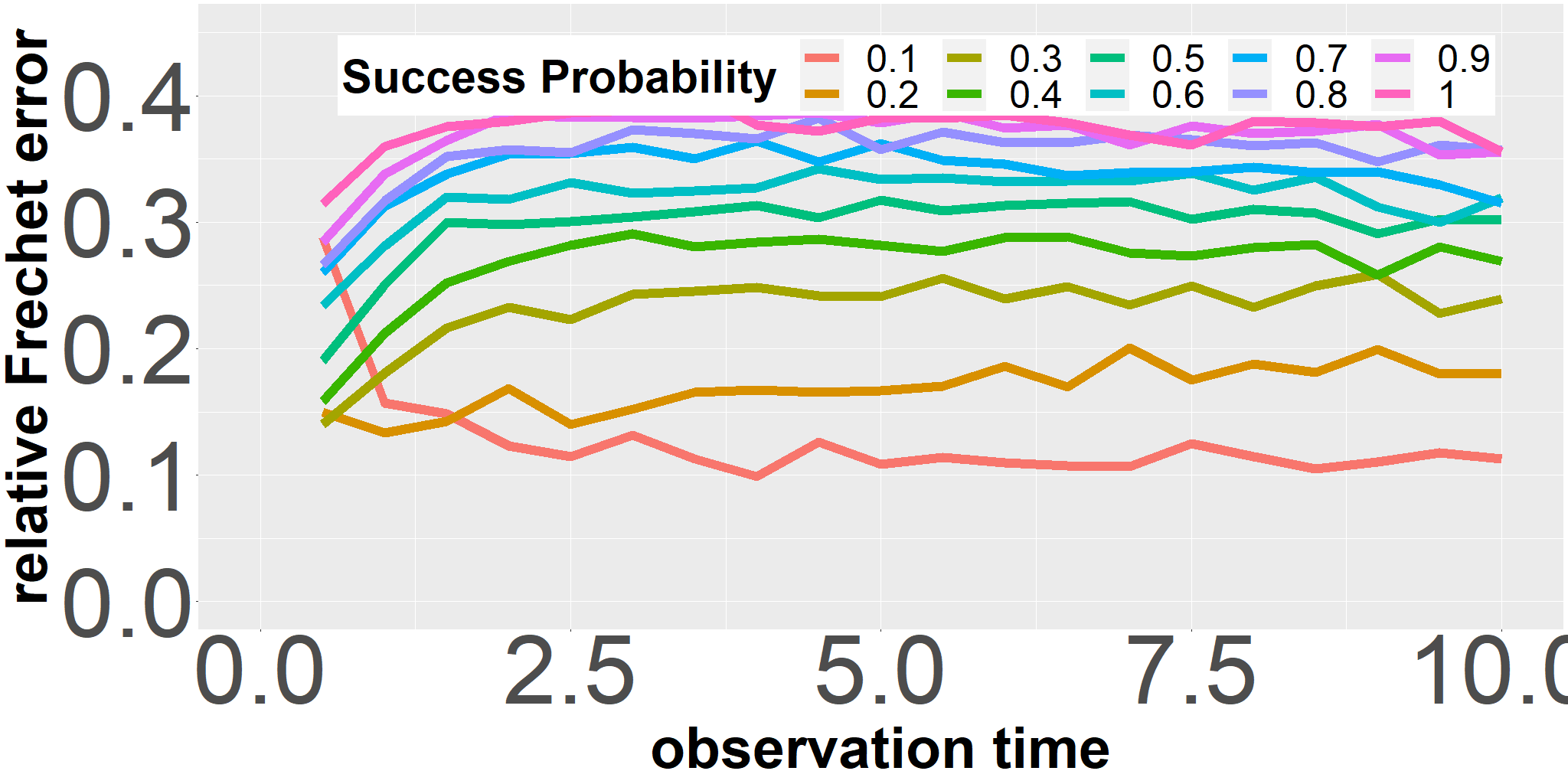} &
      \includegraphics[width=0.3\textwidth]{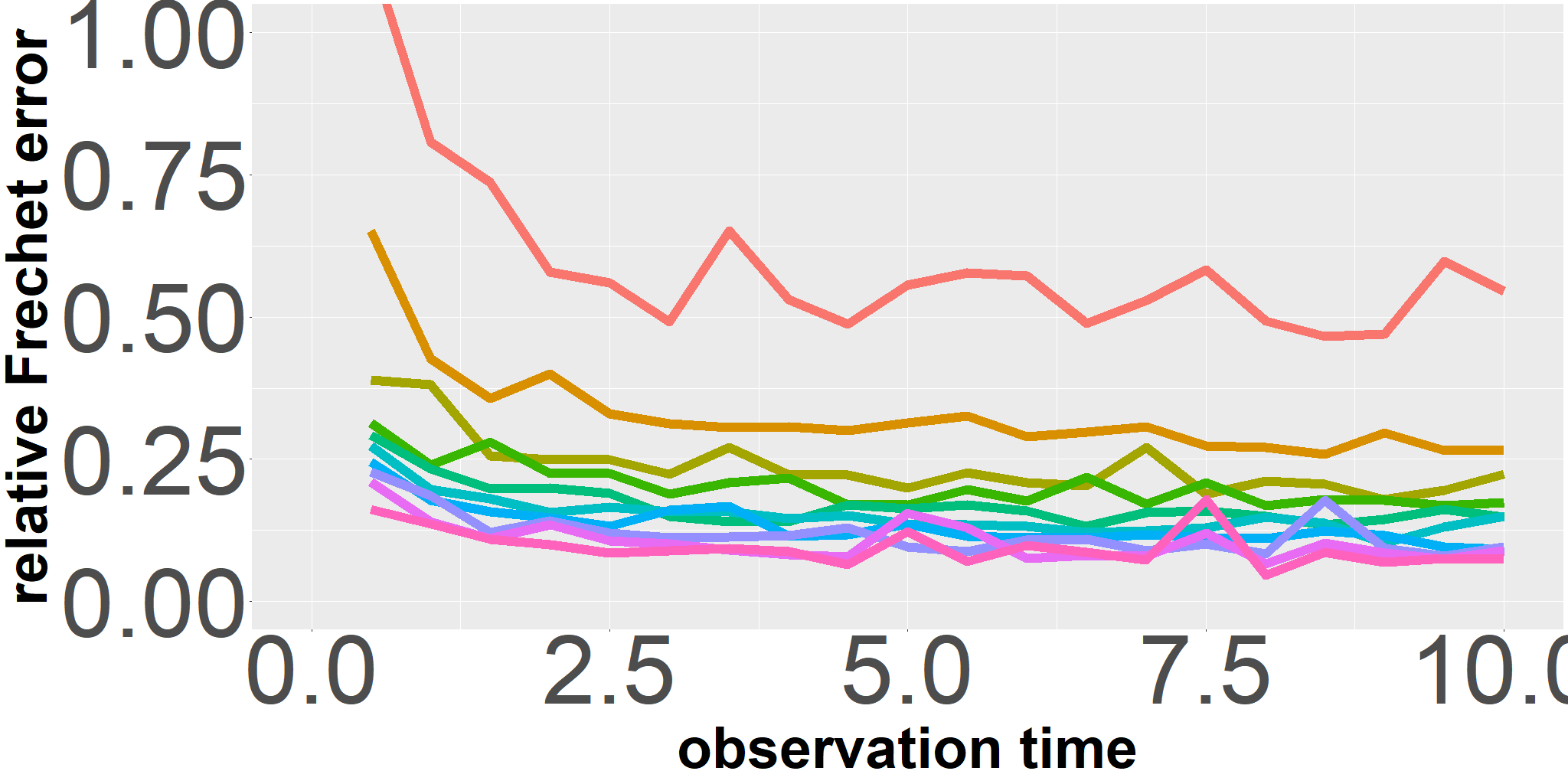} &       \includegraphics[width=0.3\textwidth]{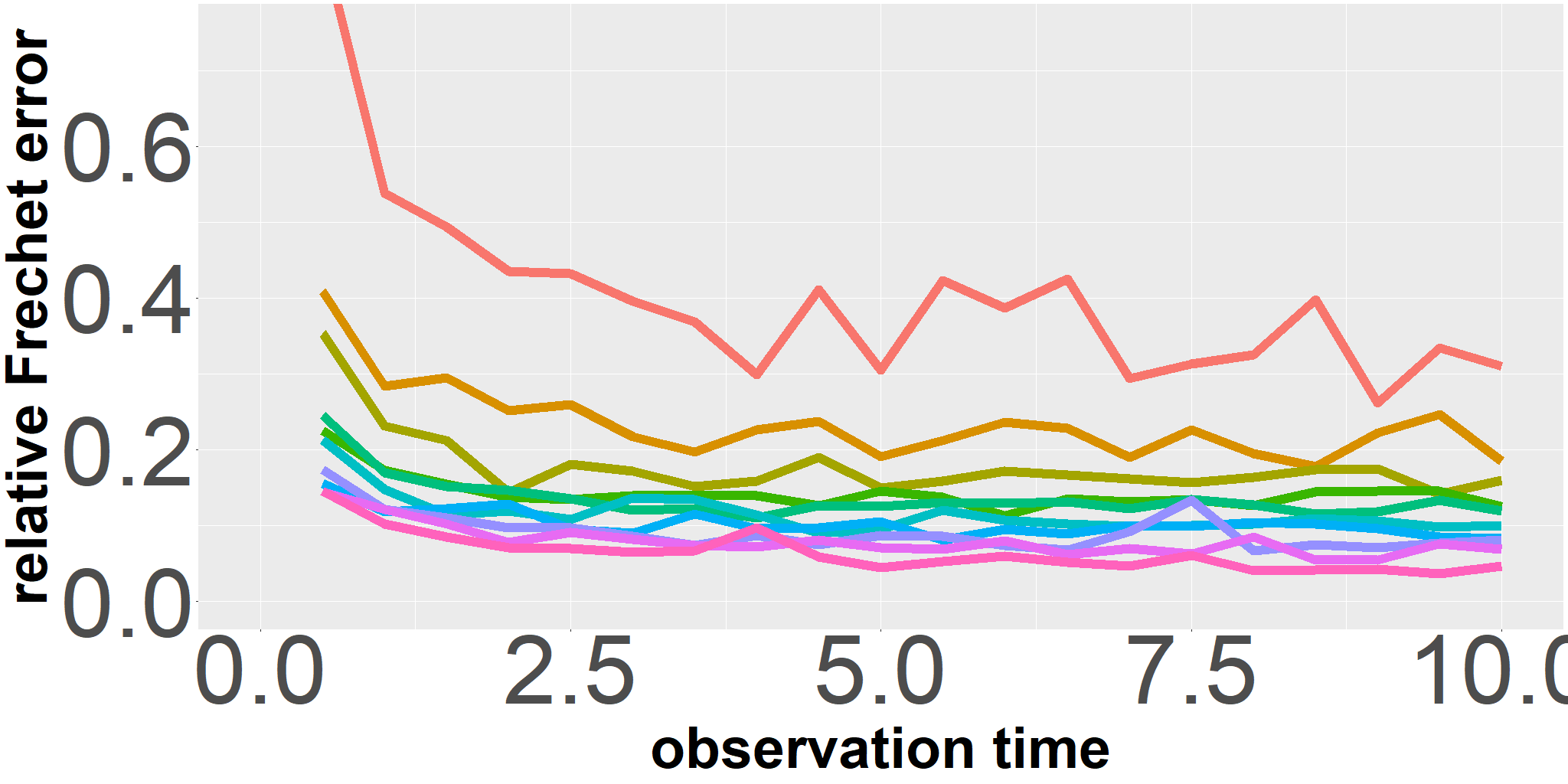} 
  \end{tabular}
  \caption{As in \Cref{fig:bary_poiPIG}, but for the SPIC class and $M=12$.}
    \label{fig:bary_poiSPIC}
  \end{figure}
  
  \begin{figure}[H]
    \begin{tabular}{ccc}
    \includegraphics[width=0.3\textwidth]{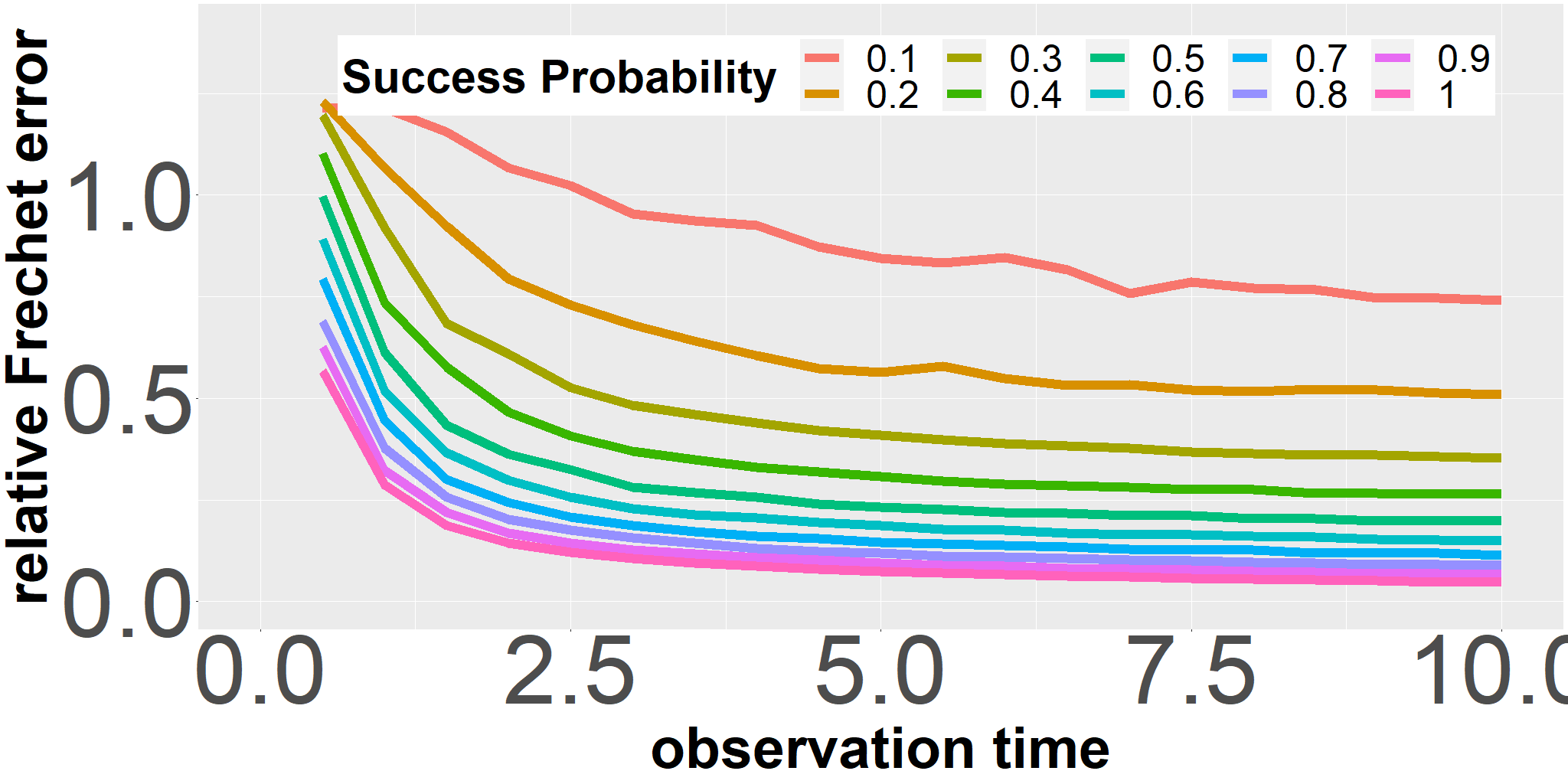} &
        \includegraphics[width=0.3\textwidth]{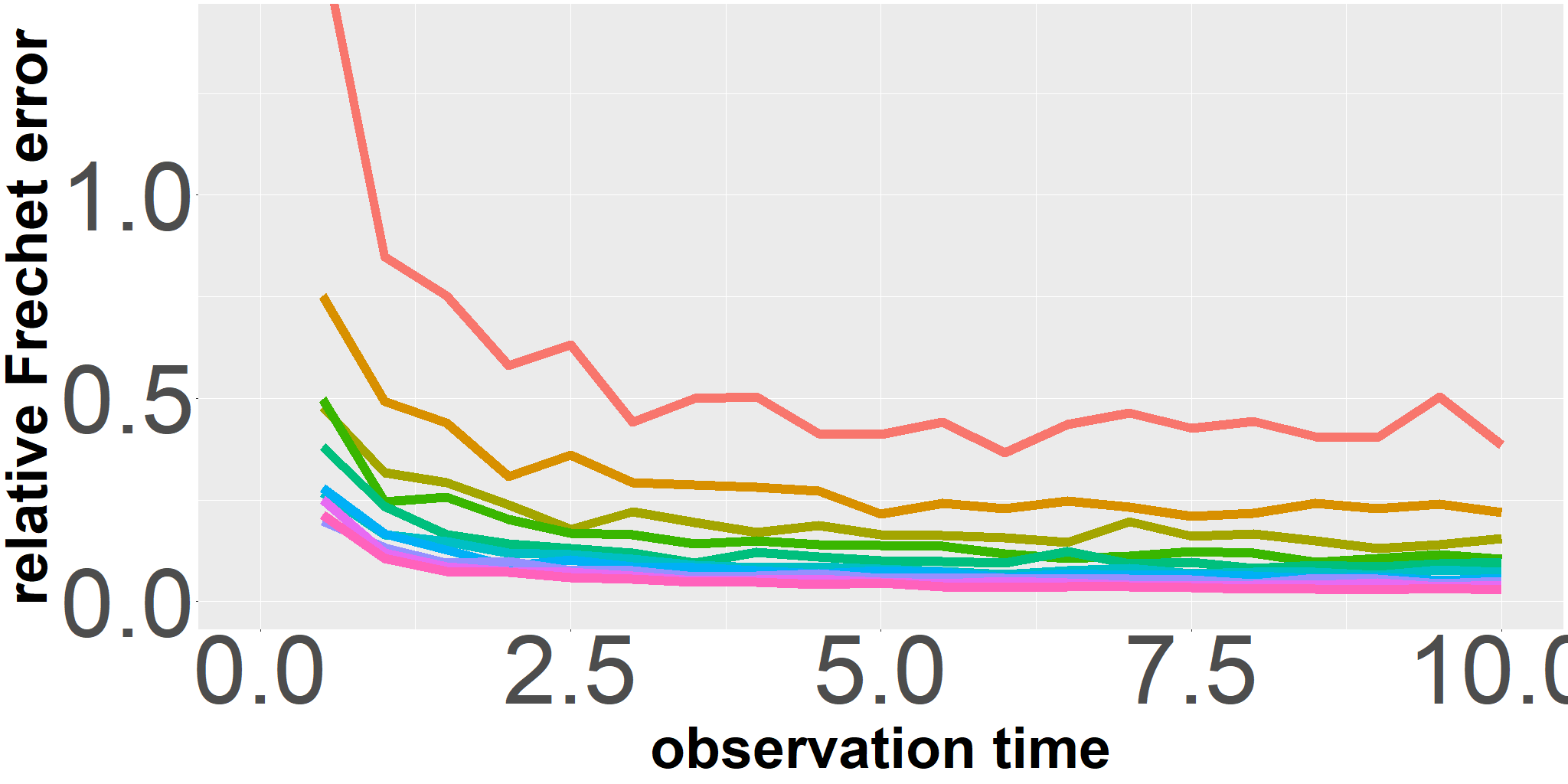} &       \includegraphics[width=0.3\textwidth]{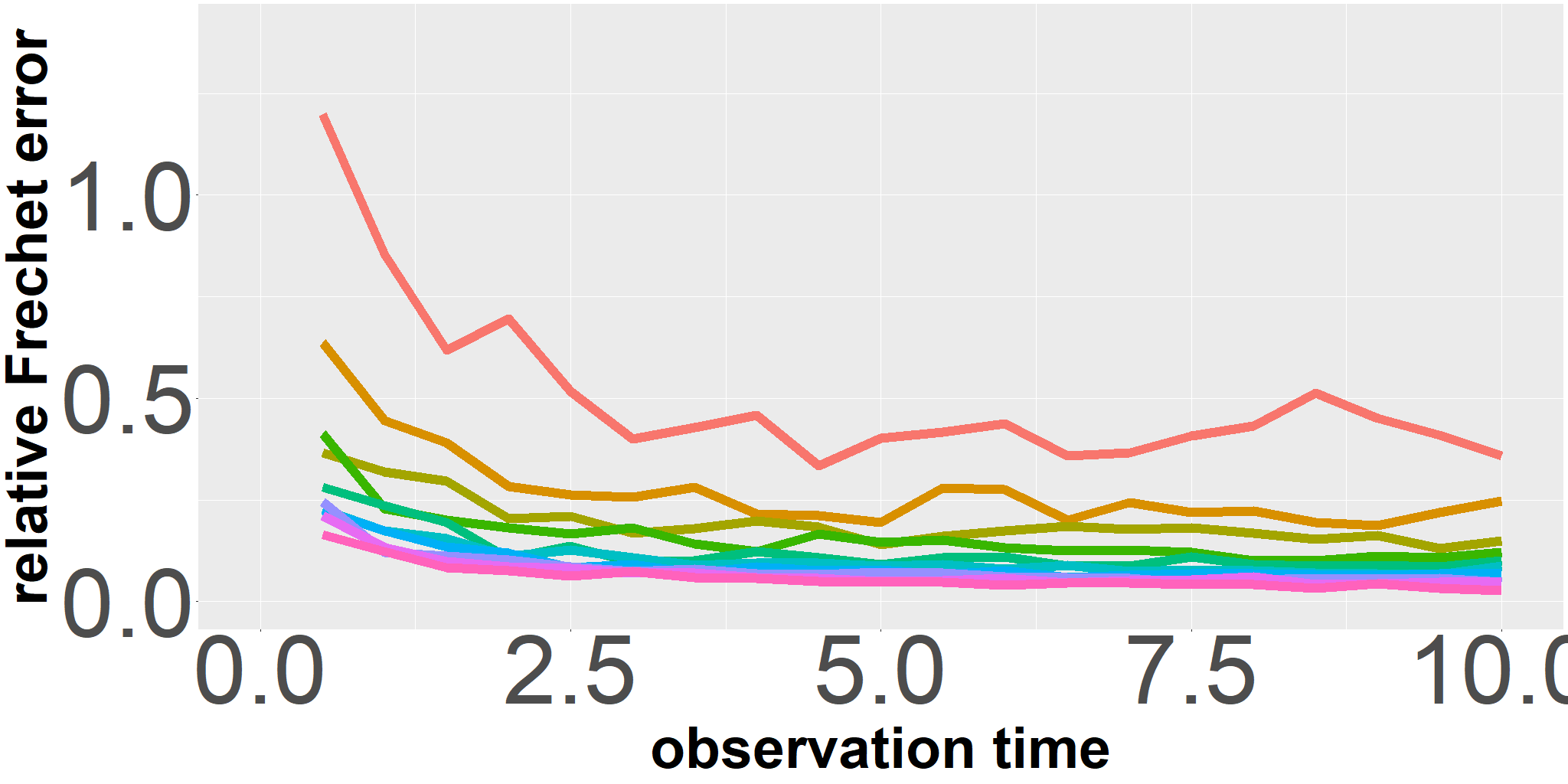} \\
    \end{tabular}
    \caption{  As in \Cref{fig:bary_poiPIG}, but for the NI class and $M=300$.}
      \label{fig:bary_poiNI}
    \end{figure}
    \begin{figure}[H]
    \begin{tabular}{ccc}
       \includegraphics[width=0.3\textwidth]{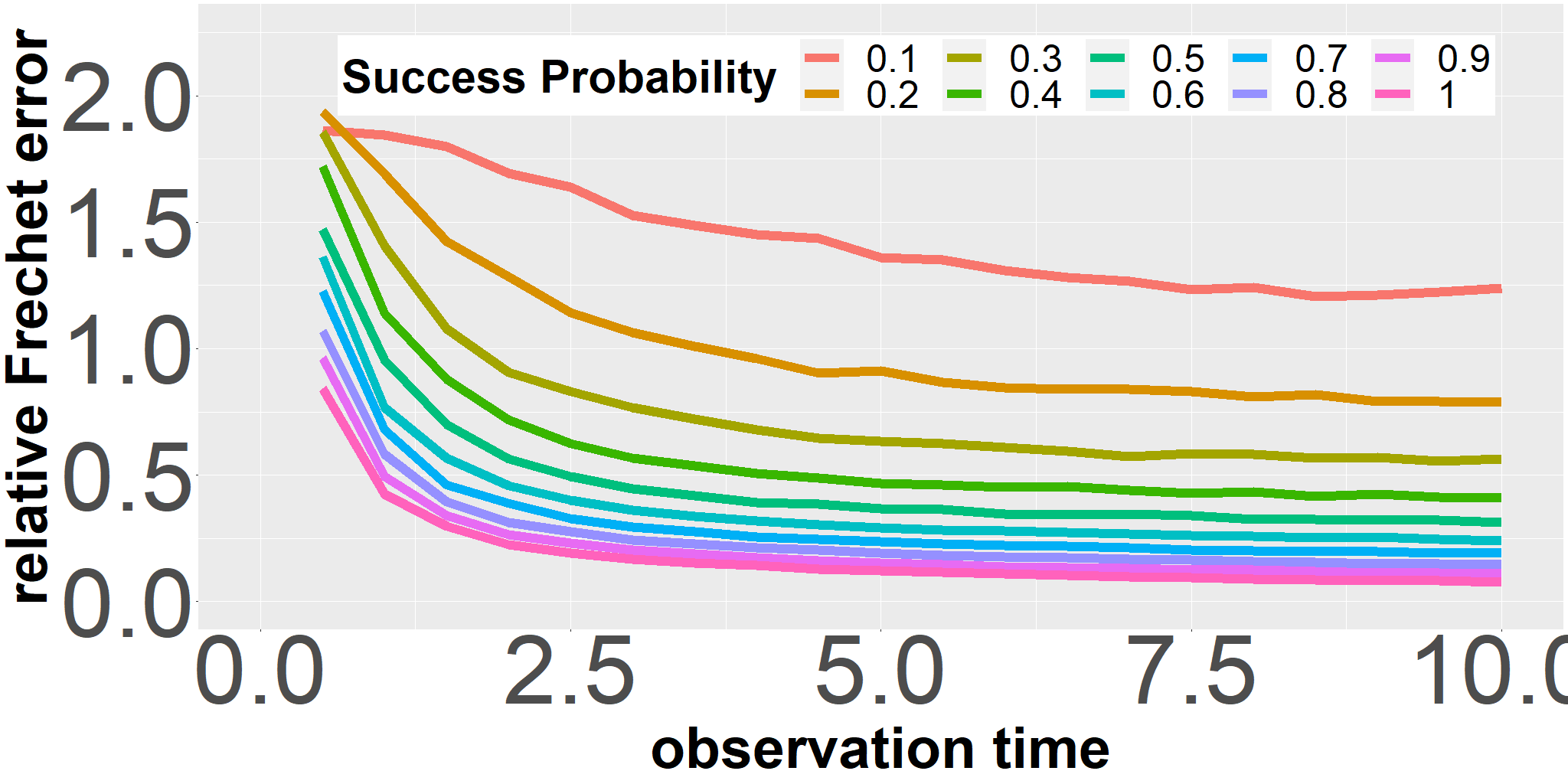} &
        \includegraphics[width=0.3\textwidth]{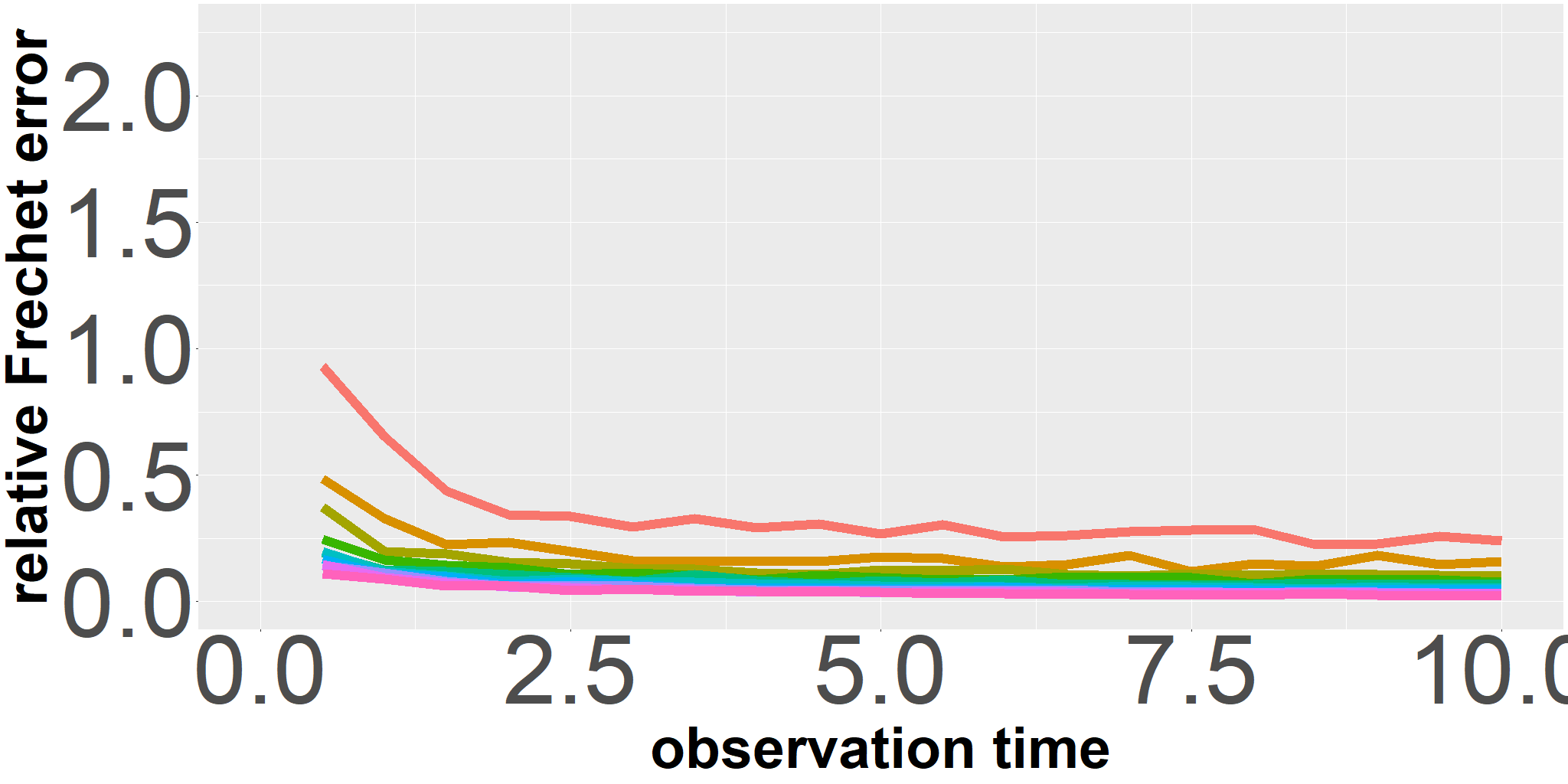} &       \includegraphics[width=0.3\textwidth]{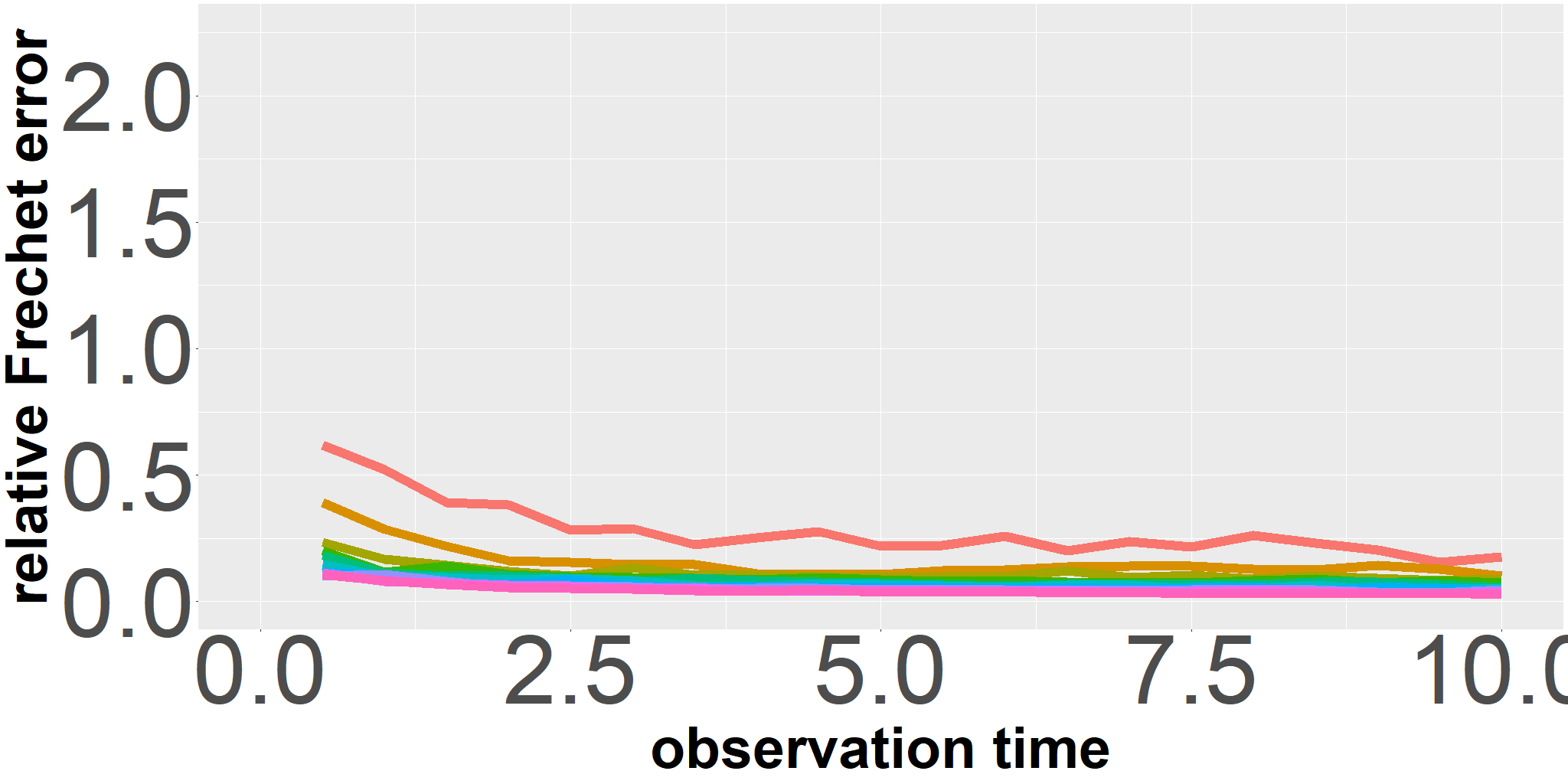} 
    \end{tabular}
    \caption{As in \Cref{fig:bary_poiPIG}, but for the NIG class and $M=17$.}
      \label{fig:bary_poiNIG}
    \end{figure}
  
  \subsection{Simulations for Multinomial Model}\label{app:sim_mult}
  In this subsection, we repeat the simulations from \Cref{sec:sims}  and Appendix \ref{sec:addfigpoi} for the multinomial model. We observe that with increasing sample size $N$, the statistical error mostly decreases with rate $N^{-1/2}$ which is what our theory in Appendix \ref{sec:multi} asserts.
  For the $(p,C)$-KRD the results (in \Cref{fig:mult_dist}) slightly differ from the results in the Poisson model. Notably, for increasing $C$, the error is decreasing. This is explained by the fact that in the multinomial scheme, we do not have to estimate the total intensities of the measures, and it is precisely this estimation error that drives the error for increasing $C$ in the Poisson model. Similarly to the Poisson scheme, we observe a decrease in error for the measure classes with clustered support structures when $C$ surpasses the distance between two individual clusters. \\
  For the $(p,C)$-barycenters under the multinomial sampling model (in \Cref{fig:mult_bary})  there is an initial increase in error for small sample sizes. Specifically, this occurs for $C=0.1$ and the NEC and SPIC classes. This value of $C$ is below the cluster size. This effect is most likely for these measure classes. For increasing $C$ there is a significant reduction in estimation error. In particular, for some classes the error reduces by two orders of magnitude going from $C=0.1$ to $C=10$. Since the total mass intensities of the individual measures do not need to be estimated in this sampling model, we already observe a decrease in error for increasing $C$ for the $(p,C)$-KRD and naturally there is a similar effect for the Fr\'echet functional. 
  \begin{figure}
      \centering
  \begin{tabular}{cc}
        \includegraphics[width=0.37\textwidth]{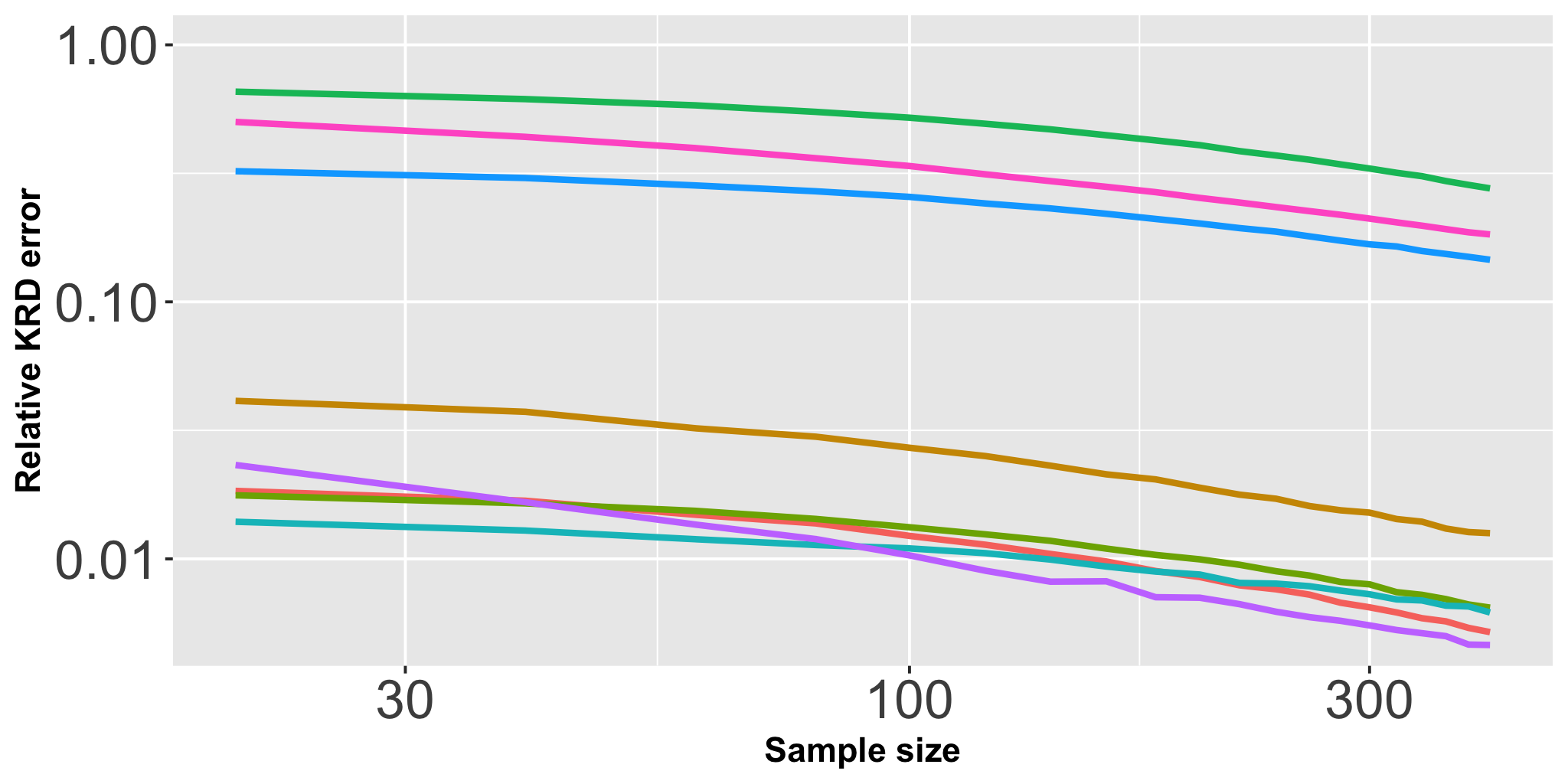} &       \includegraphics[width=0.37\textwidth]{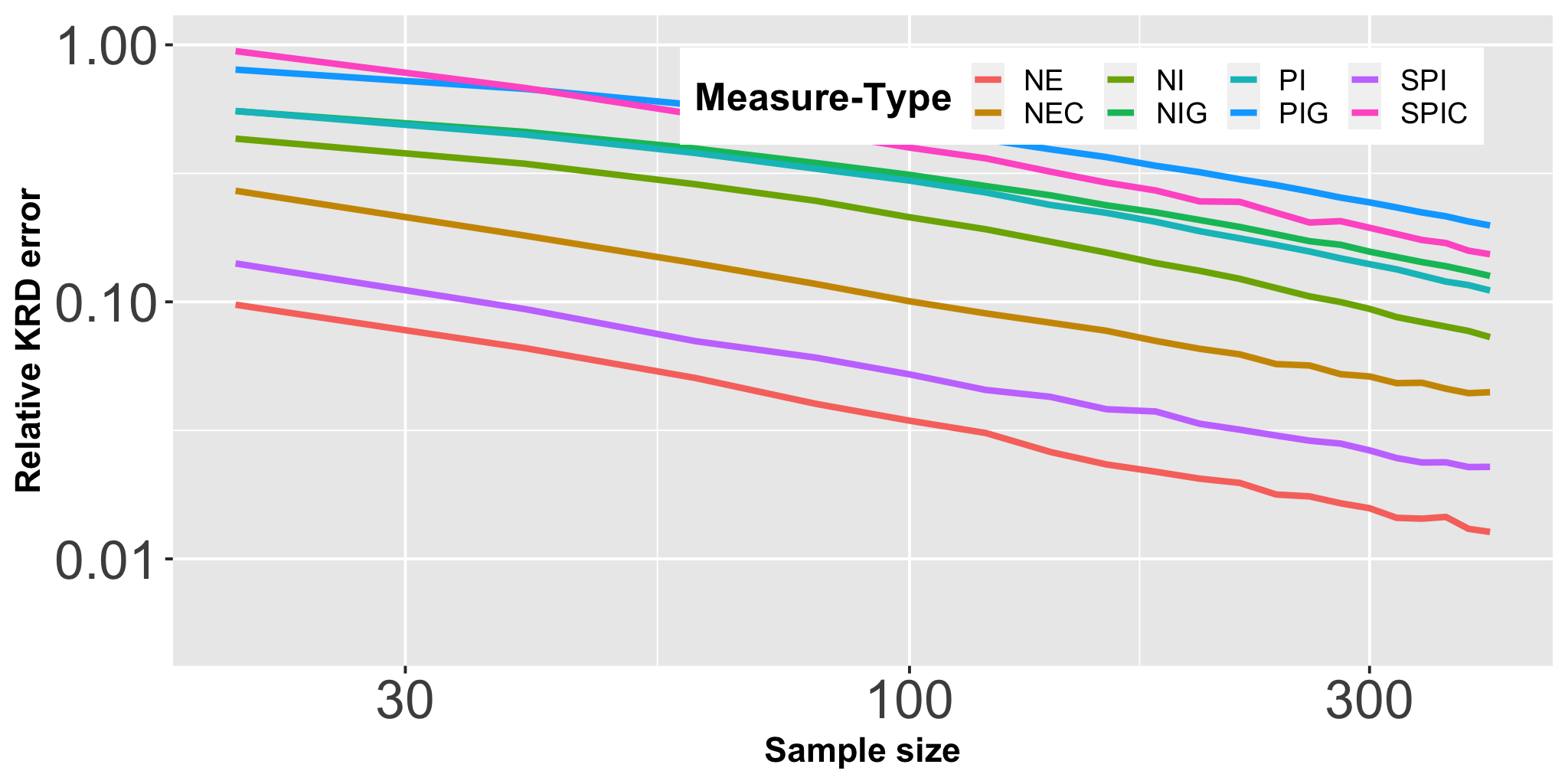} \\
      \includegraphics[width=0.37\textwidth]{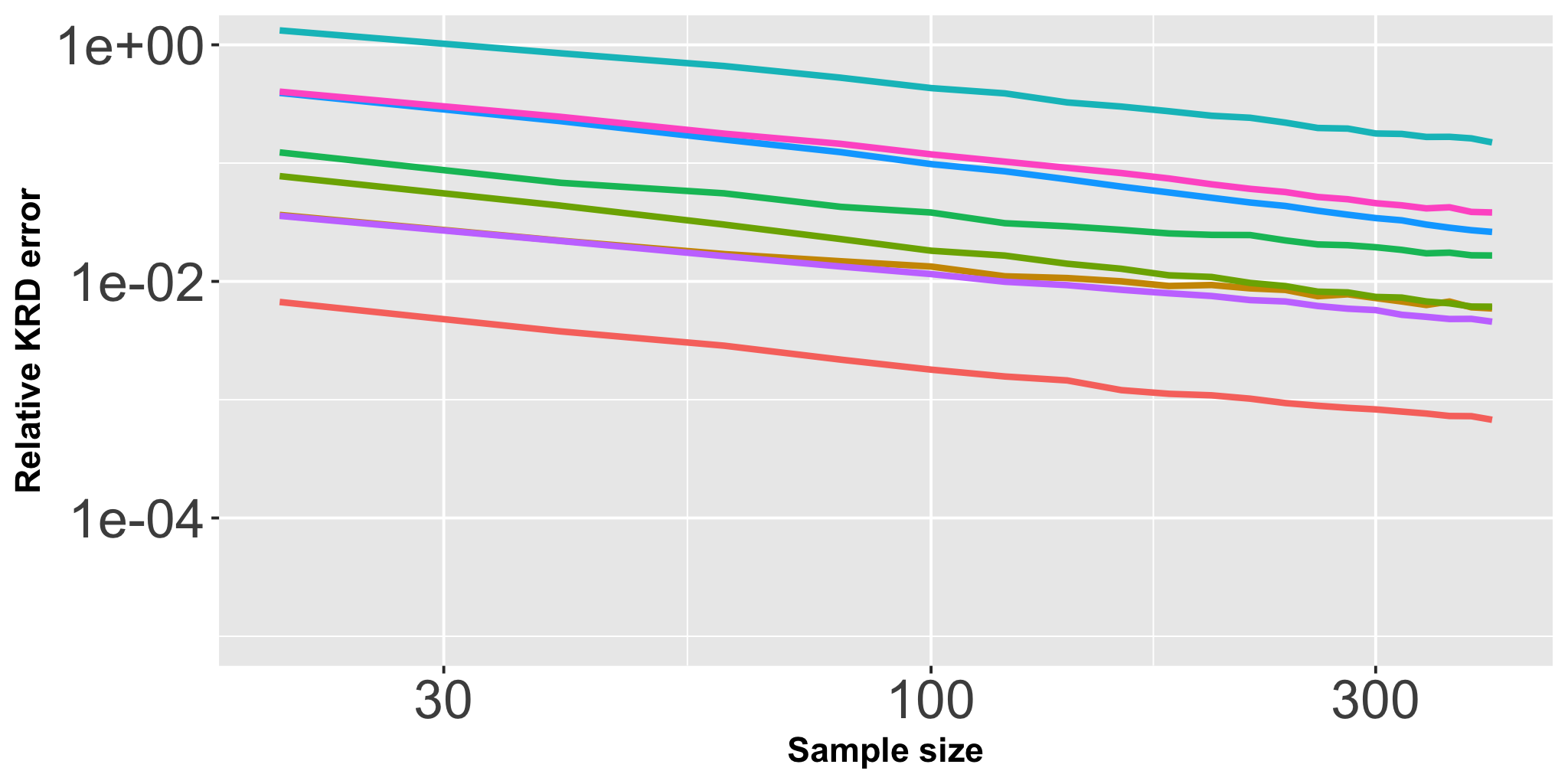} &       \includegraphics[width=0.37\textwidth]{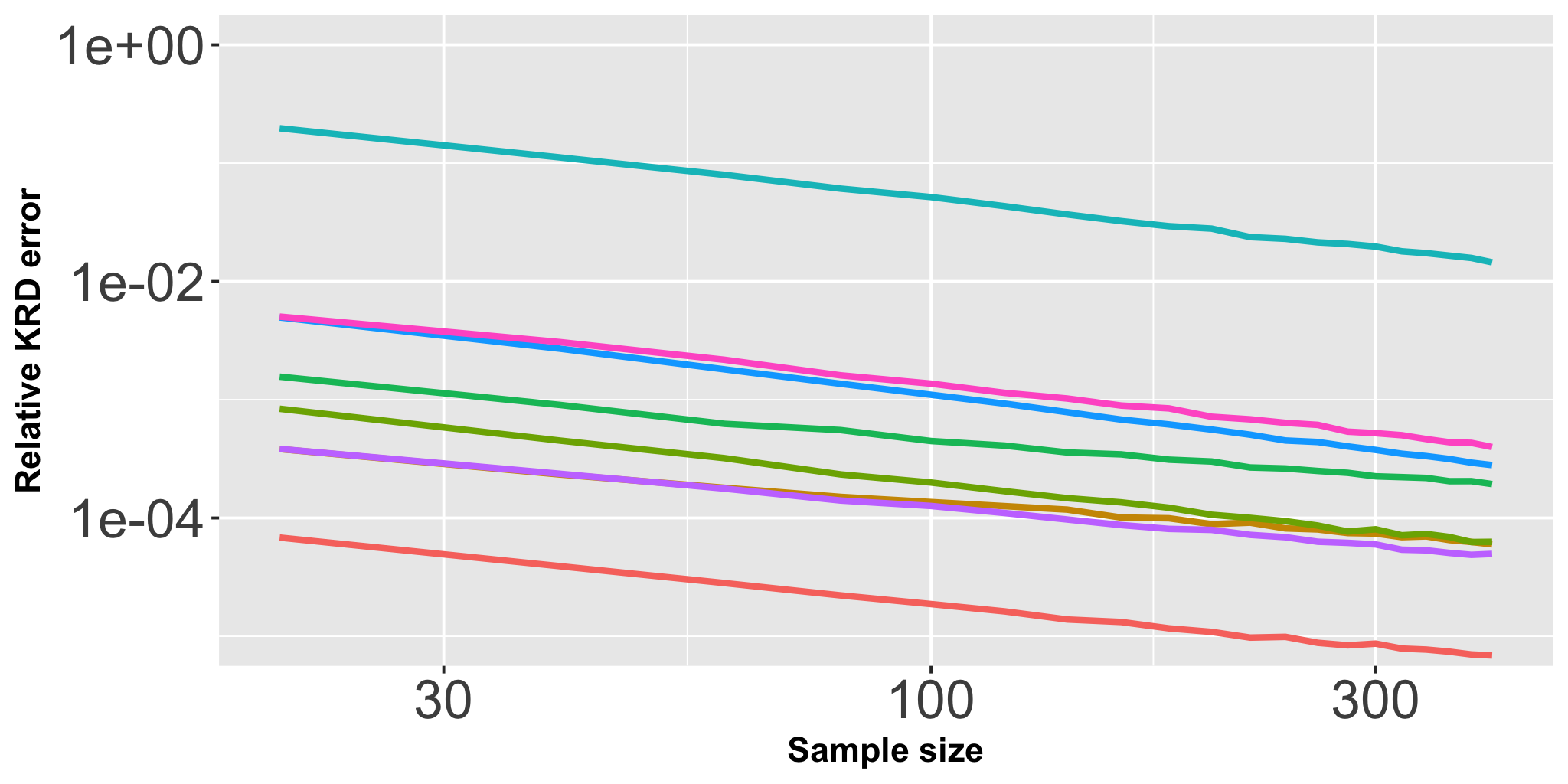} 
  \end{tabular}
  \caption{Log-log-plots of expected relative $(2,C)$-KRD error for two measures in the multinomial model for the classes in \Cref{sec:sims}. For each sampling size $N$ the expectation is estimated from $1000$ independent runs. For each class the parameters are set, such that the measures have on average $300$ support points. From top-left to bottom-right we have $C=0.01,0.1,1,10$, respectively.}
      \label{fig:mult_dist}
  \end{figure}
  \begin{figure}
      \centering
  \begin{tabular}{ccc}
        \includegraphics[width=0.31\textwidth]{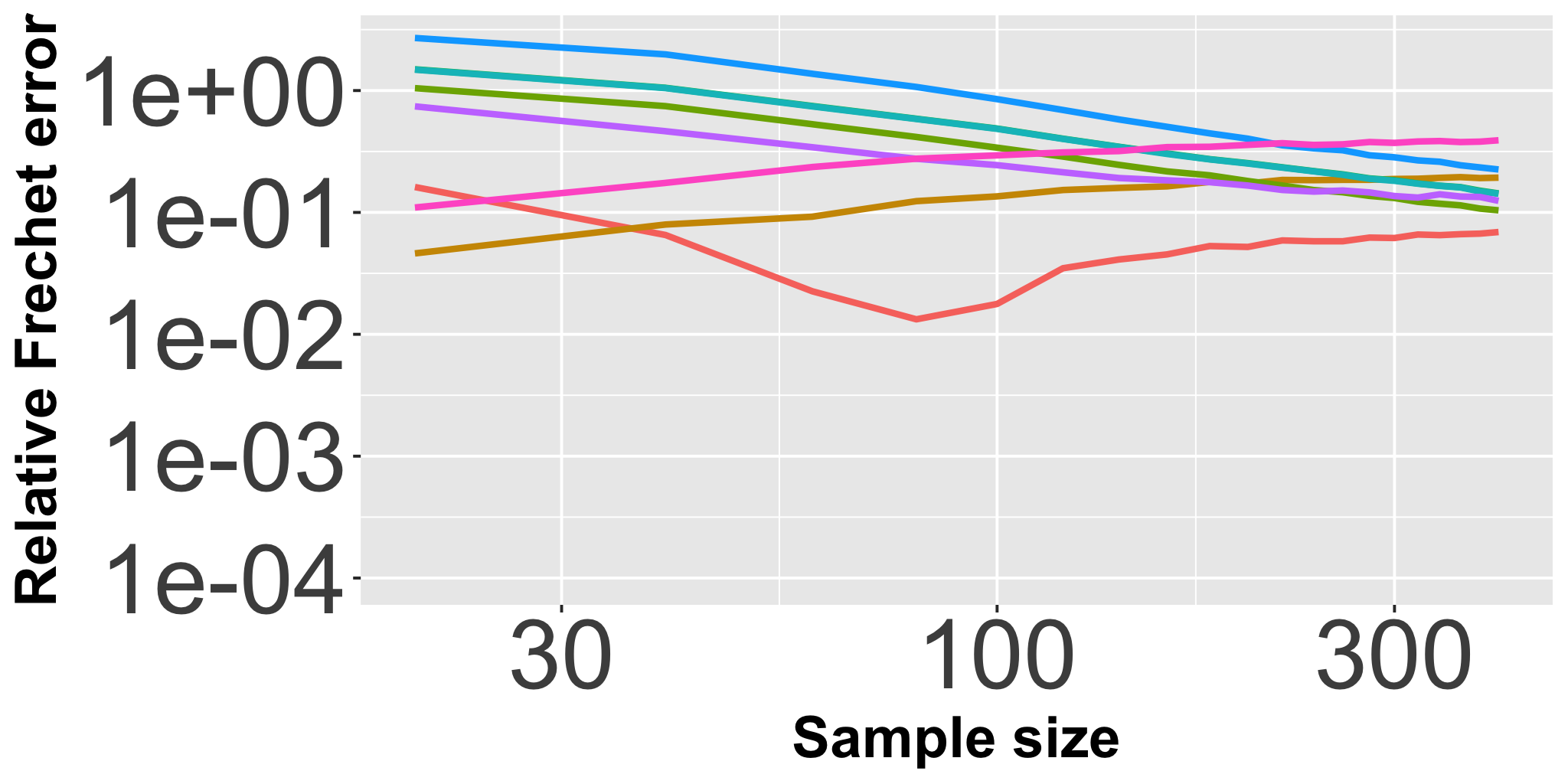} &
      \includegraphics[width=0.31\textwidth]{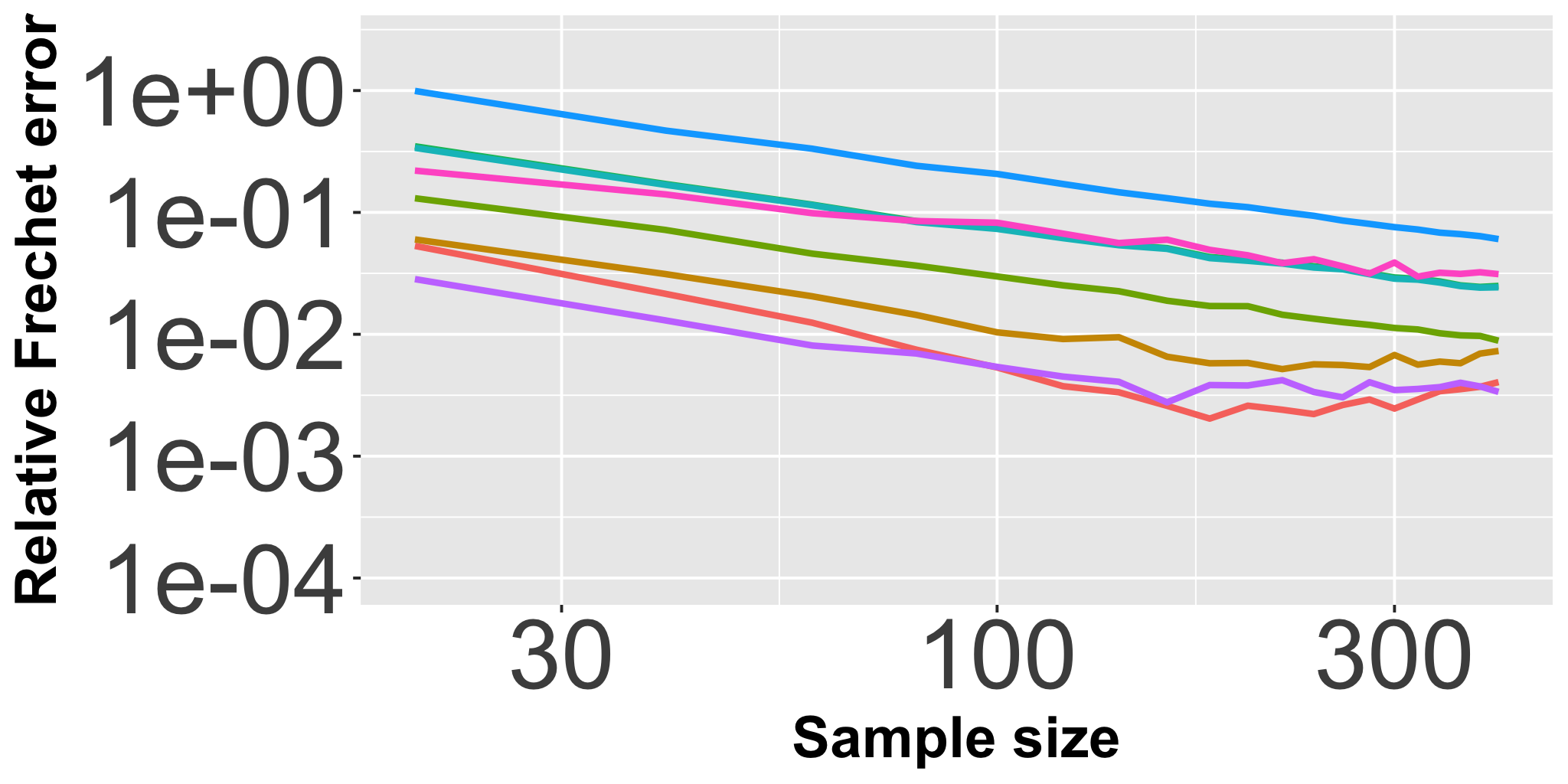} &       \includegraphics[width=0.31\textwidth]{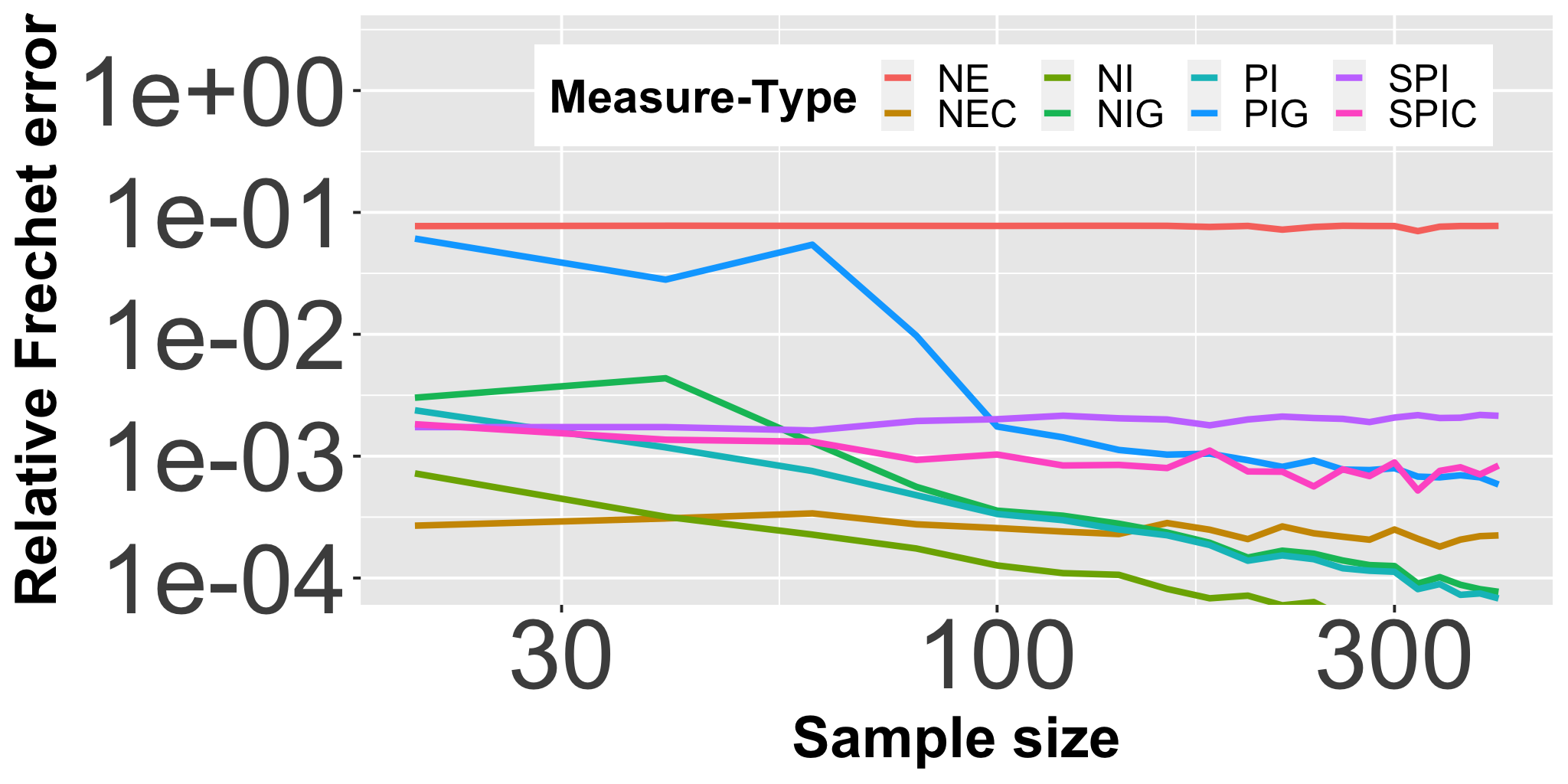} 
  \end{tabular}
  \caption{Log-log-plots of expected relative Fr\'echet error for the $(2,C)$-barycenter for $J=5$ measures from the PI class for the multinomial model with different sample sizes $N$. For each sample size the expectation is estimated from $100$ independent runs. For each class the parameters are set, such that the measures have on average $300$ support points. From left to right we have $C=0.1,1,10$, respectively.}
      \label{fig:mult_bary}
  \end{figure}
  
  \subsection{Simulations for Bernoulli Model}\label{app:sim_ber}
  To construct a reasonable framework for the simulations in the Bernoulli model, define for $s_0\in \mathbb{R}_+$, 
  \begin{align*}
      s_x=\frac{s_0}{\lVert x-(0.5,0.5)^T \rVert_2+s_0}.
  \end{align*}
  Intuitively, the success probability at a given point $x$ is larger, if $x$ is closer to the center of $[0,1]^2$ and smaller if it is further away from the center. Further, for $s_0\rightarrow \infty$ the success probability at each location converges to one. For the simulations, we now consider the error as a function of $s_0$. Note, that in this simulation study only the classes of measures with mass one at each support point are considered in accordance with the Bernoulli model in \eqref{eq:bermeasure}. One notable observation for the empirical $(p,C)$-KRD (in \Cref{fig:iber_dist}) is that the error of the SPIC class is significantly higher than for the NEC class, even though they share the same cluster locations. This can be explained by the fact that, by construction, the measures in the NEC class have a higher proportion of their mass in their central clusters, which is close to $(0.5,0.5)^T$ and thus has a high probability of being observed. This effect also carries over to the $(p,C)$-barycenter (in \Cref{fig:iber_bary}). In general, for the $(p,C)$-KRD the error in this model is increasing in $C$ (which is again explained by the estimation error for the true total mass intensity). However, the effect is less pronounced than in the Poisson model. For the clustered data types a small decrease of error for increasing $C$ over the cluster size can again be noted. Though, also this effect is less significant than in the other models. For the $(p,C)$-barycenter a decrease in error in $C$ can be observed which is consistent with the previous results for the Poisson model and again explained by the increased stability of the total mass intensity of the barycenter compared to the individual ones.
  \begin{figure}
      \centering
  \begin{tabular}{cc}
        \includegraphics[width=0.37\textwidth]{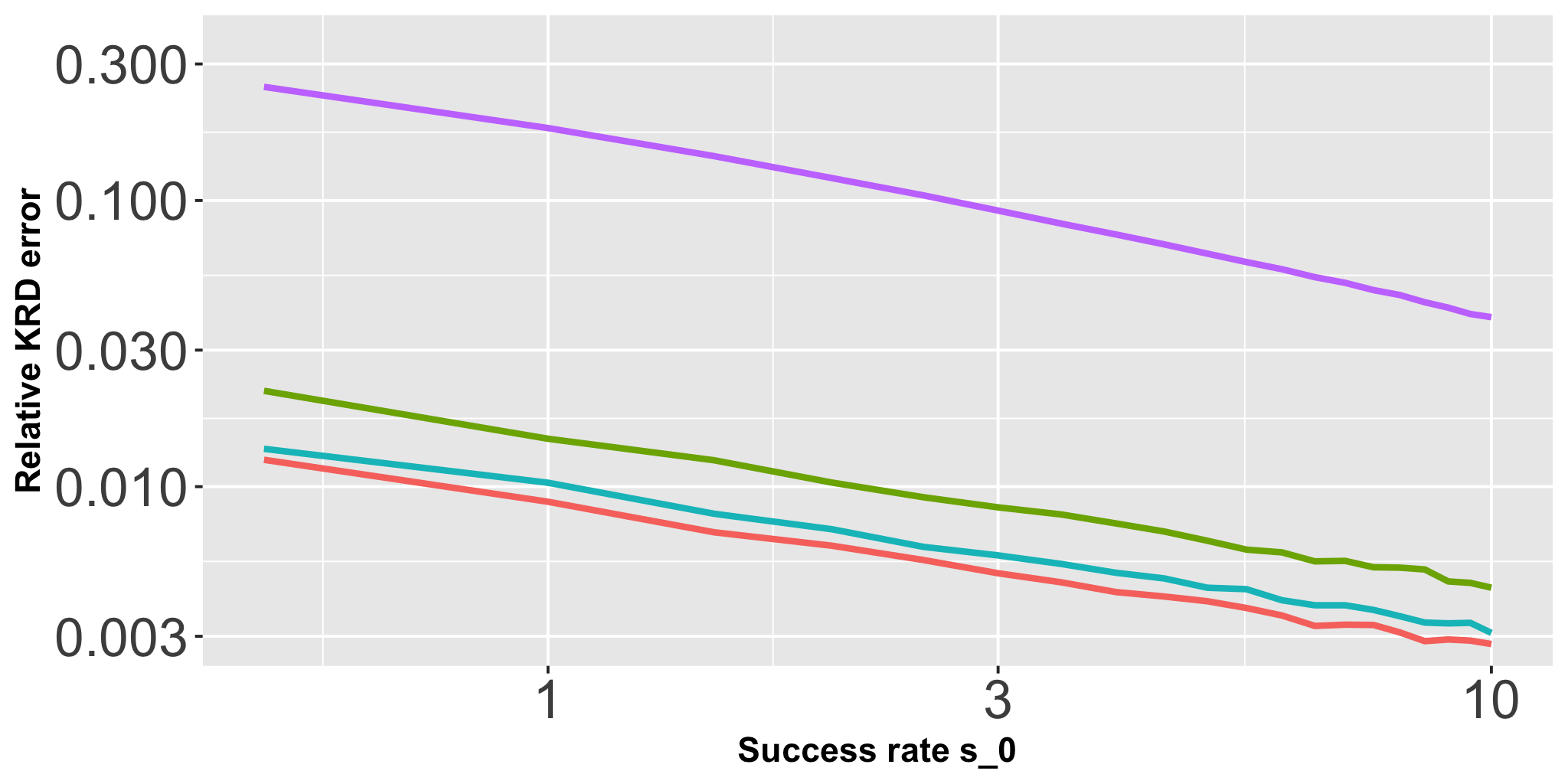} &       \includegraphics[width=0.37\textwidth]{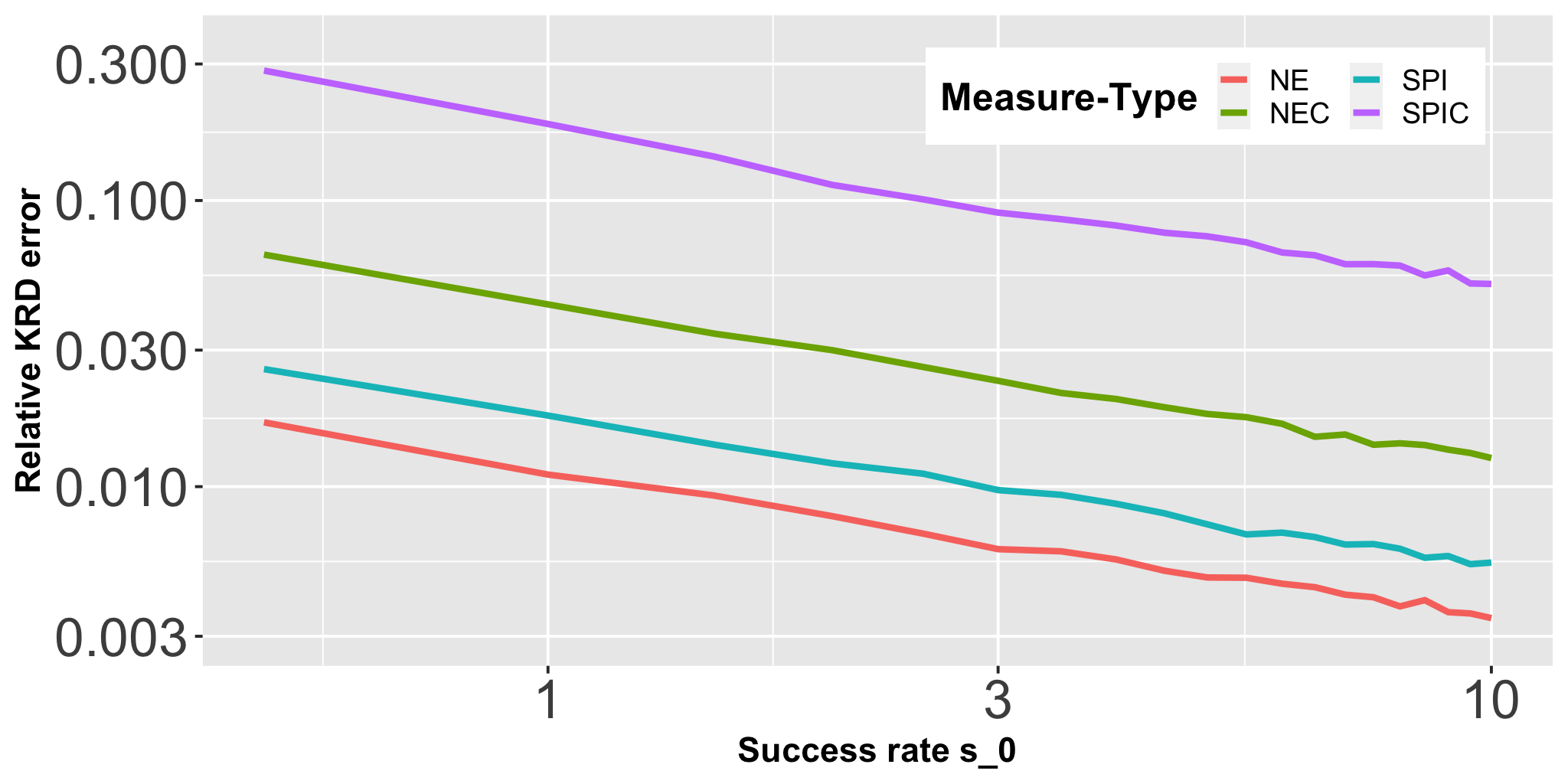} \\
      \includegraphics[width=0.37\textwidth]{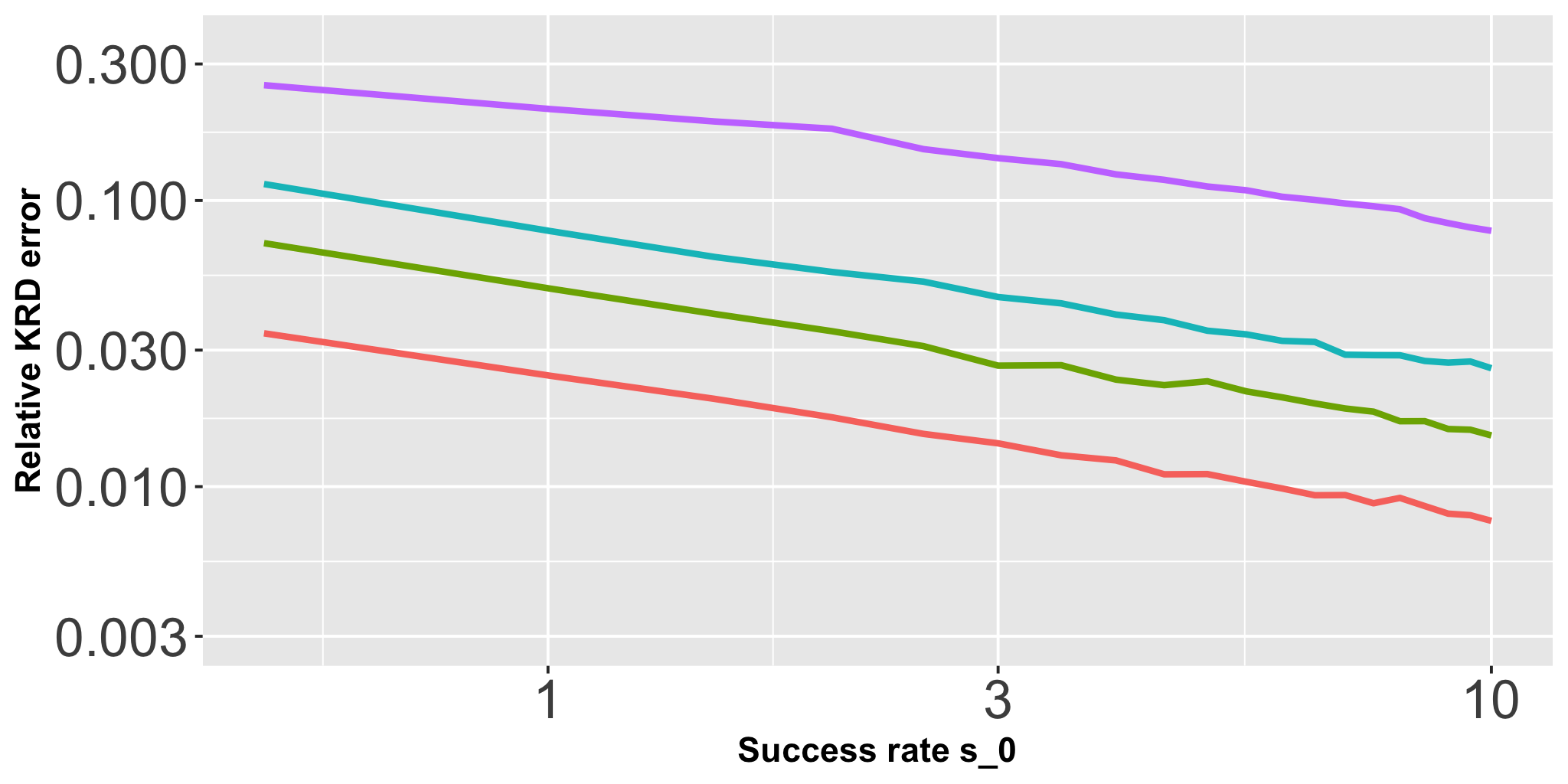} &       \includegraphics[width=0.37\textwidth]{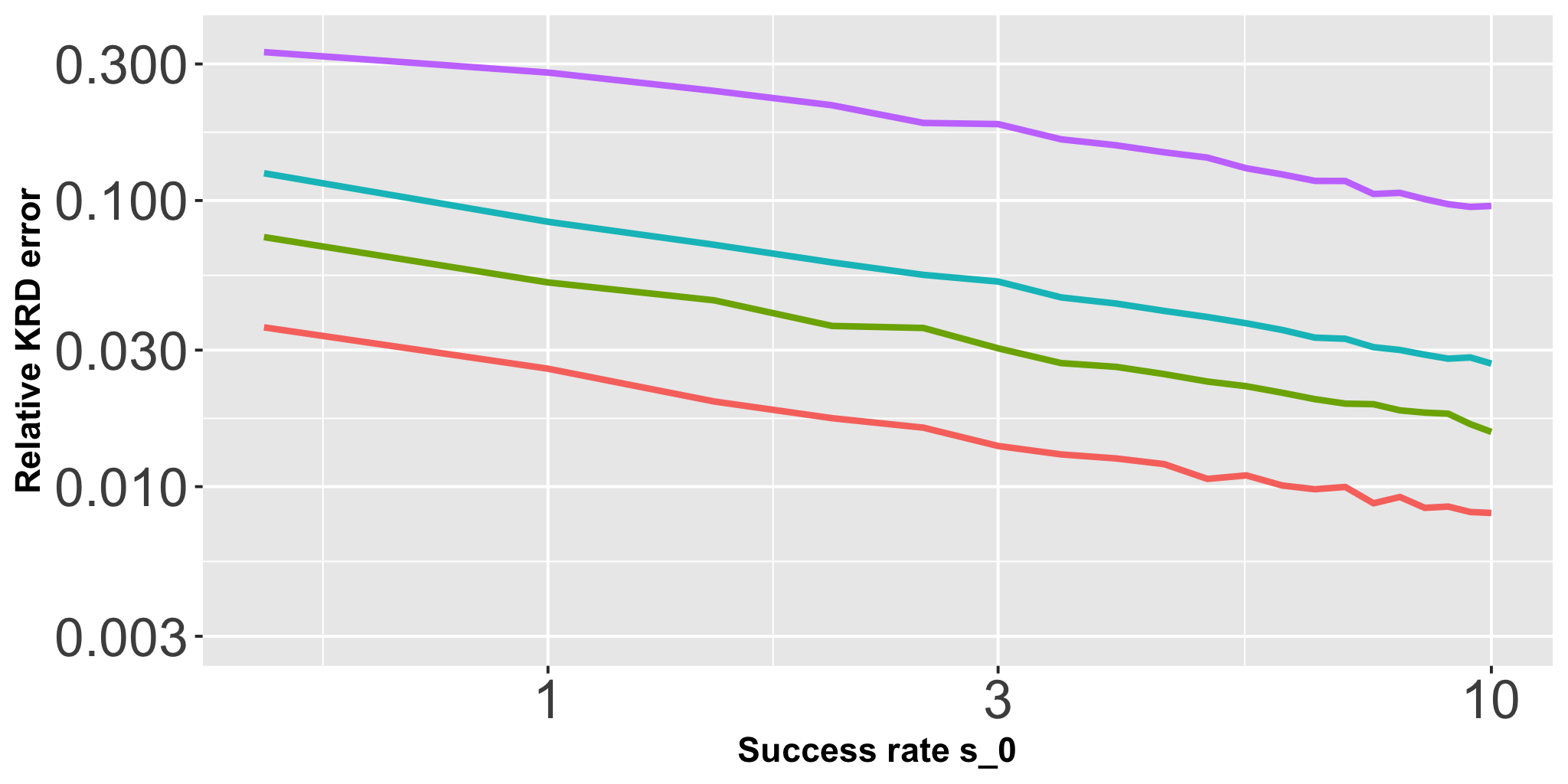} 
  \end{tabular}
  \caption{Log-log-plots of expected relative $(2,C)$-KRD error in terms of success probability $s_0$ for two measures in the Bernoulli Model for the NE, NEC, SPI, and SPIC classes from \Cref{sec:sims} and Appendix \ref{app:additionalSim}. For each success probability $s_0$ the expectation is estimated from $1000$ independent runs. The parameters are chosen such that the population measures have on average $300$ support points. From top-left to bottom-right we have $C=0.01,0.1,1,10$, respectively.}
      \label{fig:iber_dist}
  \end{figure}
  \begin{figure}
      \centering
  \begin{tabular}{ccc}
        \includegraphics[width=0.31\textwidth]{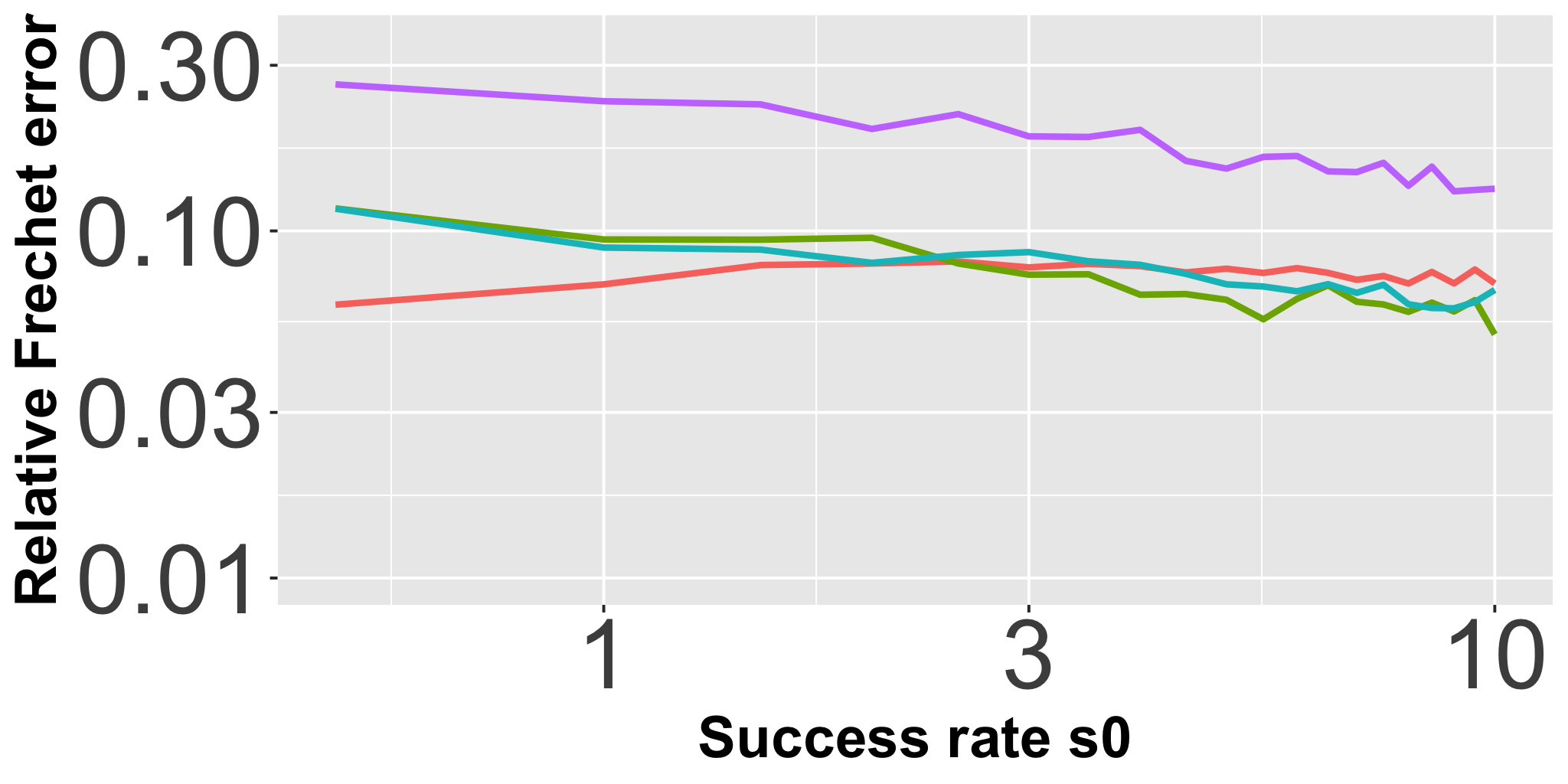} &
      \includegraphics[width=0.31\textwidth]{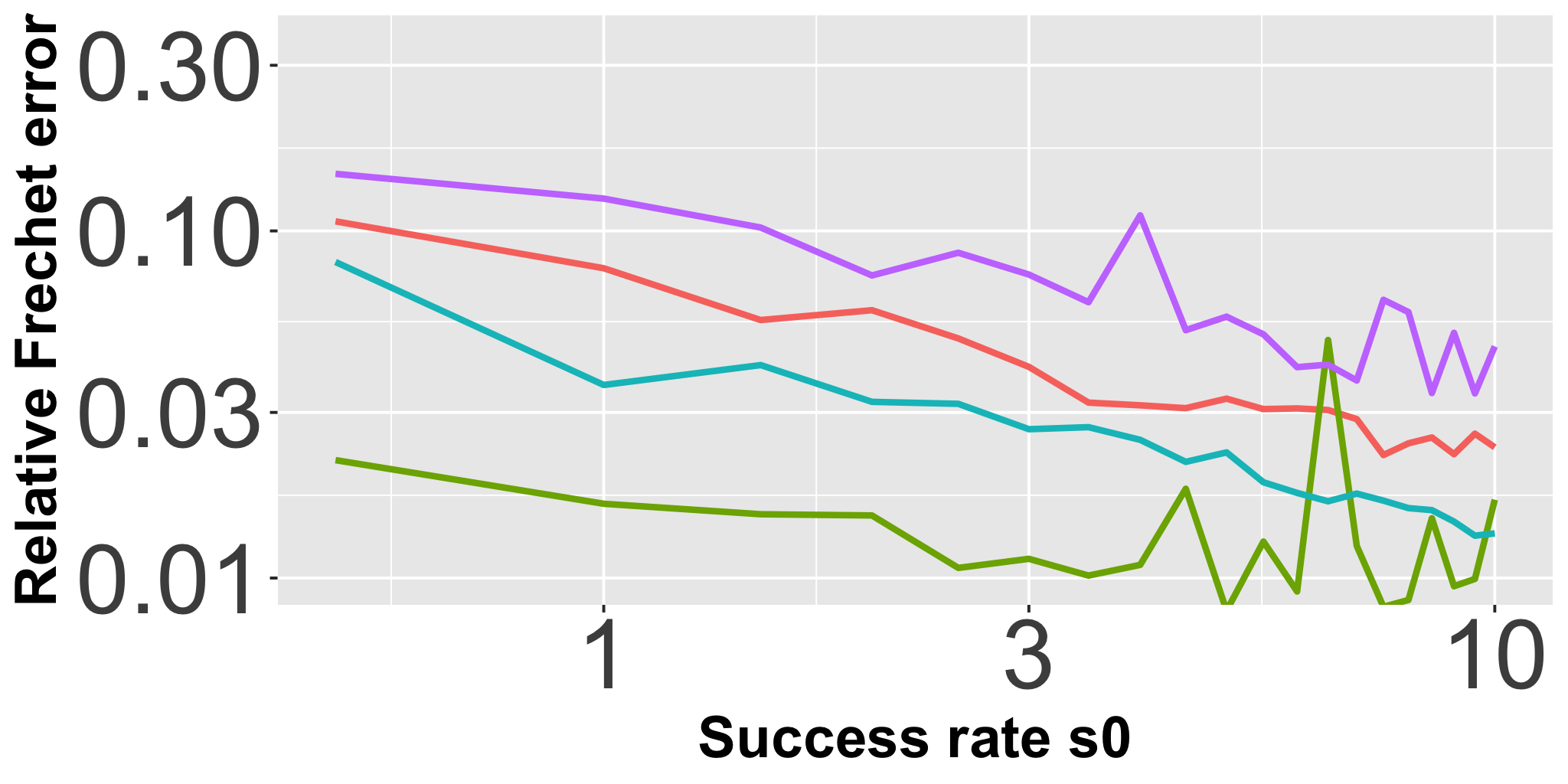} &       \includegraphics[width=0.31\textwidth]{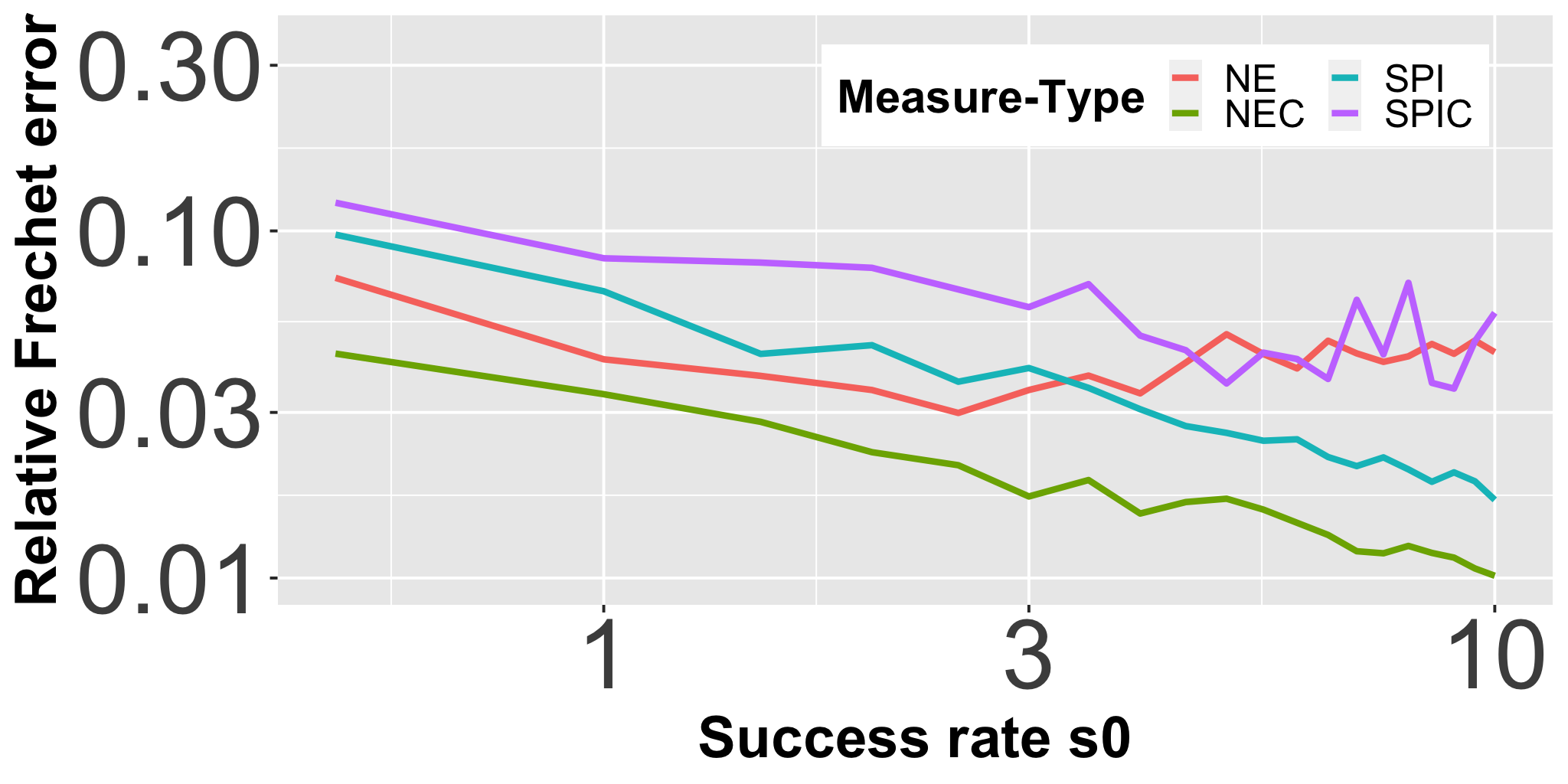} 
  \end{tabular}
  \caption{Expected relative Fr\'echet error  for the $(2,C)$-barycenter in terms of success probability $s_0$ for $J=5$ measures in the Bernoulli Model for the NE, NEC, SPI, and SPIC classes from \Cref{sec:sims} and Appendix \ref{app:additionalSim}. For  each success probability $s_0$ the expectation is estimated from $100$ independent runs. The parameters are set such that the population measures in all classes have on average $300$ support points. From left to right we have $C=0.1,1,10$, respectively.}
      \label{fig:iber_bary}
  \end{figure}

  \end{appendix}

%   \bibliographystyle{abbrvnat}
% \bibliography{kr_stats}{}

\end{document}